\documentclass[11pt]{article}
\usepackage{amsfonts}
\usepackage{amsmath,amsthm}
\usepackage{graphicx,epstopdf,epsfig,multirow,epic,bm}
\usepackage{color}
\usepackage{multicol}
\usepackage{algorithm}
\usepackage{algpseudocode}
\usepackage{paralist}

\usepackage{makeidx}
\usepackage{showidx}
\usepackage{textcomp}

\theoremstyle{definition}
\newtheorem{thm}{Theorem}
\newtheorem{cor}[thm]{Corollary}
\newtheorem{lem}[thm]{Lemma}

\theoremstyle{definition}
\newtheorem{defn}{Definition}
\theoremstyle{definition}
\newtheorem{rem}{Remark}

\theoremstyle{definition}
\newtheorem{problem}{Problem}

\theoremstyle{definition}
\newtheorem{conj}{Conjecture}
\theoremstyle{definition}
\newtheorem{example}{Example}
\theoremstyle{definition}

\DeclareGraphicsExtensions{.eps,.eps.gz}
\normalsize \rm
\usepackage{titlesec}

\newcommand{\qqed}{\hfill $\square$}
\newcommand{\paralled}{\hfill $\parallel$}

\usepackage[nottoc]{tocbibind} 

\begin{document}

\begin{center}
{\huge \textbf{Parameterized Colorings And Labellings\\[12pt]
Of Graphs In Topological Coding}}\\[14pt]
{\large \textbf{Bing YAO, Xiaohui ZHANG, Hui SUN, Jing SU, Fei Ma, Hongyu WANG}\\[8pt]
(\today)}
\end{center}

\vskip 1cm

\pagenumbering{roman}
\tableofcontents

\newpage

\setcounter{page}{1}
\pagenumbering{arabic}

\begin{center}
{\huge \textbf{Parameterized Colorings And Labellings\\[12pt]
Of Graphs In Topological Coding}}\\[14pt]
{\large Bing \textsc{Yao}\footnote{~College of Mathematics and Statistics, Northwest Normal University, Lanzhou, 730070, CHINA, email: yybb918@163.com}, Xiaohui \textsc{Zhang}\footnote{~School of Computer, Qinghai Normal University, Xining, 810001, CHINA, email: 2547851790@qq.com}, Hui \textsc{Sun}\footnote{~School of Electronics Engineering and Computer Science, Peking University, Beijing, 100871, CHINA, email: 18919104606@163.com}, Jing \textsc{Su}\footnote{~School of Electronics Engineering and Computer Science, Peking University, Beijing, 100871, CHINA, email: 1099270659@qq.com}, Fei \textsc{Ma}\footnote{~School of Electronics Engineering and Computer Science, Peking University, Beijing, 100871, CHINA, email: mafei123987@163.com}, Hongyu \textsc{Wang}\footnote{~National Computer Network Emergency Response Technical Team/Coordination Center of China, Beijing, 100029, CHINA, email: why200904@163.com}}\\[12pt]
(\today)
\end{center}

\vskip 1cm

\begin{quote}
\textbf{Abstract:} The coming quantum computation is forcing us to reexamine the cryptosystems people use. We are applying graph colorings of topological coding to modern information security and future cryptography against supercomputer and quantum computer attacks in the near future. Many of techniques introduced here are associated with many mathematical conjecture and NP-problems. We will introduce a group of $W$-constraint $(k,d)$-total colorings and algorithms for realizing these colorings in some kinds of graphs, which are used to make quickly \emph{public-keys} and \emph{private-keys} with anti-quantum computing, these $(k,d)$-total colorings are: graceful $(k,d)$-total colorings, harmonious $(k,d)$-total colorings, $(k,d)$-edge-magic total colorings, $(k,d)$-graceful-difference total colorings and $(k,d)$-felicitous-difference total colorings. One of useful tools we used is called \emph{Topcode-matrix} with elements can be all sorts of things, for example, sets, graphs, number-based strings. Most of parameterized graphic colorings/labelings are defined by \emph{Topcode-matrix algebra} here. From the application point of view, many of our coloring techniques are given by algorithms and easily converted into programs. Some operations based on Topcode-matrices help us to define new parameterized colorings, to generate number-based strings containing parameters, and to build up graphic lattices including tree-graphic lattices and tree-graphic lattice homomorphisms. More complex homomorphism phenomenon is from a colored graph homomorphism to a Topcode-matrix homomorphism, and then to a graph-set homomorphism. We investigate parameterized set-colorings, hypergraphs based on parameterized set-colorings, and set-colorings with multiple intersections, as well as graph homomorphisms with parameterized set-colorings. We point that the techniques of topological coding are related with the subgraph isomorphic problem, the problems of integer partition and integer factorization, the parameterized number-based string partition problem, the topological structure decomposition problem, and the parametric reconstitution of number-based string.\\[6pt]
\end{quote}

\section{Introduction and preliminary}

\subsection{Requirements of information security}

The coming quantum computation is forcing us to reexamine the cryptosystems we use.

The authors in \cite{Bernstein-Buchmann-Dahmen-Quantum-2009} pointed: Most research in quantum algorithms has revolved around the hidden subgroup problem. The hidden subgroup problem is a problem defined on a group, and many problems reduce to it. Factoring and discrete log reduce to the hidden subgroup problem when the underlying group is finite or countable. Pell's equation reduces to the hidden subgroup problem when the group is uncountable. For these cases there are efficient quantum algorithms to solve the hidden subgroup problem, and hence the underlying problem, because the group is \emph{abelian}. Graph isomorphism reduces to the hidden subgroup problem for the symmetric group, and the unique shortest lattice vector problem is related to the hidden subgroup problem when the group is dihedral.

``Quantum cryptography'' also called ``quantum key distribution'' expands a short shared key into an effectively infinite shared stream. The prerequisite for quantum cryptography is that the users, say Alice and Bob, both know (e.g.) 256 unpredictable secret key bits. The result of quantum cryptography is that Alice and Bob both know a stream of (e.g.) 1012 unpredictable secret bits that can be used to encrypt messages. The length of the output stream increases linearly with the amount of time that Alice and Bob spend on quantum cryptography.

Peikert, in \cite{Chris-Peikert-decade}, pointed that lattice-based ciphers have the following advantages: Conjectured security against quantum attacks; Algorithmic simplicity, efficiency, and parallelism; Strong security guarantees from worst-case hardness.

Lattice-based systems provide a good alternative since they are based on a long-standing open problem for classical computation \cite{Bernstein-Buchmann-Dahmen-Quantum-2009}.

\subsection{Techniques of topological coding}

Topological Coding consists of two different kinds of mathematical branches: Topological coding involves millions of things in the world, and people use topological structures to connects things together to form complete ``stories'' under certain constraints in asymmetric encryption system \cite{Bing-Yao-2020arXiv}. The authors, in \cite{Tian-Li-Peng-Yang-2021-102212}, use the topological graph to generate the honeywords, which is the first application of graphic labeling of topological coding in the honeywords generation. They propose a method to protect the \emph{hashed passwords} by using \emph{topological graphic sequences}.

Topsnut-gpws, the abbreviation of the sentence ``graphic passwords made by topological structures and mathematical constraints'', belong to topological coding in \cite{Wang-Xu-Yao-2016} and \cite{Wang-Xu-Yao-Key-models-Lock-models-2016}. However, Topsnut-gpws differ from the existing graphical passwords, since they are saved and implemented in computer by popular matrices. There are many NP-hard problems and conjectures in Topsnut-gpws, for example, the \textbf{Subgraph Isomorphic Problem}, a NP-hard problem. As known, a Topcode-matrix can bring a group of ``things'' together topologically by mathematical constraints, in general.

Topsnut-gpws correspond to matrices of order $3\times q$, called \emph{Topcode-matrices}, which can generate quickly number-based strings with longer bytes for ciphering digital files. And moreover, Topsnut-gpws are easy to induce asymmetric public-keys and private-keys, such as, one public-key vs two or more private-keys, more public-keys vs more private-keys. However, it is irreversible to find the original Topsnut-gpw from a number-based string made by the Topcode-matrix of this Topsnut-gpw, since a Topsnut-gpw is made by a topological structure and a group of mathematical constraints.

After years of research, we summarize the following advantages of Topsnut-gpws \cite{Yao-Wang-2106-15254v1}:
\begin{asparaenum}[(1) ]
\item \textbf{Run} fast in communication networks because they are saved in computer by popular matrices rather than pictures.

\item \textbf{Produce} easily text-based strings, or number-based strings for encrypting files in real applications.

\item \textbf{Diversity} of asymmetric ciphers, such as one \emph{public-key} corresponds to more \emph{private-keys}, or more \emph{public-keys} correspond to more \emph{private-keys}.
\item \textbf{Supported by many branches of mathematics.} Number theory (especially computational number theory), information theory, probability theory, random process, combinatorial mathematics, linear algebra, matrix algebra, discrete mathematics and graph theory are the solid foundation and strong technical support to Topsnut-gpws of topology coding.

\item \textbf{Irreversibility}, Topsnut-gpws can generate quickly text-based strings, or number-based strings with bytes as long as desired, but these strings can not reconstruct the original Topsnut-gpws, since this work must span two fundamentally different branches of mathematics: topological structures and non-topological mathematical constraints.

\item \textbf{Computational security} There are enormous numbers of graph colorings and labelings in graph theory, and new graph colorings/labelings come into being everyday. As known, there are no polynomial algorithms for creating the topological structures (see Eq.(\ref{eqa:number-graphs-23-24-vertices})) of Topsnut-gpws.

\item \textbf{Provable security}, since Topsnut-gpws are related with many long-standing mathematical conjectures, NP-complete and NP-hard problems from number theory, combinatorics, discrete mathematics and graph theory.
\end{asparaenum}

\vskip 0.2cm

In other word, Topsnut-gpws have part of the advantages of lattice-based cryptosystems.

\subsubsection{Parameterized number-based strings}

A \emph{number-based string} $S=c_1c_2\cdots c_m$ is defined by $c_i\in [0,9]=\{0,1,2,\dots ,9\}$. We can get the following number-based strings
\begin{equation}\label{eqa:3-topo-number-based-strings}
{
\begin{split}
D_1(1,2)&=001919170415171513461117159647151356639917\\
D_2(1,2)&=617936459137669151511441317151500171719190\\
D_3(1,2)&=170019191704151715134611159647151356639917
\end{split}}
\end{equation}
from a Topcode-matrix $T_{code}(1,2)$ of order $3\times 10$ shown in Eq.(\ref{eqa:example11}). It is not difficult to see that there are $30!$ different number-based strings made by $T_{code}(1,2)$, like three number-based strings shown in Eq.(\ref{eqa:3-topo-number-based-strings}), in total. The Topcode-matrix $T_{code}(1,2)$ is a particular case of a Topcode-matrix $T_{code}(k,d)$ with parameterized elements shown in Eq.(\ref{eqa:example22}), as we take $(k,d)=(1,2)$. Clearly, we have infinite different number-based strings from the Topcode-matrix $T_{code}(k,d)$.

\begin{equation}\label{eqa:example11}
\centering
{
\begin{split}
T_{code}(1,2)= \left(
\begin{array}{ccccccccccccccc}
6 & 6 & 4 & 6 & 6 & 4 & 4 & 0 & 0 & 0\\
1 & 3 & 5 & 7 & 9 & 11 & 13 & 15 & 17 & 19\\
7 & 9 & 9 & 13 & 15 & 15 & 17 & 15 & 17 & 19
\end{array}
\right)_{3\times 10}
\end{split}}
\end{equation}

{\footnotesize
\begin{equation}\label{eqa:example22}
\centering
T_{code}(k,d)= \left(
\begin{array}{ccccccccccccccc}
3d & 3d & 2d & 3d & 3d & 2d & 2d & 0 & 0 & 0\\
k & k+d & k+2d & k+3d & k+4d & k+5d & k+6d & k+7d & k+8d & k+9d\\
k+3d & k+4d & k+4d & k+6d & k+7d & k+7d & k+8d & k+7d & k+8d & k+9d
\end{array}
\right)
\end{equation}
}

Each Topsnut-gpw shown in Fig.\ref{fig:6-Topsnut-gpws} consists of a topological structure (see Fig.\ref{fig:6-Topsnut-topological}) and a parameterized total coloring. Moreover, each Topsnut-gpw $A_i$ shown in Fig.\ref{fig:6-Topsnut-gpws} admits a \emph{colored graph homomorphism} to $A$, that is $A_i\rightarrow A$ with $i\in [1,5]=\{1,2,3,4,5\}$. Thereby, each colored graph homomorphism $A_i\rightarrow A$ forms a \emph{\textbf{topological authentication}}.

\begin{figure}[h]
\centering
\includegraphics[width=16.4cm]{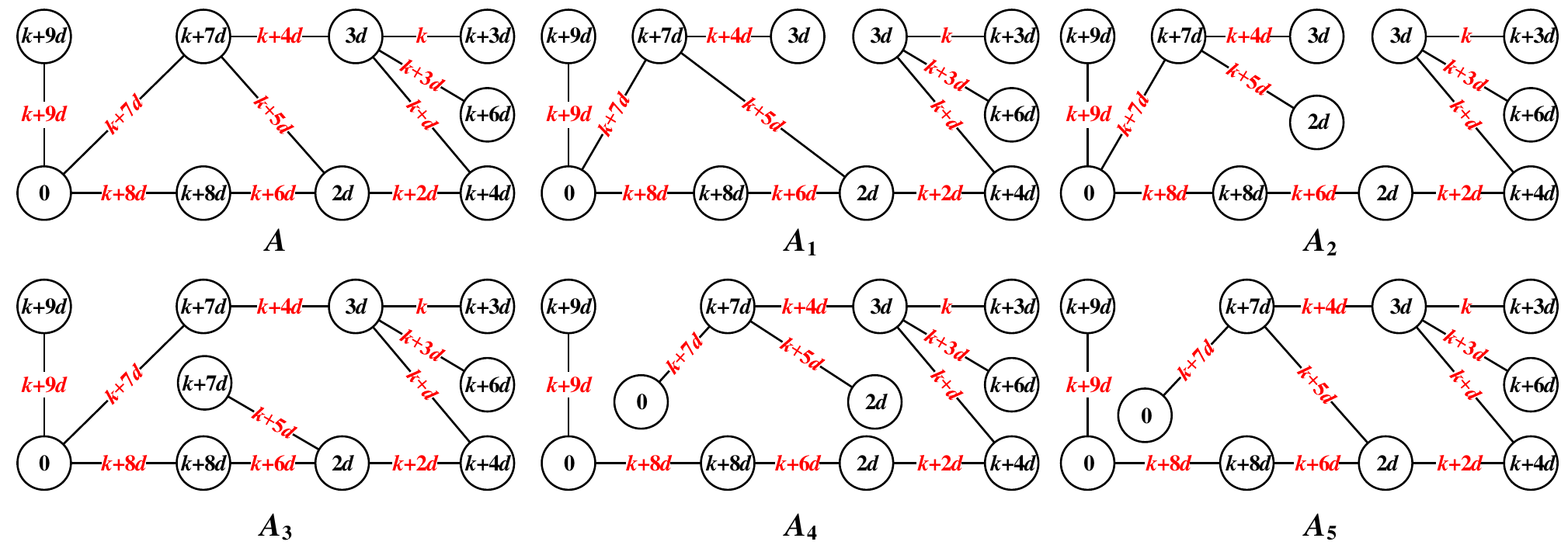}
\caption{\label{fig:6-Topsnut-gpws}{\small Six Topsnut-gpws with parameterized total colorings.}}
\end{figure}

\begin{figure}[h]
\centering
\includegraphics[width=16.4cm]{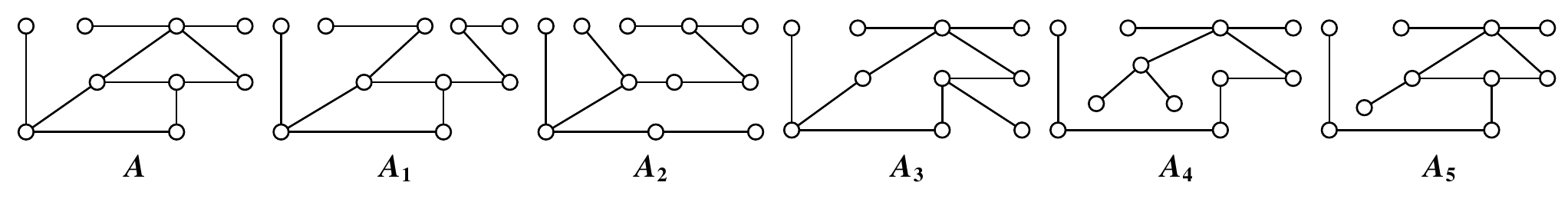}
\caption{\label{fig:6-Topsnut-topological}{\small Six uncolored graphs correspond to six Topsnut-gpws with parameterized total colorings shown in Fig.\ref{fig:6-Topsnut-gpws}.}}
\end{figure}

\begin{problem}\label{qeu:PNBSP-problem}
Motivated from the Number-based String Partition Problem proposed in \cite{Bing-Yao-Hongyu-Wang-arXiv-2020-homomorphisms}, we present the following \textbf{PNBSP-problem} (Parameterized Number-based String Partition Problem):
\begin{quote}
\textbf{PNBSP-problem}. For a given number-based string $S=c_1c_2\cdots c_m$ with $c_i\in [0,9]$, partition it into $3q$ segments $c_1c_2\cdots c_m=a_1a_2\cdots a_{3q}$ with $a_j=c_{n_j}c_{n_j+1}\cdots c_{n_{j+1}}$ with $j\in [1,3q-1]$, where $n_1=1$ and $n_{3q}=m$, and use these segments $a_1,a_2,\dots ,a_{3q}$ to reform a Topcode-matrix $T_{code}$, like the Topcode-matrix $T_{code}(1,2)$ shown in Eq.(\ref{eqa:example11}), and moreover use the Topcode-matrix $T_{code}$ to reconstruct all of Topsnut-gpws of $q$ edges, like six Topsnut-gpws corresponding to $T_{code}(k,d)$ shown in Fig.\ref{fig:6-Topsnut-gpws}. By the found Topsnut-gpws corresponding to the common Topcode-matrix $T_{code}(k,d)$, find the desired \emph{public Topsnut-gpw} $A$ and the desired \emph{private Topsnut-gpws} $A_i$, such that each mapping $\varphi_i:V(A)\rightarrow V(A_i)$ forms a \emph{colored graph homomorphism} (topological authentication) $A_i\rightarrow A$ with $i\in [1,5]$.
\end{quote}
\end{problem}

\subsubsection{The problems of integer partition and integer factorization}

The PNBSP-problem enables us to think the following problems related with the \emph{integer partition} and the \emph{integer factorization}:
\begin{asparaenum}[\textrm{\textbf{Prob}}-1. ]
\item \label{hard:coloring-labeling} Since a parameterized Topcode-matrix $T_{code}(k,d)$ with parameters $k$ and $d$ is related with one of graph colorings and graph labelings, as known there are hundreds of graph colorings and graph labelings have been introduced, so determining this Topcode-matrix $T_{code}(k,d)$ is not a slight job.

\item \label{hard:graph-isomorphic} Let $G_p$ be the number of graphs of $p$ vertices, Harary and Palmer \cite{Harary-Palmer-1973} computed
\begin{equation}\label{eqa:number-graphs-23-24-vertices}
{
\begin{split}
G_{23}&=559946939699792080597976380819462179812276348458981632\approx 2^{179}\\
G_{24}&=195704906302078447922174862416726256004122075267063365754368\approx 2^{197}
\end{split}}
\end{equation} If a parameterized Topcode-matrix $T_{code}(k,d)$ corresponding to a number-based string $D(k_0,d_0)$ has been determined, finding all non-isomorphic Topsnut-gpws (colored graphs) corresponding to $T_{code}(k,d)$ will meet the \textbf{Subgraph Isomorphic Problem}, as known, a NP-hard, see examples shown in Fig.\ref{fig:6-Topsnut-topological} and two numbers $G_{23}$, $G_{24}$ shown in Eq.(\ref{eqa:number-graphs-23-24-vertices}).

\item \label{hard:primes-decomposition} As two parameters $k$ and $d$ both are prime numbers, determining a parameterized Topcode-matrix $T_{code}(k,d)$ will meet the \textbf{Integer Factorization Problem} with an odd number factorization $\alpha =k\cdot d$ if $\alpha$ is as a \emph{public-key}, or the \textbf{Integer Partition Problem}, such as an even number partition $\beta =k+d$ if $\beta$ is as a \emph{public-key}.
\item No polynomial algorithm is for finding the colorings admitted by the \emph{public-key} $A$ and the \emph{private-key} $A_i$ with $i\in [1,5]$. And no algorithm find all possible colorings admitted by a graph, since it is a NP-type problem.
\item Since a number-based string $s=c_1c_2\cdots c_m$ can be induced by different Topcode-matrices $T^1_{code}(k,d)$, $T^2_{code}(k,d)$, $\dots$, $T^n_{code}(k,d)$, so no way is to realize the following processes:
 $$D_j(1,2)\rightarrow T_{code}(1,2)\rightarrow A;~D_j(1,2)\rightarrow T_{code}(1,2)\rightarrow A_j,~j\in [1,5]
 $$ from the number-based strings to the public Topsnut-gpw $A$ and the desired private Topsnut-gpws $A_i$, see Eq.(\ref{eqa:3-topo-number-based-strings}), Eq.(\ref{eqa:example11}) and Fig.\ref{fig:6-Topsnut-gpws}.

\item As known, a number-based string $s=c_1c_2\cdots c_m$ can be generated by degree sequences, vectors, matrices, graph colorings \emph{etc.}, so it is difficult to judge which way produces the given number-based string $s=c_1c_2\cdots c_m$. For example, the number-based string $D_3(1,2)$ shown in Eq.(\ref{eqa:3-topo-number-based-strings}), also, is an odd integer.
\item \label{hard:matrices} \textbf{Different colored graphs induce the same number-based strings.} For example, a colored graph $G_2$ shown in Fig.\ref{fig:one-string-vs-more-graphs} admits a total set-coloring $F:V(G_2)\cup E(G_2)\rightarrow M$, such that each $c\in F(uv)$ corresponds to some $a\in F(u)$ and $b\in F(v)$ holding $c=|a-b|$. So, we have a Topcode-matrix as follows:
\begin{equation}\label{eqa:same-string11}
\centering
{
\begin{split}
T_{code}(G_2)= \left(
\begin{array}{ccccccccccccccc}
11 & 11 & 11 & 11 & 44 & 44 \\
76 & 52 & 47 & 33 & 21 & 31 \\
87 & 63 & 58 & 44 & 23 & 13
\end{array}
\right)_{3\times 6}
\end{split}}
\end{equation}
Let $S_{tring}(G_2)$ be the set of number-based strings $D_k(G_2)$ induced from the Topcode-matrix $T_{code}(G_2)$ shown in Eq.(\ref{eqa:same-string11}), such as
\begin{equation}\label{eqa:same-string000}
D_k(G_2)=111111114444765247332131876358442313
\end{equation} with $36$ bytes. Clearly, the Topcode-matrix $T_{code}(G_2)$ can induce $|S_{tring}(G_2)|=(18)!$ different number-based strings, in total.

\begin{figure}[h]
\centering
\includegraphics[width=14cm]{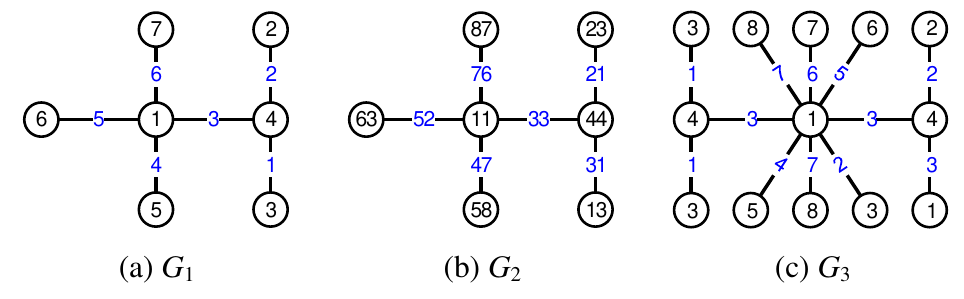}\\
\caption{\label{fig:one-string-vs-more-graphs}{\small (a) A tree $G_1$ admits a set-ordered graceful labeling; (b) $G_2$ admits a set-coloring; (c) $G_3$ admits a proper total coloring.}}
\end{figure}

\quad Another colored graph $G_3$ shown in Fig.\ref{fig:one-string-vs-more-graphs} admits a total coloring $f:V(G_3)\cup E(G_3)\rightarrow [1,10]$, such that $f(xy)=|f(x)-f(y)|$ for each edge $xy\in E(G_3)$. All number-based strings $D_i(G_3)$ induced from the Topcode-matrix $T_{code}(G_3)$ shown in Eq.(\ref{eqa:same-string22}) are collected into $S_{tring}(G_3)$, which distributes us $|S_{tring}(G_3)|=(36)!$ different number-based strings.

\begin{equation}\label{eqa:same-string22}
\centering
{
\begin{split}
T_{code}(G_3)= \left(
\begin{array}{ccccccccccccccc}
1 &1 & 1 &1 & 1 &1 & 1 &1 & 4 &4 & 4 &4 \\
7 &6 & 5 &2 & 4 &7 & 3 &3 & 2 &1 & 3 & 1\\
8 &7 & 6 &3 & 5 &8 & 4 &4 & 2 &3 & 1 &3
\end{array}
\right)_{3\times 12}
\end{split}}
\end{equation}

Since $S_{tring}(G_2)\subset S_{tring}(G_3)$, so we have more different colored graphs, which induce the same number-based strings.

\item \label{hard:no-method} No method for cutting a number-based string like $D_k(G_2)$ shown in Eq.(\ref{eqa:same-string000}) into $3q$ segments $c_1c_2\cdots c_m=a_1a_2\cdots a_{3q}$ with $a_j=c_{n_j}c_{n_j+1}\cdots c_{n_{j+1}}$ with $j\in [1,3q-1]$, where $n_1=1$ and $n_{3q}=m$, such that $a_1,a_2,\dots , a_{3q}$ are just the elements of some Topcode-matrix of order $3\times q$.
\end{asparaenum}

\begin{problem}\label{problem:hard-hard}
The above problems \textbf{Prob}-\ref{hard:coloring-labeling}, \textbf{Prob}-\ref{hard:primes-decomposition} and \textbf{Prob}-\ref{hard:no-method} form NP-hard problems. Since the complexity of the PNBSP-problem consists of the Topcode-matrix Problem and Subgraph Isomorphic Problem, so we can say PNBSP-problem to be (NP-hard$)^2$.
\end{problem}

\begin{thm}\label{thm:different-graphs-same-number-based-strings}
There are different colored graphs, which induce the same number-based strings.
\end{thm}

\subsubsection{Topological structure decomposition problems}

There are topological structure decomposition problems related with the Integer Factorization Problem and the Integer Partition Problem.
\begin{asparaenum}[\textrm{\textbf{TSDp}}-1.]
\item \textbf{Vertex-split Problem.} Let $N(x)=\{y_1,y_2,\dots, y_d\}$ with the degree $d=\textrm{deg}_G(x)\geq 2$ be the neighbor set of the vertex $x$ in a $(p,q)$-graph $G$. We vertex-split $x$ into $m$ vertices $x_1,x_2,\dots, x_m$ with $m\geq 2$, and join $x_i$ with the vertices of a subset $N(x_i)\subset N(x)$ by edges, here $|N(x_i)|\geq 1$ and $\bigcup ^m_{i=1}N(x_i)=N(x)$ and $N(x_i)\cap N(x_j)= \emptyset $ if $i\neq j$, the resultant graph, called \emph{vertex-split graph}, is denoted as $G\wedge_m x$.

\qquad Let $N(x_i)=\{y_{i,1},y_{i,2},\dots, y_{i,n_i}\}$ with $n_i=|N(x_i)|$, then we have to partition the degree $d$ into
\begin{equation}\label{eqa:vertex-split-partitions}
d=\sum^m_{i=1}|N(x_i)|=n_1+n_2+\cdots +n_m
\end{equation} For each degree partition defined in Eq.(\ref{eqa:vertex-split-partitions}), suppose that there are $n(N(x),m)$ different ways to partition $N(x)$ into different groups of subsets $N(x_1)$, $N(x_2)$, $\dots $, $N(x_m)$. We assume that there are $n(d,m)$ degree partitions holding Eq.(\ref{eqa:vertex-split-partitions}). Thereby, we have
\begin{equation}\label{eqa:555555}
\sum^{p-1}_{m=2}n(N(x),m)\cdot n(d,m)
\end{equation} different vertex-split graphs $G\wedge_m x$.
\item \textbf{Leaf-added graphs related with Integer Partition Problem.} We select $k$ vertices from a $(p,q)$-graph $G$ for adding $m$ new leaves to them, then we have $A^k_p=p(p-1)\cdots (p-k+1)$ selections, rather than ${p \choose k}=\frac{p!}{k!(p-k)!}$. Next, we partition $m$ into a sum of $k$ parts $m_1,m_2,\cdots ,m_k$ with each $m_i\neq 0$, that is $m=m_1+m_2+\cdots +m_k$. Suppose there is $n(m,k)$ groups of such $k$ parts. For a group of $k$ parts $m_1,m_2,\cdots ,m_k$, let $m_{i_1},m_{i_2},\cdots ,m_{i_k}$ be a permutation of $m_1,m_2,\cdots ,m_k$, so we have the number of such permutations is a factorial $k!$. Since the $(p,q)$-graph $G$ is colored well by a $W$-constraint coloring/labeling $f$, then we have the number $A_{leaf}(G,m)$ of all possible adding $m$ leaves as follows
\begin{equation}\label{eqa:c3xxxxx}
A_{leaf}(G,m)=\sum^m_{k=1}A^k_p\cdot n(m,k)\cdot k!=\sum^m_{k=1} n(m,k)\cdot p!
\end{equation} where $n(m,k)=\sum ^k_{r=1}n(m-k,r)$. Here, computing $n(m,k)$ can be transformed into finding the number $A(m,k)$ of solutions of \emph{Diophantine equation} $m=\sum ^k_{i=1}ix_i$. There is a recursive formula
\begin{equation}\label{eqa:c3xxxxx}
A(m,k)=A(m,k-1)+A(m-k,k)
\end{equation}
with $0 \leq k\leq m$. It is not easy to compute the exact value of $A(m,k)$, for example, the authors in \cite{Shuhong-Wu-Accurate-2007} and \cite{WU-Qi-qi-2001} computed exactly
\begin{equation}\label{eqa:c3xxxxx}
{
\begin{split}
A(m,6)=&\biggr\lfloor \frac{1}{1036800}(12m^5 +270m^4+1520m^3-1350m^2-19190m-9081)+\\
&\frac{(-1)^m(m^2+9m+7)}{768}+\frac{1}{81}\left[(m+5)\cos \frac{2m\pi}{3}\right ]\biggr\rfloor
\end{split}}
\end{equation}

\qquad On the other hands, for any odd integer $m\geq 7$, it was conjectured $m=p_1+p_2+p_3$ with three prime numbers $p_1,p_2,p_3$ from the famous Goldbach's conjecture: ``Every even integer, greater than 2, can be expressed as the sum of two prime numbers.'' In other word, determining $A(m,3)$ is difficult, also, it is difficult to express an odd integer $m=\sum^{3n}_{k=1} p\,'_k$, where each $p\,'_k$ is a prime number.

\item \textbf{Power sets and hyperedge sets. } The symbol $[1,n]^2$ is the set of all subsets of an integer set $[1,n]=\{1,2,\dots ,n\}$, called \emph{power set}. A \emph{hyperedge set} $\mathcal{E}\subset [1,n]^2$ holds $[1,n]=\bigcup _{e\in \mathcal{E}}e$ true. There is a theorem in \cite{Jianfang-Wang-Hypergraphs-2008}: If a hyperedge set $\mathcal{E}=\bigcup^m_{j=1} \mathcal{E}_j$ satisfies $\mathcal{E}_i\cap \mathcal{E}_j=\emptyset$, such that each hyperedge $e\in \mathcal{E}$ belongs to some $\mathcal{E}_j$, and $e\not\in \bigcup^m_{k=1,k\neq j} \mathcal{E}_k$. Then, $\{\mathcal{E}_1,\mathcal{E}_2,\dots ,\mathcal{E}_m\}$ is a \emph{hyperedge decomposition} of the hyperedge set $\mathcal{E}$. As $n$ is a large odd number, characterizing the power set $[1,n]^2$ will meet the integer factorization $$n=p_1\cdot p_2\cdot \cdots \cdot p_s~(s\geq 2)$$ and the integer partition $$n=m_1+m_2+\cdots +m_r~(r\geq 2)$$ since it is related with the problems of colorings, connectivity, Hamilton property and topological decomposition, as well as the application in Database structure \cite{Jianfang-Wang-Hypergraphs-2008}.
\item \textbf{Problem of Integer Partition And Trees.} Partition a positive integer $m$ into a sum $m=m_1+m_2+\cdots +m_n$ $(n\geq 2)$, such that there is a tree $T$ having vertices $u_1,u_2,\dots ,u_n$ holding just degrees $\textrm{deg}_T(u_i)=m_i$ for $i\in [1,n]$. \textbf{How many }trees with this degree sequence $m_1,m_2,\dots ,m_n$ are there?
\item \textbf{Topological form of Integer Factorization Problem.} The \emph{$4$-color triangle base} shown in Fig.\ref{fig:Triangle-split-coincide} (a)
\begin{equation}\label{eqa:4-color-triangle-base-mpgs}
\textbf{\textrm{T}}_{\textrm{riangle}}=\{T_{S,1},T_{S,2},T_{S,3},T_{S,4}\}
\end{equation} forms a \emph{4-colorable planar-graphic lattice} as follows
\begin{equation}\label{eqa:4-color-planar-lattice}
\textbf{\textrm{L}}(Z^0[~\overline{\ominus}~]\textbf{\textrm{T}}_{\textrm{riangle}})=\big \{[~\overline{\ominus}~]^4_{k=1}a_kT_{S,k}:a_k\in Z^0,T_{S,k}\in \textbf{\textrm{T}}_{\textrm{riangle}}\big \}
\end{equation} with $\sum ^4_{k=1}a_k\geq 1$, where ``$[~\overline{\ominus}~]$'' is the \emph{edge-coinciding operation} (see examples shown in Fig.\ref{fig:Triangle-split-coincide}). So, there are the following facts:

\qquad (i) Let $F_{ace}(G)$ be the face number of a maximal planar graph $G$, then $F_{ace}(G)-1=\sum^4_{k=1}a_k$.

\qquad (ii) Different $4$-colorings of a maximal planar graph $G$ induce different sums $\sum^4_{k=1}a_k$.

\qquad (iii) An integer $\sum^4_{k=1}a_k$ may correspond to two or more maximal planar graphs.
\end{asparaenum}

\begin{problem}\label{problem:66666666666}
\textbf{How partition} an integer $m=a_1+a_2+a_3+a_4$, such that the operation $[\ominus]^4_{k=1}a_kT_{S,k}$ defined in Eq.(\ref{eqa:4-color-planar-lattice}) produces \emph{4-colorable maximal planar graphs}? Or, the results of the operation $[\ominus]^4_{k=1}a_kT_{S,k}$ are just \emph{4-colorable semi-maximal planar graphs} $G^{in}$ and $G^{out}$ having 2-color no-changed cycles defined in \cite{Jin-Xu-Maximal-Science-Press-2019, Jin-Xu-55-56-configurations-arXiv-2107-05454v1}?
\end{problem}

\begin{figure}[h]
\centering
\includegraphics[width=16.4cm]{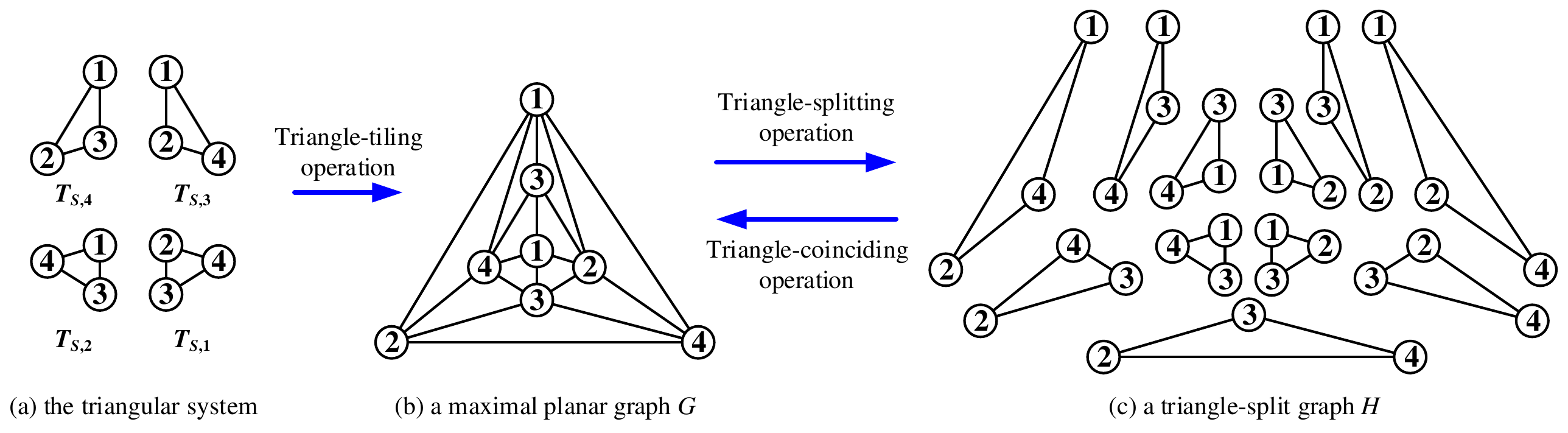}\\
\caption{\label{fig:Triangle-split-coincide}{\small The triangular system, the triangle-splitting operation and the triangle-coinciding operation.}}
\end{figure}

\begin{problem}\label{qeu:444444}
By the triangular system $\textbf{\textrm{T}}_{sys}=\{T_{S,1},T_{S,2},T_{S,3},T_{S,4}\}$ shown in Fig.\ref{fig:Triangle-split-coincide}, the vertex-coinciding operation and the edge-coinciding operation, we tile the whole $xOy$-plane, the resultant plane is called a \emph{triangular $4$-colored-plane}, denoted as $T4$-plane. At the infinity of the $T4$-plane, \textbf{what} is the $T4$-plane's boundary? \textbf{Is} any planar graph in the $T4$-plane?
\end{problem}

\begin{problem}\label{qeu:444444}
For the 4-colorable planar-graphic lattice defined in Eq.(\ref{eqa:4-color-planar-lattice}), we have the following questions:
\begin{asparaenum}[\textrm{\textbf{Mpg}}-1.]
\item In Eq.(\ref{eqa:4-color-planar-lattice}), $[~\overline{\ominus}~]^4_{k=1}a_kT_{S,k}$ produces $4$-colored maximal planar graphs, \textbf{consider} $4$-colored maximal planar graphs holding: (i) each $a_k$ is even; (ii) each $a_k$ is odd; (iii) each $a_k$ is a prime number; (iv) $4m=\sum^4_{k=1}a_k$; (v) $4n=1+\sum^4_{k=1}a_k$.
\item Since $[~\overline{\ominus}~]^4_{k=1}a_kT_{S,k}$ produces $4$-colored maximal planar graphs for a fixed group of $a_1,a_2,a_3,a_4$, \textbf{how many} different $4$-colored maximal planar graphs made by this fixed group of $a_1,a_2,a_3,a_4$?
\item A maximal planar graph $G$ admits a $4$-coloring $f$ which corresponding to a fixed group of $a_1,a_2,a_3,a_4$, and $G$ admits another $4$-coloring $g$ which corresponding to a fixed group of $b_1,b_2,b_3,b_4$, \textbf{do we have} $a_k=b_k$ for $k\in [1,4]$? If it is, then we call $G$ \emph{recursive maximal planar graph}.
\item If $G$ is a recursive maximal planar graph, \textbf{do we have} $a_k=b_k$ with $k\in [1,4]$ for any pair of $4$-colorings $f,g$ of $G$, where the $4$-coloring $f$ corresponds to $a_1,a_2,a_3,a_4$ in $G=[\overline{\ominus}]^4_{k=1}a_kT_{S,k}$, and the $4$-coloring $g$ corresponds to $b_1,b_2,b_3,b_4$ in $G=[~\overline{\ominus}~]^4_{k=1}b_kT_{S,k}$?
\item Suppose that a maximal planar graph $G$ admits a $4$-coloring $f$ which corresponding to a fixed group of $a_1,a_2,a_3,a_4$, and another maximal planar graph $H$ admits a $4$-coloring $g$ which corresponding to a fixed group of $b_1,b_2,b_3,b_4$. If $a_k=b_k$ for $k\in [1,4]$, \textbf{do we have} $G\cong H$?
\item \textbf{Characterize} a maximal planar graph $G$ admitting a $4$-coloring $f$ which corresponding to a fixed group of $a_1,a_2,a_3,a_4$, such that $|a_i-a_j|\leq 1$ for $i\neq j$, we say that this maximal planar graph $G$ admits a \emph{structure-equitable $4$-coloring}.
\end{asparaenum}
\end{problem}

\subsection{Our works}

Our goal is applying the techniques of topological coding to modern information security and future cryptography against supercomputer and quantum computer attacks, we will introduce a group of $W$-constraint $(k,d)$-total colorings and algorithms for realizing these colorings in some kinds of graphs, which are used to make quickly \emph{public-keys} and \emph{private-keys} with anti-quantum computing.

Since topological codes based on trees are easy to produce number-based strings by computer, and many Hanzi-graphs are forests consisted of trees. Another important reason is that connected graph can be vertex-split into trees for getting their own various colorings, or set-colorings obtained from that of trees. For example, we investigate trees, which obey the following $(k,d)$-total colorings: graceful $(k,d)$-total colorings, harmonious $(k,d)$-total colorings, $(k,d)$-edge-magic total colorings, $(k,d)$-graceful-difference total colorings and $(k,d)$-felicitous-difference total colorings, and so on.

We use a useful tool, called \emph{Topcode-matrix}, to bring some particular things together topologically, for example, these particular things are \emph{sets}, \emph{graphs}, \emph{number-based strings}, and so on. According to the needs of practical applications, many of our coloring techniques are given by algorithms and are easily converted into programs.

Part of $(k,d)$-type colorings/labelings were introduced in \cite{Yao-Sun-Hongyu-Wang-n-dimension-ICIBA2020}, \cite{Bing-Yao-2020arXiv} and \cite{Yao-Zhao-Zhang-Mu-Sun-Zhang-Yang-Ma-Su-Wang-Wang-Sun-arXiv2019}.

\subsection{Preliminary}

\subsubsection{Terminologies and notations}

Standard terminology and notation of graph theory used here can be found in \cite{Bondy-2008} and \cite{Gallian2021}. We often use the following terminology and notation:

\begin{asparaenum}[$\ast$ ]
\item The symbol $Z^0$ is the set of non-negative integers, and the notation $Z$ stands for the set of integers.
\item $[a,b]$ stands for an integer set $\{a,a+1,a+2,\dots, b\}$ with two integers $a,b$ subject to $a<b$; $[\alpha,\beta]^o$ denotes an \emph{odd-set} $\{\alpha,\alpha+2,\dots, \beta\}$ with odd integers $\alpha,\beta$ falling into $1\leq \alpha<\beta$.
\item For integers $r,k\geq 0$ and $s,d\geq 1$, we have two parameterized sets
\begin{equation}\label{eqa:two-parameterized-sets11}
S_{s,k,r,d}=\{k+rd,k+(r+1)d,\dots ,k+(r+s)d\}
\end{equation} and
\begin{equation}\label{eqa:two-parameterized-sets22}
O_{s,k,r,d}=\{k+[2(r+1)-1]d,k+[2(r+2)-1]d,\dots ,k+[2(r+s)-1]d\}
\end{equation}
\item The \emph{cardinality} (or \emph{number}) of elements of a set $S$ is denoted as $|S|$, so $|S_{s,k,r,d}|=s+1$ and $|O_{s,k,r,d}|=s$.
\item Let $\Lambda$ be a finite set. The set of all subsets of $\Lambda$ is denoted as $\Lambda^2=\{X:~X\subseteq \Lambda\}$, called \emph{power set}, and the power set $\Lambda^2$ contains no empty set at all. For example, for a given set $\Lambda=\{a,b,c,d,e\}$, the power set $\Lambda^2$ has its own elements $\{a\}$, $\{b\}$, $\{c\}$, $\{d\}$, $\{e\}$, $\{a,b\}$, $\{a,c\}$, $\{a,d\}$, $\{a,e\}$, $\{b,c\}$, $\{b,d\}$, $\{b,e\}$, $\{c, d\}$, $\{c, e\}$, $\{d,e\}$, $\{a,b,c\}$, $\{a,b,d\}$, $\{a,b,e\}$, $\{a,c,d\}$, $\{a,c,e\}$, $\{a,d,e\}$, $\{b,c,d\}$, $\{b,c,e\}$, $\{b,d,e\}$, $\{c,d,e\}$, $\{a,b,c,d\}$, $\{a,b,c,e\}$, $\{a,c,d,e\}$, $\{a,b,d,e\}$, $\{b,c,d,e\}$ and $\Lambda$ holding $|\Lambda^2|=2^5-1$.

\qquad Moreover, the integer set $[1,4]=\{1,2,3,4\}$ induces a power set $[1,4]^2=\big \{\{1\},\{2\},\{3\}$, $\{4\}$, $\{1,2\}$, $\{1,3\},\{1,4\}$, $\{2,3\}$, $\{2,4\}$, $\{3,4\}$, $\{1,2,3\}$, $\{1,2,4\}$, $\{1,3,4\}$, $\{2,3,4\}$, $[1,4]\big \}$, so the number of subsets is $2^4-1=15$ in total.
\item Each graph mentioned here has no multiple-edge and loop-edge.
\item A $(p,q)$-graph is a graph having $p$ vertices and $q$ edges.
\item Since the \emph{number} (also, \emph{cardinality}) of elements of a set $X$ is denoted as $|X|$, so the \emph{degree} of a vertex $x$ in a $(p,q)$-graph $G$ is denoted as $\textrm{deg}_G(x)=|N(x)|$, where $N(x)$ is the set of neighbors of the vertex $x$.
\item A \emph{leaf} is a vertex of degree one in a graph.
\item A \emph{bipartite graph} $H$ has its own vertex set $V(H)=X\cup Y$ with $X\cap Y=\emptyset$, such that each edge $uv\in E(H)$ holds $u\in X$ and $v\in Y$ true.
\item A \emph{tree} is a bipartite graph, such that each pair of vertices is connected by a unique \emph{path}. So, a tree has at least two \emph{leaves} and has no cycle.
\end{asparaenum}

\vskip 0.4cm

Let $M$ be a set of integers, and let a graph $G$ admit a mapping $f:S\subset V(G)\cup E(G)\rightarrow M$, write the set of colors assigned to the elements of $S$ by $f(S)=\{f(x):x\in S\}$. There are the following basic colorings/labelings in graph theory:
\begin{asparaenum}[\textrm{A}-1.]
\item As $S=V(G)$, $f(S)=f(V(G))$, we call $f$ \emph{vertex coloring}.
\item As $S=E(G)$, $f(S)=f(E(G))$, we call $f$ \emph{edge coloring}.
\item As $S=V(G)\cup E(G)$, $f(S)=f(V(G)\cup E(G))=f(V(G))\cup f(E(G))$, we call $f$ \emph{total coloring}.
\end{asparaenum}
\begin{asparaenum}[\textrm{B}-1. ]
\item As $f$ is a \emph{vertex coloring} and $f(u)\neq f(v)$ for each edge $uv\in E(G)$, we call $f$ \emph{proper vertex coloring}. The number $\chi(G)=\min_f\max\{f(w):w\in V(G)\}$ over all proper vertex colorings of $G$ is called \emph{chromatic number}.
\item As $f$ is an \emph{edge coloring} and $f(uv)\neq f(uw)$ for any two adjacent edges $uv$ and $uw$ of $G$, we call $f$ \emph{proper edge coloring}. The number $\chi\,'(G)=\min_f\max\{f(w):w\in E(G)\}$ over all proper edge colorings of $G$ is called \emph{chromatic index}.
\item As $f$ is a \emph{total coloring} holding $f(u)\neq f(v)$ for each edge $uv$ and $f(uv)\neq f(uw)$ for any two adjacent edges $uv,uw$, we call $f$ \emph{proper total coloring}. The number $\chi\,''(G)=\min_f\max\{f(w):w\in V(G)\cup E(G)\}$ over all proper total colorings of $G$ is called \emph{total chromatic number}.
\item As $f$ is a vertex coloring (resp. an edge coloring, or a total coloring) holding $f(u)\neq f(x)$ for any two vertices $u$ and $x$ (resp. for any two edges, or for any two elements $u,x\in V(G)\cup E(G)$), that is, $|f(V(G))|=|V(G)|$ (resp. $|f(E(G))|=|E(G)|$, or $|f(V(G)\cup E(G))|=|V(G)|+|E(G)|$), we call $f$ \emph{labeling}.
\item If a coloring/labeling $f$ satisfies a mathematical constraint or a group of mathematical constraints, we say $f$ to be a \emph{$W$-constraint coloring/labeling}.
\end{asparaenum}

\vskip 0.4cm

The labelings and colorings of graph theory not defined here can be found in \cite{Bondy-2008}, \cite{Gallian2021} and \cite{Yao-Wang-2106-15254v1}.

\begin{defn}\label{defn:topcode-matrix-definition}
\cite{Yao-Sun-Zhao-Li-Yan-2017, Yao-Zhang-Sun-Mu-Sun-Wang-Wang-Ma-Su-Yang-Yang-Zhang-2018arXiv} A \emph{Topcode-matrix} $T_{code}(G)$ of a $(p,q)$-graph $G$ is defined as
\begin{equation}\label{eqa:basic-formula-Topcode-matrix}
\centering
{
\begin{split}
T_{code}(G)= \left(
\begin{array}{ccccc}
x_{1} & x_{2} & \cdots & x_{q}\\
e_{1} & e_{2} & \cdots & e_{q}\\
y_{1} & y_{2} & \cdots & y_{q}
\end{array}
\right)_{3\times q}=\left(
\begin{array}{cccccccccccccc}
X\\
E\\
Y
\end{array}
\right)=(X,E,Y)^{T}
\end{split}}
\end{equation}\\
where the \emph{v-vector} $X=(x_1, x_2, \dots, x_q)$, the \emph{e-vector} $E=(e_1$, $e_2 $, $ \dots $, $e_q)$, and the \emph{v-vector} $Y=(y_1, y_2, \dots, y_q)$ consist of non-negative integers $e_i$, $x_i$ and $y_i$ for $i\in [1,q]$. We say $T_{code}(G)$ to be \emph{evaluated} if there exists a function $\theta$ such that $e_i=\theta(x_i,y_i)$ for $i\in [1,q]$, and call $x_i$ and $y_i$ \emph{ends} of $e_i$, denoted as $e_i=x_iy_i$, and $q$ \emph{size} of $T_{code}(G)$, such that $V(G)=X\cup Y$ and $E(G)=E$.\qqed
\end{defn}

\begin{defn}\label{defn:colored-topcode-matrix}
$^*$ Suppose that a $(p,q)$-graph $G$ admits a $W$-constraint total coloring $f:V(G)\cup E(G)\rightarrow [a,b]$. A \emph{colored Topcode-matrix} $T_{code}(G,f)$ of the graph $G$ is defined as
\begin{equation}\label{eqa:basic-colored-Topcode-matrix}
\centering
{
\begin{split}
T_{code}(G,f)= \left(
\begin{array}{ccccc}
f(x_{1}) & f(x_{2}) & \cdots & f(x_{q})\\
f(x_{1}y_{1}) & f(x_{2}y_{2}) & \cdots & f(x_{q}y_{q})\\
f(y_{1}) & f(y_{2}) & \cdots & f(y_{q})
\end{array}
\right)_{3\times q}=\left(
\begin{array}{cccccccccccccc}
X_f\\
E_f\\
Y_f
\end{array}
\right)=(X_f,E_f,Y_f)^{T}
\end{split}}
\end{equation}\\
holding the $W$-constraint $f(x_{i}y_{i})=W\langle f(x_{i}),f(y_{i})\rangle$ for $i\in [1,q]$. Moreover, if $G$ is a bipartite graph with the vertex set $V(G)=X^v\cup Y^v$ and $X^v\cap Y^v=\emptyset $, we stipulate $x_{i}\in X^v$ and $y_{i}\in Y^v$ such that $X_f\cap Y_f=\emptyset $ in Eq.(\ref{eqa:basic-colored-Topcode-matrix}), where ``$W$-constraint'' is a mathematical constraint, or a group of mathematical constraints.\qqed
\end{defn}

\begin{rem}\label{rem:generalization-colored-topcode-matrix}
In Definition \ref{defn:colored-topcode-matrix}, if the $(p,q)$-graph $G$ admits a $W$-constraint vertex coloring $g:V(G)\rightarrow [\alpha , \beta]$, we have a colored Topcode-matrix $T_{code}(G,g)$ of $G$ defined as
\begin{equation}\label{eqa:vertex-coloring-Topcode-matrix}
\centering
{
\begin{split}
T_{code}(G,g)= \left(
\begin{array}{ccccc}
g(x_{1}) & g(x_{2}) & \cdots & g(x_{q})\\
x_{1}y_{1} & x_{2}y_{2} & \cdots & x_{q}y_{q}\\
g(y_{1}) & g(y_{2}) & \cdots & g(y_{q})
\end{array}
\right)_{3\times q}=\left(
\begin{array}{cccccccccccccc}
X_g\\
E\\
Y_g
\end{array}
\right)=(X_g,E,Y_g)^{T}
\end{split}}
\end{equation} If $G$ is a bipartite graph with the vertex set $V(G)=X^v\cup Y^v$ and $X^v\cap Y^v=\emptyset $, we have $X_g\cap Y_g=\emptyset $, also, $X_g=X^v$ and $Y_g=Y^v$. And, if the $(p,q)$-graph $G$ admits a $W$-constraint edge coloring $h:E(G)\rightarrow [\lambda, \gamma]$, we have a colored Topcode-matrix $T_{code}(G,h)$ of $G$ defined as
\begin{equation}\label{eqa:edge-colored-Topcode-matrix}
\centering
{
\begin{split}
T_{code}(G,h)= \left(
\begin{array}{ccccc}
x_{1} & x_{2} & \cdots & x_{q}\\
h(x_{1}y_{1}) & h(x_{2}y_{2}) & \cdots & h(x_{q}y_{q})\\
y_{1} & y_{2} & \cdots & y_{q}
\end{array}
\right)_{3\times q}=\left(
\begin{array}{cccccccccccccc}
X\\
E_h\\
Y
\end{array}
\right)=(X,E_h,Y)^{T}
\end{split}}
\end{equation} If $G$ is a bipartite graph with the vertex set $V(G)=X^v\cup Y^v$ and $X^v\cap Y^v=\emptyset $, refer to that in Definition \ref{defn:colored-topcode-matrix}. \paralled
\end{rem}

\begin{rem}\label{rem:generalization-topcode-matrix}
In Definition \ref{defn:topcode-matrix-definition}, if each of elements $x_i,e_j,y_k$ in the v-vectors and the e-vector is a \emph{set} (resp. graph, matrix, string), we call $T_{code}(G)$ \emph{set-type Topcode-matrix} (resp. \emph{graph-type Topcode-matrix}, \emph{matrix-type Topcode-matrix}, \emph{string-type Topcode-matrix}) in this article.

In other word, the generalization of a Topcode-matrix is that each of elements in the Topcode-matrix is a \emph{thing} in the world, such that the Topcode-matrix brings these $3q$ things together topologically by a mathematical constraint, or a group of mathematical constraints for getting a completely mathematical ``story'' \cite{Yao-Wang-Su-Wanjia-Post-Quantum-2021}.

About part of Topcode-matrix algebra, readers can refer to \cite{Bing-Yao-2020arXiv} and \cite{Yao-Zhao-Zhang-Mu-Sun-Zhang-Yang-Ma-Su-Wang-Wang-Sun-arXiv2019}.\paralled
\end{rem}

\subsubsection{Basic graph operations}

Graph operations are the \emph{soul} of topological structures of graphs.

We introduce the graph operations often used. Let $G_i$ and $G_j$ be two disjoint connected graphs in the following argument.

(1) \textbf{Vertex-coinciding operation.} We coincide a vertex $u$ of $G_i$ with one vertex $x$ of $G_j$ into one vertex $w=u\odot x$, the resulting graph is denoted as $G_i\odot G_j$, called a \emph{vertex-coincided graph} based on the \emph{graph vertex-coinciding operation} ``$\odot$'', such that the vertex $w\in V(G_i\odot G_j)$ has the degree $\textrm{deg}_{G_i}(u)+\textrm{deg}_{G_j}(x)$, see an illustration in Fig.\ref{fig:22-vertex-join-split-coincide} from (a) to (b).
\begin{figure}[h]
\centering
\includegraphics[width=16cm]{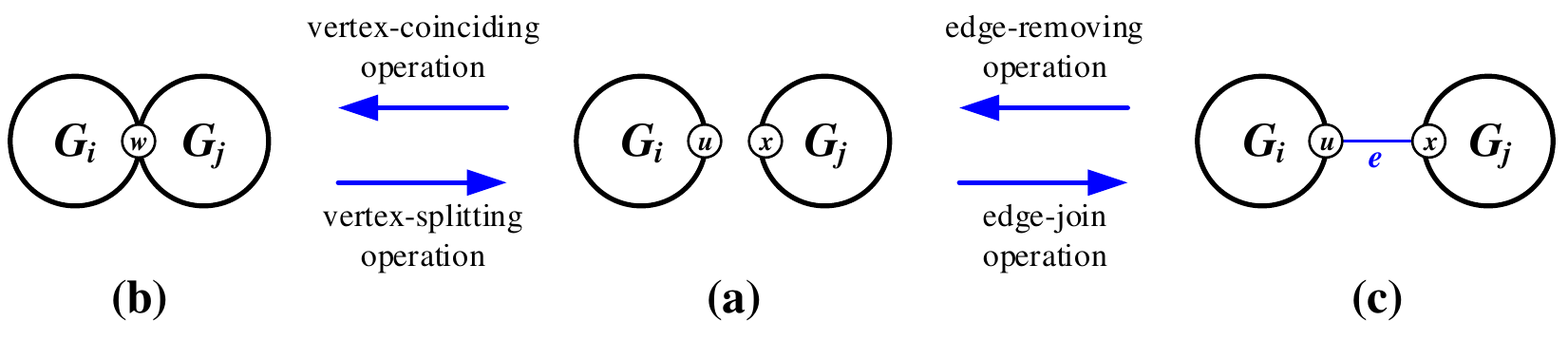}
\caption{\label{fig:22-vertex-join-split-coincide}{\small A scheme for illustrating the vertex-coinciding operation and the edge-joining operation.}}
\end{figure}

(2) \textbf{Edge-joining operation.} We use a new edge to join a vertex $u$ of $G_i$ with one vertex $x$ of $G_j$ together, so we get a connected graph $G_i\ominus G_j$, called an \emph{edge-joined graph} based on the \emph{graph edge-joining operation} ``$\ominus$'', see an example shown in Fig.\ref{fig:22-vertex-join-split-coincide} from (a) to (c).

(3) \textbf{Vertex-splitting tree-operation. }A \emph{vertex-splitting operation} on a connected graph is defined as: Let $x$ be a vertex of a cycle $C$ of a connected graph $G$, and $N(x)=\{x_1,x_2,\dots, x_d\}$ with $d=\textrm{deg}_G(x)\geq 2$ be the neighbor set of the vertex $x$. We vertex-split $x$ into two vertices $x\,'$ and $x\,''$, and let $x\,'$ to join with vertices $x_1,x_2,\dots, x_k$ by edges, and let $x\,''$ to join with vertices $x_{k+1},x_{k+2},\dots, x_d$ by edges, the resultant graph is a connected graph, denoted as $G\wedge x$ and called \emph{vertex-split graph}. See an example shown in Fig.\ref{fig:22-vertex-join-split-coincide} from (b) to (a). Clearly, the number of cycles of the vertex-split graph $G\wedge x$ is less than that of the original connected graph $G$.

If $G_1=G\wedge x$ has cycles, we do the vertex-splitting operation to $G_1$. Go on in this way, we get a graph $G_m=G_{m-1}\wedge y$ to be a tree $T$ for some $m$. The process of vertex-splitting $G$ into a tree $T$ is called the \emph{vertex-splitting tree-operation}, and write $T=\wedge(G)$ and $G=\wedge^{-1}(T)$, or $G=\odot_m (T)$.

\begin{thm}\label{thm:666666}
Each connected $(p,q)$-graph $G$ corresponds to a set $V_{\textrm{split}}(G)$ of trees of $q+1$ vertices, such that $G=\wedge^{-1}(T)$, also, $G=\odot_m (T)$ for each tree $T\in V_{\textrm{split}}(G)$.
\end{thm}

Fig.\ref{fig:33-from-graph-to-tree} shows a connected $(5,8)$-graph $G$ that can be vertex-split into trees $G_4$ and $T_4$ of $8$ edges.

\begin{figure}[h]
\centering
\includegraphics[width=13.6cm]{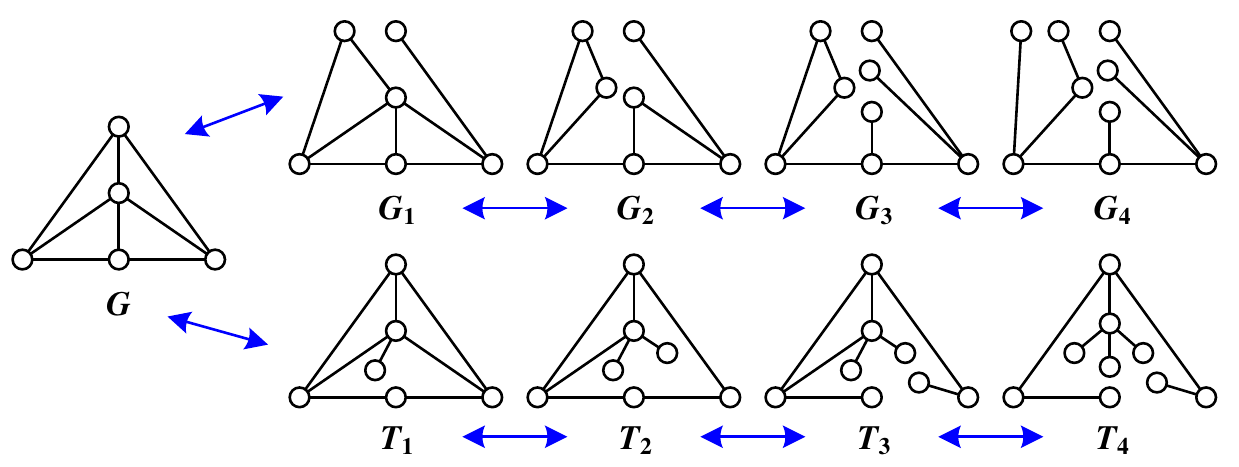}\\
\caption{\label{fig:33-from-graph-to-tree} {\small A connected graph $G$ can be vertex-split into trees $G_4$ and $T_4$.}}
\end{figure}

\begin{thm}\label{thm:2-vertex-split-graphs-isomorphic}
\cite{Yao-Su-Sun-Wang-Graph-Operations-2021, Wang-Su-Yao-divided-2020} Suppose that two connected graphs $G$ and $H$ admit a mapping $f:V(G)\rightarrow V(H)$. In general, a vertex-split graph $G\wedge u$ with $\textrm{deg}_G(u)\geq 2$ is not unique, so we have a vertex-split graph set $S_G(u)=\{G\wedge u\}$, similarly, we have another vertex-split graph set $S_H(f(u))=\{H\wedge f(u)\}$. If each vertex-split graph $L\in S_G(u)$ corresponds to another vertex-split graph $T\in S_H(f(u))$ such that $L\cong T$, and vice versa, we write this fact as $G\wedge u\cong H\wedge f(u)$ for each vertex $u\in V(G)$, then we claim that $G$ is \emph{isomorphic} to $H$, that is $G\cong H$.
\end{thm}

\begin{rem}\label{rem:333333}
Although each graph $G\wedge u$ keeps all of edges of $G$, and each graph $H\wedge f(u)$ keeps all of edges of $H$, however, it is not easy to show $G\wedge u\cong H\wedge f(u)$ for each vertex $u\in V(G)$ defined in Theorem \ref{thm:2-vertex-split-graphs-isomorphic}. The authors in \cite{Wang-Su-Yao-divided-2020} have shown two parameters of graphs based on the split connectivity:

\textbf{The v-split connectivity.} A \emph{v-split $k$-connected graph} $H$ holds: $H\wedge V^*$ (or $H\wedge \{x_i\}^{k}_1$) is disconnected, where $V^*=\{x_1,x_2,\dots,x_k\}$ is a subset of $V(H)$, each component $H_j$ of $H\wedge \{x_i\}^{k}_1$ has at least a vertex $w_j\not \in V^*$, $|V(H\wedge \{x_i\}^{k}_1)|=k+|V(H)|$ and $|E(H\wedge \{x_i\}^{k}_1)|=|E(H)|$. The smallest number of $k$ for which $H\wedge \{x_i\}^{k}_1$ is disconnected is called the \emph{v-split connectivity} of $H$, denoted as $\kappa_{d}(H)$ (see an example shown in Fig.\ref{fig:5-degree-connectivity}).

\textbf{The e-split connectivity.} An \emph{e-split $k$-connected graph} $H$ holds: $H\wedge \{e_i\}^{k}_1$ (or $H\wedge E^*$) is disconnected, where $E^*=\{e_1,e_2,\dots,e_k\}$ is a subset of $E(H)$, each component $H_j$ of $H\wedge \{e_i\}^{k}_1$ has at least a vertex $w_j$ being not any end of any edge of $E^*$, $|V(H\wedge \{e_i\}^{k}_1)|=2k+|V(H)|$ and $|E(H\wedge \{e_i\}^{k}_1)|=k+|E(H)|$. The smallest number of $k$ for which $H\wedge \{e_i\}^{k}_1$ is disconnected is called the \emph{e-split connectivity} of $H$, denoted as $\kappa'_{d}(H)$ (see an example shown in Figure \ref{fig:5-degree-connectivity}).\paralled
\end{rem}

\begin{lem}\label{thm:equivalent-proof}
\cite{Wang-Su-Yao-divided-2020} A graph $G$ is $k$-connected if and only if it is v-divided $k$-connected, namely, $\kappa_{d}(H)=\kappa(H)$.
\end{lem}

\begin{thm}\label{thm:lemma-equivalent-proof}
\cite{Wang-Su-Yao-divided-2020} If a $k$-connected graph has a property related with its $k$-connectivity, so do a $v$-divided $k$-connected graph.
\end{thm}

\begin{thm}\label{thm:vertex-divided-vs-e-divided}
\cite{Wang-Su-Yao-divided-2020} Any connected graph $G$ holds the inequalities $\kappa'_{d}(G)\leq \kappa_{d}(G)\leq 2\kappa'_{d}(G)$ true, and the boundaries are reachable.
\end{thm}

\begin{figure}[h]
\centering
\includegraphics[width=16cm]{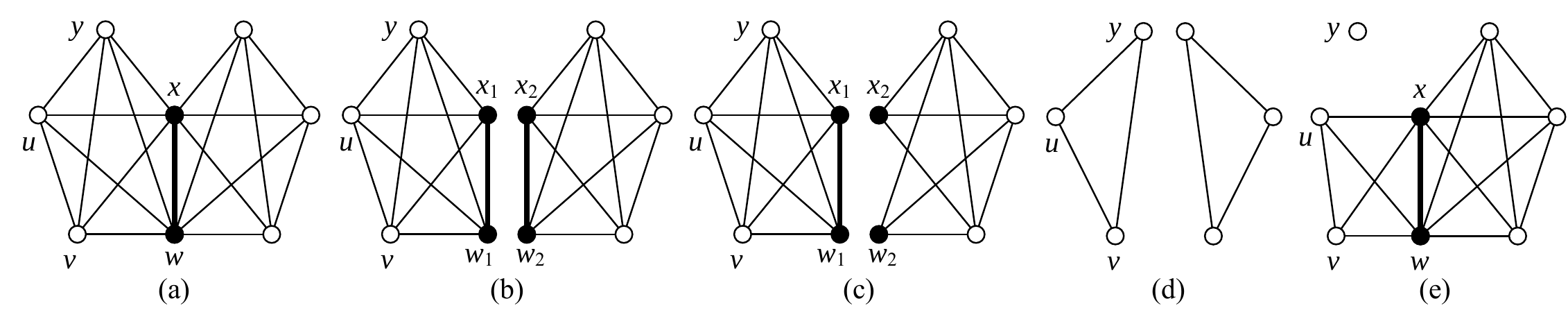}
\caption{\label{fig:5-degree-connectivity}{\small (a) A graph $H$ with minimum degree $\delta(H)=4$; (b) an e-split graph $H\wedge xw$ with $\kappa'_{d}(H)=1$; (c) a v-split graph $H\wedge \{x,w\}$ with $\kappa_{d}(H)=2$; (d) a v-deleted graph $H-\{x,w\}$ with $\kappa(H)=2$; (e) an e-deleted graph $H-\{yx,yw,yu,yv\}$ with $\kappa'(H)=4$, cited from \cite{Wang-Su-Yao-divided-2020}.}}
\end{figure}

(4) \cite{Zhou-Yao-Chen-Tao2012} \textbf{Adding leaves operation.} A graph operation called ``adding leaves operation'' is defined as: Adding new vertices $w_1,w_2,\dots, w_r$ to a graph $G$, and join each new vertex $w_i$ with some vertex $z_i$ of $G$ by an edge $w_iz_i$, the resultant graph is denoted as $G+L_{eaf}(w)$, called \emph{leaf-added graph}, where the leaf set $L_{eaf}(w)=\{w_1,w_2,\dots, w_r\}$, and each vertex $w_i$ is called a \emph{leaf} of $G+L_{eaf}(w)$. Conversely, removing all leaves of a graph $H$ produces a graph $H-L_{eaf}(H)$, where $L_{eaf}(H)$ is the set of all leaves of the graph $H$.

We will use the following trees in the later sections: Let $T$ be a tree with diameter $D(T)$ not less than 3, and let $L(T)$ be the set of all leaves of the tree $T$. We have trees $T_{i+1}=T_{i}-L(T_i)$ with $i\in[1,m]$ obtained by removing all leaves of the leaf set $L(T_i)$ such that $T_{m+1}$ is just a star $K_{1,n}$ being a tree with diameter $D(K_{1,n})\leq 2$ and $n\geq 1$, where $T_{1}=T$ and $m\geq 1$. Then we have a formula
\begin{equation}\label{eqa:555555}
m+1=\left \lceil\frac{D(T)}{2} \right \rceil
\end{equation}

See examples shown Fig.\ref{fig:tree-leaf-diameter}, we can see

(i) Trees $G_{i+1}=G_{i}-L(G_i)$ for $i=1,2$, and $G_3=K_{1,1}$ is a star of two vertices and diameter one, such that $2+1=\left \lceil\frac{D(G_1)}{2} \right \rceil=\left \lceil 2.5 \right \rceil=3$.

(ii) Trees $T_{j+1}=T_{j}-L(T_j)$ for $j=1,2$, and $T_3=K_{1,2}$ is a star of three vertices and diameter two, such that $2+1=\left \lceil\frac{D(T_1)}{2} \right \rceil=\left \lceil 3 \right \rceil=3$.

\begin{figure}[h]
\centering
\includegraphics[width=15cm]{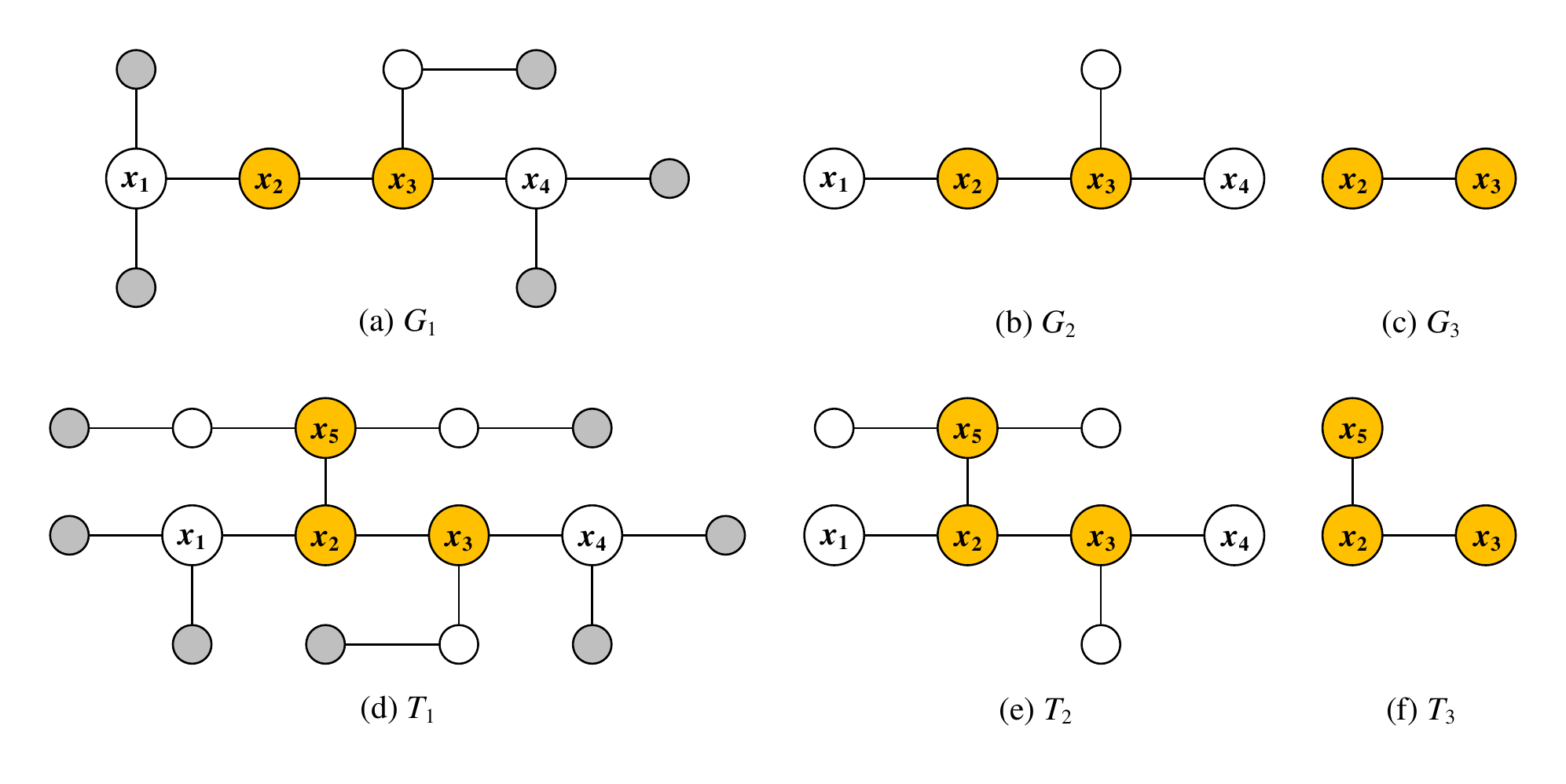}\\
\caption{\label{fig:tree-leaf-diameter} {\small (a) A tree $G_1$ with $D(G_1)=5$; (b) a tree $G_2$ obtained by removing all leaves of $G_1$; (c) a tree $G_3$ obtained by removing all leaves of $G_2$; (d) a tree $T_1$ with $D(T_1)=6$; (e) a tree $T_2$ obtained by removing all leaves of $T_1$; (f) a tree $T_3$ obtained by removing all leaves of $T_2$.}}
\end{figure}

\section{Parameterized total colorings}

\subsection{Basic labelings and colorings}

Acharya and Hegde, in \cite{Acharya-Hegde-1990}, introduced $(k,d)$-arithmetic labeling as: Let $G$ be a graph with $q$ edges and let $k$ and $d$ be positive integers. A labeling $f$ of $G$ is said to be $(k,d)$-arithmetic if the vertex labels are distinct nonnegative integers and the edge labels induced by $f(x)+f(y)$ for each edge $xy$ are $k,k+d, k+2d,\dots ,k+(q-1)d$. Clearly, this graph $G$ must be \emph{bipartite}.

Notice that a $(k,d)$-arithmetic labeling can induce a \emph{strongly c-elegant labeling} introduced by Chang, Hsu, and Rogers \cite{Gallian2021}, and moreover a strongly $c$-elegant labeling is felicitous too.

\begin{defn}\label{defn:Marumuthu-edge-magic-graceful-labeling}
\cite{Marumuthu-G-2015} If there exists a constant $k\geq 0$, such that a $(p, q)$-graph $G$ admits a total labeling $f:V(G)\cup E(G)\rightarrow [1, p+q]$, each edge $uv\in E(G)$ holds
$|f(u)+f(v)-f(uv)|=k$ and $f(V(G)\cup E(G))=[1, p+q]$ true, we call $f$ \emph{edge-magic graceful labeling} of $G$, and $k$ a \emph{magic constant}. Moreover, $f$ is called \emph{super edge-magic graceful labeling} if the vertex color set $f(V(G))=[1, p]$.\qqed
\end{defn}

\begin{defn}\label{defn:S-M-Hegde-kd-graceful-labeling}
\cite{S-M-Hegde2000} For a $(p,q)$-graph $G$, if there is a total labeling $f$ such that the vertex color set $f(V(G))\subseteq S_{b,0,0,d}\cup S_{q-1,k,0,d}$ and the edge color set
$$
f(E(G))=\{|f(u)-f(v)|;\ uv\in E(G)\}=S_{q-1,k,0,d}
$$ then we call $f$ \emph{$(k,d)$-graceful labeling} of $G$. \qqed
\end{defn}

In Definition \ref{defn:S-M-Hegde-kd-graceful-labeling}, as $d=1$, $f$ is called a \emph{$k$-graceful labeling} of $G$; as $(k,d)=(1,2)$, $f$ is called an \emph{odd-graceful labeling} of $G$.

\begin{defn}\label{defn:k-d-arithmetic-labelings}
\cite{Acharya-Hegde-1990} For a $(p, q)$-graph $G$, if there is a total labeling $f$ such that the vertex color set $f(V(G))\subseteq S_{b,0,0,d}\cup S_{q-1,k,0,d}$ and the edge color set
$$
f(E(G))=\{f(xy)=f(x)+f(y): xy\in E(G)\}=S_{q-1,k,0,d}
$$ then we call $f$ $(k,d)$-\emph{arithmetic labeling} of $G$.\qqed
\end{defn}

\begin{defn} \label{defn:k-d-harmonious}
\cite{Yao-Zhang-Sun-Mu-Wang-Zhang2018} If a $(p,q)$-graph $G$ admits a mapping $h:V(G)\rightarrow S_{b,0,0,d}\cup S_{q-1,k,0,d}$ with integers $k\geq 0$ and $d\geq 1$, such that $f(x)\neq f(y)$ for any pair of vertices $x,y$ of $G$, and each edge $uv\in E(G)$ is colored by $h(uv)=h(u)+h(v)~(\bmod^*~qd)$ defined as $h(uv)-k=[h(u)+h(v)-k](\bmod ~qd)$, and the edge color set
$$h(E(G))=S_{q-1,k,0,d}=\{k,k+d,\dots ,k+(q-1)d\}$$
then we call $h$ \emph{$(k,d)$-harmonious labeling} of $G$.\qqed
\end{defn}

Motivated from the parameterized labelings: the $(k,d)$-arithmetic labeling, the $(k,d)$-graceful labeling, the $(k,d)$-harmonious labeling and the new $(k,d)$-type labelings introduced in \cite{Sun-Zhang-Yao-ICMITE-2017}, as well as the well-defined parameterized labelings, we present the following parameterized total colorings.

\begin{defn} \label{defn:kd-w-type-colorings}
\cite{Yao-Su-Wang-Hui-Sun-ITAIC2020, Yao-Wang-2106-15254v1} Let $G$ be a bipartite and connected $(p,q)$-graph, so its vertex set $V(G)=X\cup Y$ with $X\cap Y=\emptyset$ such that each edge $uv\in E(G)$ holds $u\in X$ and $v\in Y$. If there is a total mapping
\begin{equation}\label{eqa:333333}
{
\begin{split}
&f:X\rightarrow S_{m,0,0,d}=\{0,d,\dots ,md\}\\
&f:Y\cup E(G)\rightarrow S_{n,k,0,d}=\{k,k+d,\dots ,k+nd\}
\end{split}}
\end{equation} here it is allowed $f(x)=f(y)$ for some distinct vertices $x,y\in V(G)$. Let $c$ be a fixed non-negative integer.
\begin{asparaenum}[\textrm{\textbf{Ptol}}-1. ]
\item If $f(uv)=|f(u)-f(v)|$ for each edge $uv\in E(G)$, and the edge color set $f(E(G))=S_{q-1,k,0,d}$, and the vertex color set
$$f(V(G)\cup E(G))\subseteq S_{m,0,0,d}\cup S_{q-1,k,0,d}$$ then $f$ is called a \emph{graceful $(k,d)$-total coloring}; and moreover $f$ is called a \emph{strongly graceful $(k,d)$-total coloring} if $f(x)+f(y)=k+(q-1)d$ for each matching edge $xy$ of a perfect matching $M$ of $G$.
\item If $f(uv)=|f(u)-f(v)|$ for each edge $uv\in E(G)$, and the edge color set $f(E(G))=O_{2q-1,k,d}$, and the vertex color set
$$f(V(G)\cup E(G))\subseteq S_{m,0,0,d}\cup S_{2q-1,k,0,d}$$ then $f$ is called an \emph{odd-graceful $(k,d)$-total coloring}; and moreover $f$ is called a \emph{strongly odd-graceful $(k,d)$-total coloring} if $f(x)+f(y)=k+(2q-1)d$ for each matching edge $xy$ of a perfect matching $M$ of $G$.
\item If the mixed color set $$\{f(u)+f(uv)+f(v):uv\in E(G)\}=\{2k+2ad,2k+2(a+1)d,\dots ,2k+2(a+q-1)d\}$$ with $a\geq 0$ and the total color set $f(V(G)\cup E(G))\subseteq S_{m,0,0,d}\cup S_{2(a+q-1),k,a,d}$, then $f$ is called an \emph{edge-antimagic $(k,d)$-total coloring}.
\item If $f(uv)=f(u)+f(v)~(\bmod^*qd)$ defined by $f(uv)-k=[f(u)+f(v)-k](\bmod ~qd)$ for each edge $uv\in E(G)$, and the edge color set $f(E(G))=S_{q-1,k,0,d}$, then $f$ is called a \emph{harmonious $(k,d)$-total coloring}.
\item If $f(uv)=f(u)+f(v)~(\bmod^*qd)$ defined by $f(uv)-k=[f(u)+f(v)-k](\bmod ~2qd)$ for each edge $uv\in E(G)$, and the edge color set $f(E(G))=O_{2q-1,k,d}$, then $f$ is called an \emph{odd-elegant $(k,d)$-total coloring}.

--------- \textbf{Four $W$-magic types:}

\item If there are the edge-magic constraint
\begin{equation}\label{eqa:definition-edge-magic-constraint}
f(u)+f(uv)+f(v)=c,~uv\in E(G)
\end{equation} the edge color set $f(E(G))=S_{q-1,k,0,d}$ and the total color set $f((V(G)\cup E(G))\subseteq S_{m,0,0,d}\cup S_{q-1,k,0,d}$, then $f$ is called an \emph{edge-magic $(k,d)$-total coloring}. And moreover, $f$ is called a \emph{pseudo-edge-magic $(k,d)$-total coloring} if both $|f(E(G))|< q$ and Eq.(\ref{eqa:definition-edge-magic-constraint}) hold true.
\item If there are the edge-difference constraint
\begin{equation}\label{eqa:definition-edge-difference-constraint}
f(uv)+|f(u)-f(v)|=c,~uv\in E(G)
\end{equation} and the edge color set $f(E(G))=S_{q-1,k,0,d}$, then $f$ is called an \emph{edge-difference $(k,d)$-total coloring}. And moreover $f$ is called a \emph{pseudo-edge-difference $(k,d)$-total coloring} if both $|f(E(G))|< q$ and Eq.(\ref{eqa:definition-edge-difference-constraint}) hold true.
\item If there are the graceful-difference constraint
\begin{equation}\label{eqa:definition-graceful-difference-constraint}
\big ||f(u)-f(v)|-f(uv)\big |=c,~uv\in E(G)
\end{equation} and the edge color set $f(E(G))=S_{q-1,k,0,d}$, then we call $f$ \emph{graceful-difference $(k,d)$-total coloring}. And moreover $f$ is called a \emph{pseudo-graceful-difference $(k,d)$-total coloring} if both $|f(E(G))|< q$ and Eq.(\ref{eqa:definition-graceful-difference-constraint}) hold true.
\item If there are the felicitous-difference constraint
\begin{equation}\label{eqa:definition-felicitous-difference-constraint}
|f(u)+f(v)-f(uv)|=c,~uv\in E(G)
\end{equation} and the edge color set $f(E(G))=S_{q-1,k,0,d}$, then we call $f$ \emph{felicitous-difference $(k,d)$-total coloring}. And moreover $f$ is called a \emph{pseudo-felicitous-difference $(k,d)$-total coloring} if both $|f(E(G))|< q$ and Eq.(\ref{eqa:definition-felicitous-difference-constraint}) hold true.\qqed
\end{asparaenum}
\end{defn}

\begin{defn}\label{defn:odd-edge-W-type-total-labelings-definition}
\cite{Yao-Zhang-Yang-Wang-Odd-Edge-arXiv-02477} Let $G$ be a bipartite and connected $(p,q)$-graph with vertex set $V(G)=X\cup Y$, and $X\cap Y=\emptyset$. There is a total coloring
$${
\begin{split}
&h:X\rightarrow S_{m,0,0,d}\\
&h:Y\cup E(G)\rightarrow S_{n,k,0,d}\cup O_{2q-1,k,d}
\end{split}}
$$ with integers $k\geq 0$ and $d\geq 1$. Let $c$ be an non-negative integer.

(i) If there are the edge-magic constraint $h(u)+h(uv)+h(v)=c$ for each edge $uv\in E(G)$ and the edge color set $h(E(G))=O_{2q-1,k,d}$, then $h$ is called an \emph{odd-edge edge-magic $(k,d)$-total \textbf{labeling}} when $|h(V(G))|=p$; and moreover $h$ is called an \emph{odd-edge edge-magic $(k,d)$-total \textbf{coloring}} when the vertex color set cardinality $|h(V(G))|<p$.

(ii) If there are the edge-difference constraint $h(uv)+|h(u)-h(v)|=c$ for each edge $uv\in E(G)$ and the edge color set $h(E(G))=O_{2q-1,k,d}$, then $h$ is called an \emph{odd-edge edge-difference $(k,d)$-total \textbf{labeling}} when $|h(V(G))|=p$; and furthermore $h$ is called an \emph{odd-edge edge-difference $(k,d)$-total \textbf{coloring}} when the vertex color set cardinality $|h(V(G))|<p$.

(iii) If there are the felicitous-difference constraint $|h(u)+h(v)-h(uv)|=c$ for each edge $uv\in E(G)$ and the edge color set $h(E(G))=O_{2q-1,k,d}$, then $h$ is called an \emph{odd-edge felicitous-difference $(k,d)$-total \textbf{labeling}} when $|h(V(G))|=p$; and moreover $h$ is called an \emph{odd-edge felicitous-difference $(k,d)$-total \textbf{coloring}} when the vertex color set cardinality $|h(V(G))|<p$.

(vi) If there are the graceful-difference constraint $\big ||h(u)-h(v)|-h(uv)\big |=c$ for each edge $uv\in E(G)$ and the edge color set $h(E(G))=O_{2q-1,k,d}$, then $h$ is called an \emph{odd-edge graceful-difference $(k,d)$-total \textbf{labeling}} when $|h(V(G))|=p$; and furthermore $h$ is called an \emph{odd-edge graceful-difference $(k,d)$-total \textbf{coloring}} when the vertex color set cardinality $|h(V(G))|<p$.\qqed
\end{defn}

\begin{rem} \label{rem:kd-w-tupe-colorings-definition}
In general, we call $f$ \emph{$W$-constraint $(k,d)$-total coloring} if the $W$-constraint $f(uv)=W\langle f(u), f(v)\rangle$ holds true, and the edge color set $f(E(G))$ hold one of the well-defined total colorings defined in Definition \ref{defn:kd-w-type-colorings} and Definition \ref{defn:odd-edge-W-type-total-labelings-definition}.
\begin{asparaenum}[(1) ]
\item There are new parameters of graphs based on Definition \ref{defn:kd-w-type-colorings} and Definition \ref{defn:odd-edge-W-type-total-labelings-definition}. For a graph $G$, we have two $W$-constraint $(k,d)$-total colorings $g_{\min}$ and $g_{\max}$ of $G$ such that
 $$|g_{\min}(V(G))|\leq |g(V(G))|\leq |g_{\max}(V(G))|$$ for each $W$-constraint $(k,d)$-total coloring $g$ of $G$.

\qquad As this $W$-constraint $(k,d)$-total coloring is a graceful $(1,1)$-total coloring $g_{\max}$ of a tree $T$ means that the maximal number $|g_{\max}(V(T))|$ is close to a graceful labeling of $T$. \textbf{Graceful Tree Conjecture} \cite{A-Rosa-graceful-1967, Gallian2021} says $|g_{\max}(V(T))|=|V(T)|$.

\qquad As this $W$-constraint $(k,d)$-total coloring is the graceful $(k,d)$-total coloring, then a graceful $(1,2)$-total coloring $g_{\max}$ of a tree $T$ means that the maximal number $|g_{\max}(V(T))|$ is close to an odd-graceful labeling of $T$. Gnanajothi proposed \textbf{Odd-graceful Tree Conjecture} \cite{R-B-Gnanajothi-1991} as $|g_{\max}(V(T))|=|V(T)|$.

\qquad As this $W$-constraint $(k,d)$-total coloring is a strongly graceful coloring, Wang and Yao, in \cite{Wang-Yao-equivalent-2020}, have proven that two conjectures \textbf{Graceful Tree Conjecture} and \textbf{Strongly Graceful Tree Conjecture} are equivalent from each other.

\qquad As this $W$-constraint $(k,d)$-total coloring is an edge-magic $(k,d)$-total labeling: In 1970, Anton Kotzig and Alex Rosa defined an \emph{edge-magic total labeling} of a $(p,q)$-graph $G$ as a bijective mapping $f$ from $V(G)\cup E(G)$ to $[1, p+q]$ such that for any edge $uv$, $f(u)+f(v)+f(uv)=c$, where $c$ is a fixed constant. Moreover, they \textbf{conjectured}: \emph{Every tree admits an edge-magic total labeling}.

\item A $W$-constraint $(k,d)$-total coloring $f$ defined in Definition \ref{defn:kd-w-type-colorings} is \emph{proper} if $f(u)\neq f(v)$ for each edge $uv\in E(G)$, and $f(uv)\neq f(uw)$ for any two adjacent edges $uv,uw\in E(G)$. So we call $f$ a \emph{$W$-constraint proper $(k,d)$-total coloring} of $G$. Investigating various $W$-constraint proper $(k,d)$-total colorings of graphs is interesting and challenging.
\item In practical application, we may reduce the restrictions of some $W$-constraint $(k,d)$-total colorings of graphs, such as a $W$-constraint $(k,d)$-total coloring $g$ of a bipartite graph $G$ with vertex bipartition $(X,Y)$ holds
$$g:X\rightarrow \{-md,-(m-1)d,\dots ,-2d,-d\}\cup \{0,d,\dots ,md\}$$ and
$${
\begin{split}
g:Y\cup E(G)\rightarrow &\{-k,-k-d,\dots ,-k-nd\}\cup \{-k+d,-k+2d,\dots ,-k+nd\}\\
&\cup \{k-d,k-2d,\dots ,k-nd\}\cup \{k,k+d,\dots ,k+nd\}
\end{split}}$$ where integers $m,n,k$ and $d$ are non-negative integers.\paralled
\end{asparaenum}
\end{rem}

\begin{defn} \label{defn:w-magic-kd-total-vd-ek-coloring}
$^*$ A $(p,q)$-graph $G$ admits a \emph{$W$-constraint $(k,d)$-total $(vd$-$ek)$-coloring} $f$ defined as
\begin{equation}\label{eqa:vd-ek-coloring}
{
\begin{split}
&f:V(G)\rightarrow S_{m,0,0,d}=\{0,d,\dots, md\}\\
&f:E(G)\rightarrow S_{n,k,0,d}=\{k,k+d, \dots, k+nd\}
\end{split}}
\end{equation} and each edge $uv\in E(G)$ holds the $W$-constraint $f(uv)=W\langle f(u),f(v)\rangle$ for integers $k\geq 0$ and $d\geq 1$. Moreover, $f$ is called a \emph{$W$-constraint $(k,d)$-total $(vd$-$ek)$-all-labeling} if the vertex color set $f(V(G))=S_{p-1,0,0,d}$, the edge color set $f(E(G))=S_{p-1,k,0,d}$, and each edge $uv\in E(G)$ holds the $W$-constraint $f(uv)=W\langle f(u),f(v)\rangle$ for integers $k\geq 0$ and $d\geq 1$.\qqed
\end{defn}

\begin{defn} \label{defn:w-magic-kd-total-vk-ed-labeling}
$^*$ A $(p,q)$-graph $G$ admits a \emph{$W$-constraint $(k,d)$-total $(vk$-$ed)$-coloring} $g$ defined as
\begin{equation}\label{eqa:vk-ed-labeling}
{
\begin{split}
&g:V(G)\rightarrow S_{n,k,0,d}=\{k,k+d, \dots, k+nd\}\\
&g:E(G)\rightarrow S_{m,0,0,d}=\{0,d,\dots, md\}
\end{split}}
\end{equation}
and each edge $uv\in E(G)$ holds the $W$-constraint $g(uv)=W\langle g(u),g(v)\rangle$ for integers $k\geq 0$ and $d\geq 1$. Especially, $g$ is called a \emph{$W$-constraint $(k,d)$-total $(vk$-$ed)$-all-labeling} if the vertex color set $f(V(G))=S_{p-1,k,0,d}$, the edge color set $f(E(G))=S_{p-1,0,0,d}$, and each edge $uv\in E(G)$ holds the $W$-constraint $f(uv)=W\langle f(u),f(v)\rangle$ for integers $k\geq 0$ and $d\geq 1$.\qqed
\end{defn}

\begin{defn} \label{defn:graceful-kd-total-model-coloring}
$^*$ For integers $k\geq 0$ and $d\geq 1$, a $(p,q)$-graph $G$ admits a \emph{graceful $(k,d)$-total model-coloring} $h$ is defined as
\begin{equation}\label{eqa:vk-ed-labeling}
h:V(G)\cup E(G)\rightarrow S_{m,0,0,d}\cup S_{q-1,k,0,d}=\{0,d,\dots, md\}\cup \{k,k+d, \dots, k+(q-1)d\}
\end{equation}
such that each edge $uv\in E(G)$ holds $h(uv)=|h(u)-h(v)~(\bmod^*~qd)|$, and the edge color set
$$
h(E(G))=S_{p-1,k,0,d}=\{k,k+d, \dots, k+(q-1)d\}
$$ where
$$
h(u)-h(v)~(\bmod^*~qd)=k+(q-r)d
$$ if $h(u)-h(v)=k-rd$ (or $h(u)-h(v)=-k+rd$) for $r\geq 2$.\qqed
\end{defn}

\begin{example}\label{exa:8888888888}
In Fig.\ref{fig:33-new-kd-colorings} (a) and (b), a tree $G_1$ admits a graceful-difference $(k,d)$-total labeling $h_1$ holding the graceful-difference constraint
$$
\big ||h_1(u)-h_1(v)|-h_1(uv)\big |=d,~uv\in E(G_1)
$$ and the edge color set $h_1(E(G_1))=S_{5,k,0,d}$; and a tree $G_2$ admits a graceful $(k,d)$-total model-coloring $h_2$ holding $h_2(uv)=|h_2(u)-h_2(v)~(\bmod^*~6d)|$ for each edge $uv\in E(G_2)$ and the edge color set $h_2(E(G_2))=S_{5,k,0,d}$, refer to Definition \ref{defn:graceful-kd-total-model-coloring}. Since $h_1(x)+h_2(x)=k+6d$ for each vertex $x\in V(G_1)=V(G_2)$, we say that $h_i$ is the \emph{totally vertex-dual coloring} of $h_{3-i}$ for $i=1,2$.

A colored tree $G_3$ shown in Fig.\ref{fig:33-new-kd-colorings} (c) admits an edge-difference $(k,d)$-total $(vd$-$ek)$-all-labeling $h_3$ holding the edge-difference constraint
$$h_3(uv)+|h_3(u)-h_3(v)|=k+7d,~uv\in E(G_3)
$$ according to Definition \ref{defn:w-magic-kd-total-vd-ek-coloring}. By Definition \ref{defn:w-magic-kd-total-vk-ed-labeling}, another colored tree $G_4$ shown in Fig.\ref{fig:33-new-kd-colorings} (d) admits a graceful-difference $(k,d)$-total $(vk$-$ed)$-all-labeling $h_4$ holding the graceful-difference constraint
$$\big ||h_4(u)-h_4(v)|-h_4(uv)\big |=0,~uv\in E(G_4)
$$

Notice that $|h_3(v)-h_4(v)|=k$ for each vertex $v\in V(G_3)=V(G_4)$, since $G_3\cong G_4$. So, $h_j$ is the \emph{totally vertex-dual coloring} of $h_{7-j}$ for $j=3,4$. Moreover, we have $h_3(xy)+h_4(xy)=k+7d$ for each edge $xy\in E(G_3)=E(G_4)$, then $h_j$ is the \emph{totally edge-dual coloring} of $h_{7-j}$ for $j=3,4$. Totally, $h_j$ is the \emph{total-dual coloring} of $h_{7-j}$ for $j=3,4$. \qqed
\end{example}

\begin{figure}[h]
\centering
\includegraphics[width=16.4cm]{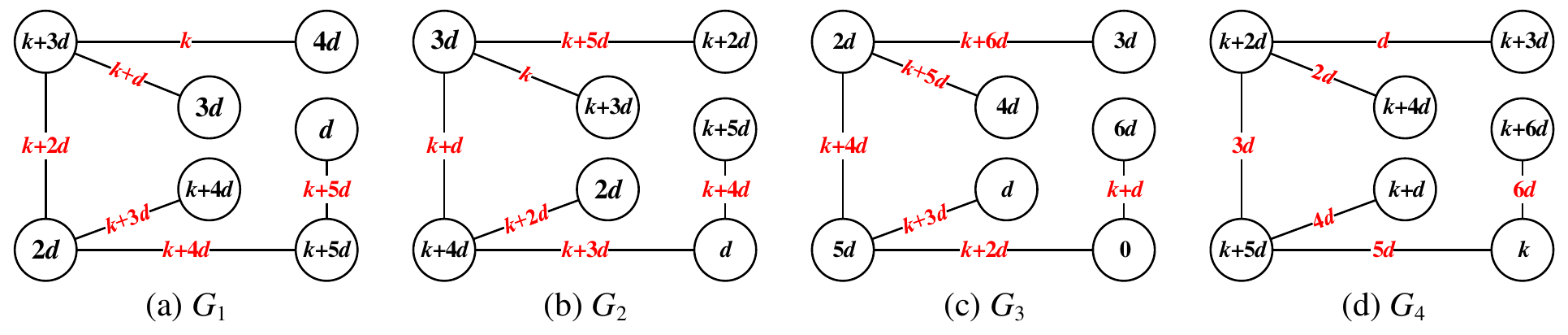}\\
\caption{\label{fig:33-new-kd-colorings}{\small (a) and (b) are two totally vertex-dual colorings holding $g_1(u)+g_2(u)=k+6d$; (b), (c) and (d) are three new $(k,d)$-total colorings/labelings defined in Definition \ref{defn:w-magic-kd-total-vd-ek-coloring}, Definition \ref{defn:w-magic-kd-total-vk-ed-labeling} and Definition \ref{defn:graceful-kd-total-model-coloring}.}}
\end{figure}

For the topological authentication, we present the following three definitions: Definition \ref{defn:vertex-dual-kd-total-coloring}, Definition \ref{defn:edge-dual-kd-total-coloring} and Definition \ref{defn:edge-dual-vertex-kd-total-coloring}.

\begin{defn} \label{defn:vertex-dual-kd-total-coloring}
$^*$ For integers $k\geq 0$ and $d\geq 1$, a $(p,q)$-graph $G$ (as a \emph{public-key}) admits a $W_i$-constraint $(k,d)$-total coloring $f$, and another $(p,q\,')$-graph $H$ (as a \emph{private-key}) admits a $W_j$-constraint $(k,d)$-total coloring $g$, and there is a bijection $\theta:V(G) \rightarrow V(H)$. If there are integers $a,r\geq 1$, such that $f(x)+g[\theta(x)]=a\cdot k+r\cdot d$ for each vertex $x\in V(G)$, then we call $f$ (resp. $g$) \emph{vertex-dual $(k,d)$-total coloring} of $g$ (resp. $f$).\qqed
\end{defn}

\begin{defn} \label{defn:edge-dual-kd-total-coloring}
$^*$ For integers $k\geq 0$ and $d\geq 1$, a $(p,q)$-graph $G$ (as a \emph{public-key}) admits a $W_i$-constraint $(k,d)$-total coloring $f$, and another $(p\,',q)$-graph $H$ (as a \emph{private-key}) admits a $W_j$-constraint $(k,d)$-total coloring $g$, and there is an edge bijection $\pi:E(G) \rightarrow E(H)$. If there are integers $b,r\geq 1$, such that $f(xy)+g[\pi(xy)]=b\cdot k+r\cdot d$ for each edge $xy\in E(G)$, then we call $f$ (resp. $g$) \emph{edge-dual $(k,d)$-total coloring} of $g$ (resp. $f$).\qqed
\end{defn}

\begin{defn} \label{defn:edge-dual-vertex-kd-total-coloring}
$^*$ For integers $k\geq 0$ and $d\geq 1$, a $(p,q)$-graph $G$ (as a \emph{public-key}) admits a $W_i$-constraint $(k,d)$-total coloring $f$, and another $(p\,',q)$-graph $H$ (as a \emph{private-key}) admits a $W_j$-constraint $(k,d)$-total coloring $g$, and there are a bijection $\theta:V(G) \rightarrow V(H)$ and an edge bijection $\phi:E(G) \rightarrow E(H)$. If there are integers $a,b,r,s\geq 1$, such that

$f(x)+g[\theta(x)]=a\cdot k+r\cdot d$ for each vertex $x\in V(G)$ and

$f(xy)+g[\phi(xy)]=b\cdot k+s\cdot d$ for each edge $xy\in E(G)$,\\
then $f$ is called (resp. $g$) \emph{$(v,e)$-dual $(k,d)$-total coloring} of $g$ (resp. $f$).\qqed
\end{defn}

\begin{example}\label{exa:8888888888}
In Fig.\ref{fig:33-dual-kd-totals}, we can see the following facts:
\begin{asparaenum}[\textrm{\textbf{Ex}}-1.]
\item Two trees $H_r\not\cong H_s$ for $r\neq s$ and $r,s\in [1,4]$.

\item Each tree $H_l$ for $l\in [1,3]$ admits a graceful $(k,d)$-total coloring $h_l$ holding $h_l(uv)=|h_l(u)-h_l(v)|$ for each edge $uv\in E(H_l)$.

\item Each tree $H_t$ for $t=4,5$ admits an edge-difference $(k,d)$-total coloring $h_t$ holding the edge-difference constraint
$$
h_t(uv)+|h_t(u)-h_t(v)|=2k+5d,\quad uv\in E(H_t)$$

\item There is a bijection $\theta_{1,i} : V(H_1)\rightarrow V(H_i)$ for $i\in [2,5]$, such that each vertex $x_j\in V(H_1)$ holds $\theta_{1,i}(x_j)=A_j$ for $j\in [1,7]$ and $A\in \{y,u,v,w\}$. Moreover, we get
$$h_1(x_j)+h_i[\theta_{1,i}(x_j)]=k+6d,~i\in [2,5],~j\in [1,7],~A\in \{y,u,v,w\}
$$ and claim that each coloring $h_i$ is a \emph{vertex-dual $(k,d)$-total coloring} of the graceful $(k,d)$-total coloring $h_1$ by Definition \ref{defn:vertex-dual-kd-total-coloring}.

\item About two trees $H_1$ and $H_2$, we have an edge bijection $\varepsilon_{1,2}:E(H_1)\rightarrow E(H_2)$, where $\varepsilon_{1,2}(x_1x_2)=y_4y_7$, $\varepsilon_{1,2}(x_2x_3)=y_4y_5$, $\varepsilon_{1,2}(x_2x_6)=y_6y_5$, $\varepsilon_{1,2}(x_5x_6)=y_2y_6$, $\varepsilon_{1,2}(x_6x_7)=y_2y_3$, $\varepsilon_{1,2}(x_4x_7)=y_1y_2$, such that each edge $x_ix_j\in E(H_1)$ holds
  $$h_1(x_ix_j)+h_2[\varphi_{1,2}(x_ix_j)]=2k+5d$$

\qquad Thereby, the coloring $h_2$ is an \emph{edge-dual $(k,d)$-total coloring} of $h_1$ according to Definition \ref{defn:edge-dual-kd-total-coloring}. Moreover, the coloring $h_2$ is an \emph{$(v,e)$-dual $(k,d)$-total coloring} of $h_1$ in term of Definition \ref{defn:edge-dual-vertex-kd-total-coloring}.\qqed
\end{asparaenum}
\end{example}

\begin{figure}[h]
\centering
\includegraphics[width=16.4cm]{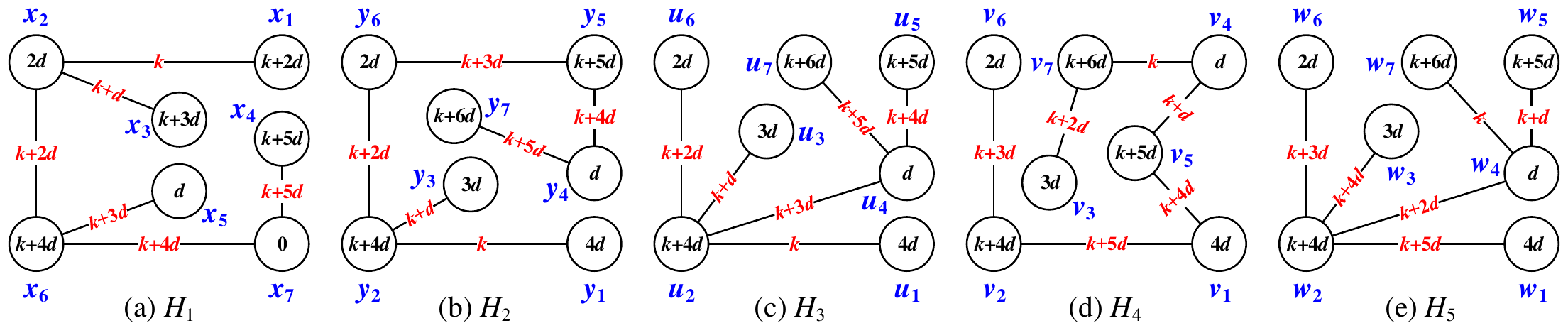}\\
\caption{\label{fig:33-dual-kd-totals}{\small Examples for illustrating Definition \ref{defn:vertex-dual-kd-total-coloring}, Definition \ref{defn:edge-dual-kd-total-coloring} and Definition \ref{defn:edge-dual-vertex-kd-total-coloring}.}}
\end{figure}

\begin{problem}\label{qeu:public-key-private-key}
For a given $(p,q)$-graph $G$ (as a \emph{public-key}) admitting a $W_i$-constraint $(k,d)$-total coloring $f$, \textbf{determine} another $(p,q\,')$-graph $H$ (as a \emph{private-key}) admitting a $W_j$-constraint $(k,d)$-total coloring $g$, such that $f$ (resp. $g$) is a \emph{vertex-dual $(k,d)$-total coloring} of $g$ (resp. $f$) defined in Definition \ref{defn:vertex-dual-kd-total-coloring}. \textbf{Do} the same works according to Definition \ref{defn:edge-dual-kd-total-coloring} and Definition \ref{defn:edge-dual-vertex-kd-total-coloring}, respectively.
\end{problem}

\begin{rem}\label{rem:333333}
About Problem \ref{qeu:public-key-private-key}, if $G=H$ is a tree admitting a set-ordered graceful labeling, then we can find the vertex-dual $(k,d)$-total coloring, the edge-dual $(k,d)$-total coloring and the $(v,e)$-dual $(k,d)$-total coloring. However, finding these colorings for $G\not \cong H$, or $G$ and $H$ being not trees in Problem \ref{qeu:public-key-private-key} is important and difficult in the real application of information security. \paralled
\end{rem}

\subsection{Labelings and results on disconnected graphs}

\begin{thm}\label{defn:group-flawed-labelings}
\cite{Yao-Mu-Sun-Sun-Zhang-Wang-Su-Zhang-Yang-Zhao-Wang-Ma-Yao-Yang-Xie2019} Let $G_1,G_2,\dots, G_m$ be disjoint connected graphs, and $E^*$ be an edge set such that each edge $uv$ of $E^*$ has one end $u$ in some $G_i$ and another end $v$ is in some $G_j$ with $i\neq j$, and $E^*$ joins $G_1,G_2,\dots, G_m$ together to form a connected graph $H$, denoted as $H=G+E^*$, where $G=\bigcup^m_{i=1}G_i$. We say $G=\bigcup^m_{i=1}G_i$ to be a disconnected $(p,q)$-graph with
$$
p=|V(H)|=\sum^m_{i=1}|V(G_i)|,\quad q=|E(H)|-|E^*|=\left (\sum^m_{i=1}|E(G_i)|\right )-|E^*|
$$ If $H$ admits an $\alpha$-labeling shown in the following:
\begin{asparaenum}[\emph{\textrm{Flawed}}-1. ]
\item $\alpha$ is a graceful labeling, or a set-ordered graceful labeling, or graceful-intersection total set-labeling, or a graceful group-labeling.
\item $\alpha$ is an odd-graceful labeling, or a set-ordered odd-graceful labeling, or an edge-odd-graceful total labeling, or an odd-graceful-intersection total set-labeling, or an odd-graceful group-labeling, or a perfect odd-graceful labeling.
\item $\alpha$ is an elegant labeling, or an odd-elegant labeling.
\item $\alpha$ is an edge-magic total labeling, or a super edge-magic total labeling, or super set-ordered edge-magic total labeling, or an edge-magic total graceful labeling.
\item $\alpha$ is an edge-antimagic $(k,d)$-total labeling, or a $(k, d)$-arithmetic.
\item $\alpha$ is a relaxed edge-magic total labeling.
\item $\alpha$ is an odd-edge-magic matching labeling, or an ee-difference odd-edge-magic matching labeling.
\item $\alpha$ is a 6C-labeling, or an odd-6C-labeling.
\item $\alpha$ is an ee-difference graceful-magic matching labeling.
\item $\alpha$ is a difference-sum labeling, or a felicitous-sum labeling.
\item $\alpha$ is a multiple edge-meaning vertex labeling.
\item $\alpha$ is a perfect $\varepsilon$-labeling.
\item $\alpha$ is an image-labeling, or a $(k,d)$-harmonious image-labeling.
\item $\alpha$ is a twin $(k,d)$-labeling, or a twin Fibonacci-type graph-labeling, or a twin odd-graceful labeling.
\end{asparaenum}

Then $G$ admits a \emph{flawed $\alpha$-labeling} too.\qqed
\end{thm}

\begin{thm}\label{thm:flawed-graceful-labeling-forest}
\cite{Yao-Mu-Sun-Sun-Zhang-Wang-Su-Zhang-Yang-Zhao-Wang-Ma-Yao-Yang-Xie2019} Let $G=\bigcup^m_{i=1} G_i$ be a forest with components $G_1,G_2,\dots ,G_m$, where each $G_i$ is a $(p_i,p_i-1)$-tree with vertex partition $(X_i,Y_i)$ and admits a set-ordered graceful labeling $f_i$ holding $\max f_i(X_i)<\min f_i(Y_i)$ true for $i\in [1,m]$. Then $G$ admits a flawed set-ordered graceful labeling.
\end{thm}
\begin{proof}
Suppose that the vertex partitions $(X_i,Y_i)$ consists of $X_i=\{x_{i,j}:~j\in[1,s_{i}]\}$ and $Y_i=\{y_{i,j}:~j\in[1,t_{i}]\}$ with $i\in [1,m]$. Let $M=\sum ^{m}_{k=1}s_{k}$ and $N=\sum ^{m}_{k=1}t_{k}$. Furthermore,
$$
0\leq \max f_i(X_i)<\min f_i(x_{i,j+1})\leq s_{i}-1,~j\in[1,s_{i}]
$$ and
$$s_{i}\leq f_i(y_{i,j})<f_i(y_{i,j+1})\leq s_{i}+t_{i}-1=p_i-1,~j\in[1,t_{i}],~i\in[1,m]$$

We defined a new labeling $g$ for the forest $G$ as follows:

(1) $g(x_{1,j})=f_1(x_{1,j})$ with $j\in[1,s_{1}]$.

(2) $g(x_{i,j})=f_i(x_{i,j})+\sum ^{i-1}_{k=1}s_{k}$ for $j\in[1,s_{i}]$ and $i\in[2,m]$.

(3) $g(y_{m,j})=f_m(y_{m,j})+M-s_m$ with $j\in[1,t_{m}]$.

(4) $g(y_{\ell,j})=f_{\ell}(y_{\ell,j})+M+\sum ^{m-\ell}_{k=1}t_{k}$ for $j\in[1,t_{\ell}]$ and $\ell\in[1,m-1]$.

Clearly,
\begin{equation}\label{eqa:flawed-labelings-00}
0=g(x_{1,1})\leq g(u)<g(v)\leq g(y_{1,t_{1}})=M+N-1
\end{equation} for $u\in \bigcup^m_{i=1} X_i$ and $v\in \bigcup^m_{i=1} Y_i$.

Notice that each edge color set
$$f_i(E(G_i))=[1,p_i-1]=\{f_{i}(y_{i,a})-f_{i}(x_{i,b}):~x_{i,b}y_{i,a}\in E(G_i)\}
$$ We compute edge labels $g(x_{i,b}y_{i,a})=g(y_{i,a})-g(x_{i,b})$ as follows: For each edge $x_{m,b}y_{m,a}\in E(G_m)$, we have
\begin{equation}\label{eqa:c3xxxxx}
{
\begin{split}
g(x_{m,b}y_{m,a})&=g(y_{m,a})-g(x_{m,b})=[f_m(y_{m,a})+M-s_m]-\left [f_m(x_{m,b})+\sum ^{m-1}_{k=1}s_{k}\right ]\\
&=f_m(y_{m,a})-f_m(x_{m,b}),
\end{split}}
\end{equation}
so the edge color set $g(E(G_m))=[1,p_m-1]$. Next, for each edge $x_{i,b}y_{i,a}\in E(G_i)$ with $i\in [1,m-1]$, we can compute
\begin{equation}\label{eqa:c3xxxxx}
{
\begin{split}
g(x_{i,b}y_{i,a})&=g(y_{i,a})-g(x_{i,b})=\left [f_{i}(y_{i,j})+M+\sum ^{m-i}_{k=1}t_{k}\right ]-\left [f_i(x_{i,j})+\sum ^{i-1}_{k=1}s_{k}\right ]\\
&=f_{i}(y_{i,j})-f_i(x_{i,j})+\sum ^{m}_{k=i}s_{k}+\sum ^{m-i}_{k=1}t_{k}
\end{split}}
\end{equation}
So, we obtain the edge color set
\begin{equation}\label{eqa:c3xxxxx}
{
\begin{split}
&\quad g(E(G_i))=[1+M_i, p_i-1+M_i]
\end{split}}
\end{equation} where $M_i=\sum ^{m}_{k=i}s_{k}+\sum ^{m-i}_{k=1}t_{k}$. Finally, we get the whole edge color set
\begin{equation}\label{eqa:c3xxxxx}
g(E(G))=\bigcup ^m_{i=1}g(E(G_i))=\left [1,\bigcup^m_{i=1}(p_i-1)\right ] \setminus g(E^*)
\end{equation} We claim that $g$ is a flawed set-ordered graceful labeling of the forest $G$.
\end{proof}

\begin{cor}\label{thm:flawed-graceful-labeling-graph}
\cite{Yao-Mu-Sun-Sun-Zhang-Wang-Su-Zhang-Yang-Zhao-Wang-Ma-Yao-Yang-Xie2019} Let $G$ be a disconnected graph with components $H_1,H_2,\dots ,H_m$, where each $H_i$ is a connected bipartite $(p_i,q_i)$-graph with vertex bipartition $(X_i,Y_i)$ and admits a set-ordered graceful labeling $f_i$ holding $\max f_i(X_i)<\min f_i(Y_i)$ with $i\in [1,m]$. Then $G$ admits a flawed set-ordered graceful labeling.
\end{cor}

In \cite{Yao-Liu-Yao-2017}, the authors have proven the following mutually equivalent labelings:

\begin{thm} \label{thm:connections-several-labelings}
\cite{Yao-Liu-Yao-2017} Let $T$ be a tree on $p$ vertices, and let $(X,Y)$ be its
bipartition of vertex set $V(T)$. For all values of integers $k\geq 1$ and $d\geq 1$, the following assertions are mutually equivalent:

$(1)$ $T$ admits a set-ordered graceful labeling $f$ with $\max f(X)<\min f(Y)$.

$(2)$ $T$ admits a super felicitous labeling $\alpha$ with $\max \alpha(X)<\min \alpha(Y)$.

$(3)$ $T$ admits a $(k,d)$-graceful labeling $\beta$ with
$\beta(x)<\beta(y)-k+d$ for all $x\in X$ and $y\in Y$.

$(4)$ $T$ admits a super edge-magic total labeling $\gamma$ with $\max \gamma(X)<\min \gamma(Y)$ and a magic constant $|X|+2p+1$.

$(5)$ $T$ admits a super $(|X|+p+3,2)$-edge antimagic total labeling $\theta$ with $\max \theta(X)<\min \theta(Y)$.

$(6)$ $T$ admits an odd-elegant labeling $\eta$ with $\eta(x)+\eta(y)\leq 2p-3$ for every edge $xy\in E(T)$.

$(7)$ $T$ admits a $(k,d)$-arithmetic labeling $\psi$ with $\max \psi(x)<\min \psi(y)-k+d\cdot |X|$ for all $x\in X$ and $y\in Y$.

$(8)$ $T$ admits a harmonious labeling $\varphi$ with $\max \varphi(X)<\min \varphi(Y\setminus \{y_0\})$ and $\varphi(y_0)=0$.
\end{thm}

We have some results similarly with that in Theorem \ref{thm:connections-several-labelings} about flawed graph labelings as follows:

\begin{thm} \label{thm:connection-flawed-labelings}
\cite{Yao-Mu-Sun-Sun-Zhang-Wang-Su-Zhang-Yang-Zhao-Wang-Ma-Yao-Yang-Xie2019} Suppose that $T=\bigcup ^m_{i=1}T_i$ is a forest having disjoint trees $T_1,T_2,\dots ,T_m$, and $(X,Y)$ be its vertex
bipartition. For all values of integers $k\geq 1$ and $d\geq 1$, the following assertions are mutually equivalent:
\begin{asparaenum}[F-1. ]
\item $T$ admits a flawed set-ordered graceful labeling $f$ with $\max f(X)<\min f(Y)$;
\item $T$ admits a flawed set-ordered odd-graceful labeling $f$ with $\max f(X)<\min f(Y)$;

\item $T$ admits a flawed set-ordered elegant labeling $f$ with $\max f(X)<\min f(Y)$;

\item $T$ admits a flawed odd-elegant labeling $\eta$ with $\eta(x)+\eta(y)\leq 2p-3$ for each edge $xy\in E(T)$.

\item $T$ admits a super flawed felicitous labeling $\alpha$ with $\max \alpha(X)<\min \alpha(Y)$.

\item $T$ admits a super flawed edge-magic total labeling $\gamma$ with $\max \gamma(X)<\min \gamma(Y)$ and a magic constant $|X|+2p+1$.
\item $T$ admits a super flawed $(|X|+p+3,2)$-edge antimagic total labeling $\theta$ with $\max \theta(X)<\min \theta(Y)$.

\item $T$ admits a flawed harmonious labeling $\varphi$ with $\max \varphi(X)<\min \varphi(Y\setminus \{y_0\})$ and $\varphi(y_0)=0$.
\end{asparaenum}
\end{thm}

We present some equivalent definitions with parameters $k,d$ for flawed $(k,d)$-labelings.

\begin{thm} \label{thm:flawed-kk-dd-labelings}
\cite{Yao-Mu-Sun-Sun-Zhang-Wang-Su-Zhang-Yang-Zhao-Wang-Ma-Yao-Yang-Xie2019} Let $T=\bigcup ^m_{i=1}T_i$ be a forest having disjoint trees $T_1,T_2,\dots ,T_m$, and $V(T)=X\cup Y$. For all values of two integers $k\geq 1$ and $d\geq 1$, the following assertions are mutually equivalent:
\begin{asparaenum}[\textbf{\textrm{Flawl}}-1. ]
\item $T$ admits a flawed set-ordered graceful labeling $f$ with $\max f(X)<\min f(Y)$.

\item $T$ admits a flawed $(k,d)$-graceful labeling $\beta$ with $\max \beta(x)<\min \beta(y)-k+d$ for all $x\in X$ and $y\in Y$.

\item $T$ admits a flawed $(k,d)$-arithmetic labeling $\psi$ with $\max \psi(x)<\min \psi(y)-k+d\cdot |X|$ for all $x\in X$ and $y\in Y$.
\item $T$ admits a flawed $(k,d)$-harmonious labeling $\varphi$ with $\max \varphi(X)<\min \varphi(Y\setminus \{y_0\})$ and $\varphi(y_0)=0$.
\end{asparaenum}
\end{thm}

\begin{problem}\label{qeu:444444}
It is meaningful to \textbf{find} out $(k,d)$-type colorings/labelings based on those colorings/labelings mentioned in Theorem \ref{defn:group-flawed-labelings}, since there are over 30 flawed graph labelings introduced in Theorem \ref{defn:group-flawed-labelings}.

Clearly, the operation $G+E^*$ produces two or more connected graphs by the different ways of joining $G_1,G_2,\dots, G_m$ together, so we get a graph set $S_{et}(G[\ominus]\{G_i\}^m_{i=1})$ such that each graph $H=G+E^*\in S_{et}(G[\ominus]\{G_i\}^m_{i=1})$ is connected and admits a $W$-constraint labeling mentioned in Theorem \ref{defn:group-flawed-labelings} when each $G_i$ of disjoint graphs $G_1,G_2,\dots, G_m$ is a connected bipartite $(p_i,q_i)$-graph with vertex bipartition $(X_i,Y_i)$ and admits a set-ordered graceful labeling $f_i$ holding $\max f_i(X_i)<\min f_i(Y_i)$ with $i\in [1,m]$, according to Theorem \ref{thm:flawed-graceful-labeling-forest}, Corollary \ref{thm:flawed-graceful-labeling-graph}, and Theorem \ref{thm:connections-several-labelings}. \textbf{Show} the topological structure and mathematical properties of this graph set $S_{et}(G[\ominus]\{G_i\}^m_{i=1})$.
\end{problem}

\begin{rem}\label{rem:333333}
A group of colored Hanzi-graphs $H_{1643}$, $H_{3907}$, $H_{3279}$, $H_{3423}$, $H_{2645}$, $H_{5238}$ is shown in Fig.\ref{fig:Hanzi-graph-00}, and they admit set-ordered graceful labelings/colorings. Hanzi-graphs are natural graphs and admit many $W$-constraint labelings/colorings mentioned in Theorem \ref{defn:group-flawed-labelings}, Theorem \ref{thm:flawed-graceful-labeling-forest}, Corollary \ref{thm:flawed-graceful-labeling-graph}, Theorem \ref{thm:connections-several-labelings}, Theorem \ref{thm:connection-flawed-labelings} and Theorem \ref{thm:flawed-kk-dd-labelings}.

Suppose that each Hanzi-graph $H_i$ of disjoint Hanzi-graphs $H_{1}$, $H_{2}$, $\dots$, $H_{n}$ is a connected bipartite $(p_i,q_i)$-graph with vertex bipartition $(X_i,Y_i)$ and admits a group of colorings $g_{i,j}$ with $j\in [1,a_i]$ and $i\in [1,n]$, corresponding to a group of Topcode-matrices $T_{code}(H_i,g_{i,1})$, $T_{code}(H_i,g_{i,2})$, $\dots $, $T_{code}(H_i,g_{i,a_i})$, where each Topcode-matrix $T_{code}(H_i,g_{i,1})$ is $3\times q_i$ and produces $(3q_i)!$ number-based strings, we put them into a \emph{string set} $S_{tring}(H_i,g_{i,j})$ with $j\in [1,a_i]$ and $i\in [1,n]$.

\textbf{The complex analysis.} We take randomly take a number-based string $s(r_{i,j})\in S_{tring}(H_i,g_{i,j})$ with $r_{i,j}\in [1,(3q_i)!]$, and make a large scale of number-based string as follows
\begin{equation}\label{eqa:Hanzi-number-based-strings}
S(\{H_i\}^n_{i=1},\{r_{i,j_t}\}^n_{i=1})=s(r_{1,j_1})s(r_{2,j_2})\cdots s(r_{n,j_n})
\end{equation} for $j_t\in [1,a_t]$, $r_{i,j_t}\in [1,(3q_t)!]$ and $t\in [1,n]$. If we do not think about the permutations on the strings $s(r_{1,j_1}),s(r_{2,j_2}),\dots ,s(r_{n,j_n})$ and Hanzi-graphs $H_{1}$, $H_{2}$, $\dots$, $H_{n}$, we get $\prod^n_{i=1}(a_i)!\prod^n_{i=1}(q_i)!$ different number-based strings like $S(\{H_i\}^n_{i=1},\{r_{i,j_t}\}^n_{i=1})$ shown in Eq.(\ref{eqa:Hanzi-number-based-strings}). Otherwise, we have at least
\begin{equation}\label{eqa:Hanzi-number-based-string-complex}
n_{string}(\{H_i\}^n_{i=1})=(n!)^2\prod^n_{i=1}(a_i)!\prod^n_{i=1}(q_i)!
\end{equation} different number-based strings $S(\{H_i\}^n_{i=1},\{r_{i,j_t}\}^n_{i=1})$ define in Eq.(\ref{eqa:Hanzi-number-based-strings}).

A Chinese sentence $H_{1643}$$H_{3907}$$H_{3279}$$H_{3423}$$H_{2645}$$H_{5238}$ consists of disjoint connected Hanzi-graphs $H_{1}$, $H_{2}$, $\dots$, $H_{13}$ with the edge numbers $q_1=$2, $q_2=$2, $q_3=$6, $q_4=$18, $q_5=$4, $q_6=$7, $q_7=$6, $q_8=$4, $q_9=$7, $q_{10}=$1, $q_{11}=$4, $q_{12}=$1, $q_{13}=$7. This group of connected Hanzi-graphs provide us
$$n_{string}(\{H_i\}^{13}_{i=1})=(13!)^2\prod^{13}_{i=1}(a_i)!\prod^{13}_{i=1}(q_i)!
$$ different number-based strings, where each $a_i\geq 30$ by Theorem \ref{defn:group-flawed-labelings}, Theorem \ref{thm:flawed-graceful-labeling-forest}, Corollary \ref{thm:flawed-graceful-labeling-graph}, Theorem \ref{thm:connections-several-labelings}, Theorem \ref{thm:connection-flawed-labelings} and Theorem \ref{thm:flawed-kk-dd-labelings}.

The above work illustrates the fact that Chinese speakers can generate their own passwords made by long byte number-based strings through speaking and writing their own easy-to-remember Chinese sentences, in response to cryptographic security in the age of supercomputers and quantum computers. \paralled
\end{rem}

\begin{figure}[h]
\centering
\includegraphics[width=16.4cm]{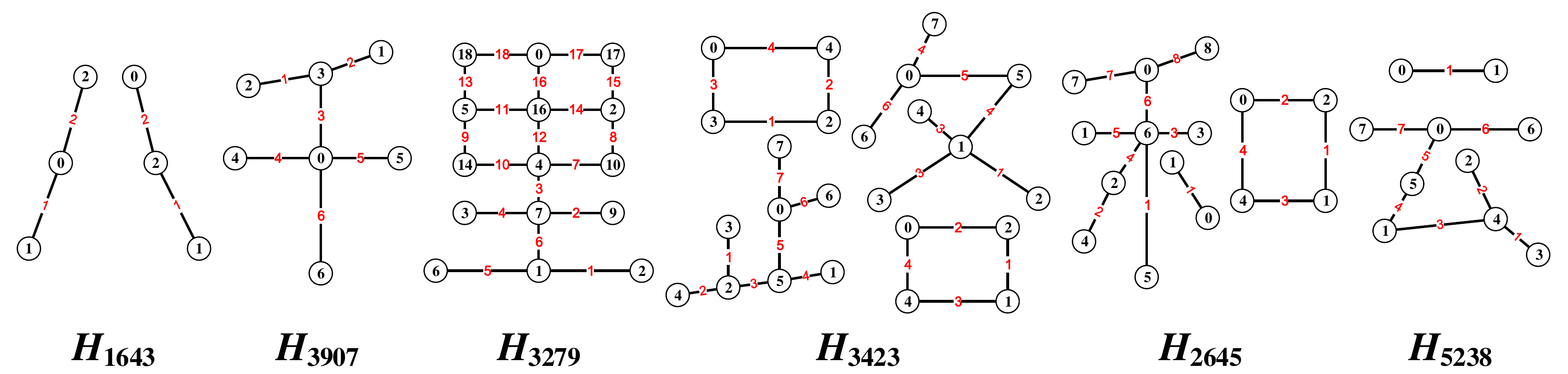}
\caption{\label{fig:Hanzi-graph-00}{\small A sentence $H_{1643}$$H_{3907}$$H_{3279}$$H_{3423}$$H_{2645}$$H_{5238}$ in Chinese \cite{GB2312-80}.}}
\end{figure}

\section{Graphs admitting parameterized total colorings}

\subsection{Graphs admitting $W$-constraint $(k,d)$-total colorings}

\begin{example}\label{exa:bipartite-graph-graceful-kd-total}
It is not hard to find a graceful $(k,d)$-total coloring of a \emph{complete bipartite graph} $K_{m,n}$. Let $V(K_{m,n})=X\cup Y$ with $X=\{x_j:j\in [1,m]\}$ and $Y=\{y_i:i\in [1,n]\}$ be the vertex sets of $K_{m,n}$, and $E(K_{m,n})=\{x_jy_i:j\in [1,m],i\in [1,n]\}$ be the edge set of $K_{m,n}$.

We define a $(k,d)$-total coloring $f$ for $K_{m,n}$ in the following way:

(i) $f(x_j)=(j-1)d$ with $j\in [1,m]$;

(ii) $f(y_i)=k+(m-1)d+m(i-1)d$ with $i\in [1,n]$; and

(iii) $f(x_jy_i)=f(y_i)-f(x_j)=k+(mi-j)d$ with $j\in [1,m]$.\\
Clearly, the edge color set $f(E(K_{m,n}))=S_{mn-1,k,0,d}$, so we claim that $f$ is a graceful $(k,d)$-total coloring of $K_{m,n}$.

See examples shown in Fig.\ref{fig:22-G-book}, where $K_{2,3}=G_1$, $K_{2,4}=G_4$, $K_{2,2}=G_5$, and $K_{2,1}=G_6$.\qqed
\end{example}

\begin{figure}[h]
\centering
\includegraphics[width=16.4cm]{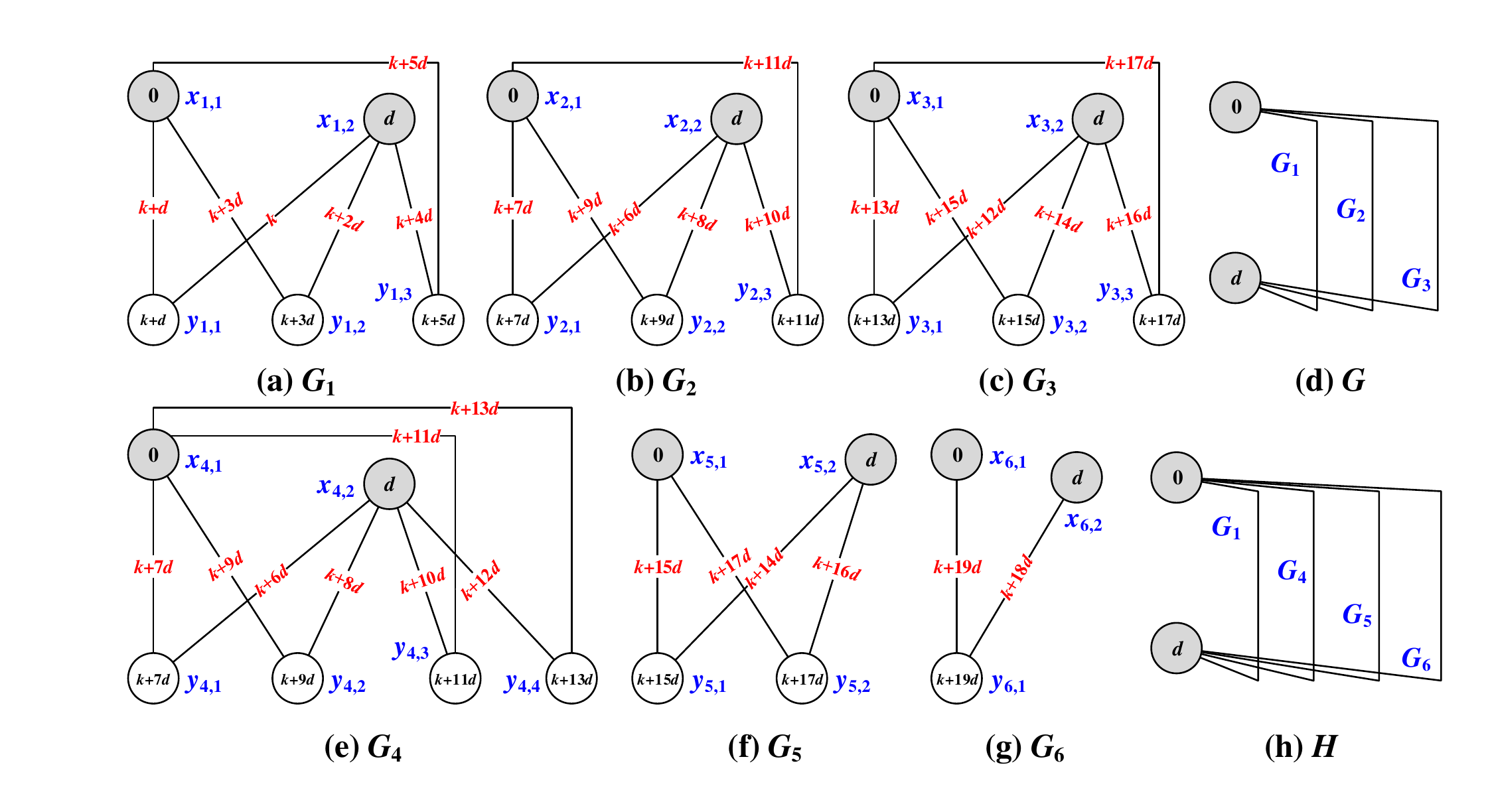}
\caption{\label{fig:22-G-book}{\small Two $W$-constraint colored $\{T_i\}^m_1$-books: (d) $G=(G_1\odot_2G_2)\odot_2G_3$; (h) $H=[(G_1\odot_2G_4)\odot_2G_5]\odot_2G_6$.}}
\end{figure}

\begin{example}\label{exa:regular-graph-book-expression}
\textbf{The regular graph-books.} A bipartite and connected $(p,q)$-graph $G$ admits a graceful $(k,d)$-total coloring $f$ such that the $X$-vertex color set $f(X)\subseteq S_{m,k,0,d}$, the $Y$-vertex color set $f(Y)\subseteq S_{q-1,k,0,d}$ and the edge color set $f(E(G))=S_{q-1,k,0,d}$, where $(X,Y)$ is the bipartition of $V(G)$, $X=\{x_i:i\in [1,s]\}$ and $Y=\{y_j:j\in [1,t]\}$ holding $s+t=p$.

Each copy $G_j$ of $G$ has its own bipartition $(X_j,Y_j)$, such that $X_j=\{x_{j,i}:i\in [1,s]\}$ and $Y_j=\{y_{j,i}:i\in [1,t]\}$ with $j\in [1,n]$. We define a total coloring $f_j$ of $G_j$ as:

(1) $f_j(x_{j,i})=f(x_i)$ for $x_i\in X$ and $x_{j,i}\in X_j$ with $i\in [1,s]$;

(2) $f_j(y_{j,i})=f(y_i)+(j-1)qd$ for $y_i\in Y$ and $y_{j,i}\in Y_j$ with $i\in [1,t]$;

(3) The edge $x_ry_i\in E(G)$ and the edge $x_{j,r}y_{j,i}\in E(G_j)$ hold
$$f_j(x_{j,r}y_{j,i})=f_j(y_{j,i})-f_j(x_{j,i})=f(y_i)+(j-1)qd-f(x_r)$$
Obviously, we have edge colors
$$
f_j(x_{j,r}y_{j,i})\in\{k+(j-1)qd,k+[(j-1)q+1]d,\dots , k+(jq-1)d\},~j\in [1,n]
$$ Then, we vertex-coincide these vertices $x_i,x_{1,i},x_{2,i},\dots , x_{n,i}$ into one vertex
$$
x_i=x_{1,i}\odot x_{2,i}\odot \cdots \odot x_{n,i}=x_1\odot \left ([\odot ]^n_{j=1}x_{j,i}\right )\in X
$$ The resulting graph is called \emph{regular graph-book}, denoted as
\begin{equation}\label{eqa:regular-graph-book-expression}
R_{book}=[\odot_X]^n_{j=1}G_j
\end{equation} the vertex set $X$ of $G$ is called \emph{spine}, and each colored graph $G_i$ is called \emph{book-page}.

Thereby, this regular graph-book admits a graceful $(k,d)$-total coloring $F$ defined by the compound of the colorings $f_1,f_2,\dots ,f_n$. See a regular graph-book $G$ shown in Fig.\ref{fig:22-G-book} (d), where $G$ admits a graceful $(k,d)$-total coloring.\qqed
\end{example}

\begin{example}\label{exa:general-graph-book-expression}
\textbf{The general graph-books.} Suppose that each bipartite and connected $(p_j,q_j)$-graph $H_j$ for $j\in [1,n]$ has its own vertex bipartition $(X_j,Y_j)$ as
$$X_j=\{u_{j,i}:i\in [1,s_j]\},~Y_j=\{v_{j,r}:r\in [1,t_j]\},~p_j=s_j+t_j
$$ and admits a $W$-constraint $(k,d)$-total coloring $g_j$ defined by
\begin{equation}\label{eqa:graph-book-expression11}
{
\begin{split}
&g_j:X_j\rightarrow S_{m_j,0,0,d}=\{0,d,\dots ,m_jd\}\\
&g_j:Y_j\cup E(G_j)\rightarrow S_{n_j,k,0,d}=\{k,k+d,\dots ,k+n_jd\}
\end{split}}
\end{equation} and holding the $W$-constraint $g_j(u_{j,i}v_{j,r})=W\langle g_j(u_{j,i}), g_j(v_{j,r})\rangle$ for each edge $u_{j,i}v_{j,r}\in E(H_j)$.

For constructing general graph-books, we set $s=s_j$ for $j\in [1,n]$ and $g_a(u_{a,i})=g_b(u_{b,i})$ for $i\in [1,s]$ and $a,b\in [1,n]$. We vertex-coincide the vertices $u_{1,i},u_{2,i},\dots ,u_{n,i}$ into one vertex
\begin{equation}\label{eqa:general-general-graph-book-vertex}
w_i=u_{1,i}\odot u_{2,i}\odot \cdots \odot u_{n,i}=[\odot]^n_{j=1}u_{j,i}
\end{equation} with $i\in [1,s]$, the resultant graph is just a \emph{general graph-book} denoted as
\begin{equation}\label{eqa:general-general graph-book}
G_{book}=H_1\odot _sH_2\odot_s \cdots \odot_s H_n=[\odot_s]^n_{j=1}H_j
\end{equation}

Now, we define a $(k,d)$-total coloring $F$ for the general graph-book $G_{book}$ by the following algorithm:

\textbf{Step 1.} $F(w_i)=g_1(u_{1,i})$ for $i\in [1,s]$, where $u_{1,i}\in X_1\subset V(H_1)$ and $w_i\in V(G_{book})$.

\textbf{Step 2.} Each vertex $v_{1,r}\in Y_1\subset V(H_1)$ is colored as $F(v_{1,r})=g_1(v_{1,r})$ for $r\in [1,t_1]$, and each edge $u_{1,i}v_{1,r}\in E(H_1)$ is colored by $F(u_{1,i}v_{1,r})=g_1(u_{1,i}v_{1,r})$.

\textbf{Step 3.} Each vertex $v_{j,r}\in Y_j\subset V(H_j)$ with $j\in [2,n]$ is colored as
\begin{equation}\label{eqa:555555}
F(v_{j,r})=g_j(v_{j,r})+\sum ^{j-1}_{i=1}q_i\cdot d,~r\in [1,t_j]
\end{equation} By Eq.(\ref{eqa:general-general-graph-book-vertex}), we have each edge $w_{i}v_{j,r}=u_{j,i}v_{j,r}\in E(H_j)\subset E(G_{book})$, so the edge color is
\begin{equation}\label{eqa:555555}
F(w_{i}v_{j,r})=F(u_{j,i}v_{j,r})=g_j(u_{j,i}v_{j,r})+\sum ^{j-1}_{i=1}q_i\cdot d
\end{equation}

\textbf{Step 3.} The coloring $F$ holds the $W$-constraint $F(xy)=W\langle F(x), F(y)\rangle$ for $xy\in E(G_{book})$.

\textbf{Step 4.} Return the $W$-constraint $(k,d)$-total coloring $F$ of the general graph-book $G_{book}$.

\vskip 0.2cm

See a general graph-book $H$ shown in Fig.\ref{fig:22-G-book} (h).\qqed
\end{example}

\begin{defn} \label{defn:general-graph-book-defn}
$^*$ We write a general graph-book $G_{book}=B\langle S_X,W,\{H_j\}^M_{j=1}\rangle$ for the convenience of further study, such that

(i) $S_X$ is the \emph{spine};

(ii) each \emph{book-page} $H_j$ is a bipartite and connected graph with its vertex bipartition $V(H_j)=(S_X,Y_j)$ for $j\in [1,M]$;

(iii) edge number $q_j=|E(H_j)|$ od each book-page $H_j$ is a \emph{book-page weight}, and the sum $\sum^M_{j=1}q_j$ is the \emph{volume} of the general graph-book;

(iv) two book-pages $H_i$ and $H_j$ hold $V(H_i)\cap V(H_j)=S_X$ for $i\neq j$.\qqed
\end{defn}

\begin{rem}\label{rem:333333}
We say that the general graph-book $B\langle S_X,W,\{H_j\}^M_{j=1}\rangle$ defined in Definition \ref{defn:general-graph-book-defn} to be \emph{Fibonacci graph-book} if the \emph{book-page weight sequence} $\{q_j\}^M_{j=1}$ is the \emph{Fibonacci sequence}. And moreover, a general graph-book $B\langle S_X,W,\{H_j\}^M_{j=1}\rangle$ is called \emph{$Q$-graph-book} if $\{q_j\}^M_{j=1}$ is a \emph{$Q$-sequence}, in general.

In the topological authentication of view, a general graph-book $B\langle S_X,W,\{H_j\}^M_{j=1}\rangle$ is just a topological authentication, such that each book-page $H_j$ is just a \emph{private-key}.\paralled
\end{rem}

\textbf{Graph-book lattices.} In a \emph{base} $\textbf{\textrm{T}}=(T_1,T_2,\dots ,T_n)$, suppose that each bipartite and connected $(p_j,q_j)$-graph $T_j\in \textbf{\textrm{T}}$ for $j\in [1,n]$ admits a $W$-constraint total coloring (or $(k,d)$-total coloring) $g_j$ defined in Eq.(\ref{eqa:graph-book-expression11}). The vertex bipartition $(X_i,Y_i)=V(T_i)$ of each graph $T_i$ holds $X_i=S_X$ true for $i\in [1,n]$ according to Definition \ref{defn:general-graph-book-defn}. Let $H_1,H_2,\dots ,H_M$ in a general graph-book $B\langle S_X,W,\{H_j\}^M_{j=1}\rangle$ defined in Definition \ref{defn:general-graph-book-defn} be a permutation of graphs $a_1T_1$, $a_2T_2$, $\dots $, $a_nT_n$ with $M=\sum ^n_{k=1}a_k\geq 1$, so we rewrite this general graph-book as
$$B\langle S_X,W,\{H_j\}^M_{j=1}\rangle=B\langle S_X,W,[\odot_X]^n_{k=1}a_kT_k\rangle
$$ Notice that there are $M!$ permutations, in total. We get a general uniform-$W$-constraint graph-book lattice
\begin{equation}\label{eqa:general-graph-book-lattice}
\textbf{\textrm{L}}(Z^0[\odot_X]\textbf{\textrm{T}},W)=\big \{B\langle S_X,W,[\odot_X]^n_{k=1}a_kT_k\rangle:~a_k\in Z^0,T_k\in \textbf{\textrm{T}}\big \}
\end{equation} under the base $\textbf{\textrm{T}}=(T_1,T_2,\dots ,T_n)$.

If each bipartite and connected $(p_j,q_j)$-graph $T_j$ of the base $\textbf{\textrm{T}}=(T_1,T_2,\dots ,T_n)$ admits a $W_t$-constraint total coloring (or $(k,d)$-total coloring) $f_{j,t}$ for $t\in [1,A]$, then we get a group of general uniform-$W_t$-constraint graph-book lattices $\textbf{\textrm{L}}(Z^0[\odot_X]\textbf{\textrm{T}},W_t)$ for each $t\in [1,A]$.

\begin{problem}\label{qeu:444444}
Does there exist a general graph-book lattice $\textbf{\textrm{L}}(Z^0[\odot_X]\textbf{\textrm{T}},\{W_t\}^A_{i=1})$ as each bipartite and connected $(p_j,q_j)$-graph $T_j\in \textbf{\textrm{T}}$ admitting a $W_j$-constraint total coloring (or $(k,d)$-total coloring) $g_j$ for $j\in [1,n]$, and $W_j\neq W_i$ for some $j\neq i$ and $1\leq j,i\leq A$?
\end{problem}

\subsection{Graphs admitting graceful $(k,d)$-total colorings}

\begin{thm}\label{thm:adding-leaves-keep-total-coloring}
\cite{Yao-Sun-Hongyu-Wang-n-dimension-ICIBA2020} Suppose that a bipartite and connected $(p,q)$-graph $G$ admits a graceful $(k,d)$-total coloring defined in Definition \ref{defn:kd-w-type-colorings}, adding randomly leaves to $G$ produces a bipartite and connected $(p\,',q\,')$-graph admitting a graceful $(k,d)$-total coloring.
\end{thm}
\begin{proof} Let $G$ be a bipartite and connected $(p,q)$-graph with its vertex set $V(G)=X\cup Y$ with $X\cap Y=\emptyset$, and $X=\{x_i:i\in [1,s]\}$ and $Y=\{y_j:j\in [1,t]\}$ holding $s+t=p$, as well as the edge set $E(G)=\{e_i:i\in [1,q]\}$.

Suppose that $G$ admits a graceful $(k,d)$-total coloring $f$ defined in Definition \ref{defn:kd-w-type-colorings}, so we have
$$f:X\rightarrow S_{m,0,0,d},~f:Y\cup E(G)\rightarrow S_{q-1,k,0,d},~k\geq 1
$$ and moreover $0=f(x_1)\leq f(x_i)\leq f(x_{i+1})$ for $i\in [1,s-1]$ and
$$f(y_j)\leq f(y_{j+1})\leq f(y_{t})=k+(q-1)d,~j\in [1,t-1]
$$ as well as the edge color set $f(E(G))=S_{q-1,k,0,d}$.

Adding randomly leaves to the vertices $x_i$ and $y_j$ of $G$ produces leaf sets
$$L_{eaf}(x_i)=\{x_{i,1},x_{i,2},\dots ,x_{i,a_i}\},~i\in [1,s]\textrm{ and }L_{eaf}(y_j)=\{y_{j,1},y_{j,2},\dots ,y_{j,b_j}\},~j\in [1,t]
$$ Here, it is allowed that $L_{eaf}(x_i)=\emptyset$ (or $L_{eaf}(y_j)=\emptyset$) if some vertex $x_i$ (or some vertex $y_j$) is not joined with any new added vertex. So, the leaf-added graph $G+L_{eaf}(w)$ has its own vertex set
$$
V(G+L_{eaf}(w))=X\,'\cup Y\,',~X\,'=X\bigcup \left (\bigcup^t_{j=1} L_{eaf}(y_j)\right ),~Y\,'=Y\bigcup \left (\bigcup^s_{i=1} L_{eaf}(x_i)\right )
$$ and the new leaf set
$$L_{eaf}(w)=\left (\bigcup^s_{i=1} L_{eaf}(x_i)\right )\bigcup \left (\bigcup^t_{j=1} L_{eaf}(y_j)\right )$$
Let $A=\sum ^{s}_{k=1} a_k$ and $B=\sum ^{t}_{k=1} b_k$. Thereby,
$${
\begin{split}
|V(G+L_{eaf}(w))|&=|V(G)|+|L_{eaf}(w)|=|V(G)|+A+B\\
|E(G+ L_{eaf}(w))|&=|E(G)|+|L_{eaf}(w)|=|E(G)|+A+B
\end{split}}
$$

We define a total coloring $g$ for the connected graph $G+L_{eaf}(w)$ in the following steps:

\textbf{Step 1.} Color the edges $x_ix_{i,r}$ made by the leaves of $L_{eaf}(x_i)$ joined with the vertices $x_i$ of $X$ as:
$$
g(x_ix_{i,r})=k+d\cdot \left (r-1+\sum ^{i-1}_{k=1} a_k\right )
$$ for $r\in [1,a_i]$ and $i\in [1,s]$, here, $\sum ^{0}_{k=1} a_k=0$.

\textbf{Step 2.} Color the edges $y_jy_{j,r}$ made by the leaves of $L_{eaf}(y_j)$ joined with the vertices $y_j$ of $Y$ as:
$$
g(y_ty_{t,r})=k+d\cdot (A+r-1),~r\in [1,b_t]
$$ and
$$
g(y_{t-j}y_{t-j,r})=k+d\cdot \left (A+r-1+\sum^{j}_{k=1} b_{t-k+1}\right ),~r\in [1,b_{t-j}],~j\in [1,t-1]
$$

\textbf{Step 3.} Color the vertices $x_i\in X\,'$ by $g(x_i)=f(x_i)$ with $i\in [1,s]$, and set $g(x_{i,r})=f(x_i)+g(x_ix_{i,r})$ for each leaf $x_{i,r}\in L_{eaf}(x_i)\subset Y\,'$.

\textbf{Step 4.} Color the vertices $y_j\in Y\,'$ by $g(y_j)=f(y_j)+(A+B)d$ with $j\in [1,t]$, and each leaf $y_{j,r}\in L_{eaf}(y_j)\subset Y\,'$ is colored with
$$
g(y_{j,r})=g(y_j)-g(y_{j}y_{j,r})=f(y_j)+(A+B)d-\left [k+d\cdot \left(A+r-1+\sum^{t-j}_{k=1} b_{t-k+1}\right )\right ]
$$ for $j\in [1,t]$, where $\sum^{0}_{k=1} b_{t-k+1}=0$.

\textbf{Step 5.} Color each edge $x_iy_j\in E(G)\subset E(G+L_{eaf}(w))$ by
$$
g(x_iy_j)=g(y_j)-g(x_i)=f(y_j)+(A+B)d-f(x_i)=f(x_iy_j)+(A+B)d
$$ with $x_iy_j\in E(G+L_{eaf}(w))$. Thereby, we have the edge color set
$${
\begin{split}
g(E(G+L_{eaf}(w)))=&\{k,k+d,\dots, k+(A+B-1)d\}\cup\\
&\cup \{f(x_iy_j)+(A+B)d:x_iy_j\in E(G+L_{eaf}(w))\}\\
=&\{k,k+d,\dots, k+(A+B-1)d, k+(A+B)d, \dots, k+(q+A+B-1)d\}\\
=&S_{q+A+B-1,k,d}
\end{split}}
$$ and the smallest vertex color is $\min g(V(G+L_{eaf}(w)))=g(x_1)=0$, and the largest vertex color is
$$\max g(V(G+L_{eaf}(w)))=g(y_t)=k+(q-1+A+B)d=k+(|E(G+L_{eaf}(w))|-1)d$$ as well as the vertex color sets $g(X\,')\subset S_{n,0,d}$ and $g(Y\,')\subset S_{A+B+q,k,d}$. Finally, we claim that $g$ is a graceful $(k,d)$-total coloring of the leaf-added graph $G+L_{eaf}(w)$.

The proof of the theorem is complete.
\end{proof}

\begin{figure}[h]
\centering
\includegraphics[width=16.4cm]{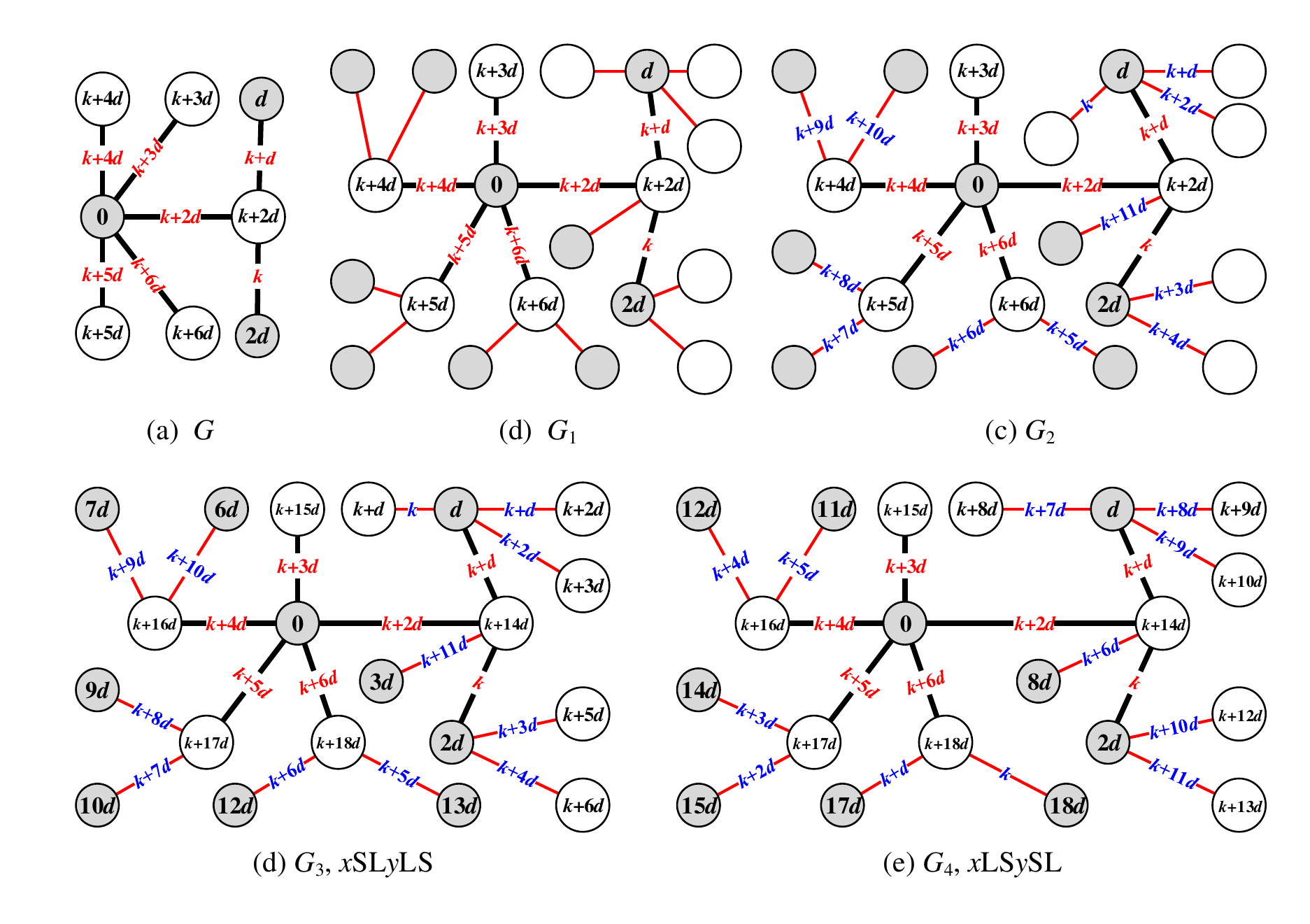}
\caption{\label{fig:k-d-add-leaves-1}{\small The processes for adding leaves: (a) A graph $G$ admitting a graceful $(k,d)$-total coloring; (b) adding randomly leaves to $G$; (c) coloring new added edges by the Leaf-adding $x$SL$y$LS-graceful algorithm; (d) coloring added leaves; (e) coloring new added edges by the Leaf-adding $x$LS$y$SL-graceful algorithm and coloring added leaves.}}
\end{figure}

\begin{figure}[h]
\centering
\includegraphics[width=16.4cm]{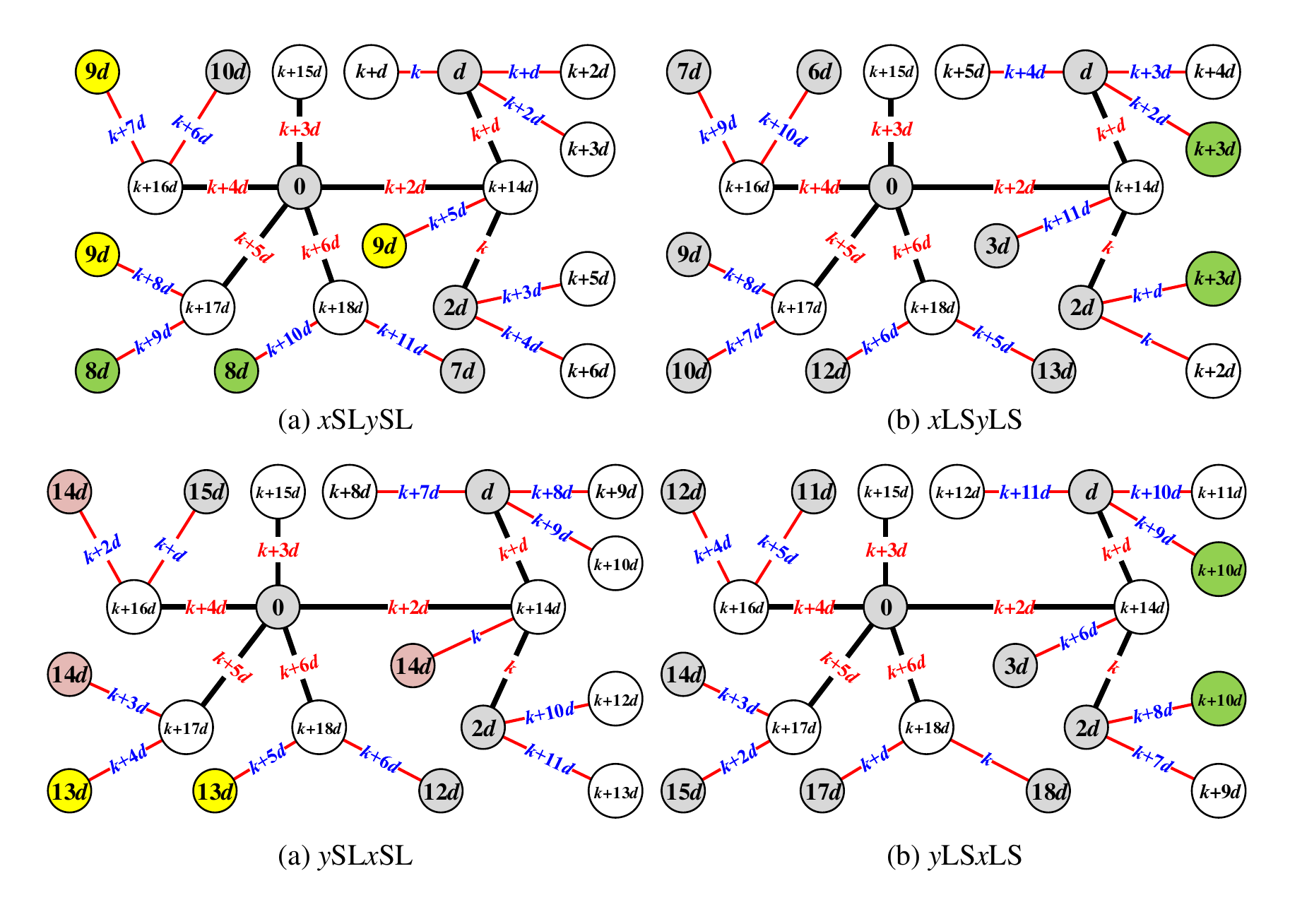}
\caption{\label{fig:k-d-add-leaves-1a}{\small Four Leaf-adding $W$-graceful algorithms with $W=$$x$SL$y$SL, $x$LS$y$LS, $y$SL$x$SL and $y$LS$x$LS based on $G_2$ shown in Fig.\ref{fig:k-d-add-leaves-1}.}}
\end{figure}

\begin{figure}[h]
\centering
\includegraphics[width=13.6cm]{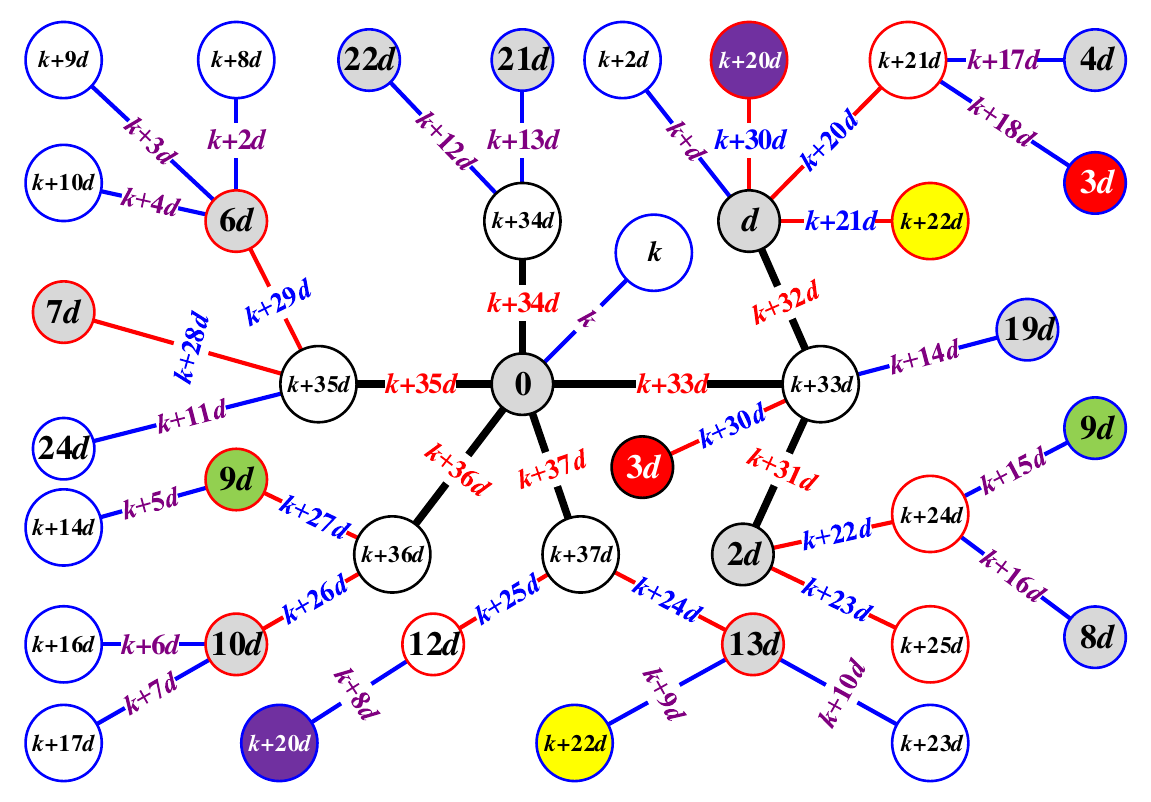}
\caption{\label{fig:k-d-add-leaves-2}{\small An example for understanding the Leaf-adding $x$SL$y$LS-graceful algorithm by the positive order of coloring the new added edges, where two vertices are colored with $3d$, two vertices are colored with $9d$, two vertices colored with $k+20d$, and two vertices are colored with $k+22d$.}}
\end{figure}

\begin{rem}\label{rem:333333}
We call the process of showing the graceful $(k,d)$-total coloring $g$ of the leaf-added graph $G+L_{eaf}(w)$ obtained by ``adding-leaf operation'' in the proof of Theorem \ref{thm:adding-leaves-keep-total-coloring} as \emph{Leaf-adding $x$SL$y$LS-graceful algorithm} hereafter, since we have colored new added edges from the leaves of $x_1$ to the leaves of $x_s$ (from small to large, $x$SL) and then from the leaves of $y_t$ to the leaves of $y_1$ (from large to small, $y$LS). The coloring technique in the proof of Theorem \ref{thm:adding-leaves-keep-total-coloring} is as the same as that in \cite{Zhou-Yao-Chen-Tao2012}. See examples shown in Fig.\ref{fig:k-d-add-leaves-1} and Fig.\ref{fig:k-d-add-leaves-2} for illustrating the Leaf-adding $x$SL$y$LS-graceful algorithm, where $G_3$ is obtained by adding randomly leaves to $G$.

Conversely, we have a Leaf-adding $x$LS$y$SL-graceful algorithm when we color new added edges from the leaves of $x_s$ to the leaves of $x_1$ and then from the leaves of $y_1$ to the leaves of $y_t$ in the proof of Theorem \ref{thm:adding-leaves-keep-total-coloring}. See an example shown in Fig.\ref{fig:anti-oder} for knowing the Leaf-adding $x$LS$y$SL-graceful algorithm.\paralled
\end{rem}

\begin{figure}[h]
\centering
\includegraphics[width=13.6cm]{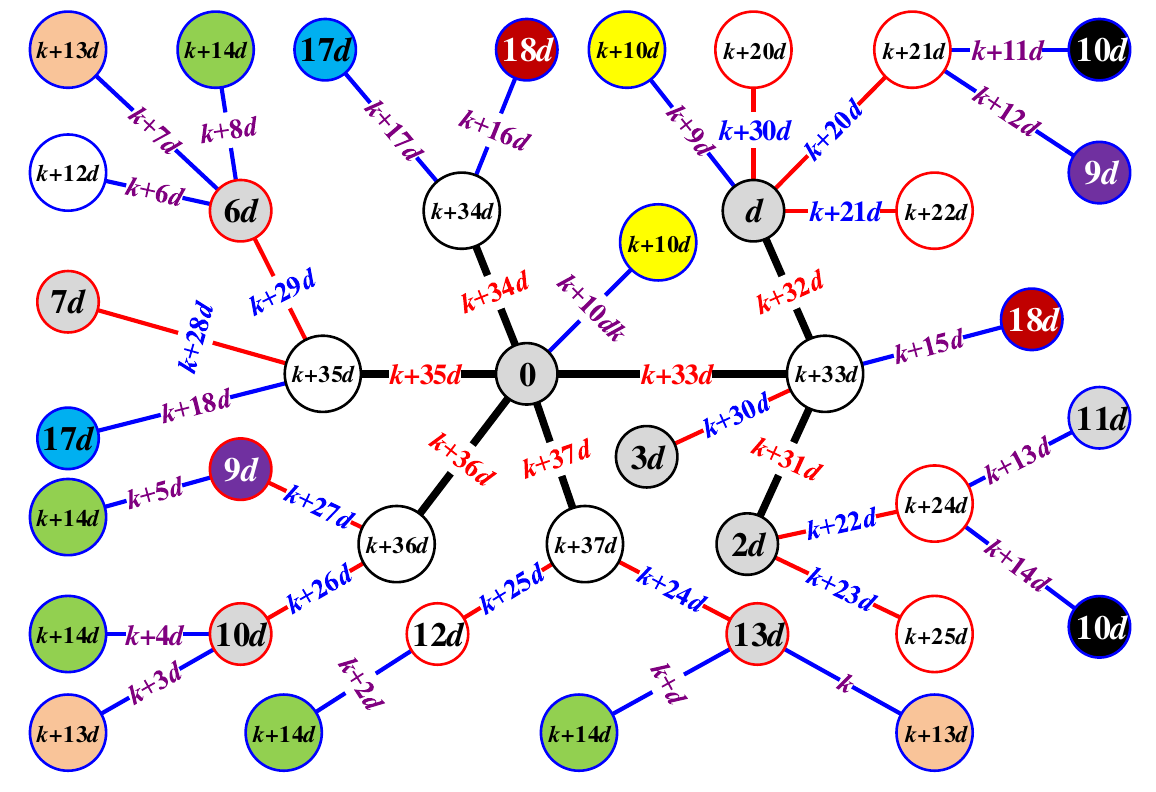}
\caption{\label{fig:anti-oder}{\small An example for illustrating the Leaf-adding $x$LS$y$SL-graceful algorithm by the reverse order of coloring the new added edges, where two vertices are colored with $9d$, three vertices are colored with $10d$, two vertices are colored with $17d$, two vertices are colored with $18d$, two vertices are colored with $k+10d$, three vertices are colored with $k+13d$, and five vertices are colored with $k+14d$.}}
\end{figure}

\begin{thm}\label{thm:any-tree-k-d-graceful-total-coloring}
$^*$ Any tree $T$ with diameter $D(T)\geq 3$ admits at least $2^m$ different graceful $(k,d)$-total colorings for $m+1=\left \lceil \frac{D(T)}{2}\right \rceil $.
\end{thm}
\begin{proof} Let $D(T_1)$ be the diameter of a tree $T_1$ with $D(T_1)\geq 3$, and let $L_{eaf}(T_1)$ be the set of leaves of $T_1$. We get another tree $T_2=T_1-L_{eaf}(T_1)$ after removing the leaves of $L_{eaf}(T_1)$. Go on in this way, we have trees $T_{i+1}=T_i-L_{eaf}(T_i)$ with $i\in [1,m]$, and $T_{m+1}$ is just a star with vertex set $V(T_{m+1})=\{x\}\cup \{x_j:j\in [1,n]\}$ and edge set $E(T_{m+1})=\{xx_j:j\in [1,n]\}$.

We define a graceful $(k,d)$-total coloring $f_{m+1}$ for the star $T_{m+1}$ as:

(i) $f_{m+1}(x)=0$, $f_{m+1}(x_j)=k+(j-1)d$ and

(ii) $f_{m+1}(xx_j)=f_{m+1}(x_j)-f_{m+1}(x)$ with $j\in [1,n]$.\\
Thereby, each tree $T_i$ is obtained by adding the leaves of $L_{eaf}(T_i)$ to the tree $T_{i+1}$, that is, $T_i=T_{i+1}+ L_{eaf}(T_i)$ is a \emph{leaf-added tree}. Moreover, $T_i$ admits a graceful $(k,d)$-total coloring $f_{i}$ obtained from the graceful $(k,d)$-total coloring $f_{i+1}$ of $T_{i+1}$ according to one of the Leaf-adding $x$SL$y$LS-graceful algorithm and the Leaf-adding $x$LS$y$SL-graceful algorithm. By the induction, the tree $T_1$ admits a graceful $(k,d)$-total coloring $f_{1}$ obtained from the graceful $(k,d)$-total coloring $f_{2}$ of $T_{2}$.

Notice that there are two ways from choosing one of the Leaf-adding $x$SL$y$LS-graceful algorithm and the Leaf-adding $x$LS$y$SL-graceful algorithm, since each tree has at least two leaves, that is, $|L_{eaf}(T_i)|\geq 2$ with $i\in [1,m]$.

Then the tree $T_1$ admits at least $2^m$ different graceful $(k,d)$-total colorings as desired.
\end{proof}

\subsection{Graphs admitting harmonious $(k,d)$-total colorings}

\begin{thm}\label{thm:tree-harmonious-k-d-total-colorings}
$^*$ Each tree admits a harmonious $(k,d)$-total coloring.
\end{thm}
\begin{proof} Let $L_{eaf}(T_i)$ be the leaf set of a tree $T_i$, and let $T_{i+1}=T_i-L_{eaf}(T_i)$ with $i\in [1,m]$. Suppose that $T_{m+1}=T_{m}-L_{eaf}(T_{m})$ is a star with vertex set $V(T_{m+1})=\{u_0\}\cup L_{eaf}(T_{m+1})$ and the edge set $E(T_{m+1})=\{u_0v_j:j\in [1,n]\}$, where $L_{eaf}(T_{m+1})=\{v_1,v_2,\dots ,v_n\}$ is the leaf set of the star $T_{m+1}$.

\textbf{Initialization. }We define a harmonious $(k,d)$-total coloring $f_{m+1}$ for the star $T_{m+1}$ as follows: $f_{m+1}(u_0)=0$, $f_{m+1}(v_j)=k+jd$ with $k\geq 0$ and $d\geq 1$, and color each edge $u_0u_j$ with $$f_{m+1}(u_0u_j)=f_{m+1}(v_j)+f_{m+1}(u_0)=k+jd,~j\in [1,n]
$$ So, $f_{m+1}$ is a harmonious $(k,d)$-total coloring of $T_{m+1}$, since
$$f_{m+1}(u_0u_j)=f_{m+1}(v_j)+f_{m+1}(u_0)=k+jd~(\bmod^* ~q_md),~q_m=|E(T_{m+1})|=n
$$ where the operation ``$\bmod^* q_md$'' means $f_{m+1}(u_0u_j)-k~(\bmod~q_md)$.

\textbf{Iteration. }Suppose that a tree $T_{i+1}=T_{i}-L_{eaf}(T_{i})$ admits a harmonious $(k,d)$-total coloring $f_{i+1}$, and $V(T_{i+1})=X_{i+1}\cup Y_{i+1}$ with $X_{i+1}=\{u_{i+1,r}:r\in [1,c_{i+1}]\}$ and $Y_{i+1}=\{v_{i+1,j}:j\in [1,b_{i+1}]\}$ for some $i\in [1,m]$. Thereby, the $X$-vertex set
$$f_{i+1}(X_{i+1})\subseteq S_{l,0,0,d}=\{0,d,\dots ,ld\}
$$ and the $Y$-vertex set
$$
f_{i+1}(Y_{i+1})\subseteq S_{q_{i+1}-1,k,0,d}=\{k,k+d,\dots ,k+(q_{i+1}-1)d\}
$$
Since the edge color set $f_{i+1}(E(T_{i+1}))=S_{q_{i+1}-1,k,0,d}$, so we have $E(T_{i+1})=E_{\textrm{non}}\cup E_{\textrm{mod}}$ such that the edge color set
$$
f_{i+1}(E(T_{i+1}))=f_{i+1}(E_{non})\cup f_{i+1}(E_{\textrm{mod}})
$$ where $f_{i+1}(uv)~(\bmod^* ~q_{i+1}d)=f_{i+1}(u)+f_{i+1}(v)$ for each edge $uv\in E_{non}$, and
$$f_{i+1}(xy)~(\bmod^* ~q_{i+1}d)=f_{i+1}(x)+f_{i+1}(y)-k~(\bmod~q_{i+1}d),~xy\in E_{\textrm{mod}}
$$
Clearly, two edge color sets $f_{i+1}(E_{\textrm{non}})=\{k,k+d,\dots ,k+\alpha d\}$ and
$$f_{i+1}(E_{\textrm{mod}})=\{k+(\alpha +1)d,k+(\alpha +2)d,\dots ,k+(q_{i+1}-1)d\},~\alpha\geq 0
$$ Without loss of generality, we have
$$0=f_{i+1}(u_{i,1})\leq f_{i+1}(u_{i+1,r})\leq f_{i+1}(u_{i,r+1}),~r\in [1,c_{i+1}-1]
$$
and
$$
k\leq f_{i+1}(v_{i+1,j})\leq f_{i+1}(v_{i,j+1})\leq f_{i+1}(v_{i,b_{i+1}})=k+(q_{i+1}-1)d,~j\in [1,b_{i+1}-1]
$$

Adding the leaves of the leaf set $L_{eaf}(T_{i})$ to the tree $T_{i+1}$ produces the tree $T_i$, and each vertex $u_{i+1,r}$ of $X_{i+1}$ has its own leaf set
$$L_{eaf}(u_{i+1,r})=\{u\,'_{i+1,r,j}:j\in [1,c_{i+1,r}]\}\textrm{ with }|L_{eaf}(u_{i+1,r})|=c_{i+1,r},~r\in [1,c_{i+1}]$$ and each vertex $v_{i+1,s}$ of $Y_{i+1}$ has its own leaf set
$$
L_{eaf}(v_{i+1,s})=\{v\,'_{i+1,s,j}:j\in [1,b_{i+1,s}]\}\textrm{ with }|L_{eaf}(v_{i+1,s})|=b_{i+1,s},~s\in [1,b_{i+1}]
$$ Thereby, we have the leaf set
$$
L_{eaf}(T_i)=\left (\bigcup^{c_{i+1}}_{r=1}L_{eaf}(u_{i,r})\right )\bigcup \left (\bigcup^{b_{i+1}}_{s=1}L_{eaf}(v_{i+1,s})\right )
$$ and $A=|L_{eaf}(T_{i})|=M_X+M_Y$, where $M_X=\sum^{c_{i+1}}_{r=1} |L_{eaf}(u_{i,r})|$ and $M_Y=\sum^{b_{i+1}}_{s=1} |L_{eaf}(v_{i,s})|$.

We, now, define a total coloring $f_i$ for $T_i$ with $i\in [1,m]$ in the following steps:

\textbf{Step 1.} $f_i(v_{i+1,s})=d\cdot A+f_{i+1}(v_{i+1,s})$ for $v_{i+1,s}\in Y_{i+1}\subset Y_{i}\subset V(T_i)$;

\textbf{Step 2.} $f_i(u_{i+1,r})=f_{i+1}(u_{i+1,r})$ for $u_{i+1,r}\in X_{i+1}\subset Y_{i}\subset V(T_i)$;

\textbf{Step 3.} $f_i(xy)=f_{i+1}(xy)$ for each edge $xy\in E_{\textrm{mod}}\subset E(T_{i+1})\subset E(T_i)$;

\textbf{Step 4.} $f_i(uv)=f_i(u)+f_i(v)=d\cdot A+f_{i+1}(u)+f_{i+1}(v)$ for each edge $uv\in E_{\textrm{non}}\subset E(T_{i+1})\subset E(T_i)$ holding $u\in Y_{i+1}$ and $v\in X_{i+1}$;

\textbf{Step 5.} For the new added leaves to $T_{i+1}$, we, first, color new edges as:
$$f_i(v_{i+1,b_{i+1}}v\,'_{i+1,b_{i+1},j})=k+(\alpha +j)d,~v\,'_{i+1,b_{i+1},j}\in L_{eaf}(v_{i+1,b_{i+1}})
$$ in general,
\begin{equation}\label{eqa:c3xxxxx}
f_i(v_{i+1,s}v\,'_{i+1,s,j})=jd+k+\alpha d+d\sum^{b_{i+1}-s}_{k=1} b_{i+1,b_{i+1}-k+1}
\end{equation}for $v\,'_{i+1,s,j}\in L_{eaf}(v_{i+1,s})$ with $s\in [1,b_{i+1}-1]$, where $\sum^{0}_{k=1} b_{i+1,b_{i+1}-k+1}=0$; and moreover
$$f_i(v_{i+1,1}v\,'_{i+1,1,b_{i+1,1}})=d\cdot b_{i+1,1}+k+\alpha d+d\sum^{b_{i+1}-1}_{k=1} b_{i+1,b_{i+1}-k+1}=k+(\alpha +M_Y)d.$$
We color the leaves $v\,'_{i+1,s,j}\in \bigcup^{b_{i+1}}_{s=1}L_{eaf}(v_{i+1,s})$ by
$$f_i(v\,'_{i+1,s,j})=d\cdot (A+q_{i+1})+f_i(v_{i+1,s}v\,'_{i+1,s,j})-f_i(v_{i+1,s})
$$ so $f_i(v\,'_{i+1,s,j})\in S_{q_{i}-1,k,0,d}$, and $f_i(v_{i+1,s}v\,'_{i+1,s,j})\in S_{q_{i}-1,k,0,d}$.

\textbf{Step 6.} Let $N_Y=k+(\alpha +M_Y)d$. We, first, color edges $u_{i+1,c_{i+1}}u\,'_{i+1,c_{i+1},j}$ by
$$f_i(u_{i+1,c_{i+1}}u\,'_{i+1,c_{i+1},j})=jd+N_Y
$$ with $j\in [1,c_{i+1,c_{i+1}}]$ for $u\,'_{i+1,c_{i+1},j}\in L_{eaf}(u_{i+1,c_{i+1}})$. In general, we color edges $u_{i+1,r}u\,'_{i+1,r,j}$ by
\begin{equation}\label{eqa:c3xxxxx}
f_i(u_{i+1,r}u\,'_{i+1,r,j})=jd+N_Y+d \sum^{c_{i+1}-r}_{k=1} c_{i+1,c_{i+1}-k+1}
\end{equation} with $j\in [1,c_{i+1,r}]$ for $u\,'_{i+1,r,j}\in L_{eaf}(u_{i+1,r})$ and $r\in [1,c_{i+1}]$. The last edge $u_{i+1,1}u\,'_{i+1,1,c_{i+1,1}}$ is colored with
$$f_i(u_{i+1,1}u\,'_{i+1,1,c_{i+1,1}})=c_{i+1,1}d+N_Y+d \sum^{c_{i+1}-1}_{k=1} c_{i+1,c_{i+1}-k+1}=N_Y+N_X$$
where $N_X=\sum^{c_{i+1}}_{k=1} c_{i+1,c_{i+1}-k+1}$.

Now, we color the leaves of the leaf set $\bigcup^{c_{i+1}}_{r=1}L_{eaf}(u_{i,r})$ by
$$f_i(u\,'_{i+1,r,j})=f_i(u_{i+1,r}u\,'_{i+1,r,j})-f_i(u_{i+1,r}),~j\in [1,c_{i+1,r}]
$$ for $u\,'_{i+1,r,j}\in L_{eaf}(u_{i+1,r})$ and $r\in [1,c_{i+1}]$. Clearly,
$$f_i(u\,'_{i+1,r,j})\in S_{m,0,0,d},\quad f_i(u_{i+1,r}u\,'_{i+1,r,j})\in S_{q_{i}-1,k,0,d}$$

Thereby, there are the vertex color set $f_i(V(T_i))\subseteq S_{m,0,0,d}\cup S_{q_{i}-1,k,0,d}$ and the edge color set $f_i(E(T_i))=S_{q_{i}-1,k,0,d}$ with $i\in [1,m]$, then we claim that the coloring $f_i$ is a harmonious $(k,d)$-total coloring of $T_i$ with $i\in [1,m]$.

The theorem follows the induction on adding leaves.
\end{proof}

The algorithm in the proof of Theorem \ref{thm:tree-harmonious-k-d-total-colorings} is called \emph{Leaf-adding LS-harmonious algorithm}, conversely, we can color the edges $u_{i+1,1}u\,'_{i+1,1,j}$ for $u\,'_{i+1,1,j}\in L_{eaf}(u_{i+1,1})$ first, and then color the edges $u_{i+1,2}u\,'_{i+1,2,j}$ for $u\,'_{i+1,2,j}\in L_{eaf}(u_{i+1,2})$, go on in this way, to the last edges $v_{i+1,b_{i+1}}v\,'_{i+1,b_{i+1},j}$ for $u\,'_{i+1,b_{i+1},j}\in L_{eaf}(u_{i+1,b_{i+1}})$. This process is called \emph{Leaf-adding SL-harmonious algorithm}. So, we have the following result:

\begin{thm}\label{thm:different-k-d-harmonious-colorings}
$^*$ Each tree $T_1$ admits at least $2^m$ different harmonious $(k,d)$-total colorings, where each tree $T_{i+1}=T_{i}-L_{eaf}(T_{i})$ with $i\in [1,m]$, each $L_{eaf}(T_{i})$ is the leaf set of the tree $T_{i+1}$, and $T_{m+1}$ is a star such that $T_{m+1}-L_{eaf}(T_{m+1})$ is a graph having one vertex only.
\end{thm}

See examples shown in Fig.\ref{fig:harmonious-1}, Fig.\ref{fig:harmonious-2} and Fig.\ref{fig:harmonious-3} for understanding the proof of Theorem \ref{thm:tree-harmonious-k-d-total-colorings}, and the Leaf-adding LS-harmonious algorithm and the Leaf-adding SL-harmonious algorithm, respectively.

\begin{figure}[h]
\centering
\includegraphics[width=15cm]{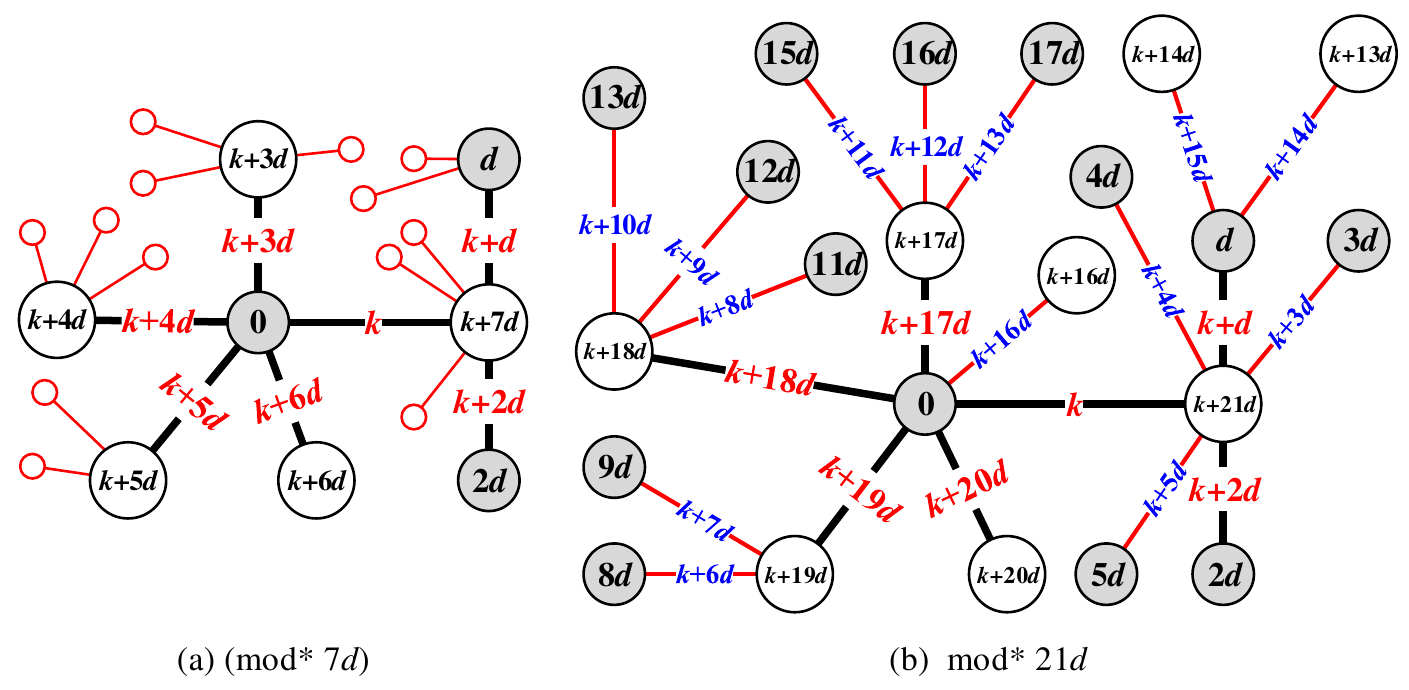}
\caption{\label{fig:harmonious-1}{\small An example for illustrating the proof of Theorem \ref{thm:tree-harmonious-k-d-total-colorings}.}}
\end{figure}

\begin{figure}[h]
\centering
\includegraphics[width=14cm]{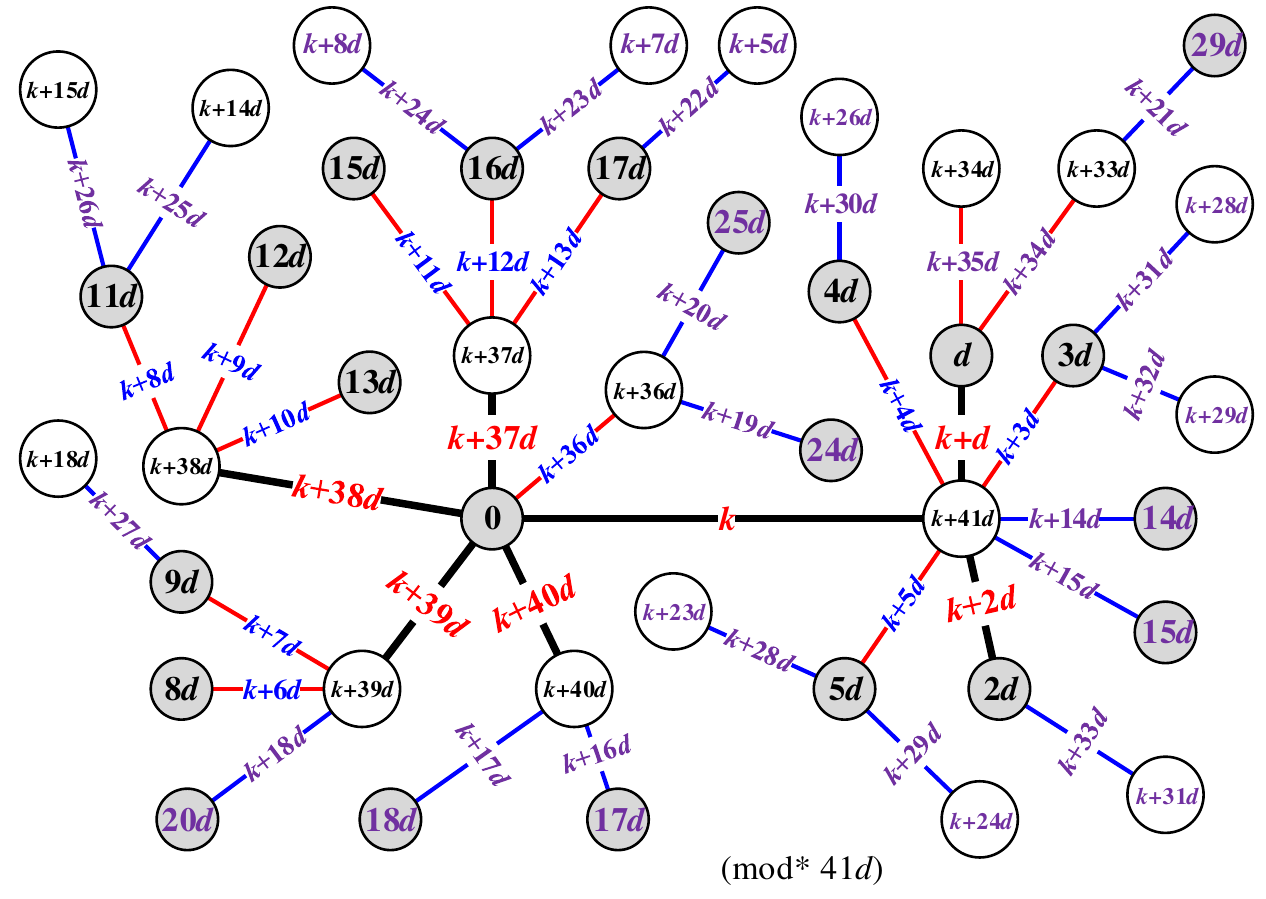}
\caption{\label{fig:harmonious-2}{\small An example for illustrating Leaf-adding LS-harmonious algorithm differing from Leaf-adding SL-harmonious algorithm.}}
\end{figure}

\begin{figure}[h]
\centering
\includegraphics[width=13.6cm]{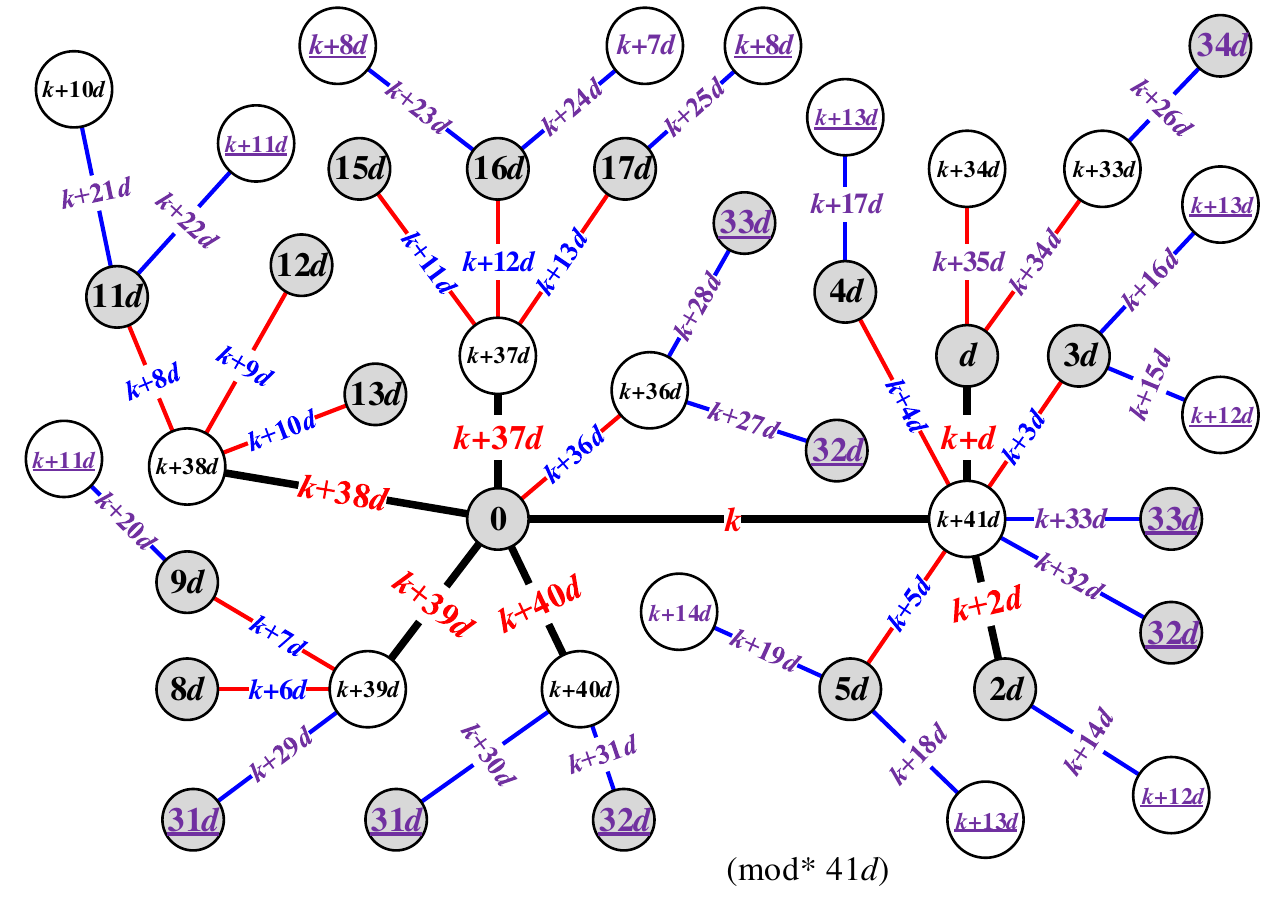}
\caption{\label{fig:harmonious-3}{\small An example for illustrating Leaf-adding SL-harmonious algorithm differing from Leaf-adding LS-harmonious algorithm.}}
\end{figure}

\subsection{Connections between parameterized total colorings}

\begin{defn} \label{defn:kd-w-type-coloring-transfoemations}
$^*$ Let $G$ be a bipartite and connected $(p,q)$-graph with vertex bipartition $V(G)=X\cup Y$ with $X\cap Y=\emptyset$, where
$$
X=\{x_i:i\in [1,s]\},~Y=\{y_j:j\in [1,t]\},~s+t=p
$$ and $E(G)=\{e_i:i\in [1,q]\}$. Suppose that $G$ admits a $W$-constraint $(k,d)$-total coloring $f$ holding

(i) $f:X\rightarrow S_{m,0,0,d}$ and $f:Y\cup E(G)\rightarrow S_{q-1,k,a,d}$;

(ii) $f(x_i)\leq f(x_{i+1})$ for $i\in [1,s-1]$ and $f(y_j)\leq f(y_{j+1})$ for $j\in [1,t-1]$;

(iii) $f(e_i)\leq f(e_{i+1})$ with $i\in [1,q-1]$.

\noindent We have the following transformations $g$:
\begin{asparaenum}[\textrm{\textbf{Tra}}-1. ]
\item \label{aspa:vertex-unchange} $g(u)=f(u)$ for $u\in V(G)$;
\item \label{aspa:x-vertex-unchange} $g(x_i)=f(x_i)$ for $x_i\in X$;
\item \label{aspa:y-vertex-unchange} $g(y_j)=f(y_j)$ for $y_j\in Y$;
\item \label{aspa:edge-unchange} $g(e_i)=f(e_i)$ for each edge $e_i\in E(G)$;
\item \label{aspa:totally-v-image} (totally vertex-image) $g(w)=[f(y_{t})+f(x_1)]-f(w)$ for $w\in V(G)$;
\item \label{aspa:totally-e-image} (totally edge-image) $g(e_i)=[f(e_q)+f(e_{1})]-f(e_i)$ for each edge $e_i\in E(G)$;
\item \label{aspa:partly-x-v-image} (partly X-image) $g(x_i)=[f(x_s)+f(x_1)]-f(x_i)$ for $x_i\in X$;
\item \label{aspa:partly-y-v-image} (partly Y-image) $g(y_j)=[f(y_{t})+f(y_1)]-f(y_j)$ for $y_j\in Y$;
\item \label{aspa:e-reciprocal} (edge-reciprocal) $g(e_i)=f(e_{q-i+1})$ for each edge $e_i\in E(G)$;
\item \label{aspa:X-reciprocal} (X-vertex-reciprocal) $g(x_i)=f(x_{s-i+1})$ for $x_i\in X$;
\item \label{aspa:Y-reciprocal} (Y-vertex-reciprocal) $g(y_j)=f(y_{t-j+1})$ for $y_j\in Y$;
\end{asparaenum}

\noindent\textbf{We call $g$:}
\begin{asparaenum}[\textrm{\textbf{Col}}-1. ]
\item \label{color:totally-v-image} a \emph{totally vertex-image $(k,d)$-total coloring} if Tra-\ref{aspa:edge-unchange} and Tra-\ref{aspa:totally-v-image} hold true;
\item \label{color:totally-e-image} a \emph{totally edge-image $(k,d)$-total coloring} if Tra-\ref{aspa:vertex-unchange} and Tra-\ref{aspa:totally-e-image} hold true;
 \item \label{color:totally-e-image} a \emph{partly X-image totally-e-image $(k,d)$-total coloring} if Tra-\ref{aspa:y-vertex-unchange}, Tra-\ref{aspa:totally-e-image} and Tra-\ref{aspa:partly-x-v-image} hold true;
\item \label{color:totally-e-image} a \emph{partly Y-image totally-e-image $(k,d)$-total coloring} if Tra-\ref{aspa:x-vertex-unchange}, Tra-\ref{aspa:totally-e-image} and Tra-\ref{aspa:partly-y-v-image} hold true;
\item \label{color:totally-ve} a \emph{totally ve-image $(k,d)$-total coloring} if Tra-\ref{aspa:totally-v-image} and Tra-\ref{aspa:totally-e-image} hold true;
\item \label{color:totally-ve} a \emph{partly XY-image $(k,d)$-total coloring} if Tra-\ref{aspa:edge-unchange}, Tra-\ref{aspa:partly-x-v-image} and Tra-\ref{aspa:partly-y-v-image} hold true;
\item \label{color:X-Y-image-totally-e-image} a \emph{partly XY-image totally-e-image $(k,d)$-total coloring} if Tra-\ref{aspa:totally-e-image}, Tra-\ref{aspa:partly-x-v-image} and Tra-\ref{aspa:partly-y-v-image} hold true;
\item \label{color:totally-e-image} a \emph{partly X-image e-reciprocal $(k,d)$-total coloring} if Tra-\ref{aspa:y-vertex-unchange}, Tra-\ref{aspa:partly-x-v-image} and Tra-\ref{aspa:e-reciprocal} hold true;
\item \label{color:totally-e-image} a \emph{partly Y-image e-reciprocal $(k,d)$-total coloring} if Tra-\ref{aspa:x-vertex-unchange}, Tra-\ref{aspa:partly-y-v-image} and Tra-\ref{aspa:e-reciprocal} hold true;
\item \label{color:e-reciprocal} an \emph{e-reciprocal $(k,d)$-total coloring} if Tra-\ref{aspa:e-reciprocal} hold true;
\item \label{color:X-reciprocal} an \emph{X-reciprocal $(k,d)$-total coloring} if Tra-\ref{aspa:y-vertex-unchange}, Tra-\ref{aspa:edge-unchange} and Tra-\ref{aspa:X-reciprocal} hold true;
\item \label{color:X-reciprocal} a \emph{Y-reciprocal $(k,d)$-total coloring} if Tra-\ref{aspa:x-vertex-unchange}, Tra-\ref{aspa:edge-unchange} and Tra-\ref{aspa:Y-reciprocal} hold true;
\item \label{color:X-e-reciprocal} an \emph{X-e-reciprocal $(k,d)$-total coloring} if Tra-\ref{aspa:y-vertex-unchange}, Tra-\ref{aspa:e-reciprocal} and Tra-\ref{aspa:X-reciprocal} hold true;
\item \label{color:Y-e-reciprocal} a \emph{Y-e-reciprocal $(k,d)$-total coloring} if Tra-\ref{aspa:x-vertex-unchange}, Tra-\ref{aspa:e-reciprocal} and Tra-\ref{aspa:Y-reciprocal} hold true;

\item \label{color:XY-reciprocal} an \emph{XY-reciprocal $(k,d)$-total coloring} if Tra-\ref{aspa:edge-unchange}, Tra-\ref{aspa:X-reciprocal} and Tra-\ref{aspa:Y-reciprocal} hold true;
\item \label{color:XY-reciprocal} a \emph{ve-reciprocally $(k,d)$-total coloring} if Tra-\ref{aspa:e-reciprocal}, Tra-\ref{aspa:X-reciprocal} and Tra-\ref{aspa:Y-reciprocal} hold true.\qqed
\end{asparaenum}
\end{defn}

\begin{figure}[h]
\centering
\includegraphics[width=16.4cm]{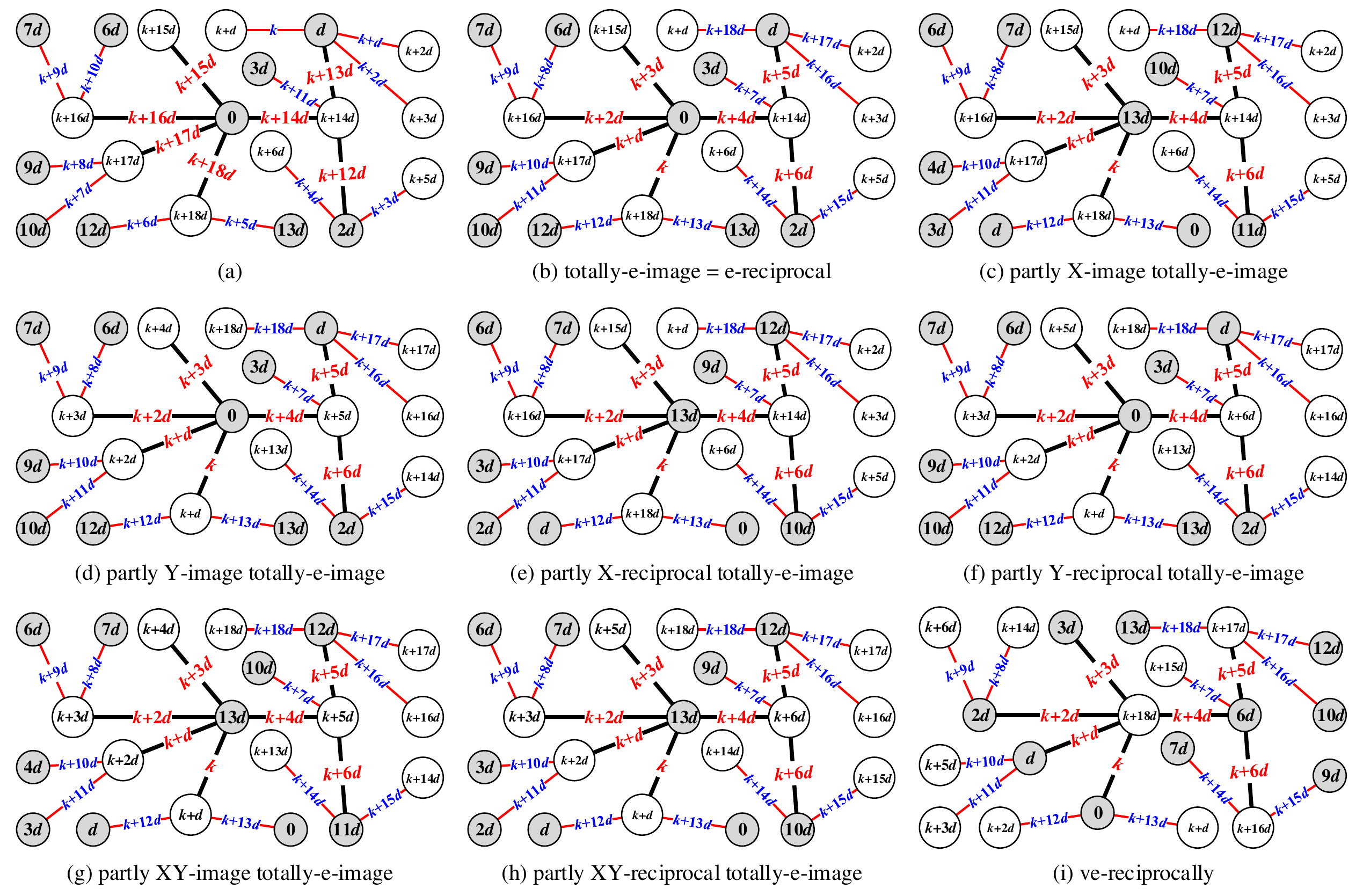}
\caption{\label{fig:transformation-1}{\small A scheme for illustrating Definition \ref{defn:kd-w-type-coloring-transfoemations}.}}
\end{figure}

\begin{figure}[h]
\centering
\includegraphics[width=16.4cm]{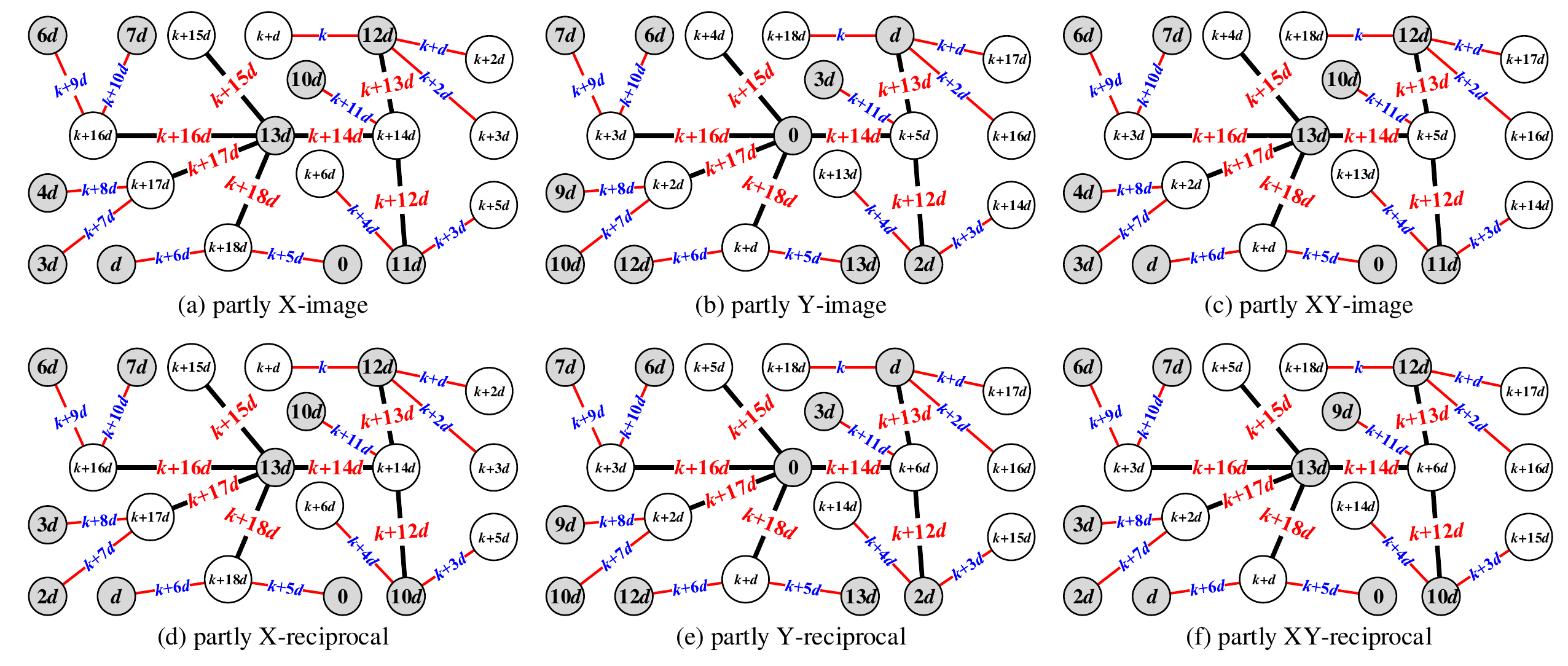}
\caption{\label{fig:transformation-2}{\small Another scheme for illustrating Definition \ref{defn:kd-w-type-coloring-transfoemations}.}}
\end{figure}

Definition \ref{defn:kd-w-type-coloring-transfoemations} enables us to get the following result:
\begin{thm}\label{thm:equivalent-k-d-total-colorings}
$^*$ A bipartite and connected $(p,q)$-graph $G$ admits a graceful $(k,d)$-total coloring if and only if $G$ admits each one of edge-magic $(k,d)$-total coloring, graceful-difference $(k,d)$-total coloring, edge-difference $(k,d)$-total coloring (also, 3C-$(k,d)$-total coloring), felicitous-difference $(k,d)$-total coloring, harmonious $(k,d)$-total coloring and edge-antimagic $(k,d)$-total coloring.
\end{thm}
\begin{proof} Let $G$ be a bipartite and connected $(p,q)$-graph with its vertex set $V(G)=X\cup Y$ with $X\cap Y=\emptyset$, and
$$
X=\{x_i:i\in [1,s]\},~Y=\{y_r:j\in [1,t]\},~s+t=p
$$ and its edge set $E(G)=\{e_i:i\in [1,q]\}$.

Suppose that $G$ admits a graceful $(k,d)$-total coloring $f$ defined in Definition \ref{defn:kd-w-type-colorings}, without loss of generality, we have $0=f(x_1)\leq f(x_i)\leq f(x_{i+1})$ for $i\in [1,s-1]$ and
$$f(y_r)\leq f(y_{j+1})\leq f(y_{t})=k+(q-1)d,~j\in [1,t-1]
$$ as well as
$$
k=f(e_1)< f(e_2)<\cdots < f(e_{q})=k+(q-1)d
$$ holding $f(uv)=|f(u)-f(v)|$ for each edge $uv\in E(G)$ and the edge color set $f(E(T))=S_{q-1,k,0,d}$. Notice that each edge $e_r=x_ry_r\in E(G)$ holds $f(e_r)=|f(x_r)-f(y_r)|=f(y_r)-f(x_r)$ true. We show that the above graceful $(k,d)$-total coloring $f$ can induce the following $W$-constraint $(k,d)$-total colorings:
\begin{asparaenum}[\textbf{\textrm{Wtype}}-1. ]
\item By Definition \ref{defn:kd-w-type-coloring-transfoemations}, the coloring $g$ is a \emph{partly X-image totally-e-image $(k,d)$-total coloring} holding $g(e_r)=[f(e_q)+f(e_{1})]-f(e_r)$ for each edge $e_r\in E(G)$, $g(x_r)=[f(x_s)+f(x_1)]-f(x_r)$ for $x_r\in X$ and $g(y_r)=f(y_r)$ for $y_r\in Y$. Then, for each edge $e_r=x_ry_r\in E(G)$, we have the edge-magic constraint
$${
\begin{split}
g(x_r)+g(y_r)+g(x_ry_r)=&[f(x_s)+f(x_1)]-f(x_r)+f(y_r)+[f(e_q)+f(e_{1})]-f(e_r)\\
=&[f(x_s)+f(x_1)]+[f(e_q)+f(e_{1})]\\
=&2k+(q-1)d+f(x_s),~x_ry_r\in E(G)
\end{split}}$$ to be a constant. So, $g$ is an \emph{edge-magic $(k,d)$-total coloring} defined in Definition \ref{defn:kd-w-type-colorings}. If $f(x_r)=(r-1)d$ for $r\in [1,s]$, we can use $g(x_r)=f(x_{s-r+1})$ for $x_r\in X$, so $f(x_r)+f(x_{s-r+1})=(s-1)d$, and moreover we have the edge-magic constraint
$${
\begin{split}
g(x_r)+g(y_r)+g(x_ry_r)&=f(x_{s-r+1})+f(y_r)+[f(e_q)+f(e_{1})]-f(e_r)\\
&=f(x_{s-r+1})+f(y_r)+2k+(q-1)d-[f(y_r)-f(x_r)]\\
&=2k+(q-1)d+(s-1)d
\end{split}}$$ holds true for each edge $x_ry_r\in E(G)$.
\item Since each edge $e_r=x_ry_r\in E(G)$ holds $f(e_r)=|f(x_r)-f(y_r)|=f(y_r)-f(x_r)$ true, then the coloring $g$ is a \emph{graceful-difference $(k,d)$-total coloring}.
\item According to Definition \ref{defn:kd-w-type-coloring-transfoemations}, the coloring $g$ is a \emph{totally edge-image $(k,d)$-total coloring} holding $g(u)=f(u)$ for $u\in V(G)$ and $g(e_r)=[f(e_q)+f(e_{1})]-f(e_r)$ for each edge $e_r\in E(G)$. For each edge $e_r=x_ry_r\in E(G)$, we have $f(e_r)=|f(x_r)-f(y_r)|=f(y_r)-f(x_r)$, and moreover the edge-difference constraint
$${
\begin{split}
g(e_r)+|g(y_r)-g(x_r)|&=[f(e_q)+f(e_{1})]-f(e_r)+|f(y_r)-f(x_r)|\\
&=f(e_q)+f(e_{1})=2k+(q-1)d
\end{split}}
$$
is equal to a constant for each edge $e_r=x_ry_r\in E(G)$; there exists an edge $u_xu_y$ such that
$$f(u_xu_y)=[f(e_q)+f(e_{1})]-f(e_r),~g(e_r)=f(u_xu_y)=|f(u_y)-f(u_x)|=|g(u_y)-g(u_x)|$$
we get $g(e_r)=|g(u_y)-g(u_x)|$; and
$${
\begin{split}
s(e_r)&=|g(y_r)-g(x_r)|-g(e_r)=|f(y_r)-f(x_r)|-[f(e_q)+f(e_{1})]+f(e_r)\\
&=|f(y_r)-f(x_r)|-f(u_xu_y)
\end{split}}
$$
and
$$s(u_xu_y)=|g(u_y)-g(u_x)|-g(u_xu_y)=|f(u_y)-f(u_x)|-[f(e_q)+f(e_{1})]+f(e_r)$$
so $s(e_r)+s(u_xu_y)=f(e_q)+f(e_{1})$ is a constant.

\quad We claim that $g$ is a \emph{\textbf{3C-$(k,d)$-total coloring}}, also, an \emph{edge-difference $(k,d)$-total coloring} defined in Definition \ref{defn:kd-w-type-colorings}.
\item By Definition \ref{defn:kd-w-type-coloring-transfoemations}, the coloring $g$ is a \emph{partly Y-image totally-e-image $(k,d)$-total coloring} holding $g(e_r)=[f(e_q)+f(e_{1})]-f(e_r)$ for each edge $e_r\in E(G)$, $g(x_r)=f(x_r)$ for $x_r\in X$ and $g(y_r)=[f(y_t)+f(y_1)]-f(y_r)$ for $y_r\in Y$. For each edge $e_r=x_ry_r\in E(G)$, we compute the felicitous-difference constraint
$${
\begin{split}
|g(x_r)+g(y_r)-g(e_r)|&=f(x_r)+[f(y_t)+f(y_1)]-f(y_r)-\{[f(e_q)+f(e_{1})]-f(e_r)\}\\
&=f(y_t)+f(y_1)-[f(e_q)+f(e_{1})]\\
&=f(y_1)-k
\end{split}}
$$ which enables us to claim that $g$ is a \emph{felicitous-difference $(k,d)$-total coloring}.

\item We defined a transformation $g$ as:

\quad (i) $g(x_r)=f(x_r)$ for $x_r\in X$;

\quad (ii)  $g(y_r)=[f(y_{t})+f(y_1)]-f(y_r)$ for $y_r\in Y$; and

\quad (iii) $g(e_r)=g(x_r)+g(y_r)~(\bmod^*qd)$ by $g(uv)-k=[g(u)+g(v)-k](\bmod ~qd)$ for each edge $uv\in E(G)$.

And furthermore we have
$${
\begin{split}
g(e_r)-k&=g(x_r)+g(y_r)-k~(\bmod ~qd)=f(x_r)+[f(y_{t})+f(y_1)]-f(y_r)-k~(\bmod ~qd)\\
&=2k+(q-1)d-f(e_r)-k~(\bmod ~qd)\\
&=k+(q-1)d-f(e_r)\in S_{q-1,0,0,d}
\end{split}}
$$ so the edge color set $g(E(G))=S_{q-1,k,0,d}$, which means that $g$ is a \emph{harmonious $(k,d)$-total coloring} defined in Definition \ref{defn:kd-w-type-colorings}.
\item By Definition \ref{defn:kd-w-type-coloring-transfoemations}, a \emph{partly Y-image e-reciprocal $(k,d)$-total coloring} $g$ holds:

\quad (i) $g(x_r)=f(x_r)$ for $x_r\in X$;

\quad (ii) $g(y_r)=[f(y_{t})+f(y_1)]-f(y_r)$ for $y_r\in Y$; and

\quad (iii) $g(e_i)=f(e_{q-i+1})$ for each edge $e_i\in E(G)$.

We have the edge-antimagic constraint
$${
\begin{split}
g(x_r)+g(e_r)+g(y_r)&=f(x_r)+f(e_{q-r+1})+[f(y_{t})+f(y_1)]-f(y_r)\\
&=[f(y_{t})+f(y_1)]-2f(e_r)+f(e_{q-r+1})+f(e_{r})\\
&=2k+(q-1)d+2k+(q-1)d-2f(e_r)\\
&=4k+2(q-1)d-2f(e_r)\in \{2k, 2k+2d,2k+4d,\dots, 2k+2(q-1)d\}
\end{split}}
$$ Thereby, the coloring $g$ is an \emph{edge-antimagic $(k,d)$-total coloring} defined in Definition \ref{defn:kd-w-type-colorings}.
\end{asparaenum}

The proof of the theorem is complete.
\end{proof}

\begin{thm}\label{thm:connected-tree-k-d-graceful-total-coloring}
$^*$ A connected graph $G$ admits a graceful $(k,d)$-total coloring if and only if there exists a tree $T$ obtained from $G$ by the vertex-splitting tree-operation and $T$ admits a graceful $(k,d)$-total coloring, such that $T$ admits a colored graph homomorphism to $G$, that is, $T\rightarrow G$.
\end{thm}

\subsection{Trees admitting magic-constraint total colorings}

For dealing with issues based on magic-constraint total colorings we introduce the Leaf-adding $(k,d)$-magic-constraint algorithms, ``SL'' means ``From Small to Large'', and ``LS'' means ``From Large to Small'' in the following algorithms: (1) Leaf-adding LS-edge-difference (SL-edge-difference) algorithm; (2) Leaf-adding LS-graceful-difference (SL-graceful-difference) algorithm.

\begin{defn}\label{defn:leaf-added-trees}
$^*$ \textbf{Leaf-added trees based on the adding leaves operation. }Let $(X,Y)$ be the bipartition of a tree $H$, where $X=\{x_1,x_2,\dots ,x_s\}$ and $Y=\{y_1,y_2,\dots ,y_t\}$ with
$$
s+t=p=|V(H)|=1+|E(H)|
$$ We add new leaves to $H$ as: Each vertex $x_i\in X$ is added the leaves $x_{i,1},x_{i,2},\dots ,x_{i,a_i}$ of a leaf set $L_{eaf}(x_i)$ with $i\in [1,s]$, and each vertex $y_j\in Y$ is added the leaves $y_{j,1},y_{j,2},\dots ,y_{j,b_j}$ of a leaf set $L_{eaf}(y_j)$ with $j\in [1,t]$. The resultant tree is denoted as $H_L$, called \emph{leaf-added tree}. So, the leaf-added tree $H_L$ has its own bipartition $(X\,',Y\,')$, where
$$X\,'=X\bigcup \left (\bigcup^t_{k=1}L_{eaf}(y_k)\right ),~Y\,'=Y\bigcup \left (\bigcup^s_{k=1}L_{eaf}(x_k)\right )
$$ Let $A=\sum^s_{k=1}a_k$ and $B=\sum^t_{k=1}b_k$.
\end{defn}

\begin{thm}\label{thm:edge-difference-algorithm}
$^*$ Each tree $T$ with diameter $D(T)\geq 3$ admits at least $2^m$ different edge-difference $(k,d)$-total colorings with $m+1=\left \lceil \frac{D(T)}{2}\right \rceil $.
\end{thm}
\begin{proof} Use the notation and terminology of defining a leaf-added tree $H_L$ obtained from a tree $H$ in Definition \ref{defn:leaf-added-trees}. Suppose that this tree $H$ admits an edge-difference $(k,d)$-total coloring $f$ with $$f(x_iy_j)+|f(y_j)-f(x_i)|=f(x_iy_j)+f(y_j)-f(x_i)=2k+\alpha\cdot d$$ for each edge $x_iy_j\in E(H)$, and the edge color set $f(E(H))=S_{p-2,k,0,d}$. Without loss of generality, we have
$$f(x_1)\leq f(x_2)\leq \cdots \leq f(x_s),\quad f(y_1)\leq f(y_2)\leq \cdots \leq f(y_t)
$$ We define a total coloring $g$ of the leaf-added tree $H_L$ by the following algorithms.

\vskip 0.4cm

\textbf{1. Leaf-adding LS-edge-difference algorithm.}

\textbf{LSed-1.} Vertex colors are $g(x_i)=f(x_i)$ for $x_i\in X\subset X\,'$ and $g(y_j)=f(y_j)+(A+B)d$ for $y_j\in Y\subset Y\,'$.

\textbf{LSed-2.} Edge $x_iy_j$ are colored with $g(x_iy_j)=f(x_iy_j)$ for each edge $x_iy_j\in E(H)\subset E(H_L)$, so $g(x_iy_j)=f(x_iy_j)\leq e^*=k+(p-2)d$, and we have the edge-difference constraints
$$g(x_iy_j)+|g(y_j)-g(x_i)|=f(x_iy_j)+f(y_j)+(A+B)d-f(x_i)=t^*,~x_iy_j\in E(H)\subset E(H_L)$$
where $t^*=2k+(\alpha+A+B)d$.

\textbf{LSed-3.} \textbf{From Large to Small.} Set $g(y_ty_{t,i})=id+e^*$ for $i\in [1,b_t]$, and moreover
\begin{equation}\label{eqa:c3xxxxx}
g(y_{t-r}y_{t-r,i})=id+e^*+\sum ^r_{k=1}b_{t-k+1},~i\in [1,b_{t-r}],~r\in [1,t-1]
\end{equation} Clearly, $g(y_{1}y_{1,b_1})=e^*+B$. And set $g(y_{j,i})=g(y_jy_{j,i})+g(y_j)-t^*$ for each edge $y_jy_{j,i}\in E(H_L)$.

\textbf{LSed-4.} \textbf{From Large to Small.} Set $g(x_sx_{s,i})=id+e^*+B$ for $i\in [1,a_s]$, in general, we have edge colors
\begin{equation}\label{eqa:c3xxxxx}
g(x_{s-r}x_{s-r,i})=id+e^*+B+\sum ^r_{k=1}a_{s-k+1},~i\in [1,a_{s-r}],~r\in [1,s-1]
\end{equation} and $g(x_{1}x_{1,a_1})=e^*+A+B$. And we set $g(x_{j,i})=t^*+g(x_j)-g(x_jx_{j,i})$ for each edge $x_jx_{j,i}\in E(H_L)$.

Thereby, we get the edge-difference constraint $g(uv)+|g(u)-g(v)|=t^*$ for each edge $uv\in E(H_L)$ such that the edge color set
$$g(E(H_L))=S_{p-2,k,0,d}\cup \{e^*+d,e^*+2d,\dots ,e^*+A+B\}=S_{p-2+A+B,k,0,d}$$
so the Leaf-adding LS-edge-difference algorithm determines an edge-difference $(k,d)$-total coloring $g$ of the leaf-added tree $H_L$.

\vskip 0.4cm

We show another algorithm for determining another edge-difference $(k,d)$-total coloring $h$ of the leaf-added tree $H_L$ as follows:

\vskip 0.4cm

\textbf{2. Leaf-adding SL-edge-difference algorithm.}

\textbf{SLed-1.} Each vertex $x_i\in X\subset X\,'$ is colored by $h(x_i)=f(x_i)$, and each vertex $y_j\in Y\subset Y\,'$ is cored by $h(y_j)=f(y_j)+(A+B)d$.

\textbf{SLed-2.} Edge colors are $h(x_iy_j)=f(x_iy_j)$ for each edge $x_iy_j\in E(H)\subset E(H_L)$, so $h(x_iy_j)=f(x_iy_j)\leq e^*=k+(p-2)d$, and there are edge-difference constraints
$$h(x_iy_j)+|h(y_j)-h(x_i)|=f(x_iy_j)+f(y_j)+(A+B)d-f(x_i)=t^*,~x_iy_j\in E(H)\subset E(H_L)$$
where $t^*=2k+(\alpha+A+B)d$.

\textbf{SLed-3.} \textbf{From Small to Large.} Set $h(x_{1}x_{1,i})=id+e^*$ for $i\in [1,a_1]$, in general, there is
\begin{equation}\label{eqa:c3xxxxx}
h(x_{r}x_{r,i})=id+e^*+\sum ^{r-1}_{k=1}a_{k},~i\in [1,a_{r}],~r\in [2,s]
\end{equation} and $h(x_{s}x_{s,a_s})=e^*+A$. And we color vertex $x_{j,i}$ with $h(x_{j,i})=t^*+h(x_j)-h(x_{j}x_{j,i})$ for each edge $x_{j}x_{j,i}\in E(H_L)$.

\textbf{SLed-4.} \textbf{From Small to Large.} Set $h(y_1y_{1,i})=id+e^*+A$ for $i\in [1,b_1]$, and moreover
\begin{equation}\label{eqa:c3xxxxx}
h(y_{r}y_{r,i})=id+e^*+A+\sum ^{r-1}_{k=1}b_{k},~i\in [1,b_{r}],~r\in [2,t]
\end{equation} so, $h(y_{t}y_{t,b_t})=e^*+A+B$. We set $h(y_{j,i})=h(y_jy_{j,i})+h(y_j)-t^*$ for each edge $y_jy_{j,i}\in E(H_L)$.

Clearly, the edge-difference constraint $h(uv)+|h(u)-h(v)|=t^*$ holds for each edge $uv\in E(H_L)$, and the edge color set
$$
h(E(H_L))=S_{p-2+A+B,k,0,d}
$$ that is the Leaf-adding LS-edge-difference algorithm determines another edge-difference $(k,d)$-total coloring $h$ of the leaf-added tree $H_L$.

\vskip 0.4cm

Let $T_1$ be a tree being not a star, and let $L_{eaf}(T_1)$ be the set of all leaves of $T_1$. Furthermore, we have trees $T_{i+1}=T_i-L_{eaf}(T_i)$ with $i\in [1,m]$, and each tree $T_i$ has its own leaf set $L_{eaf}(T_i)$, $T_{m+1}$ is a star. Since the star $T_{m+1}$ admits an edge-difference $(k,d)$-total coloring, we apply two algorithms introduced above to each tree $T_{i+1}$ for determining two different edge-difference $(k,d)$-total colorings of the leaf-added tree $T_i=T_{i+1}+L_{eaf}(T_i)$, since each number $|L_{eaf}(T_i)|\geq 2$ according to each tree having at least two leaves.

In the above each step, we have two ways to realize edge-difference $(k,d)$-total colorings, so $T_1$ admits at least $2^m$ different edge-difference $(k,d)$-total colorings.

This is the complete proof of the theorem.
\end{proof}

\begin{figure}[h]
\centering
\includegraphics[width=16.4cm]{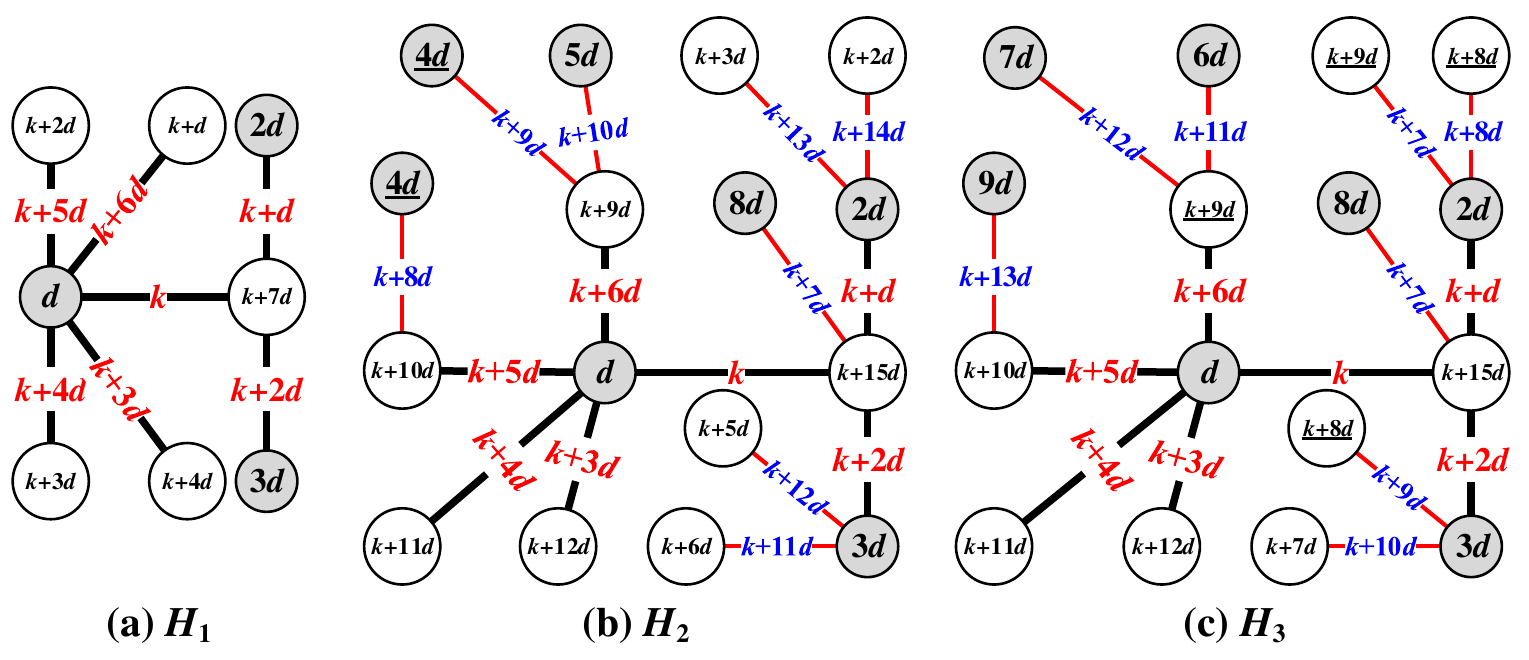}
\caption{\label{fig:1-edge-difference}{\small (a) $H_1$ admits an edge-difference $(k,d)$-total coloring $f_1$ holding the edge-difference constraint $f_1(uv)+|f_1(u)-f_1(v)|=2k+6d$; (b) $H_2$ admits an edge-difference $(k,d)$-total coloring $f_2$ holding the edge-difference constraint $f_2(uv)+|f_2(u)-f_2(v)|=2k+14d$ from Large to Small; (c) $H_3$ admits an edge-difference $(k,d)$-total coloring $f_3$ holding the edge-difference constraint $f_3(uv)+|f_3(u)-f_3(v)|=2k+14d$ from Small to Large.}}
\end{figure}

\begin{thm}\label{thm:graceful-difference-algorithm}
$^*$ Each tree $T$ with diameter $D(T)\geq 3$ admits at least $2^m$ different graceful-difference $(k,d)$-total colorings with $m+1=\left \lceil \frac{D(T)}{2}\right \rceil $.
\end{thm}
\begin{proof} Use the notation and terminology of defining a leaf-added tree $H_L$ obtained from a tree $H$ in Definition \ref{defn:leaf-added-trees}. Suppose that this tree $H$ admits a graceful-difference $(k,d)$-total coloring $f$ with the graceful-difference constraint
$$
\big ||f(u)-f(v)|-f(uv)\big |=\beta\cdot d,~uv\in E(H)
$$ and the edge color set $f(E(H))=S_{p-2,k,0,d}$, and let $e^*=k+(p-2)d$. Without loss of generality, there are vertex color inequalities
$$f(x_1)\leq f(x_2)\leq \cdots \leq f(x_s),~f(y_1)\leq f(y_2)\leq \cdots \leq f(y_t)
$$ We define a total coloring $g$ for the leaf-added tree $H_L$ by the following algorithm.

\vskip 0.2cm

\textbf{1. Leaf-adding SL-graceful-difference algorithm.}

\textbf{SLgd-1.} Set vertex colors $g(x_i)=f(x_i)+(A+B)d$ for $x_i\in X\subset X\,'$, and edge colors $g(y_j)=f(y_j)+(A+B)d$ for $y_j\in Y\subset Y\,'$, and $g(x_iy_j)=f(x_iy_j)+(A+B)d$ for each edge $x_iy_j\in E(H)$. Thus, each edge $uv\in E(H)\subset E(H_L)$ holds the graceful-difference constraint
$$\big ||g(u)-g(v)|-g(uv)\big |=\big ||f(u)-f(v)|-f(uv)-(A+B)d\big |=(\beta+A+B)d$$
and let $t^*=(\beta+A+B)d$.

\textbf{SLgd-2.} \textbf{From Small to Large.} Set edge colors $g(x_{1}x_{1,i})=id+e^*$ for $i\in [1,a_1]$, in general, we get
\begin{equation}\label{eqa:c3xxxxx}
g(x_{r}x_{r,i})=id+e^*+\sum ^{r-1}_{k=1}a_{k},~i\in [1,a_{r}],~r\in [2,s]
\end{equation} and $g(x_{s}x_{s,a_s})=e^*+A$. And we set $g(x_{j,i})=g(x_{j}x_{j,i})+g(x_j)-t^*$ for each edge $x_{j}x_{j,i}\in E(H_L)$.

\textbf{SLgd-3.} \textbf{From Small to Large.} Set edge colors $g(y_1y_{1,i})=id+e^*+A$ for $i\in [1,b_1]$, and moreover
\begin{equation}\label{eqa:c3xxxxx}
g(y_{r}y_{r,i})=id+e^*+A+\sum ^{r-1}_{k=1}b_{k},~i\in [1,b_{r}],~r\in [2,t]
\end{equation} so, $g(y_{t}y_{t,b_t})=e^*+A+B$. We color $y_{j,i}$ with $g(y_{j,i})=g(y_j)-g(y_jy_{j,i})-t^*$, or $g(y_{j,i})=g(y_j)+t^*-g(y_jy_{j,i})$ for each edge $y_jy_{j,i}\in E(H_L)$. Thereby, we get the following graceful-difference constraints

1. $\big ||g(x_j)-g(x_{j,i})|-g(x_jx_{j,i})\big |=t^*$ for each edge $x_{j}x_{j,i}\in E(H_L)$, and

2. $\big ||g(y_j)-g(y_{j,i})|-g(y_jy_{j,i})\big |=t^*$ for each edge $y_jy_{j,i}\in E(H_L)$. \\
So, $g$ is a graceful-difference $(k,d)$-total coloring of the leaf-added tree $H_L$.

\vskip 0.4cm

We show another algorithm for defining another graceful-difference $(k,d)$-total coloring $h$ of the leaf-added tree $H_L$ as follows:

\vskip 0.4cm

\textbf{2. Leaf-adding LS-graceful-difference algorithm.}

\textbf{LSgd-1.} Set vertex colors $h(x_i)=f(x_i)+(A+B)d$ for $x_i\in X\subset X\,'$, and $h(y_j)=f(y_j)+(A+B)d$ for $y_j\in Y\subset Y\,'$, and edge colors are $h(x_iy_j)=f(x_iy_j)+(A+B)d$ for each edge $x_iy_j\in E(H)$. So, each edge $uv\in E(H)\subset E(H_L)$ holds the graceful-difference constraint
$$\big ||h(u)-h(v)|-h(uv)\big |=\big ||f(u)-f(v)|-f(uv)-(A+B)d\big |=(\beta+A+B)d$$
and let $t^*=(\beta+A+B)d$.

\textbf{LSgd-2.} \textbf{From Large to Small.} Set edge colors $h(y_ty_{t,i})=id+e^*$ for $i\in [1,b_t]$, and moreover
\begin{equation}\label{eqa:c3xxxxx}
h(y_{t-r}y_{t-r,i})=id+e^*+\sum ^r_{k=1}b_{t-k+1},~i\in [1,b_{t-r}],~r\in [1,t-1]
\end{equation} Clearly, $h(y_{1}y_{1,b_1})=e^*+B$. And we set $h(y_{j,i})=h(y_j)-h(y_jy_{j,i})-t^*$ for each edge $y_jy_{j,i}\in E(H_L)$, also, the graceful-difference constraint
$$\big ||h(y_{j,i})-h(y_j)|-h(y_jy_{j,i})\big |=t^*,~y_jy_{j,i}\in E(H_L)$$
holds true.

\textbf{LSgd-3.} \textbf{From Large to Small.} Set edge colors $h(x_sx_{s,i})=id+e^*+B$ for $i\in [1,a_s]$, in general, we have
\begin{equation}\label{eqa:c3xxxxx}
h(x_{s-r}x_{s-r,i})=id+e^*+B+\sum ^r_{k=1}a_{s-k+1},~i\in [1,a_{s-r}],~r\in [1,s-1]
\end{equation} and $h(x_{1}x_{1,a_1})=e^*+A+B$. And we set $h(x_{j,i})=h(x_j)+h(x_jx_{j,i})-t^*$ for each edge $x_jx_{j,i}\in E(H_L)$, that is the graceful-difference constraint
$$\big ||h(x_{j,i})-h(x_j)|-h(x_jx_{j,i})\big |=t^*,~x_jx_{j,i}\in E(H_L)$$
holds true.

Thereby, the leaf-added tree $H_L$ admits $h$ as its graceful-difference $(k,d)$-total coloring.

\vskip 0.4cm

Let $T_1$ be a tree being not a star, and let $L_{eaf}(T_1)$ be the set of all leaves of $T_1$. Furthermore, we have trees $T_{i+1}=T_i-L_{eaf}(T_i)$ with $i\in [1,m]$, and each tree $T_i$ has its own leaf set $L_{eaf}(T_i)$, $T_{m+1}$ is a star. Since the star $T_{m+1}$ admits an graceful-difference $(k,d)$-total coloring, we apply two algorithms introduced above to each tree $T_{i+1}$ for determining two different graceful-difference $(k,d)$-total colorings of the leaf-added tree $T_i=T_{i+1}+L_{eaf}(T_i)$, because of each number $|L_{eaf}(T_i)|\geq 2$ from each tree having at least two leaves.

In the above each step, we have two ways to realize graceful-difference $(k,d)$-total colorings, so $T_1$ admits at least $2^m$ different graceful-difference $(k,d)$-total colorings.

This theorem is covered.
\end{proof}

\begin{figure}[h]
\centering
\includegraphics[width=16.4cm]{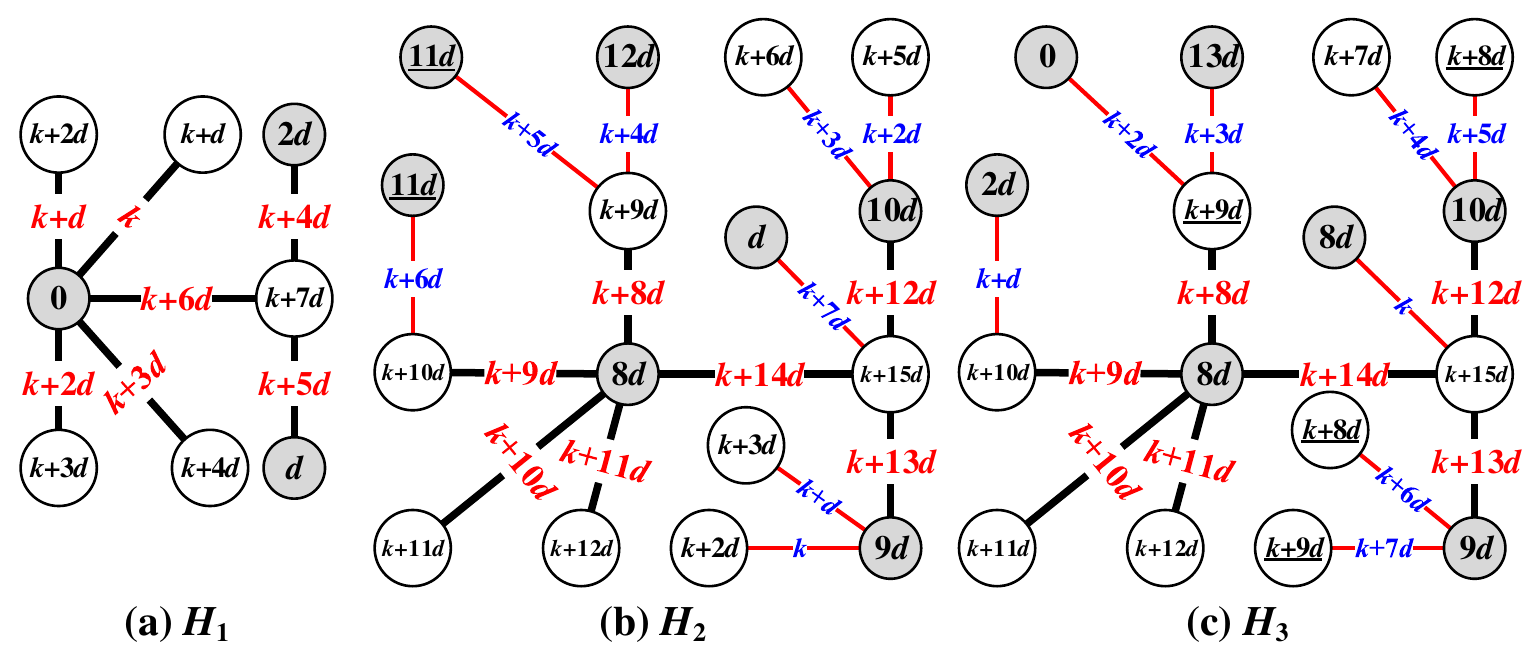}
\caption{\label{fig:2-graceful-difference}{\small (a) $H_1$ admits a graceful-difference $(k,d)$-total coloring $g_1$ holding the graceful-difference constraint $\big ||g_1(u)-g_1(v)|-g_1(uv)\big |=d$; (b) $H_2$ admits a graceful-difference $(k,d)$-total coloring $g_2$ holding the graceful-difference constraint $\big ||g_2(u)-g_2(v)|-g_2(uv)\big |=7d$ from Small to Large; (c) $H_3$ admits a graceful-difference $(k,d)$-total coloring $g_3$ holding the graceful-difference constraint $\big ||g_3(u)-g_3(v)|-g_3(uv)\big |=7d$ from Large to Small.}}
\end{figure}

By Theorem \ref{thm:any-tree-k-d-graceful-total-coloring} and the transformation from a graceful $(k,d)$-total coloring to a felicitous-difference $(k,d)$-total coloring, or to an edge-magic $(k,d)$-total coloring in the proof of Theorem \ref{thm:equivalent-k-d-total-colorings}, we have the following two results in Theorem \ref{thm:felicitous-difference-algorithm} and Theorem \ref{thm:edge-magic-algorithm}:

\begin{thm}\label{thm:felicitous-difference-algorithm}
$^*$ Each tree $T$ with diameter $D(T)\geq 3$ admits at least $2^m$ different felicitous-difference $(k,d)$-total colorings for $m+1=\left \lceil \frac{D(T)}{2}\right \rceil $.
\end{thm}

See examples of trees admitting felicitous-difference $(k,d)$-total colorings shown in Fig.\ref{fig:33-felicitous-difference} and Fig.\ref{fig:3-felicitous-difference}.

\begin{figure}[h]
\centering
\includegraphics[width=16.4cm]{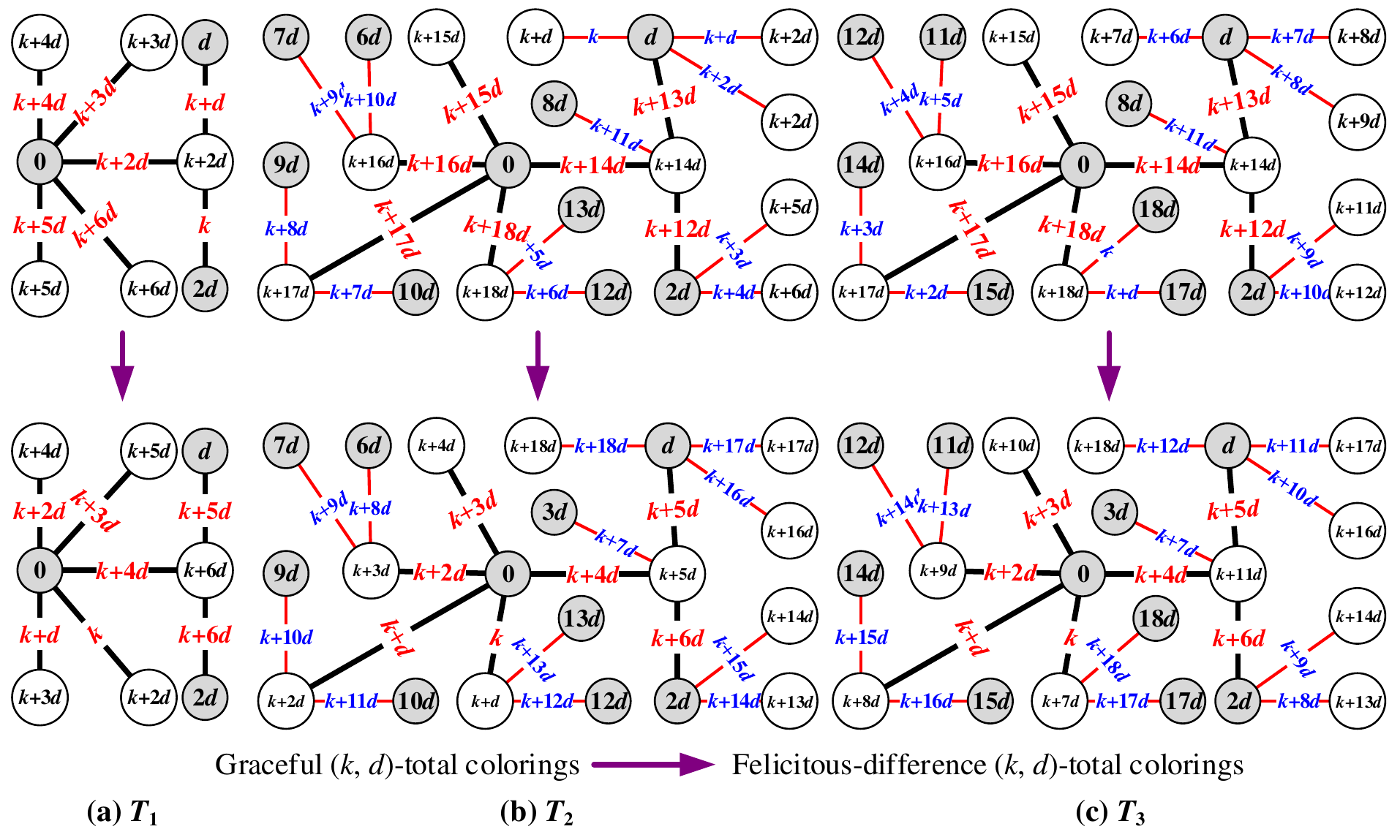}
\caption{\label{fig:33-felicitous-difference}{\small A scheme for illustrating Theorem \ref{thm:felicitous-difference-algorithm}, where (a) $T_1$ admits a felicitous-difference $(k,d)$-total coloring $f_1$ holding the felicitous-difference constraint $|f_1(u)+f_1(v)-f_1(uv)|=2d$; (b) $T_2$ admits a felicitous-difference $(k,d)$-total coloring $f_2$ holding the felicitous-difference constraint $|f_2(u)+f_2(v)-f_2(uv)|=d$; (c) $T_3$ admits a felicitous-difference $(k,d)$-total coloring $f_3$ holding the felicitous-difference constraint $|f_3(u)+f_3(v)-f_3(uv)|=7d$.}}
\end{figure}

\begin{figure}[h]
\centering
\includegraphics[width=16.4cm]{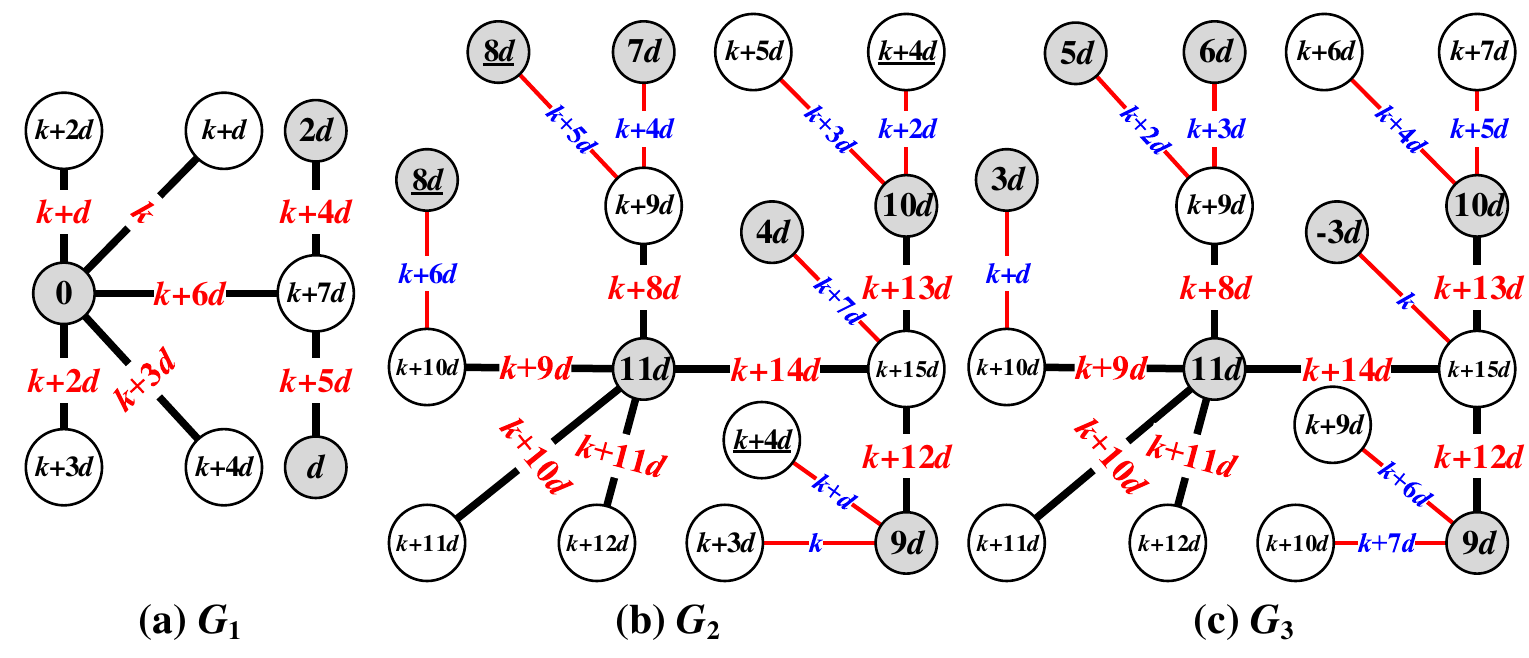}
\caption{\label{fig:3-felicitous-difference}{\small (a) $G_1$ admits a felicitous-difference $(k,d)$-total coloring $h_1$ holding the felicitous-difference constraint $|h_1(u)+h_1(v)-h_1(uv)|=4d$; (b) $G_2$ admits a felicitous-difference $(k,d)$-total coloring $h_2$ holding the felicitous-difference constraint $|h_2(u)+h_2(v)-h_2(uv)|=12d$ from Small to Large; (c) $G_3$ admits a total coloring $h_3$ holding the felicitous-difference constraint $|h_3(u)+h_3(v)-h_3(uv)|=12d$ from Large to Small.}}
\end{figure}

\begin{thm}\label{thm:edge-magic-algorithm}
$^*$ Each tree $T$ with diameter $D(T)\geq 3$ admits at least $2^m$ different edge-magic $(k,d)$-total colorings with $m+1=\left \lceil \frac{D(T)}{2}\right \rceil $.
\end{thm}

See examples shown in Fig.\ref{fig:00-edge-magic} and Fig.\ref{fig:0-edge-magic} for understanding Theorem \ref{thm:edge-magic-algorithm}.

\begin{figure}[h]
\centering
\includegraphics[width=16.2cm]{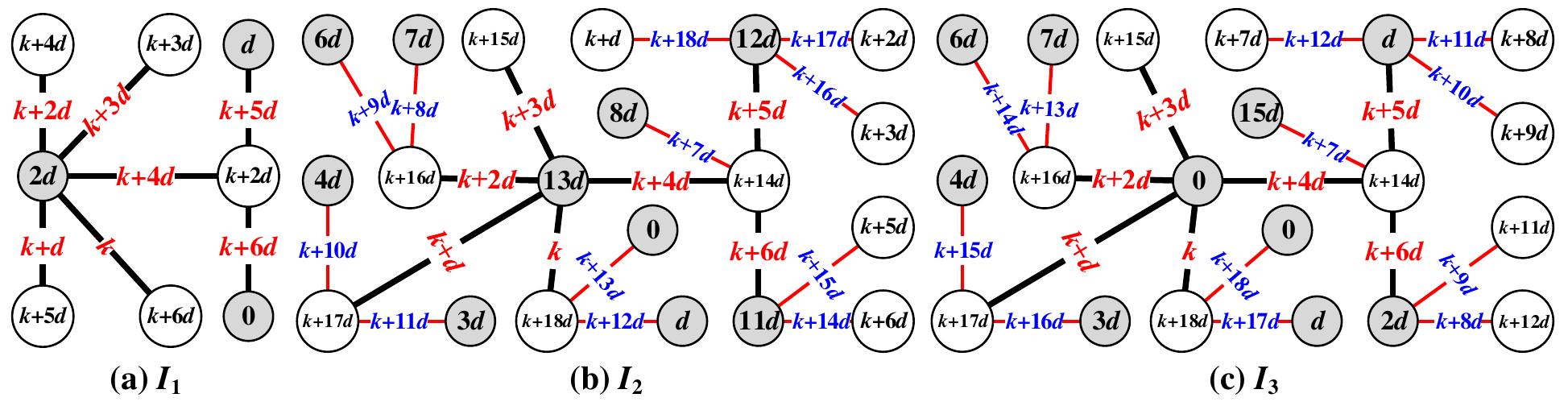}
\caption{\label{fig:00-edge-magic}{\small (a) $I_1$ admits an edge-magic $(k,d)$-total coloring $s_1$ holding the edge-magic constraint $s_1(u)+s_1(uv)+s_1(v)=2k+8d$; (b) $I_2$ admits an edge-magic $(k,d)$-total coloring $s_2$ holding the edge-magic constraint $s_2(u)+s_2(uv)+s_2(v)=2k+31d$ from Small to Large; (c) $I_3$ admits an edge-magic $(k,d)$-total coloring $s_3$ holding the edge-magic constraint $s_3(u)+s_3(uv)+s_3(v)=2k+36d$ from Large to Small.}}
\end{figure}

\begin{figure}[h]
\centering
\includegraphics[width=16.4cm]{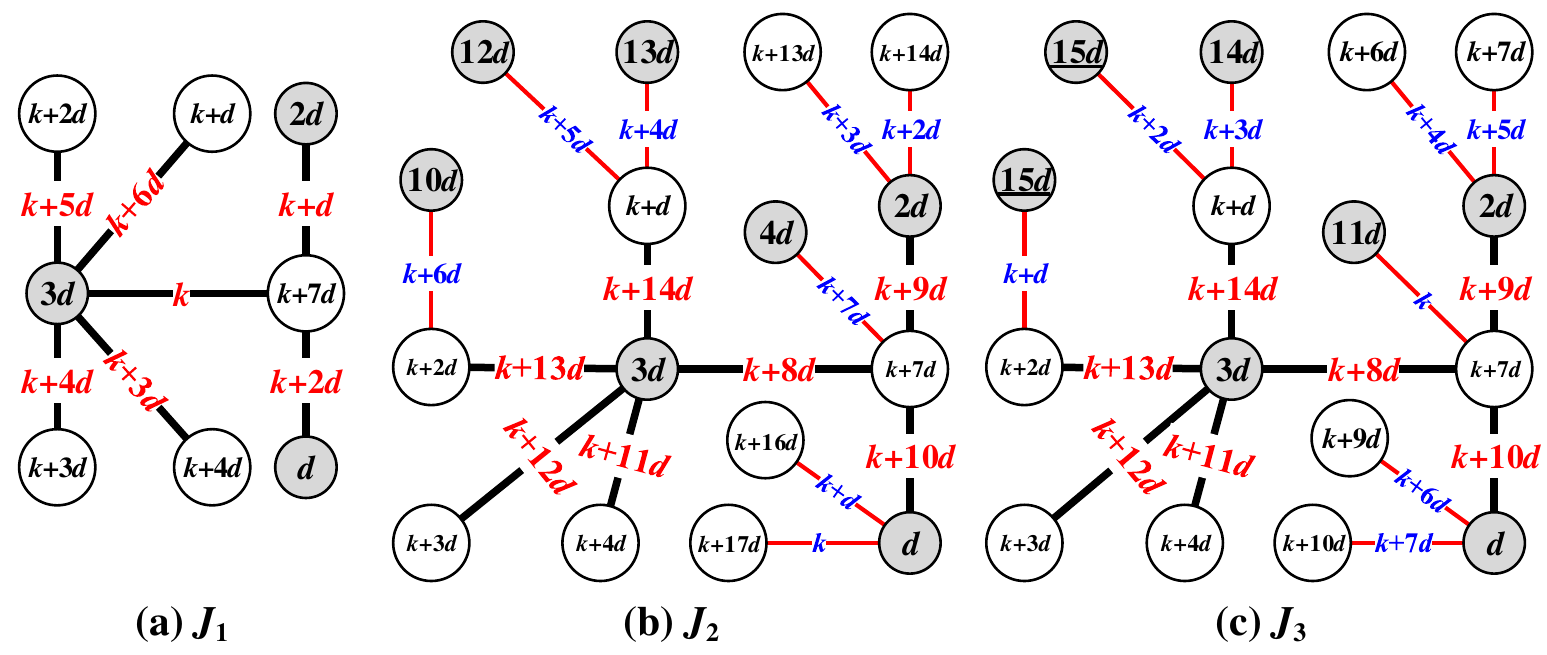}
\caption{\label{fig:0-edge-magic}{\small (a) $J_1$ admits an edge-magic $(k,d)$-total coloring $l_1$ holding the edge-magic constraint $l_1(u)+l_1(uv)+l_1(v)=2k+10d$; (b) $J_2$ admits an edge-magic $(k,d)$-total coloring $l_2$ holding the edge-magic constraint $l_2(u)+l_2(uv)+l_2(v)=2k+18d$ from Small to Large; (c) $J_3$ admits an edge-magic $(k,d)$-total coloring $l_3$ holding the edge-magic constraint $l_3(u)+l_3(uv)+l_3(v)=2k+18d$ from Large to Small.}}
\end{figure}

\section{Magic-constraint parameterized algorithms}

In this subsection, the sentence ``RANDOMLY-LEAF-adding algorithm'' is abbreviated as ``RLA-algorithm''.

\subsection{Parameterized graceful-difference total labelings/colorings}

\noindent $^*$ \textbf{RLA-algorithm-A for the graceful-difference $(k,d)$-total coloring}.

\textbf{Input:} A connected bipartite $(p,q)$-graph $G$ admitting a set-ordered graceful labeling $f$.

\textbf{Output:} A connected bipartite $(p+m,q+m)$-graph $G^*$ admitting a graceful-difference $(k,d)$-total coloring, where $G^*$ is the result of adding randomly $m$ leaves to $G$.

\textbf{Step A-1.} \textbf{Initialization.} Since a connected bipartite $(p,q)$-graph $G$ admitting a set-ordered graceful labeling $f$, we have the vertex set $V(G)=X\cup Y$ with $X\cap Y=\emptyset$, where
$$
X=\{x_1,x_2,\dots ,x_s\},~Y=\{y_1,y_2,\dots ,y_t\},~s+t=p=|V(G)|
$$ such that
$$
0=f(x_1)<f(x_2)<\cdots <f(x_s)<f(y_1)<f(y_2)<\cdots <f(y_t)=q
$$ with the edge color set
$$f(E(G))=\{f(x_iy_j)=f(y_j)-f(x_i):x_iy_j\in E(G)\}=[1,q]$$

Adding randomly $a_i$ new leaves $u_{i,j}\in L(x_i)$ to each vertex $x_i$ by adding new edges $x_iu_{i,j}$ for $j\in [1,a_i]$ and $i\in[1,s]$, and adding randomly $b_j$ new leaves $v_{j,r}\in L(y_j)$ to each vertex $y_j$ by adding new edges $y_jv_{j,r}$ for $r\in [1,b_j]$ and $j\in[1,t]$, it is allowed that some $a_i=0$ or some $b_j=0$. The resultant graph is denoted as $G^*$. Let $A=\sum ^s_{i=1}a_i$ and $B=\sum ^t_{j=1}b_j$, so $m=A+B$.

\textbf{Step A-2.} Define a graceful-difference total labeling $f_{gr}$ for $G$ by the set-ordered graceful labeling $f$ in the following substeps:

\textbf{Step A-2.1.} \textbf{$X$-dual transformation.} Set
$$
f_{gr}(x_i)=\max f(X)+\min f(X)-f(x_i)=\max f(X)-f(x_i),~x_i\in X
$$

\textbf{Step A-2.2.} \textbf{$Y$-dual transformation.} Set
$$
f_{gr}(y_j)=\max f(Y)+\min f(Y)-f(y_j)=q+\min f(Y)-f(y_j),~y_j\in Y
$$

\textbf{Step A-2.3.} \textbf{Edge-dual transformation.} Each edge $x_iy_j\in E(G)$ is colored as
$$
f_{gr}(x_iy_j)=\max f(E(G))+\min f(E(G))-f(x_iy_j)=q+1-f(x_iy_j)
$$

There is the following the graceful-difference constraint
\begin{equation}\label{eqa:graceful-difference-total-labeling-11}
{
\begin{split}
&\quad \big | |f_{gr}(x_i)-f_{gr}(y_j)|-f_{gr}(x_iy_j)\big |=\big | f_{gr}(y_j)-f_{gr}(x_i)-f_{gr}(x_iy_j)\big |\\
&=\big |q+\min f(Y)-f(y_j)-[\max f(X)-f(x_i)]-[q+1-f(x_iy_j)]\big |\\
&=\big |\min f(Y)-\max f(X)+f(x_iy_j)-f(y_j)+f(x_i)-1\big |\\
&=\min f(Y)-\max f(X)-1
\end{split}}
\end{equation} for each edge $x_iy_j\in E(G)$, which means that $f_{gr}$ is a graceful-difference total labeling of $G$.

\textbf{Step A-3.} \textbf{Parameterizing the graceful-difference total labeling.} For integers $k\geq 0$ and $d\geq 1$, we construct a $(k,d)$-constraint labeling $g_{gr}$ for $G$ as follows:

A-3.1. $g_{gr}(x_i)=f_{gr}(x_i)\cdot d$ for $x_i\in X$;

A-3.2. $g_{gr}(y_j)=k+[f_{gr}(y_j)-1]\cdot d$ for $y_j\in Y$; and

A-3.3. $g_{gr}(x_iy_j)=k+[f_{gr}(x_iy_j)-1]\cdot d$ for each edge $x_iy_j\in E(G)$.

So, we obtain the vertex color set
\begin{equation}\label{eqa:555555}
g_{gr}(V(G))=\{f(x_i)\cdot d:x_i\in X\}\cup \{k+[f_{gr}(y_j)-1]\cdot d:y_j\in Y\}=g_{gr}(X)\cup g_{gr}(Y)
\end{equation}
and the edge color set
\begin{equation}\label{eqa:555555}
g_{gr}(E(G))=\{k+[f_{gr}(x_iy_j)-1]\cdot d:x_iy_j\in E(G)\}=\{k,k+d,\dots, k+(q-1)\cdot d\}
\end{equation}
Moreover, for each edge $x_iy_j\in E(G)$, we use Eq.(\ref{eqa:graceful-difference-total-labeling-11}) to get the graceful-difference constraint
\begin{equation}\label{eqa:555555}
{
\begin{split}
&\big | |g_{gr}(x_i)-g_{gr}(y_j)|-g_{gr}(x_iy_j)\big |\\
=&\big | k+[f_{gr}(y_j)-1]\cdot d-f_{gr}(x_i)\cdot d-(k+[f_{gr}(x_iy_j)-1]\cdot d)\big |\\
=&[\min f(Y)-\max f(X)-1]\cdot d
\end{split}}
\end{equation} Thereby, we claim that $g_{gr}$ is a graceful-difference $(k,d)$-total labeling of $G$ by Definition \ref{defn:kd-w-type-colorings}.

\textbf{Step A-4.} Define a total coloring $h_{gr}$ for $G^*$ in the following substeps:

\textbf{Step A-4.1.} Recolor the vertices and edges of $G$ as:

A-4.1. $h_{gr}(x_i)=g_{gr}(x_i)\cdot d$ for $x_i\in X$;

A-4.2. $h_{gr}(y_j)=g_{gr}(y_j)+m\cdot d$ for $y_j\in Y$; and

A-4.3. $h_{gr}(x_iy_j)=g_{gr}(x_iy_j)+m\cdot d$ for each edge $x_iy_j\in E(G)$.

Then, there is the graceful-difference constraint
\begin{equation}\label{eqa:555555}
{
\begin{split}
\big | |h_{gr}(x_i)-h_{gr}(y_j)|-h_{gr}(x_iy_j)\big |=&\big | |g_{gr}(y_j)+m\cdot d-g_{gr}(x_i)|-g_{gr}(x_iy_j)-m\cdot d\big |\\
=&[\min f(Y)-\max f(X)-1]\cdot d
\end{split}}
\end{equation} for each edge $x_iy_j\in E(G)$. Notice that the edge color set
$$h_{gr}(E(G))=\{k+m\cdot d,k+(m+1)\cdot d,\dots, k+(q+m-1)\cdot d\}$$

\textbf{Step A-4.2.} Color edges $y_jv_{j,r}$ for $v_{j,r}\in L(y_j)$ with $r\in [1,b_j]$ and $j\in[1,t]$.

Since each edge $y_1v_{1,r}$ generated by the leaves of $L(y_1)$ is colored with $h_{gr}(y_1v_{1,r})=k+(r-1)\cdot d$ for $r\in [1,b_1]$, and $h_{gr}(y_1v_{1,r})=k+(b_1-1)\cdot d$; we have a formula \begin{equation}\label{eqa:555555}
h_{gr}(y_jv_{j,r})=k+(r-1)\cdot d+d\sum ^{j-1}_{i=1}b_i,~r\in [1,b_j]
\end{equation} and
$$
h_{gr}(y_jv_{j,b_j})=k+(b_j-1)\cdot d+d\sum ^{j-1}_{i=1}b_i=k-d+d\sum ^{j}_{i=1}b_i
$$ The last edge $y_tv_{t,b_t}$ is colored with
$$
h_{gr}(y_tv_{t,b_t})=k+(b_t-1)\cdot d+d\sum ^{t-1}_{i=1}b_i=k-d+d\sum ^{t}_{i=1}b_i=k+(B-1)\cdot d
$$

\textbf{Step A-4.3.} Color edges $x_iu_{i,j}$ for $u_{i,j}\in L(x_i)$ with $j\in [1,a_i]$ and $i\in[1,s]$.

Since $h_{gr}(x_su_{s,j})=k+(j-1)\cdot d+B\cdot d$ for $j\in [1,a_s]$, $h_{gr}(x_su_{s,a_s})=k+(a_s-1)\cdot d+B\cdot d$, we have a formula
\begin{equation}\label{eqa:555555}
h_{gr}(x_iu_{i,j})=k+(j-1)\cdot d+B\cdot d+d\sum^{s-i}_{r=1}a_{s-r+1},~j\in [1,a_i],~i\in[1,s-1]
\end{equation} The last edge $x_1u_{1,a_1}$ is colored with
$$
h_{gr}(x_1u_{1,a_1})=k+(a_1-1)\cdot d+B\cdot d+d\sum^{s-1}_{r=1}a_{s-r+1}=k+(A+B-1)\cdot d=k+(m-1)\cdot d
$$

\textbf{Step A-5.} Let $M_{gr}=[\min f(Y)-\max f(X)-1]$. Color each added leaf in the following substeps:

\textbf{Step A-5.1.} Color each added leaf $u_{i,j}\in L(x_i)$ for $j\in [1,a_i]$ and $i\in[1,s]$ with
$$
h_{gr}(u_{i,j})=h_{gr}(x_i)+h_{gr}(x_iu_{i,j})+M_{gr}
$$ so we have the graceful-difference constraints
$$
\big | |h_{gr}(u_{i,j})-h_{gr}(x_i)|-h_{gr}(x_iu_{i,j})\big |=M_{gr}\cdot d,\quad x_iu_{i,j}\in E(G^*)
$$ hold true.

\textbf{Step A-5.2.} Color each added leaf $v_{j,r}\in L(y_j)$ with $r\in [1,b_j]$ and $j\in[1,t]$ with $M_{gr}+h_{gr}(v_{j,r})=h_{gr}(y_j)-h_{gr}(y_jv_{j,r})$, that are the graceful-difference constraints
$$\big | |h_{gr}(y_j)-h_{gr}(v_{j,r})|-h_{gr}(y_jv_{j,r})\big |=M_{gr}\cdot d,~y_jv_{j,r}\in E(G^*)$$

\textbf{Step A-6.} Return the graceful-difference $(k,d)$-total coloring $h_{gr}$ of $G^*$.

\vskip 0.4cm

The RLA-algorithm-A enables us to get the following result:

\begin{thm}\label{thm:trees-admits-graceful-difference-k-d}
$^*$ Each tree admits a graceful-difference $(k,d)$-total coloring and an odd-edge graceful-difference $(k,d)$-total coloring.
\end{thm}

An example shown in Fig.\ref{fig:4-graceful-difference-leaves} is for illustrating the RLA-algorithm-A, and we can see:

(a) A connected graph $G$ admits a set-ordered graceful-difference total labeling $f_{gr}$ holding the graceful-difference constraint
$$
\big | |f_{gr}(x)-f_{gr}(y)|-f_{gr}(xy)\big |=0,~xy\in E(G)
$$

(b) a parameterized connected graph $G_{\textrm{p}}$ admitting a graceful-difference $(k,d)$-total labeling $g_{gr}$ with the graceful-difference constraint
$$\big | |g_{gr}(x)-g_{gr}(y)|-g_{gr}(xy)\big |=0,~xy\in E(G_{\textrm{p}})
$$

(c) adding leaves to $G_{\textrm{p}}$ produces a connected graph $G_{\textrm{p-leaf}}$;

(d) a connected graph $Gh_{\textrm{p-leaf}}$ admits a graceful-difference $(k,d)$-total coloring $h_{gr}$ with the graceful-difference constraint
$$\big | |h_{gr}(x)-h_{gr}(y)|-h_{gr}(xy)\big |=2d,~
$$

\begin{figure}[h]
\centering
\includegraphics[width=16.4cm]{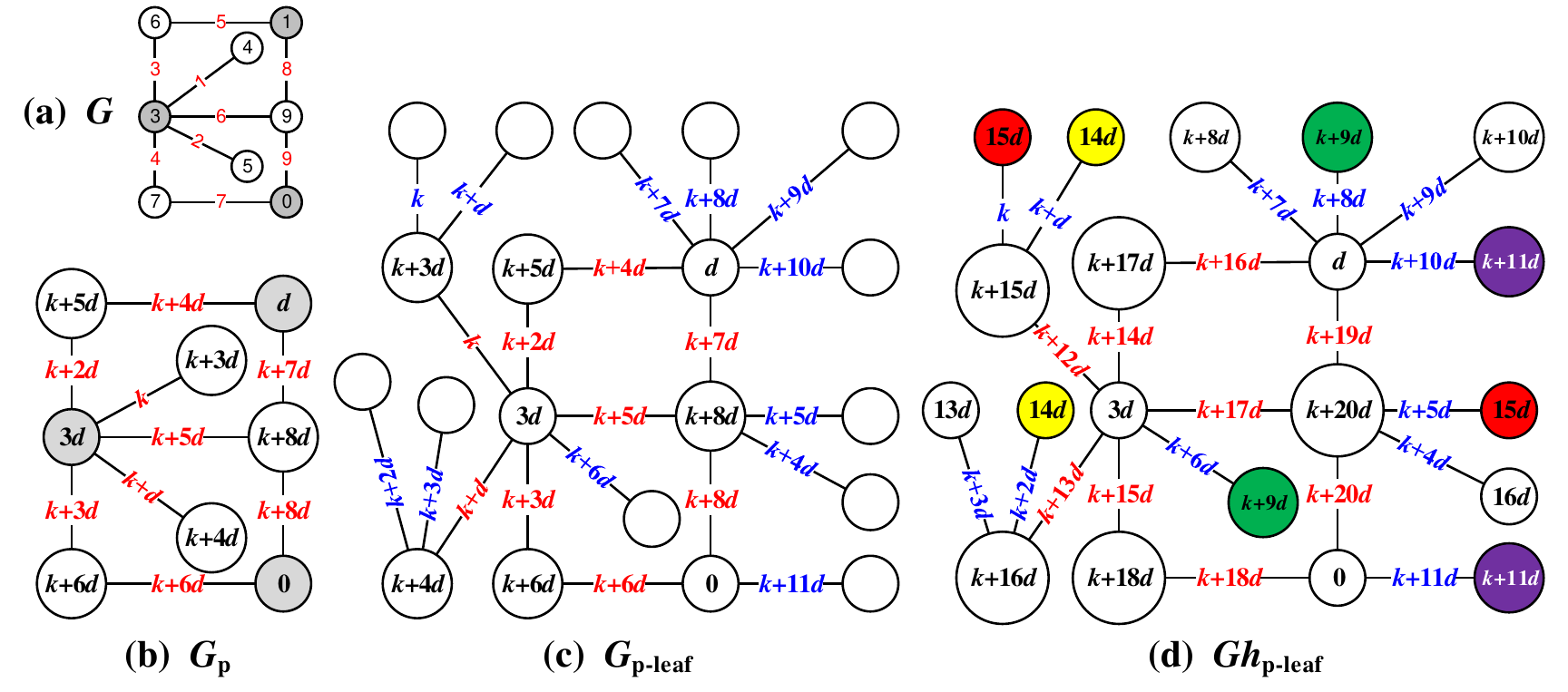}\\
\caption{\label{fig:4-graceful-difference-leaves} {\small Adding leaves to a connected graph $G$ makes a large scale connected graph $Gh_{\textrm{p-leaf}}$ admitting a graceful-difference $(k,d)$-total coloring.}}
\end{figure}

\begin{thm}\label{thm:666666}
$^*$ If a connected graph $G$ admits a set-ordered graceful labeling, then $G$ admits a graceful-difference $(k,d)$-total coloring and an odd-edge graceful-difference $(k,d)$-total coloring.
\end{thm}

\begin{thm}\label{thm:666666}
(i) $^*$ A connected graph $G$ admits a graceful-difference $(k,d)$-total coloring if and only if there exists a tree $T$ admitting a graceful-difference $(k,d)$-total coloring such that the result of vertex-coinciding each group of vertices colored the same colors into one vertex is just $G$.

(ii) $^*$ A connected bipartite graph $H$ admits an odd-edge graceful-difference $(k,d)$-total coloring if and only if there exists a tree $T^*$ admitting an odd-edge graceful-difference $(k,d)$-total coloring such that the result of vertex-coinciding each group of vertices colored the same color in $T^*$ into one vertex is just $H$.
\end{thm}

\subsection{Parameterized edge-difference total labelings/colorings}

\vskip 0.4cm

\noindent $^*$ \textbf{RLA-algorithm-B for the edge-difference $(k,d)$-total coloring}.

\textbf{Input:} A connected bipartite $(p,q)$-graph $G$ admitting a set-ordered graceful labeling $f$.

\textbf{Output:} A connected bipartite $(p+m,q+m)$-graph $G^*$ admitting an edge-difference $(k,d)$-total coloring, where $G^*$ is the result of adding randomly $m$ leaves to $G$.

\textbf{Step B-1.} \textbf{Initialization.} Since a connected bipartite $(p,q)$-graph $G$ admitting a set-ordered graceful labeling $f$, we have the vertex set $V(G)=X\cup Y$ with $X\cap Y=\emptyset$, where
$$
X=\{x_1,x_2,\dots ,x_s\},~Y=\{y_1,y_2,\dots ,y_t\},~s+t=p=|V(G)|
$$ such that
$$
0=f(x_1)<f(x_2)<\cdots <f(x_s)<f(y_1)<f(y_2)<\cdots <f(y_t)=q
$$ with the edge color set
$$f(E(G))=\{f(x_iy_j)=f(y_j)-f(x_i):x_iy_j\in E(G)\}=[1,q]$$

Adding randomly $a_i$ new leaves $u_{i,j}\in L(x_i)$ to each vertex $x_i$ by adding new edges $x_iu_{i,j}$ for $j\in [1,a_i]$ and $i\in[1,s]$, and adding randomly $b_j$ new leaves $v_{j,r}\in L(y_j)$ to each vertex $y_j$ by adding new edges $y_jv_{j,r}$ for $r\in [1,b_j]$ and $j\in[1,t]$, it is allowed that some $a_i=0$ or some $b_j=0$. The resultant graph is denoted as $G^*$. Let $A=\sum ^s_{i=1}a_i$ and $B=\sum ^t_{j=1}b_j$, so $m=A+B$.

\textbf{Step B-2.} Define an edge-difference total labeling $f_{ed}$ for $G$ in the following way:

B-2.1. \textbf{$X$-dual transformation.} Set
$$
f_{ed}(x_i)=\max f(X)+\min f(X)-f(x_i)=\max f(X)-f(x_i),~x_i\in X
$$

B-2.2. \textbf{$Y$-dual transformation.} Set
$$
f_{ed}(y_j)=\max f(Y)+\min f(Y)-f(y_j)=q+\min f(Y)-f(y_j),~y_j\in Y
$$

B-2.3. $f_{ed}(x_iy_j)=f(x_iy_j)$ for each edge $x_iy_j\in E(G)$.\\
Since the edge-difference constraint
\begin{equation}\label{eqa:edge-difference-total-labelings}
{
\begin{split}
f_{ed}(x_iy_j)+ |f_{ed}(x_i)-f_{ed}(y_j)|&=f_{ed}(x_iy_j)+f_{ed}(y_j)-f_{ed}(x_i)\\
&=f(x_iy_j)+ q+\min f(Y)-f(y_j)-[\max f(X)-f(x_i)] \\
&=f(x_iy_j)+q+\min f(Y)-\max f(X)-[f(y_j)-f(x_i)]\\
&=q+\min f(Y)-\max f(X),\quad x_iy_j\in E(G)
\end{split}}
\end{equation} hold true, so we claim that that $f_{ed}$ is a set-ordered edge-difference total labeling of $G$.

\textbf{Step B-3.} \textbf{Parameterizing the graceful-difference total labeling.} For integers $k\geq 0,d\geq 1$, we define a $(k,d)$-constraint labeling $g_{ed}$ for $G$ as follows:

B-3.1. $g_{ed}(x_i)=f_{ed}(x_i)\cdot d$ for $x_i\in X$;

B-3.2. $g_{ed}(y_j)=k+[f_{ed}(y_j)-1]\cdot d$ for $y_j\in Y$;

B-3.3. $g_{ed}(x_iy_j)=k+[f_{ed}(x_iy_j)-1]\cdot d$ for each edge $x_iy_j\in E(G)$. \\
So, the vertex color set
$$
g_{ed}(V(G))=\{f(x_i)\cdot d:x_i\in X\}\cup \{k+[f_{ed}(y_j)-1]\cdot d:y_j\in Y\}=g_{ed}(X)\cup g_{ed}(Y)
$$
and the edge color set
$$
g_{ed}(E(G))=\{k+[f_{ed}(x_iy_j)-1]\cdot d:x_iy_j\in E(G)\}=\{k,k+d,\dots, k+(q-1)\cdot d\}
$$

Let $M_{ed}=q+\min f(Y)-\max f(X)$ shown in Eq.(\ref{eqa:edge-difference-total-labelings}). We have that the edge-difference constraint
\begin{equation}\label{eqa:555555}
{
\begin{split}
g_{ed}(x_iy_j)+|g_{ed}(x_i)-g_{ed}(y_j)|=&k+[f_{ed}(x_iy_j)-1]\cdot d+| k+[f_{ed}(y_j)-1]\cdot d-f_{ed}(x_i)\cdot d |\\
=&2(k-d)+f_{ed}(x_iy_j)\cdot d+[f_{ed}(y_j)-f_{ed}(x_i)]\cdot d\\
=&2(k-d)+M_{ed}\cdot d
\end{split}}
\end{equation} for each edge $x_iy_j\in E(G)$ holds true. Thereby, we claim that $g_{ed}$ is an edge-difference $(k,d)$-total labeling of $G$ according to Definition \ref{defn:kd-w-type-colorings}.

\textbf{Step B-4.} Define a total coloring $h_{ed}$ for $G^*$ in the following substeps:

\textbf{Step B-4.1.} Recolor the vertices and edges of $G$ by setting

B-4.1. $h_{ed}(x_i)=g_{ed}(x_i)$ for $x_i\in X\subset V(G)$;

B-4.2. $h_{ed}(y_j)=g_{ed}(y_j)+m\cdot d$ for $y_j\in Y\subset V(G)$; and

B-4.3. $h_{ed}(x_iy_j)=g_{ed}(x_iy_j)$ for each edge $x_iy_j\in E(G)$.

Then, each edge $x_iy_j\in E(G)$ holds the edge-difference constraint
\begin{equation}\label{eqa:para-edge-difference-total-coloring}
{
\begin{split}
h_{ed}(x_iy_j)+|h_{ed}(y_j)-h_{ed}(x_i)|&=g_{ed}(x_iy_j)+g_{ed}(y_j)+m\cdot d-g_{ed}(x_i)\\
&=2(k-d)+(M_{ed}+m)\cdot d
\end{split}}
\end{equation} true.

\textbf{Step B-4.2.} Color edges $x_iu_{i,j}$ for $u_{i,j}\in L(x_i)$ with $j\in [1,a_i]$ and $i\in[1,s]$.
Since
$$h_{ed}(x_1u_{1,j})=k+j\cdot d+(q-1)\cdot d,~j\in [1,a_1],\textrm{ and }h_{ed}(x_1u_{1,a_1})=k+(q+a_1-1)\cdot d
$$ we have a formula
\begin{equation}\label{eqa:555555}
h_{ed}(x_iu_{i,j})=k+j\cdot d+(q-1)\cdot d+d\sum^{i-1}_{r=1}a_{r},~j\in [1,a_i],~i\in[1,s]
\end{equation}
The last edge $x_su_{s,a_s}$ is colored with
$$
h_{ed}(x_su_{s,a_s})=k+a_s\cdot d+(q-1)\cdot d+d\sum^{s-1}_{r=1}a_{r}=k+(A+B-1)\cdot d=k+(q+A-1)\cdot d
$$

\textbf{Step B-4.3.} Color edges $y_jv_{j,r}$ for $v_{j,r}\in L(y_j)$ with $r\in [1,b_j]$ and $j\in[1,t]$.
Because of $h_{ed}(y_tv_{t,r})=k+r\cdot d+(q+A-1)\cdot d$ for $r\in [1,b_t]$, and $h_{ed}(y_1v_{1,r})=k+b_t\cdot d+(q+A-1)\cdot d$. There is a formula
\begin{equation}\label{eqa:555555}
h_{ed}(y_jv_{j,r})=k+r\cdot d+(q+A-1)\cdot d+d\sum ^{j-1}_{i=1}b_{t-i+1},~r\in [1,b_j],~j\in[1,t]
\end{equation} as well as
$$
h_{ed}(y_jv_{j,b_j})=k+b_j\cdot d+(q+A-1)\cdot d+d\sum ^{j-1}_{i=1}b_{t-i+1}=k+(q+A-1)\cdot d+d\sum ^{j}_{i=1}b_{t-i+1}
$$ The last edge $x_su_{s,a_s}$ is colored with
$$
h_{ed}(y_1v_{1,b_1})=k+b_1\cdot d+(q+A-1)\cdot d+d\sum ^{t-1}_{i=1}b_{t-i+1}=k+(q+m-1)\cdot d
$$

\textbf{Step B-4.4.} Let $M^*=2(k-d)+(M_{ed}+m)\cdot d$ defined in Eq.(\ref{eqa:para-edge-difference-total-coloring}). Color added leaves $u_{i,j}\in L(x_i)$ for $j\in [1,a_i]$ and $i\in[1,s]$ by  $$h_{ed}(u_{i,j})=M^*-h_{ed}(x_iu_{i,j})+h_{ed}(x_i),~j\in [1,a_i],~i\in[1,s]$$

\textbf{Step B-4.5.} Color added leaves $v_{j,r}\in L(y_j)$ with $r\in [1,b_j]$ and $j\in[1,t]$ as:
$$h_{ed}(v_{j,r})=h_{ed}(y_j)-[M^*-h_{ed}(y_jv_{j,r})],~r\in [1,b_j],~j\in[1,t]
$$ Also, we get the edge-difference constraint $h_{ed}(y_jv_{j,r})+h_{ed}(y_j)-h_{ed}(v_{j,r})=M^*$.

\textbf{Step B-5.} Return the edge-difference $(k,d)$-total coloring $h_{ed}$ of $G^*$.

\vskip 0.4cm

Fig.\ref{fig:4-edge-difference-leaves} shows us examples for understanding the RLA-algorithm-B:

(a) A connected graph $T$ admits a set-ordered edge-difference total labeling $f_{ed}$ holding the edge-difference constraint
$$
f_{ed}(xy)+|f_{ed}(x)-f_{ed}(y)|=10,~xy\in E(T)
$$

(b) a parameterized connected graph $T_{\textrm{p}}$ admitting a graceful-difference $(k,d)$-total labeling $g_{ed}$ with the edge-difference constraint
$$g_{ed}(xy)+|g_{ed}(x)-g_{ed}(y)|=2k+8d,~xy\in E(T_{\textrm{p}})
$$

(c) adding leaves to $T_{\textrm{p}}$ produces a connected graph $T_{\textrm{p-leaf}}$;

(d) a connected graph $Th_{\textrm{p-leaf}}$ admits a graceful-difference $(k,d)$-total coloring $h_{ed}$ with the edge-difference constraint
$$h_{ed}(xy)+|h_{ed}(x)-h_{ed}(y)|=2k+22d,~
$$

\begin{figure}[h]
\centering
\includegraphics[width=16.4cm]{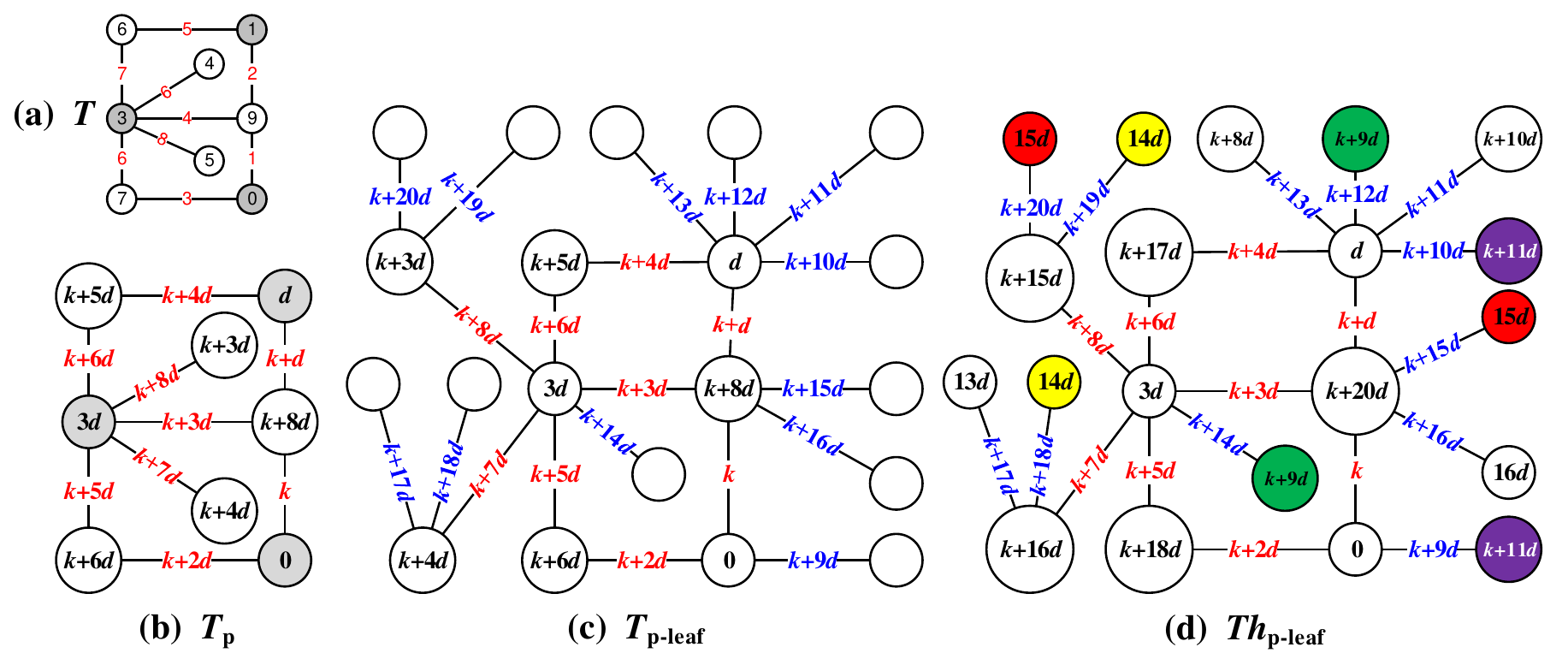}\\
\caption{\label{fig:4-edge-difference-leaves} {\small For understanding the RLA-algorithm-B for the edge-difference $(k,d)$-total coloring.}}
\end{figure}

The RLA-algorithm-B enables us to get the following result:

\begin{thm}\label{thm:edge-difference-kd-trees}
$^*$ Each tree admits an edge-difference $(k,d)$-total coloring and an odd-edge edge-difference $(k,d)$-total coloring.
\end{thm}

\begin{thm}\label{thm:edge-difference-kd-total-coloring}
$^*$ If a connected graph $G$ admits a set-ordered graceful labeling, then $G$ admits an edge-difference $(k,d)$-total coloring and an odd-edge edge-difference $(k,d)$-total coloring.
\end{thm}

\begin{thm}\label{thm:666666}
(i) $^*$ A connected graph $G$ admits an edge-difference $(k,d)$-total coloring if and only if there exists a tree $T$ admitting an edge-difference $(k,d)$-total coloring such that the result of vertex-coinciding each group of vertices colored the same colors into one vertex is just $G$.

(ii) $^*$ A connected bipartite graph $H$ admits an odd-edge edge-difference $(k,d)$-total coloring if and only if there exists a tree $T^*$ admitting an odd-edge edge-difference $(k,d)$-total coloring such that the result of vertex-coinciding each group of vertices colored the same color in $T^*$ into one vertex is just $H$.
\end{thm}

\subsection{Parameterized felicitous-difference total labelings/colorings}

\vskip 0.4cm

\noindent $^*$ \textbf{RLA-algorithm-C for the felicitous-difference $(k,d)$-total coloring}.

\textbf{Input:} A connected bipartite $(p,q)$-graph $G$ admitting a set-ordered graceful labeling $f$.

\textbf{Output:} A connected bipartite $(p+m,q+m)$-graph $G^*$ admitting a felicitous-difference $(k,d)$-total coloring, where $G^*$ is the result of adding randomly $m$ leaves to $G$.

\textbf{Step C-1.} \textbf{Initialization.} Since a connected bipartite $(p,q)$-graph $G$ admitting a set-ordered graceful labeling $f$, we have the vertex set $V(G)=X\cup Y$ with $X\cap Y=\emptyset$, where
$$
X=\{x_1,x_2,\dots ,x_s\},~Y=\{y_1,y_2,\dots ,y_t\},~s+t=p=|V(G)|
$$ such that
$$
0=f(x_1)<f(x_2)<\cdots <f(x_s)<f(y_1)<f(y_2)<\cdots <f(y_t)=q
$$ with the edge color set
$$f(E(G))=\{f(x_iy_j)=f(y_j)-f(x_i):x_iy_j\in E(G)\}=[1,q]
$$

Adding randomly $a_i$ new leaves $u_{i,j}\in L(x_i)$ to each vertex $x_i$ by adding new edges $x_iu_{i,j}$ for $j\in [1,a_i]$ and $i\in[1,s]$, and adding randomly $b_j$ new leaves $v_{j,r}\in L(y_j)$ to each vertex $y_j$ by adding new edges $y_jv_{j,r}$ for $r\in [1,b_j]$ and $j\in[1,t]$, it is allowed that some $a_i=0$ or some $b_j=0$. The resultant graph is denoted as $G^*$. Let $A=\sum ^s_{i=1}a_i$ and $B=\sum ^t_{j=1}b_j$, so $m=A+B$.

\textbf{Step C-2.} Define a felicitous-difference total labeling $f_{fel}$ of $G$ in the following way:

C-2.1. \textbf{$X$-dual transformation.} Each vertex $x_i\in X$ holds
$$f_{fel}(x_i)=\max f(X)+\min f(X)-f(x_i)=\max f(X)-f(x_i)$$

C-2.2. $f_{fel}(y_j)=f(y_j)$ for $y_j\in Y$; and

C-2.3. $f_{fel}(x_iy_j)=f(x_iy_j)$ for each edge $x_iy_j\in E(G)$.\\
Because of the felicitous-difference constraint
\begin{equation}\label{eqa:felicitous-difference-total-labelings}
{
\begin{split}
\big |f_{fel}(x_i)+f_{fel}(y_j)-f_{fel}(x_iy_j)\big |&=\big |\max f(X)-f(x_i)+f(y_j)-f(x_iy_j)\big |=\max f(X)
\end{split}}
\end{equation} for each edge $x_iy_j\in E(G)$ holds true, so we claim that that $f_{fel}$ is a set-ordered felicitous-difference total labeling of $G$.

\textbf{Step C-3.} \textbf{Parameterizing the felicitous-difference total labeling.} For integers $k\geq 0,d\geq 1$, we define a $(k,d)$-constraint labeling $g_{fel}$ for $G$ as follows:

C-3.1. $g_{fel}(x_i)=f_{fel}(x_i)\cdot d$ for $x_i\in X$;

C-3.2. $g_{fel}(y_j)=k+[f_{fel}(y_j)-1]\cdot d$ for $y_j\in Y$;

C-3.3. $g_{fel}(x_iy_j)=k+[f_{fel}(x_iy_j)-1]\cdot d$ for each edge $x_iy_j\in E(G)$. \\
So, the vertex color set
$$
g_{fel}(V(G))=\{f(x_i)\cdot d:x_i\in X\}\cup \{k+[f_{fel}(y_j)-1]\cdot d:y_j\in Y\}=g_{fel}(X)\cup g_{fel}(Y)
$$ and the edge color set
$$
g_{fel}(E(G))=\{k+[f_{fel}(x_iy_j)-1]\cdot d:x_iy_j\in E(G)\}=\{k,k+d,\dots, k+(q-1)\cdot d\}
$$

By Eq.(\ref{eqa:felicitous-difference-total-labelings}), the felicitous-difference constraint
\begin{equation}\label{eqa:555555}
{
\begin{split}
& \big |g_{fel}(x_i)+g_{fel}(y_j)-g_{fel}(x_iy_j)\big |\\
=&|f_{fel}(x_i)\cdot d+k+[f_{fel}(y_j)-1]\cdot d-(k+[f_{fel}(x_iy_j)-1]\cdot d)|\\
=&f_{fel}(x_i)\cdot d+f_{fel}(y_j)\cdot d-f_{fel}(x_iy_j)\cdot d\\
=&\max f(X)\cdot d
\end{split}}
\end{equation} for each edge $x_iy_j\in E(G)$ holds true. Thereby, we claim that $g_{fel}$ is a felicitous-difference $(k,d)$-total labeling of $G$ from Definition \ref{defn:kd-w-type-colorings}.

\textbf{Step C-4.} Define a total coloring $h_{fel}$ for $G^*$ in the following substeps:

\textbf{Step C-4.1.} Recolor the vertices and edges of $G$ by setting

C-4.1. $h_{fel}(x_i)=g_{fel}(x_i)$ for $x_i\in X\subset V(G)$;

C-4.2. $h_{fel}(y_j)=g_{fel}(y_j)+m\cdot d$ for $y_j\in Y\subset V(G)$; and

C-4.3. $h_{fel}(x_iy_j)=g_{fel}(x_iy_j)$ for each edge $x_iy_j\in E(G)$.

So, there are the following felicitous-difference constraints
\begin{equation}\label{eqa:para-felicitous-difference-total-coloring}
{
\begin{split}
|h_{fel}(y_j)+h_{fel}(x_i)-h_{fel}(x_iy_j)|&=g_{fel}(x_i)+g_{fel}(y_j)+m\cdot d-g_{fel}(x_iy_j)\\
&=[\max f(X)+m]\cdot d,~x_iy_j\in E(G)
\end{split}}
\end{equation}

\textbf{Step C-4.2.} Color edges $x_iu_{i,j}$ for $u_{i,j}\in L(x_i)$ with $j\in [1,a_i]$ and $i\in[1,s]$.

From $h_{fel}(x_1u_{1,j})=k+j\cdot d+(q-1)\cdot d$ for $j\in [1,a_1]$, $h_{fel}(x_1u_{1,a_1})=k+(q+a_1-1)\cdot d$, we get a formula

\begin{equation}\label{eqa:555555}
h_{fel}(x_iu_{i,j})=k+j\cdot d+(q-1)\cdot d+d\sum^{i-1}_{r=1}a_{r},~j\in [1,a_i],~ i\in[1,s]
\end{equation} The last edge $x_su_{s,a_s}$ is colored with
$$
h_{fel}(x_su_{s,a_s})=k+a_s\cdot d+(q-1)\cdot d+d\sum^{s-1}_{r=1}a_{r}=k+(A+B-1)\cdot d=k+(q+A-1)\cdot d
$$

\textbf{Step C-4.3.} Color edges $y_jv_{j,r}$ for $v_{j,r}\in L(y_j)$ with $r\in [1,b_j]$ and $j\in[1,t]$.

Since
$$h_{fel}(y_tv_{t,r})=k+r\cdot d+(q+A-1)\cdot d,~r\in [1,b_t],\textrm{ and }h_{fel}(y_1v_{1,r})=k+b_t\cdot d+(q+A-1)\cdot d
$$ we have a formula

\begin{equation}\label{eqa:555555}
h_{fel}(y_jv_{j,r})=k+r\cdot d+(q+A-1)\cdot d+d\sum ^{j-1}_{i=1}b_{t-i+1},~r\in [1,b_j]
\end{equation} and
$$
h_{fel}(y_jv_{j,b_j})=k+b_j\cdot d+(q+A-1)\cdot d+d\sum ^{j-1}_{i=1}b_{t-i+1}=k+(q+A-1)\cdot d+d\sum ^{j}_{i=1}b_{t-i+1}
$$ The last edge $y_1v_{1,b_1}$ is colored with
$$
h_{fel}(y_1v_{1,b_1})=k+b_1\cdot d+(q+A-1)\cdot d+d\sum ^{t-1}_{i=1}b_{t-i+1}=k+(q+m-1)\cdot d
$$

\textbf{Step C-4.4.} Let $M\,'=[\max f(X)+m]\cdot d$ defined in Eq.(\ref{eqa:para-felicitous-difference-total-coloring}). Color added leaves $u_{i,j}\in L(x_i)$ for $j\in [1,a_i]$ and $i\in[1,s]$ as:
$$
h_{fel}(u_{i,j})=M\,'+h_{fel}(x_iu_{i,j})-h_{fel}(x_i),~j\in [1,a_i],~i\in[1,s]
$$ which is just the felicitous-difference constraint.

\textbf{Step C-4.5.} Color the added leaves $v_{j,r}\in L(y_j)$ with $r\in [1,b_j]$ and $j\in[1,t]$ as:
$$h_{fel}(v_{j,r})=M\,'+h_{fel}(y_jv_{j,r})-h_{fel}(y_j),~r\in [1,b_j],~j\in[1,t]
$$ Also, the felicitous-difference constraint
$$h_{fel}(v_{j,r})+h_{fel}(y_j)-h_{fel}(y_jv_{j,r})=M\,'
$$ holds true for each edge $y_jv_{j,r}\in E(G^*)$.

Summarizing all steps together, we confirm that $h_{fel}$ is really a felicitous-difference $(k,d)$-total coloring of $G^*$.

\textbf{Step C-5.} Return the felicitous-difference $(k,d)$-total coloring $h_{fel}$ of $G^*$.

\vskip 0.4cm

Fig.\ref{fig:4-felicitous-difference-leaves} is for understanding the RLA-algorithm-C, we can see:

(a) A connected graph $L$ admits a set-ordered felicitous-difference total labeling $f_{fel}$ holding the felicitous-difference constraint
$$\big | f_{fel}(x)+f_{fel}(y)-f_{fel}(xy)\big |=3,~xy\in E(L)
$$

(b) a parameterized connected graph $L_{\textrm{p}}$ admitting a felicitous-difference $(k,d)$-total labeling $g_{fel}$ with the felicitous-difference constraint
$$\big | g_{fel}(x)+ g_{fel}(y)-g_{fel}(xy)\big |=3d,~xy\in E(L_{\textrm{p}})
$$

(c) adding leaves to $L_{\textrm{p}}$ produces a connected graph $L_{\textrm{p-leaf}}$;

(d) a connected graph $Lh_{\textrm{p-leaf}}$ admits a felicitous-difference $(k,d)$-total coloring $h_{fel}$ with the felicitous-difference constraint
$$\big | h_{fel}(x)+ h_{fel}(y)-h_{fel}(xy)\big |=17d,~xy\in E(Lh_{\textrm{p-leaf}})
$$

\begin{figure}[h]
\centering
\includegraphics[width=16.4cm]{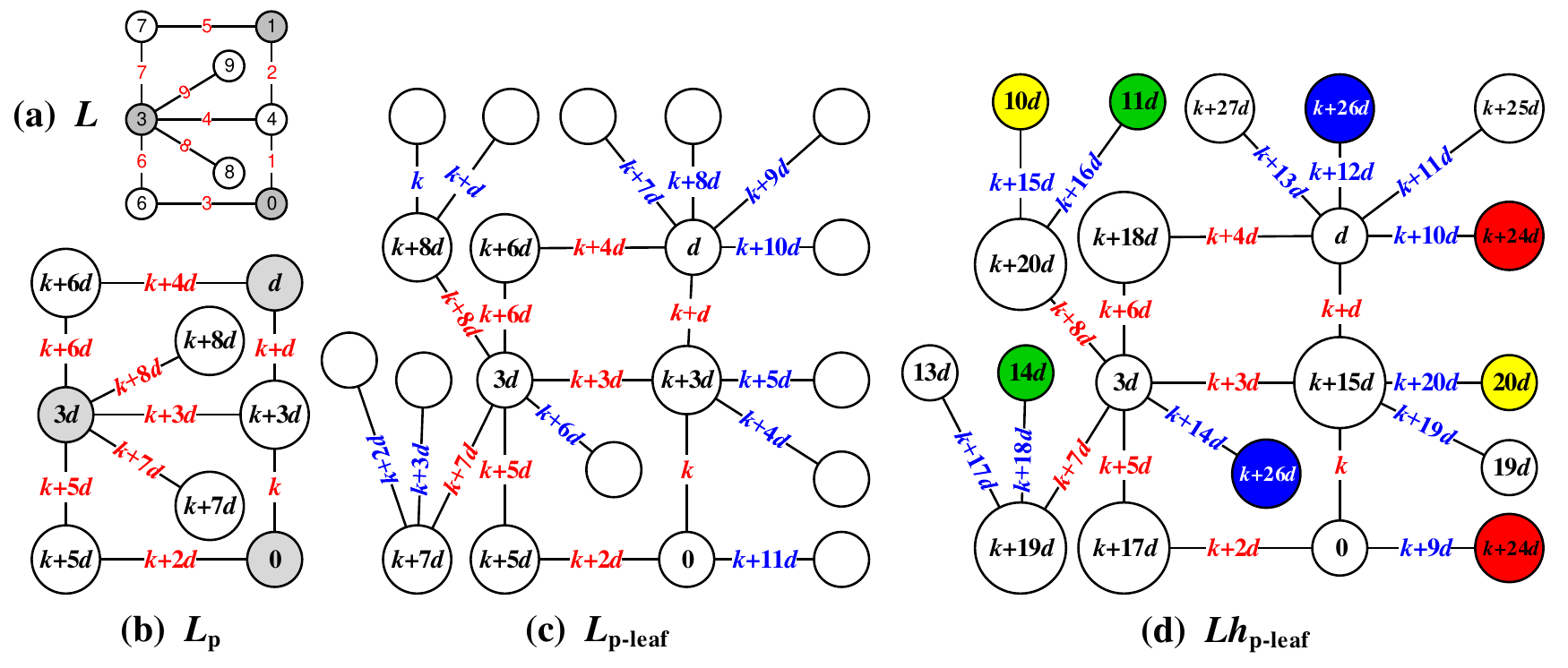}\\
\caption{\label{fig:4-felicitous-difference-leaves} {\small For understanding the RLA-algorithm-C for the felicitous-difference $(k,d)$-total coloring.}}
\end{figure}

The RLA-algorithm-C enables us to get the following result:

\begin{thm}\label{thm:trees-admits-felicitous-difference-k-d}
$^*$ Each tree admits a felicitous-difference $(k,d)$-total coloring and an odd-edge felicitous-difference $(k,d)$-total coloring.
\end{thm}

\begin{thm}\label{thm:666666}
$^*$ If a connected graph $G$ admits a set-ordered graceful labeling, then $G$ admits a felicitous-difference $(k,d)$-total coloring and an odd-edge felicitous-difference $(k,d)$-total coloring.
\end{thm}

\begin{thm}\label{thm:666666}
(i) $^*$ A connected graph $G$ admits a felicitous-difference $(k,d)$-total coloring if and only if there exists a tree $T$ admitting a felicitous-difference $(k,d)$-total coloring such that such that the result of vertex-coinciding each group of vertices colored the same colors into one vertex is just $G$.

(ii) $^*$ A connected bipartite graph $H$ admits an odd-edge felicitous-difference $(k,d)$-total coloring if and only if there exists a tree $T^*$ admitting an odd-edge felicitous-difference $(k,d)$-total coloring such that such that the result of vertex-coinciding each group of vertices colored the same color in $T^*$ into one vertex is just $H$.
\end{thm}

\subsection{Parameterized edge-magic total labelings/colorings}

\vskip 0.4cm

\noindent $^*$ \textbf{RLA-algorithm-D for the edge-magic $(k,d)$-total coloring}.

\textbf{Input:} A connected bipartite $(p,q)$-graph $G$ admitting a set-ordered graceful labeling $f$.

\textbf{Output:} A connected bipartite $(p+m,q+m)$-graph $G^*$ admitting an edge-magic $(k,d)$-total coloring, where $G^*$ is the result of adding randomly $m$ leaves to $G$.

\textbf{Step D-1.} \textbf{Initialization.} Since a connected bipartite $(p,q)$-graph $G$ admitting a set-ordered graceful labeling $f$, we have the vertex set $V(G)=X\cup Y$ with $X\cap Y=\emptyset$, where
$$
X=\{x_1,x_2,\dots ,x_s\},~Y=\{y_1,y_2,\dots ,y_t\},~+t=p=|V(G)|
$$ such that
$$
0=f(x_1)<f(x_2)<\cdots <f(x_s)<f(y_1)<f(y_2)<\cdots <f(y_t)=q
$$ with the edge color set
$$f(E(G))=\{f(x_iy_j)=f(y_j)-f(x_i):x_iy_j\in E(G)\}=[1,q]$$

Adding randomly $a_i$ new leaves $u_{i,j}\in L(x_i)$ to each vertex $x_i$ by adding new edges $x_iu_{i,j}$ for $j\in [1,a_i]$ and $i\in[1,s]$, and adding randomly $b_j$ new leaves $v_{j,r}\in L(y_j)$ to each vertex $y_j$ by adding new edges $y_jv_{j,r}$ for $r\in [1,b_j]$ and $j\in[1,t]$, it is allowed that some $a_i=0$ or some $b_j=0$. The resultant graph is denoted as $G^*$. Let $A=\sum ^s_{i=1}a_i$ and $B=\sum ^t_{j=1}b_j$, so $m=A+B$.

\textbf{Step D-2.} Define an edge-magic total labeling $f_{mag}$ for $G$ in the following way:

D-2.1. \textbf{$X$-dual transformation.} Each vertex $x_i\in X$ holds
$$f_{mag}(x_i)=\max f(X)+\min f(X)-f(x_i)=\max f(X)-f(x_i)$$

D-2.2. $f_{mag}(y_j)=f(y_j)$ for $y_j\in Y$; and

D-2.3. \textbf{Edge-dual transformation.} Each edge $x_iy_j\in E(G)$ holds
$$f_{mag}(x_iy_j)=\max f(E(G))+\min f(E(G))-f(x_iy_j)$$
Since each edge $x_iy_j$ of $G$ holds the edge-magic constraint
\begin{equation}\label{eqa:felicitous-difference-total-labelings}
{
\begin{split}
&\quad f_{mag}(x_i)+f_{mag}(x_iy_j)+f_{mag}(y_j)\\
&=\max f(X)-f(x_i)+\max f(E(G))+\min f(E(G))-f(x_iy_j)+f(y_j)\\
&=\max f(E(G))+\min f(E(G))+\max f(X)\\
&=\max f(X)+q+1
\end{split}}
\end{equation} true, so $\lambda_{mag}$ is an edge-magic total labeling of $G$.

\textbf{Step D-3.} \textbf{Parameterizing the felicitous-difference total labeling.} For integers $k\geq 0,d\geq 1$, we define a $(k,d)$-constraint labeling $g_{mag}$ for $G$ as follows:

D-3.1. $g_{mag}(x_i)=f_{mag}(x_i)\cdot d$ for $x_i\in X$;

D-3.2. $g_{mag}(y_j)=k+[f_{mag}(y_j)-1]\cdot d$ for $y_j\in Y$;

D-3.3. $g_{mag}(x_iy_j)=k+[f_{mag}(x_iy_j)-1]\cdot d$ for each edge $x_iy_j\in E(G)$. \\
So, the vertex color set
$$
g_{mag}(V(G))=\{f(x_i)\cdot d:x_i\in X\}\cup \{k+[f_{mag}(y_j)-1]\cdot d:y_j\in Y\}=g_{mag}(X)\cup g_{mag}(Y)
$$ and the edge color set
$$
g_{mag}(E(G))=\{k+[f_{mag}(x_iy_j)-1]\cdot d:x_iy_j\in E(G)\}=\{k,k+d,\dots, k+(q-1)\cdot d\}
$$ By Eq.(\ref{eqa:felicitous-difference-total-labelings}), we have that the edge-magic constraint
\begin{equation}\label{eqa:555555}
{
\begin{split}
& g_{mag}(x_i)+g_{mag}(x_iy_j)+g_{mag}(y_j)\\
=&f_{mag}(x_i)\cdot d+(k+[f_{mag}(x_iy_j)-1]\cdot d)+(k+[f_{mag}(y_j)-1]\cdot d)\\
=&2(k-d)+[f_{mag}(x_i)+f_{mag}(x_iy_j)+f_{mag}(y_j)]\\
=&2k+[\max f(X)+q-1]\cdot d
\end{split}}
\end{equation} for each edge $x_iy_j\in E(G)$ holds true. Thereby, we claim that $g_{mag}$ is an edge-magic $(k,d)$-total labeling of $G$ by Definition \ref{defn:kd-w-type-colorings}.

\textbf{Step D-4.} Define a total coloring $h_{mag}$ for $G^*$ in the following substeps:

\textbf{Step D-4.1.} Recolor the vertices and edges of $G$ by setting

D-4.1. $h_{mag}(x_i)=g_{mag}(x_i)$ for $x_i\in X\subset V(G)$;

D-4.2. $h_{mag}(y_j)=g_{mag}(y_j)+m\cdot d$ for $y_j\in Y\subset V(G)$; and

D-4.3. $h_{mag}(x_iy_j)=g_{mag}(x_iy_j)+m\cdot d$ for each edge $x_iy_j\in E(G)$.

Thereby, we have that the edge-magic constraint
\begin{equation}\label{eqa:para-edge-magic-total-coloring}
{
\begin{split}
h_{mag}(y_j)+h_{mag}(x_iy_j)+h_{mag}(x_i)&=g_{mag}(x_i)+g_{mag}(x_iy_j)+m\cdot d+g_{mag}(y_j)+m\cdot d\\
&=2k+[\max f(X)+q+2m-1]\cdot d
\end{split}}
\end{equation} for each edge $x_iy_j\in E(G)$ holds true.

\textbf{Step D-4.2.} Color edges $y_jv_{j,r}$ for $v_{j,r}\in L(y_j)$ with $r\in [1,b_j]$ and $j\in[1,t]$.

Since $h_{mag}(y_tv_{t,r})=k+(r-1)\cdot d$ for $r\in [1,b_t]$ and $h_{mag}(y_tv_{t,b_t})=k+(b_t-1)\cdot d$, we have a formula

\begin{equation}\label{eqa:555555}
h_{mag}(y_jv_{j,r})=k+(r-1)\cdot d+d\sum ^{j-1}_{i=1}b_{t-i+1},~r\in [1,b_j], j\in[2,t]
\end{equation} and
$$
h_{mag}(y_jv_{j,b_j})=k+(b_j-1)\cdot d+d\sum ^{j-1}_{i=1}b_{t-i+1}=k+d\sum ^{j}_{i=1}b_{t-i+1}
$$ The last edge $y_1v_{1,b_1}$ is colored with
$$
h_{mag}(y_1v_{1,b_1})=k+(b_1-1)\cdot d+d\sum ^{t-1}_{i=1}b_{t-i+1}=k+(B-1)\cdot d
$$

\textbf{Step D-4.3.} Color edges $x_iu_{i,j}$ for $u_{i,j}\in L(x_i)$ with $j\in [1,a_i]$ and $i\in[1,s]$.

From $h_{mag}(x_su_{s,j})=k+j\cdot d+(B-1)\cdot d$ for $j\in [1,a_1]$, and
$$h_{mag}(x_su_{s,a_s})=k+a_s\cdot d+(B-1)\cdot d
$$ we get a formula
\begin{equation}\label{eqa:555555}
h_{mag}(x_iu_{i,j})=k+j\cdot d+(B-1)\cdot d+d\sum^{s-i}_{r=1}a_{s-r+1},~j\in [1,a_i],~ i\in[2,s]
\end{equation} The last edge $x_1u_{1,a_1}$ is colored with
$$
h_{mag}(x_1u_{1,a_1})=k+a_1\cdot d+(B-1)\cdot d+d\sum^{s-1}_{r=1}a_{s-r+1}=k+(A+B-1)\cdot d=k+(m-1)\cdot d
$$

\textbf{Step D-4.4.} Let $M\,''=2k+[\max f(X)+q+2m-1]\cdot d$ defined in Eq.(\ref{eqa:para-edge-magic-total-coloring}). Color added leaves $u_{i,j}\in L(x_i)$ for $j\in [1,a_i]$ and $i\in[1,s]$ as: $$h_{mag}(u_{i,j})=M\,''-h_{mag}(x_iu_{i,j})-h_{mag}(x_i),~j\in [1,a_i],~i\in[1,s]$$
which is just the edge-magic constraint.

\textbf{Step D-4.5.} Color added leaves $v_{j,r}\in L(y_j)$ with $r\in [1,b_j]$ and $j\in[1,t]$ as:
$$h_{mag}(v_{j,r})=M\,''-h_{mag}(y_jv_{j,r})-h_{mag}(y_j),~r\in [1,b_j],~j\in[1,t]$$
which is just the edge-magic constraint $h_{mag}(v_{j,r})+h_{mag}(y_jv_{j,r})+h_{mag}(y_j)=M\,''$.

Thereby, we claim that $h_{mag}$ is really an edge-magic $(k,d)$-total coloring of $G^*$.

\textbf{Step D-5.} Return the edge-magic $(k,d)$-total coloring $h_{mag}$ of $G^*$.

\vskip 0.4cm

Fig.\ref{fig:4-edge-magic-leaves} is for understanding the RLA-algorithm-D for the edge-magic $(k,d)$-total coloring, there are:

(a) A connected graph $H$ admits a set-ordered edge-magic total labeling $f_{mag}$ holding the edge-magic constraint
$$f_{mag}(x)+f_{mag}(xy)+f_{mag}(y)=3,~xy\in E(L)
$$

(b) a parameterized connected graph $H_{\textrm{p}}$ admitting an edge-magic $(k,d)$-total labeling $g_{mag}$ with the edge-magic constraint
$$g_{mag}(x)+g_{mag}(xy)+ g_{mag}(y)=2k+11d,~xy\in E(H_{\textrm{p}})
$$

(c) adding leaves to $H_{\textrm{p}}$ produces a connected graph $H_{\textrm{p-leaf}}$;

(d) a connected graph $Hh_{\textrm{p-leaf}}$ admits an edge-magic $(k,d)$-total coloring $h_{mag}$ with the edge-magic constraint
$$h_{mag}(x)+h_{mag}(xy)+ h_{mag}(y)=2k+35d,~xy\in E(Hh_{\textrm{p-leaf}})
$$

\begin{figure}[h]
\centering
\includegraphics[width=16.4cm]{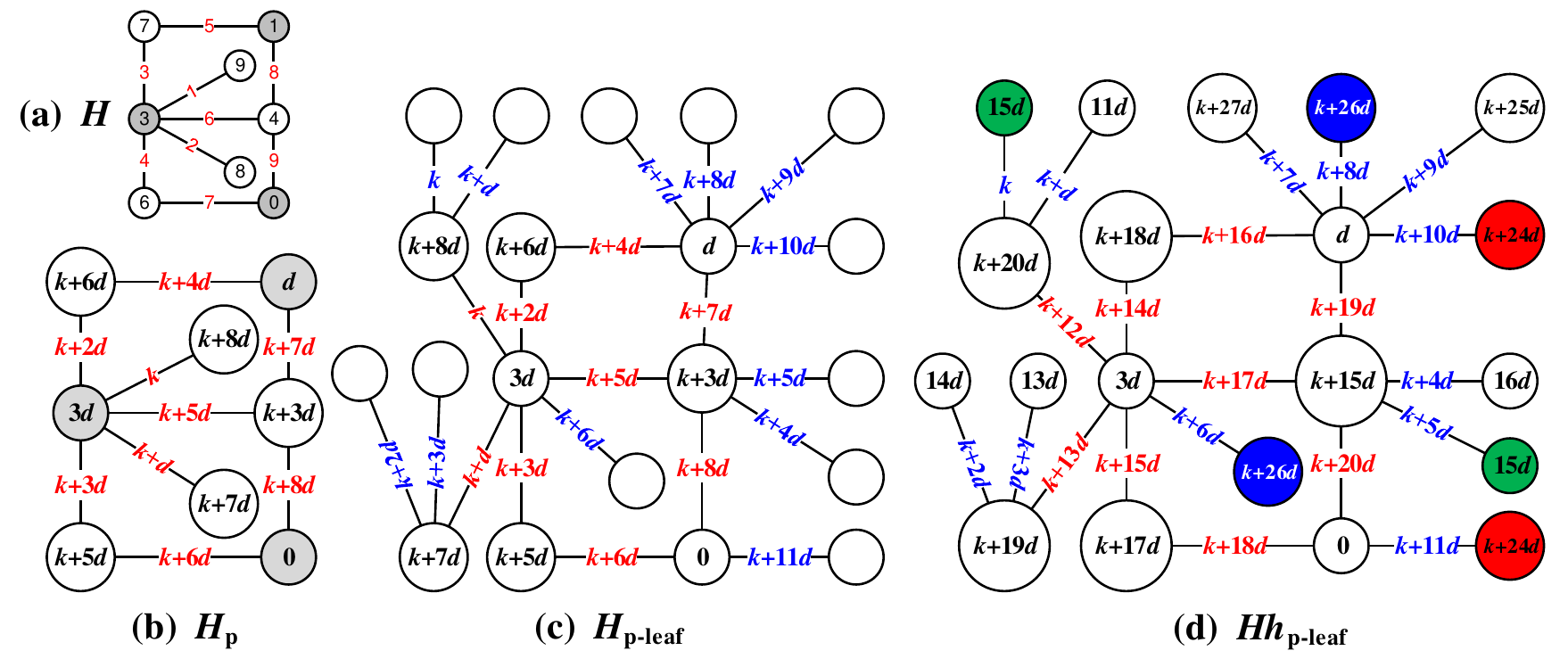}\\
\caption{\label{fig:4-edge-magic-leaves} {\small For understanding the RLA-algorithm-D for the edge-magic $(k,d)$-total coloring.}}
\end{figure}

\begin{figure}[h]
\centering
\includegraphics[width=16.4cm]{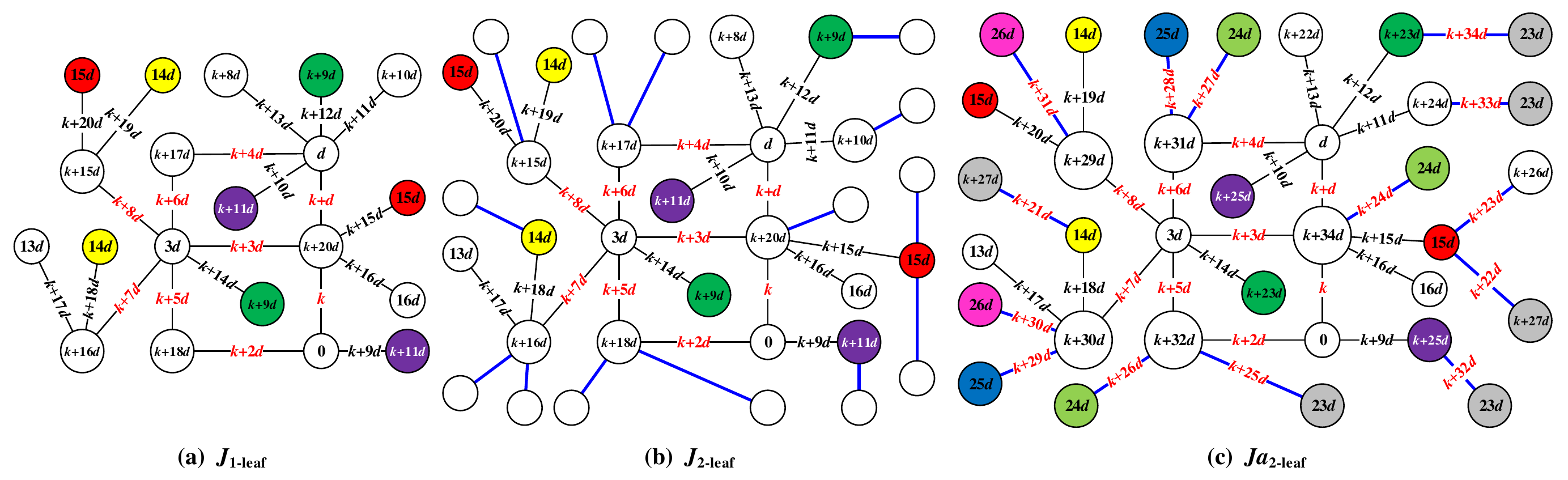}\\
\caption{\label{fig:2add-leaf-edge-difference} {\small Adding leaf two times, edge-difference $(k,d)$-total coloring.}}
\end{figure}

\begin{figure}[h]
\centering
\includegraphics[width=16.4cm]{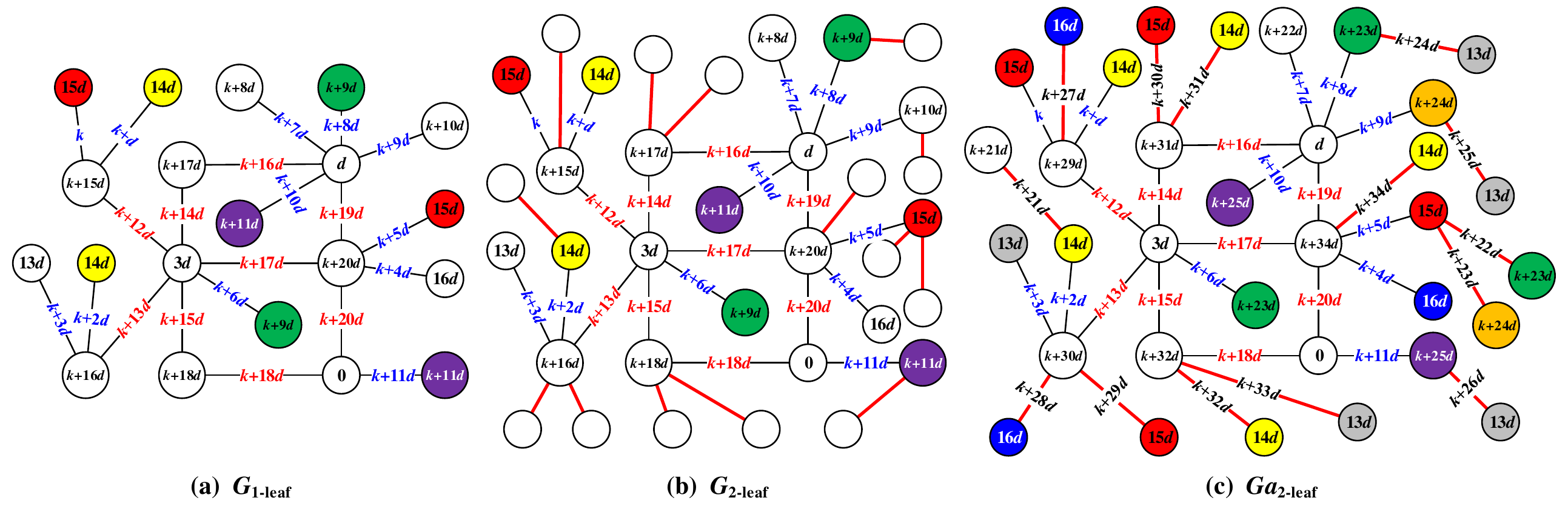}\\
\caption{\label{fig:2add-leaf-graceful-difference} {\small Adding leaf two times, graceful-difference $(k,d)$-total coloring.}}
\end{figure}

\begin{figure}[h]
\centering
\includegraphics[width=16.4cm]{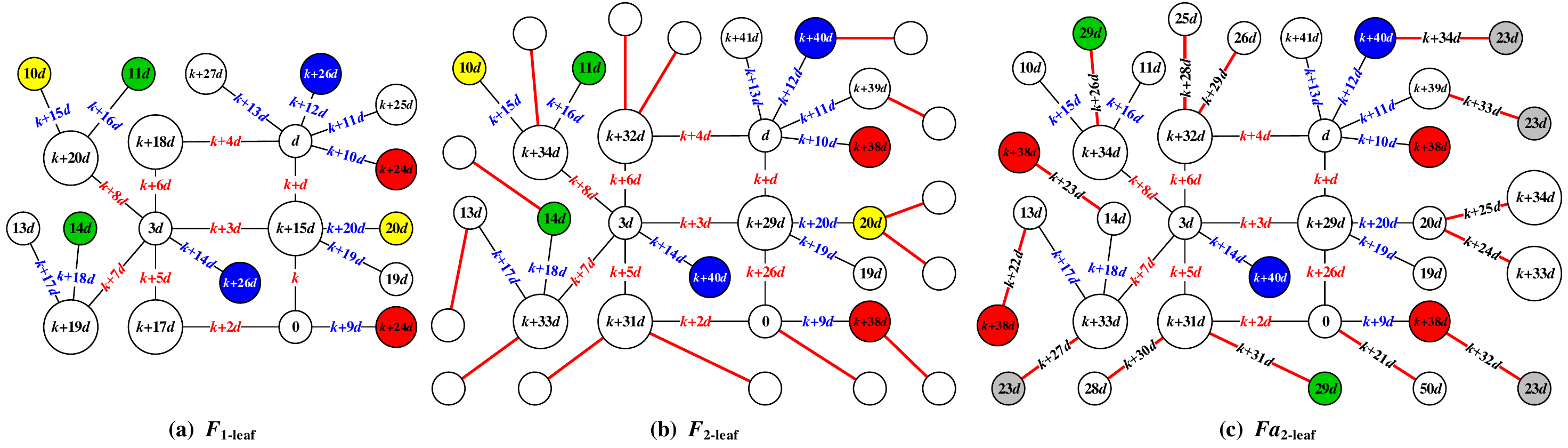}\\
\caption{\label{fig:2add-leaf-felicitous-difference} {\small Adding leaf two times, felicitous-difference $(k,d)$-total coloring.}}
\end{figure}

The RLA-algorithm-D enables us to get the following result:

\begin{thm}\label{thm:trees-admits-edge-magic-k-d}
$^*$ Each tree admits an edge-magic $(k,d)$-total coloring and an odd-edge edge-magic $(k,d)$-total coloring.
\end{thm}

\begin{thm}\label{thm:666666-edge-magic-11}
$^*$ If a connected graph $G$ admits a set-ordered graceful labeling, then $G$ admits an edge-magic $(k,d)$-total coloring and an odd-edge edge-magic $(k,d)$-total coloring.
\end{thm}

\begin{thm}\label{thm:666666-edge-magic-22}
(i) $^*$ A connected graph $G$ admits an edge-magic $(k,d)$-total coloring if and only if there exists a tree $T$ admitting an edge-magic $(k,d)$-total coloring such that such that the result of vertex-coinciding each group of vertices colored the same colors into one vertex is just $G$.

(ii) $^*$ A connected bipartite graph $H$ admits an odd-edge edge-magic $(k,d)$-total coloring if and only if there exists a tree $T^*$ admitting an odd-edge edge-magic $(k,d)$-total coloring such that such that the result of vertex-coinciding each group of vertices colored the same color in $T^*$ into one vertex is just $H$.
\end{thm}

\section{Parameterized labelings on trees having perfect matchings}

\subsection{Basic labelings on trees having perfect matchings}

\begin{defn} \label{defn:arm-odd-elegant-perfect-matchings}
\cite{SH-ZXH-YB-2017-1, SH-ZXH-YB-2017-2} Let $T$ be a $(p,q)$-tree having a perfect matching $M$. Suppose that each leaf of $T$ is an end of some matching edge of $M$, and $M(L)$ is the set of all leaves of $T$, and furthermore the graph $T-M(L)$ obtained by deleting all leaves of $T$ is just a tree having $k=|V(T-M(L))|$ vertices.

(1) \cite{SH-ZXH-YB-2017-1} Suppose that $T$ admits an odd-elegant labeling~$f_1$. If $|f_1(u)-f_1(v)|=1$ for each matching edge $uv \in M$ holds true, and the elements of $f_1(E(M))$ form an arithmetic progression having $|M|$ terms with the first term $1$ and the common difference 4. Then we call $f_1$ \emph{arithmetic matching odd-elegant labeling} (\emph{arm-odd-elegant labeling}) of $T$, and $T$ is an \emph{arm-odd-elegant graph}.

(2) \cite{SH-ZXH-YB-2017-1} Suppose that $T$ admits a super felicitous labeling~$f_2$. If $|f_2(u)-f_2(v)|=k$ for every matching edge $uv \in M$, we call $f_2$ \emph{arithmetic matching super felicitous labeling} (\emph{arm-super felicitous labeling}) of $T$, and $T$ an \emph{arm-super felicitous graph}.

(3) \cite{SH-ZXH-YB-2017-1} Suppose that $T$ admits a super edge-magic graceful labeling~$f_3$. If $|f_3(u)-f_3(v)|=k$ for each matching edge $uv \in M$, and the elements of $f_3(E(T))$ produce an arithmetic progression having $|M|$ terms with the first term $p+1$ and the common difference $2$. Then we call $f_3$ \emph{arithmetic matching super edge-magic graceful labeling} (\emph{arm-super edge-magic graceful labeling}) of $T$, and $T$ an \emph{arm-super edge-magic graceful graph}.

(4) \cite{SH-ZXH-YB-2017-2} Suppose that $T$ admits a super edge-magic total labeling~$f_4$. If $|f_4(u)-f_4(v)|=k$ for each matching edge $uv \in M$, and the elements of $f_4(E(T))$ yield an arithmetic progression having $|M|$ terms with the first term $2p-1$ and the common difference $-2$. Then we call $f_4$ \emph{arithmetic matching super edge-magic total labeling} (\emph{arm-super edge-magic total labeling}) of $T$, and $T$ an \emph{arm-super edge-magic total graph}.\qqed
\end{defn}

\begin{defn} \label{defn:am-parameterized-matchings}
\cite{SH-ZXH-YB-2017-2} Let $T$ be a tree with a perfect matching $M$.

(1) Suppose that $T$ admits a super $(k,d)$-edge antimagic labeling $g$ defined in Definition \ref{defn:k-d-arithmetic-labelings}. If $|g(u)-g(v)|=k$ for each matching edge $uv \in M$, and

(i) the elements of $g(E(T))$ produce an arithmetic progression having $|M|$ terms with the first term $p+1$ and the common difference $2$;

(ii) $g(u)+g(v)=b_i$ forms a sequence $\{b_i\}$ being an arithmetic progression having $|M|$ terms with the first term $k+2$ and the common difference $2$, where $n=|M|$.

Then we call $g$ an \emph{arithmetic matching super edge antimagic $(k,d)$-total labeling} (\emph{arm-super edge-antimagic $(k,d)$-total labeling}) of $T$, and $T$ an \emph{arm-super edge-antimagic $(k,d)$-total labeling graph}.

(2) Suppose that $T$ admits a $(k,d)$-arithmetic labeling $h$ defined in Definition \ref{defn:k-d-arithmetic-labelings}. If $|h(u)-h(v)|=k$ for each matching edge $uv \in M$, and the elements of $h(E(T))$ produce an arithmetic progression having $|M|$ terms with the first term $k$ and the common difference $2d$. Then we call $h$ an \emph{arithmetic matching $(k,d)$-arithmetic labeling} (\emph{arm-$(k,d)$-arithmetic labeling}) of $T$, and $T$ an \emph{arm-$(k,d)$-arithmetic graph}.\qqed
\end{defn}

\begin{figure}[h]
\centering
\includegraphics[width=16.4cm]{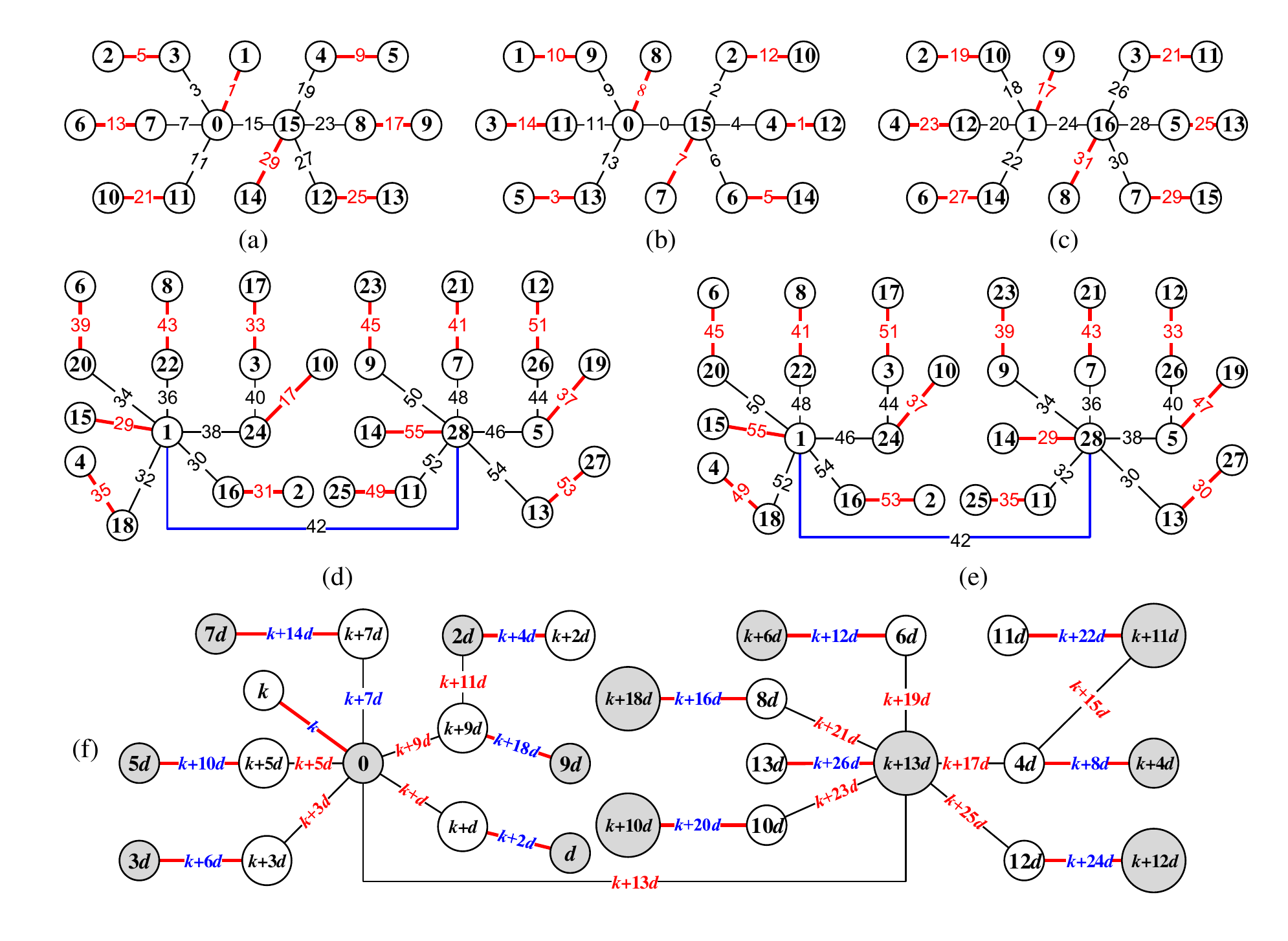}
\caption{\label{fig:for-illustrating-de-59}{\small For illustrating Definition \ref{defn:arm-odd-elegant-perfect-matchings} and Definition \ref{defn:am-parameterized-matchings}: (a) An arm-odd-elegant labeling; (b) an arm-super felicitous labeling; (c) an arm-super edge-magic graceful labeling; (d) an arm-super $(|X|+p+3,2)$-edge antimagic total labeling; (e) an arm-super edge-total graceful labeling; (f) an arm-$(k,d)$-arithmetic labeling.}}
\end{figure}

\subsection{$(k,d)$-constraint labelings on trees with perfect matchings}

\begin{defn} \label{defn:totally-odd-even-graceful-labeling}
\cite{Sun-Zhang-Yao-ICMITE-2017} Let $T$ be a $(p,q)$-tree having a perfect matching $M$, and let

$S_{q-1,k,d}=\{k, k+d, \dots , k +(q-1)d\}$,

$S\,'_{k,d}=\{k, k+2d, \dots , k +2(q-1)d\}$,

$S\,''_{k,d}=\{k, k+d, \dots , k +(p+q-1)d\}$\\
for integers $k\geq 1$ and $d\geq 1$. Suppose that each leaf of $T$ is an end of some matching edge of $M$, and $M(L)$ is the set of all leaves of $T$, and furthermore the graph $T-M(L)$ obtained by deleting all leaves of $T$ is a tree having $c=|V(T-M(L))|$ vertices. The abbreviation of ``arithmetic matching'' is ``arm''.
\begin{asparaenum}[S-1. ]
 \item An \emph{arm-$(k,d)$-totally odd (resp. even) graceful labeling} $f$ of $G$ holds $f(V(G))\subseteq [0, k+2(q-1)d]$, $f(x)\neq f(y)$ for distinct vertices $x,y\in V(G)$ and
 $$f(E(G))=\{|f(u)-f(v)|:\ uv\in E(G)\}=S\,'_{k,d}$$
Meanwhile, $f(u)+f(v)=k+2(q-1)d$ for each matching edge $uv \in M$.

 \item An \emph{arm-$(k,d)$-felicitous labeling} $f$ of $G$ holds $f(V(G))\subseteq [0, k+(q-1)d]$, $f(x)\neq f(y)$ for distinct vertices $x,y\in V(G)$ and
 $$
 f(E(G))=\{f(uv)=f(u)+f(v)~ (\bmod~qd): uv\in E(G)\}=S_{q-1,k,d}
 $$
and furthermore, $f$ is \emph{super} if
 $$
 f(V(G))=\left [0,\left (\frac{p}{2}-1\right )d\right ]\bigcup \left [k+\frac{pd}{2},k +(p-1)d\right ]
 $$
Meanwhile, $|f(u)-f(v)|=k+cd$ for each matching edge $uv \in M$.

 \item An \emph{arm-edge-magic $(k,d)$-total labeling} $f$ of $G$ holds $f(V(G)\cup E(G))\subseteq [0, k+(p+q-1)d]$ such that $f(u)+f(v)+f(uv)=a$ for any edge $uv\in E(G)$, where the magic constant $a$ is a fixed integer; and furthermore, $f$ is \emph{super} if $f(V(G))=[1,pd/2]\cup[k+pd/2,k +qd]$. Meanwhile, $|f(u)-f(v)|=k+(c-1)d$ for each matching edge $uv \in M$ and the elements of the edge color set $f(E(T))$ yield an arithmetic progression having $|M|$ terms with the first term $k+2(p-1)d$ and the common difference $-2d$.

 \item An \emph{arm-edge-magic $(k,d)$-graceful labeling} $f$ of $G$ holds $f(V(G)\cup E(G))\subseteq [0, k+(p+q-1)d]$ such that $\mid f(u)+f(v)-f(uv)\mid=a$ for any edge $uv\in E(G)$, where the magic constant $a$ is a fixed integer; and furthermore, $f$ is \emph{super} if $f(V(G))=[1,pd/2]\cup[k+pd/2,k+qd]$. Meanwhile, $|f(u)-f(v)|=k+(c-1)d$ for each matching edge $uv \in M$ and the elements of the edge color set $f(E(T))$ yield an arithmetic progression having $|M|$ terms with the first term $k+pd$ and the common difference $2d$.

 \item Suppose that $T$ admits a $(k,d)$-arithmetic labeling $f$. If $|f(u)-f(v)|=k$ for each matching edge $uv \in M$, and the elements of the edge color set $f(E(T))$ produce an arithmetic progression having $|M|$ terms with the first term $k$ and the common difference $2d$. Then we call $f$ \emph{arm-$(k,d)$-arithmetic labeling} of $T$.

 \item An \emph{arm-$(k,d)$-totally odd (resp. even)-elegant labeling} $f$ of $G$ holds $f(V(G))\subset [0,k+2(q-1)d]$, $f(u)\neq f(v)$ for distinct vertices $u,v\in V(G)$, and
 $$f(E(G))=\{f(uv)=f(u)+f(v)~(\bmod~k+2qd):uv\in E(G)\}=S\,'_{k,d}
 $$ Meanwhile, $|f(u)-f(v)|=k$ for each matching edge $uv \in M$ and the elements of the edge color set $f(E(T))$ yield an arithmetic progression having $|M|$ terms with the first term $k$ and the common difference $4d$.\qqed
\end{asparaenum}
\end{defn}

\begin{figure}[h]
\centering
\includegraphics[width=16.4cm]{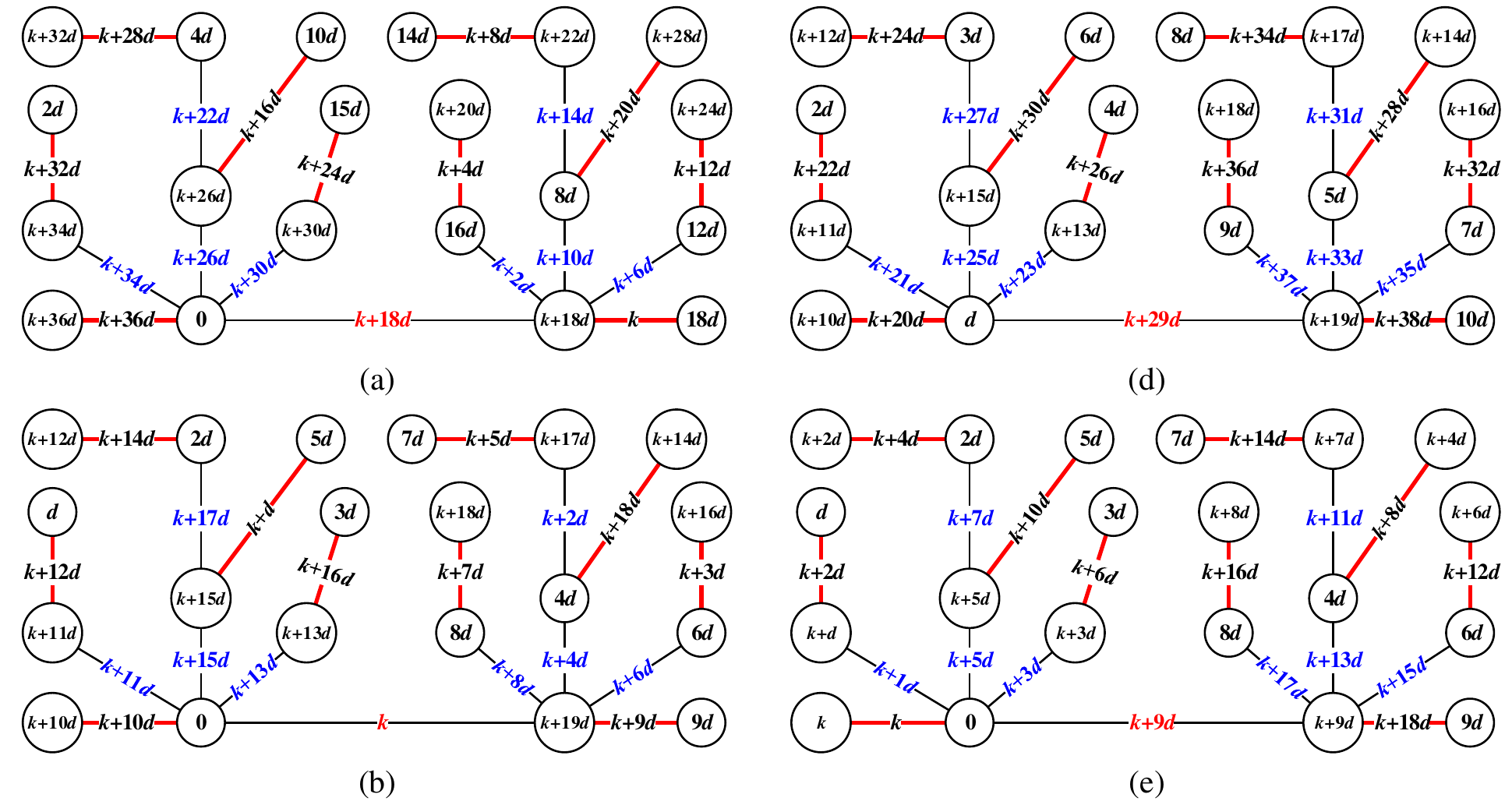}
\caption{\label{fig:matching-labelings-k-d}{\small For illustrating Definition \ref{defn:totally-odd-even-graceful-labeling}: (a) An arm-$(k,d)$-totally odd (resp. even) graceful labeling; (b) an arm-$(k,d)$-super felicitous labeling; (c) an arm-super edge-magic $(k,d)$-total labeling; (d) an arm-$(k,d)$-super edge-magic graceful labeling.}}
\end{figure}

\begin{defn} \label{defn:perfect-matching-definitions}
\cite{Sun-Wang-Yao-2020} Suppose that $T$ is a tree of $p$ vertices and $q$ edges, and has a perfect matching $M$ and its vertex bipartition $V(T)=X\cup Y$ with $X\cap Y=\emptyset$ such that each edge $xy\in E(T)$ holds $x\in X$ and $y\in Y$. Let $S_{m,\alpha,k\beta}=\{\alpha, \alpha + k\beta, \dots , \alpha +k(m-1)\beta\}$ with integers $k\geq 1$, $m\geq 2$, $\alpha\geq 0$ and $\beta\geq 1$, and let $c$ be a constant, as well as $\lambda=\lfloor \frac{p}{2} \rfloor$. There are the following restrictive conditions:

\begin{asparaenum}[C-1. ]
\item \label{parameter:totally-odd-even-graceful-v} $\sigma: V(T)\rightarrow S_{p-1,0,\beta}\cup S_{2p-3,\alpha,\beta}$, and $\sigma(x)\neq \sigma(w)$ for any pair of vertices $x,w\in V(T)$;
\item \label{parameter:felicitous-v} $\sigma: V(T)\rightarrow S_{\lambda,0,\beta}\cup S_{p,\alpha,\beta}$, and $\sigma(x)\neq \sigma(w)$ for any pair of vertices $x,w\in V(T)$;
\item \label{parameter:totally-odd-even-elegant} $\sigma(E(T))=\{\sigma(xy)=\sigma(x)+\sigma(y):xy\in E(T)\}=S_{q,\alpha,2\beta}$;
\item \label{parameter:arithmetic-labeling-e} $\sigma(E(T))=\{\sigma(xy)=\sigma(x)+\sigma(y):xy\in E(T)\}=S_{q,\alpha,\beta}$;

\item \label{parameter:totally-odd-even-graceful-e} Induced edge label $\sigma(xy)=|\sigma(x)-\sigma(y)|$ for each edge $xy\in E(T)$, and $\sigma(E(T))=\{\sigma(xy):xy\in E(T)\}=S_{q,\alpha,\beta}$;
\item \label{parameter:totally-odd-even-graceful-2e} Induced edge label $\sigma(xy)=|\sigma(x)-\sigma(y)|$ for each edge $xy\in E(T)$, and $\sigma(E(T))=\{\sigma(xy):xy\in E(T)\}=S_{q,\alpha,2\beta}$;
\item \label{parameter:felicitous-e} Induced edge label $\sigma(xy)-\alpha=\sigma(x)+\sigma(y)-\alpha~ (\bmod ~q\beta)$ for each edge $xy\in E(T)$, and $\sigma(E(T))=\{\sigma(xy):xy\in E(T)\}=S_{q,\alpha,\beta}$;

\item \label{parameter:edge-magic-total} $\sigma: V(T)\cup E(T)\rightarrow S_{\lambda+1,0,\beta}\cup S_{p+q,\alpha,\beta}$, and $\sigma(x)\neq \sigma(w)$ for any pair of vertices $x,w\in V(T)$;

\item \label{parameter:strongly-grace}$\sigma(X)=S_{\lambda,0,\beta}$ and $\sigma(Y)=S_{p-1,\alpha,\beta}\setminus S_{\lambda-1,\alpha,\beta}$;
\item \label{parameter:super-0}$\sigma(X)=S_{\lambda,0,\beta}$ and $\sigma(Y)=S_{p,\alpha,\beta}\setminus S_{\lambda,\alpha,\beta}$;
\item \label{parameter:super} $\sigma(X)=S_{\lambda,0,\beta}\setminus \{0\}$ and $\sigma(Y)=S_{p,\alpha,\beta}\setminus S_{\lambda,\alpha,\beta}$;

------------ \emph{matching-set}

\item \label{parameter:felicitous-matching} $|\sigma(x)-\sigma(y)|=\alpha+\frac{1}{2}p \beta$ for each matching edge $xy \in M$;
\item \label{parameter:edge-magic-total-matching} $|\sigma(x)-\sigma(y)|=\alpha+\frac{1}{2}(p-2)\beta$ for each matching edge $xy \in M$;
\item \label{parameter:totally-odd-even-graceful-matching} $\sigma(x)+\sigma(y)=\alpha+2(q-1)\beta$ for each matching edge $xy \in M$;
\item \label{parameter:totally-odd-even-elegant-matching} $|\sigma(x)-\sigma(y)|=\alpha$ for each matching edge $xy \in M$;
\item \label{parameter:edge-difference-matching} $|\sigma(x)-\sigma(y)|=\alpha+(q-1)\beta$ for each matching edge $xy \in M$;

\item \label{parameter:arithmetic-matching-set} $\sigma(M)=S_{|M|,\alpha,2\beta}$;
\item \label{parameter:matching-set-2} $\sigma(M)=S_{|M|,\alpha,4\beta}$;
\item \label{parameter:matching-set-0}$\sigma(M)=\{\alpha+2(p-1)\beta, \alpha+2(p-2)\beta,\dots ,\alpha+2(p-|M|)\beta\}$;
\item \label{parameter:matching-set-1} $\sigma(M)=\{\alpha+p\beta, \alpha+(2+p)\beta,\dots ,\alpha+[2(|M|-1)+p]\beta\}$;

------------ \emph{edge-magic}

\item \label{parameter:edge-magic-total-number} $\sigma(x)+\sigma(y)+\sigma(xy)=c$ for each edge $xy\in E(T)$;
\item \label{parameter:edge-difference} $\sigma(xy)+|\sigma(x)-\sigma(y)|=c$ for each edge $xy\in E(T)$;
\item \label{parameter:edge-magic-total-graceful} $|\sigma(x)+\sigma(y)-\sigma(xy)|=c$ for each edge $xy\in E(T)$;
\item \label{parameter:graceful-difference} $\big |\sigma(xy)-|\sigma(x)-\sigma(y)|\big |=c$ for each edge $xy\in E(T)$.
\end{asparaenum}
\textbf{Then the labeling $\sigma$ is}:
\begin{asparaenum}[\textrm{Label}-1. ]
\item \cite{S-M-Hegde2000} an \emph{$(\alpha,\beta)$-graceful labeling} if C-\ref{parameter:felicitous-v}, C-\ref{parameter:totally-odd-even-graceful-e} and C-\ref{parameter:edge-difference-matching} hold true;
\item a \emph{set-ordered $(\alpha,\beta)$-graceful labeling} if C-\ref{parameter:felicitous-v}, C-\ref{parameter:totally-odd-even-graceful-e}, C-\ref{parameter:edge-difference-matching} and C-\ref{parameter:strongly-grace} hold true (see an example shown in Fig.\ref{fig:1-definition} (a));

\item an \emph{$(\alpha,\beta)$ totally odd/even-graceful labeling} if C-\ref{parameter:totally-odd-even-graceful-v}, C-\ref{parameter:totally-odd-even-graceful-2e} and C-\ref{parameter:totally-odd-even-graceful-matching} hold true (see an example shown in Fig.\ref{fig:1-definition} (b));

\item an \emph{arm-$(\alpha,\beta)$ totally odd/even-elegant labeling} if C-\ref{parameter:totally-odd-even-graceful-v}, C-\ref{parameter:totally-odd-even-elegant}, C-\ref{parameter:totally-odd-even-elegant-matching} and C-\ref{parameter:matching-set-2} hold true (see an example shown in Fig.\ref{fig:1-definition} (c));
\item an \emph{arm-$(\alpha,\beta)$ felicitous labeling} if C-\ref{parameter:felicitous-v}, C-\ref{parameter:felicitous-e} and C-\ref{parameter:felicitous-matching} hold true;
\item a \emph{super arm-$(\alpha,\beta)$ felicitous labeling} if C-\ref{parameter:felicitous-v}, C-\ref{parameter:felicitous-e}, C-\ref{parameter:felicitous-matching} and C-\ref{parameter:super-0} hold true (see an example shown in Fig.\ref{fig:2-definition} (e));

\item an \emph{arm-$(\alpha,\beta)$ arithmetic labeling} if C-\ref{parameter:edge-magic-total}, C-\ref{parameter:totally-odd-even-elegant-matching}, C-\ref{parameter:arithmetic-labeling-e} and C-\ref{parameter:arithmetic-matching-set} hold true (see an example shown in Fig.\ref{fig:2-definition} (d));

------------ \emph{$W$-magic}

\item an \emph{arm-$(\alpha,\beta)$ edge-magic total labeling} if C-\ref{parameter:edge-magic-total}, C-\ref{parameter:edge-magic-total-number}, C-\ref{parameter:edge-magic-total-matching} and C-\ref{parameter:matching-set-0} hold true;
\item a \emph{super arm-$(\alpha,\beta)$ edge-magic total labeling} if C-\ref{parameter:edge-magic-total}, C-\ref{parameter:edge-magic-total-number}, C-\ref{parameter:edge-magic-total-matching}, C-\ref{parameter:super} and C-\ref{parameter:matching-set-0} hold true (see an example shown in Fig.\ref{fig:2-definition} (f));
\item an \emph{arm-$(\alpha,\beta)$ edge-magic graceful labeling} if C-\ref{parameter:edge-magic-total}, C-\ref{parameter:edge-magic-total-matching}, C-\ref{parameter:edge-magic-total-graceful} and C-\ref{parameter:matching-set-1} hold true;
\item a \emph{super arm-$(\alpha,\beta)$ edge-magic graceful labeling} if C-\ref{parameter:edge-magic-total}, C-\ref{parameter:edge-magic-total-matching}, C-\ref{parameter:edge-magic-total-graceful}, C-\ref{parameter:super} and C-\ref{parameter:matching-set-1} hold true (see an example shown in Fig.\ref{fig:3-definition} (g));

\item an \emph{arm-$(\alpha,\beta)$ edge-difference labeling} if C-\ref{parameter:totally-odd-even-graceful-v}, C-\ref{parameter:edge-difference} and C-\ref{parameter:edge-difference-matching} hold true;
\item a \emph{super arm-$(\alpha,\beta)$ edge-difference labeling} if C-\ref{parameter:totally-odd-even-graceful-v}, C-\ref{parameter:edge-difference}, C-\ref{parameter:edge-difference-matching} and C-\ref{parameter:strongly-grace} hold true (see an example shown in Fig.\ref{fig:3-definition} (h));

\item an \emph{arm-$(\alpha,\beta)$ graceful-difference labeling} if C-\ref{parameter:totally-odd-even-graceful-v}, C-\ref{parameter:graceful-difference} and C-\ref{parameter:edge-difference-matching} hold true;
\item a \emph{super arm-$(\alpha,\beta)$ graceful-difference labeling} if C-\ref{parameter:totally-odd-even-graceful-v}, C-\ref{parameter:graceful-difference}, C-\ref{parameter:edge-difference-matching} and C-\ref{parameter:strongly-grace} hold true (see an example shown in Fig.\ref{fig:3-definition} (i)).\qqed
\end{asparaenum}
\end{defn}

\begin{figure}[h]
\centering
\includegraphics[width=16.4cm]{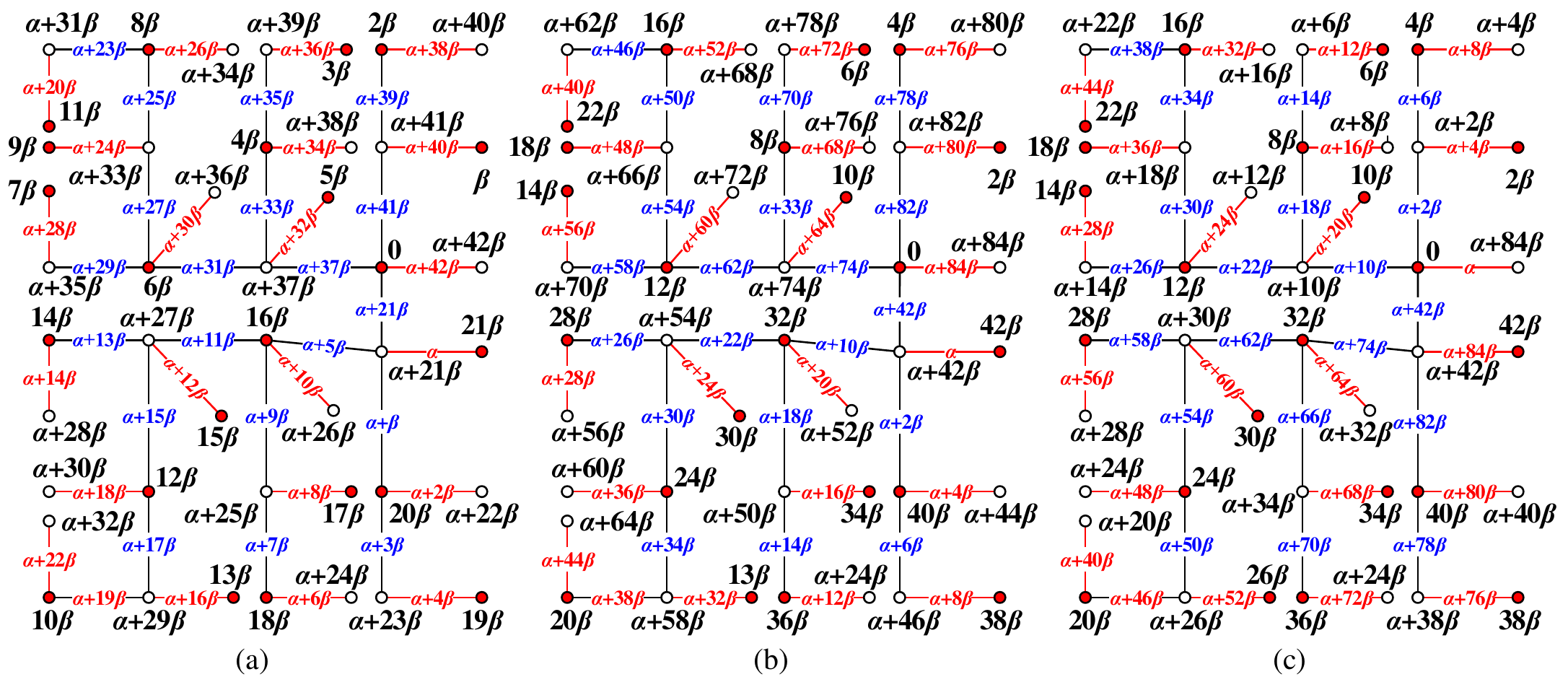}
\caption{\label{fig:1-definition}{\small $p=44$, $q=43$ and $|M|=22$ in Definition \ref{defn:perfect-matching-definitions}: (a) A strongly set-ordered $(\alpha,\beta)$-graceful labeling; (b) an arm-$(\alpha,\beta)$ totally odd/even-graceful labeling; (c) an arm-$(\alpha,\beta)$ totally odd/even-elegant labeling, cited from \cite{Sun-Wang-Yao-2020}.}}
\end{figure}

\begin{figure}[h]
\centering
\includegraphics[width=16.4cm]{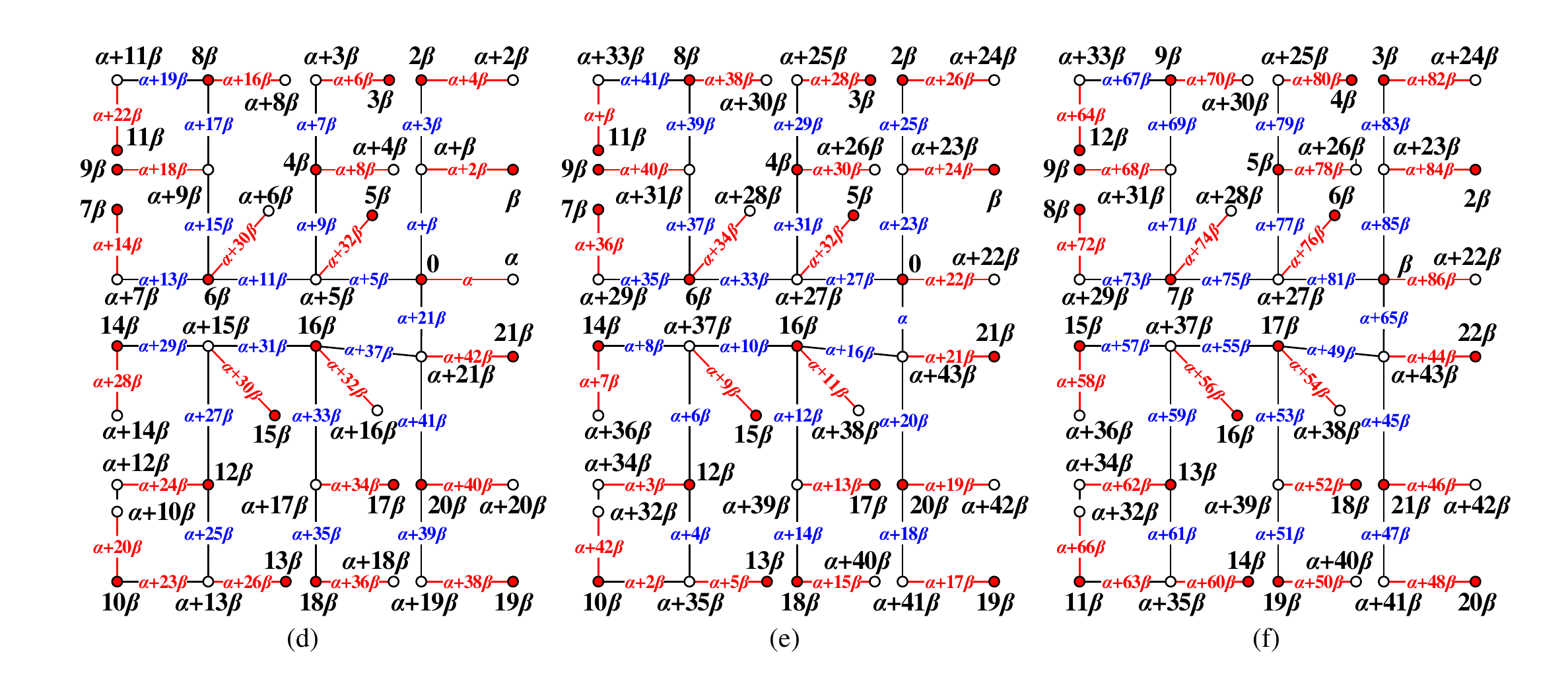}
\caption{\label{fig:2-definition}{\small $p=44$, $q=43$ and $|M|=22$ in Definition \ref{defn:perfect-matching-definitions}: (d) An arm-$(\alpha,\beta)$ arithmetic labeling; (e) a super arm-$(\alpha,\beta)$ felicitous labeling; (f) a super arm-$(\alpha,\beta)$ edge-magic total labeling, cited from \cite{Sun-Wang-Yao-2020}.}}
\end{figure}

\begin{figure}[h]
\centering
\includegraphics[width=16.4cm]{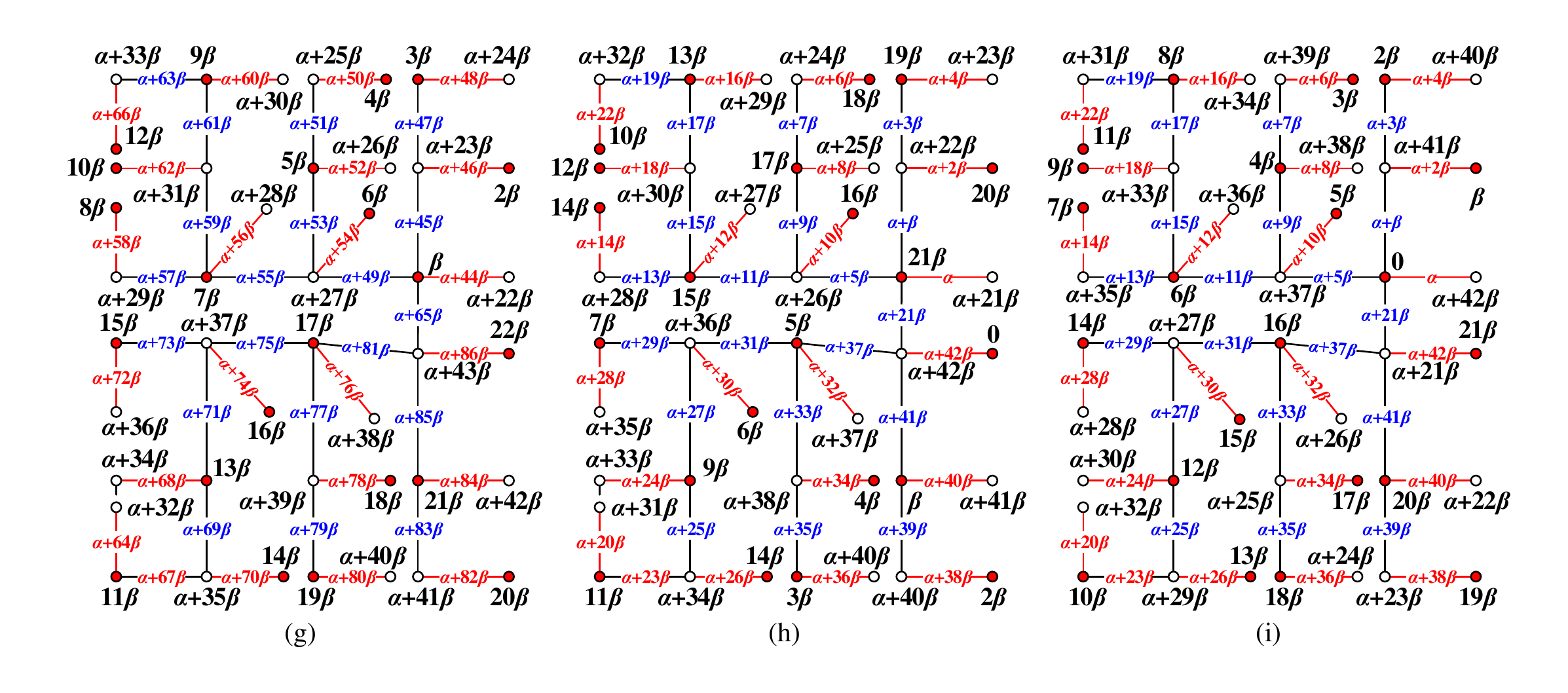}
\caption{\label{fig:3-definition}{\small $p=44$, $q=43$ and $|M|=22$ in Definition \ref{defn:perfect-matching-definitions}: (g) a super arm-$(\alpha,\beta)$ edge-magic graceful labeling; (h) a super arm-$(\alpha,\beta)$ graceful-difference labeling; (i) a super arm-$(\alpha,\beta)$ graceful-difference labeling, cited from \cite{Sun-Wang-Yao-2020}.}}
\end{figure}

\begin{defn} \label{defn:111111}
\cite{Sun-Wang-Yao-2020} Let $f$ be a strongly set-ordered $(\alpha,\beta)$-graceful labeling of a tree $T$, we have the following equivalent transformations:

(1) An arm-$(\alpha,\beta)$ graceful-difference labeling $\sigma$ is defined as: $\sigma(w)=f(w)$ for $w\in V(T)$, $\sigma(xy)=2\alpha+(q-1)\beta-f(xy)$ for each edge $xy\in E(T)$.

(2) An arm-$(\alpha,\beta)$ edge-difference labeling $\sigma$ is defined as: $\sigma(x)=\max f(X)+\min f(X)-f(x)$ for $x\in X$, $\sigma(y)=\max f(Y)+\min f(Y)-f(y)$ for $y\in Y$, $\sigma(xy)=\max f(E(T))+\min f(E(T))-f(xy)$ for each edge $xy\in E(T)$.\qqed
\end{defn}

\begin{lem} \label{them:particular-trees-perfect-matching}
\cite{Sun-Wang-Yao-2020} A \emph{leaf-matching tree} $T$ with a perfect matching $M$ has its leaf set $L(T)$ such that each matching edge $uv\in M$ has one end $u$ in $L(T)$. If the leaf-matching tree $T$ admits a strongly set-ordered graceful labeling, then it admits the following pairwise equivalent labelings:

(1) strongly set-ordered graceful labeling;

(2) strongly $(\alpha,\beta)$-graceful labeling;

(3) arm-$(\alpha,\beta)$-totally odd/even-graceful labeling;

(4) arm-$(\alpha,\beta)$-totally odd/even-elegant labeling;

(5) arm-$(\alpha,\beta)$ arithmetic labeling;

(6) arm-$(\alpha,\beta)$ super felicitous labeling;

(7) arm-$(\alpha,\beta)$ super edge-magic total labeling;

(8) arm-$(\alpha,\beta)$ super edge-magic graceful labeling.
\end{lem}

\begin{lem} \label{thm:exchange-labels-in-perfect-matching}
Suppose that a tree $T$ has a perfect matching $M(T)$ and admits a strongly graceful labeling $\theta$. There are three operations:

$(i)$ \textbf{Exchanging-edge operation}. A new strongly graceful labeling $\varphi$ of $T$ holding $\varphi(x)=\theta(y)$ and $\varphi(y)=\theta(x)$ for
each matching edge $xy\in M(T)$.

$(ii)$ \textbf{Decreasing-leaf operation}. If a matching edge $xy\in M(T)$ satisfies the neighborhood $N(x)=\{y,x_i:i\in [1,d_x-1]\}$ with $d_x=\textrm{deg}_T(x)\geq 3$ and $\textrm{deg}_T(y)=1$. Deleting the edges $xx_j$ and joins $y$ with $x_j$ for $j\in [s+1,d_x-1]$ with $0\leq s<d_x-1$ produces a new tree $H$ such that $H$ is also strongly graceful, and its perfect matching $M(H)$ keeps
$$
|L(T)\cap V(M(T))|-1=|L(H)\cap V(M(H))|
$$

$(iii)$ \textbf{Increasing-leaf operation}. If a matching edge $xy\in M(T)$ satisfies $\textrm{deg}_T(x)\geq 2$ and $\textrm{deg}_T(y)\geq 2$. A process of (a) moving the edge $xy$ from $T$; (b) identifying the vertex $x$ with the vertex $y$ into one denoted as $x_0$; and (c) adding a new vertex $y_0$ to join with $x_0$ will build up another tree $T\,'$ having a strongly graceful labeling. Furthermore, the perfect matching $M(T\,')=\{x_0y_0\}\cup \big (M(T)\setminus\{xy\}\big )$ holds
$$
|L(T)\cap V(M(T))|+1=|L(T\,')\cap V(M(T\,'))|
$$
\end{lem}
\begin{proof} In order to prove the operation $(i)$, let $\theta$ be a strongly graceful labeling of a tree $T$ with $n$ vertices and a perfect matching $M(T)$. We define a labeling $\varphi$ for $T$ as: $\varphi(u)=\theta(v)$ and $\varphi(v)=\theta(u)$ for each edge $uv\in M(T)$. Clearly, $\varphi(u)+\varphi(v)=\theta(v)+\theta(u)=n-1$ and $|\varphi(u)-\varphi(v)|=|\theta(v)-\theta(u)|$ for each edge $uv\in M(T)$. For each edge $xy\in E(T)\setminus M(T)$, there are two edges
$xx\,',yy\,'\in M(T)$ so that
$${
\begin{split}
|\varphi(x)-\varphi(y)|&=|\theta(x\,')-\theta(y\,')|=\big |(n-1)-\theta(y\,')-\big [(n-1)-\theta(x\,')\big ]\big |=|\theta(y)-\theta(x)|
\end{split}
}$$

We obtain the vertex color set $\theta(V(T))=\varphi(V(T))$ and the edge color set $\theta(E(T))=\varphi(E(T))$, as well as $\varphi(u)+\varphi(v)=n-1$ for each edge $uv\in M(T)$. Therefore, $\varphi$ is really a strongly graceful labeling of $T$.

To prove the operation $(ii)$, we assume that there is an edge $uv\in M(T)$ such that the neighborhood $N(u)=\{v,u_i:i\in
[1,d_u-1]\}$ and $N(u)\cap L(T)=\{v\}$ since $N(T)$ is a perfect matching, where $d_u=\textrm{deg} _T(u)\geq 3$. We have a tree $H$ defined by the hypothesis of the operation $(ii)$. Since $H-uv$ has two components $T\,'$ and $T\,''$, without loss of generality, the
vertex $u$ lies in $T\,'$. Notice that $M(H)=M(T)$. We define a labeling $\varphi$ for $H$ as: $\varphi(x)=\theta(x)$ with $x\in V(T\,')$;
$\varphi(v)=\theta(v)$; $\varphi(x)=\theta(y)$ and $\varphi(y)=\theta(x)$ for each edge $xy\in E(T\,'')\cap M(T)$.

Since $V(H)=V(T)$, so two vertex color sets $\varphi(V(H))=\theta(V(T))$. For edges $xy\in E(T\,')$, we have
$$
|\varphi(x)-\varphi(y)|=|\theta(x)-\theta(y)|,~|\varphi(u)-\varphi(v)|=|\theta(u)-\theta(v)|
$$ and every edge $xy\in E(T\,'')\cap M(T)$ holds $|\varphi(x)-\varphi(y)|=|\theta(x)-\theta(y)|$ true.

Considering each edge $xy\in E(T\,'')\setminus M(T)$, so there are two edges $xx\,',yy\,'\in E(T\,'')\cap M(T)$.

Notice that $\varphi(x)=\theta(x\,')$ and $\varphi(y)=\theta(y\,')$, and $n=|V(H)|$~$(=|V(T)|)$. We have
$$
|\varphi(x)-\varphi(y)|=|\theta(x\,')-\theta(y\,')|=|(n-1)-\theta(y\,')-[(n-1)-\theta(x\,')]|=|\theta(y)-\theta(x)|
$$
since $\theta(x)+\theta(x\,')=n-1$ and $\theta(y)+\theta(y\,')=n-1$ according to the definition of a strongly graceful labeling. Therefore,
$\varphi(E(H))=\theta(E(T))$, which means that $\varphi$ is a strongly graceful labeling of $H$. Clearly,
$$
|L(T)\cap V(M(T))|-1=|L(H)\cap V(M(H))|
$$

Notice that in the above proof of the operation $(ii)$, the strongly graceful tree $H$ and a strongly graceful labeling $\varphi$ of $H$
enable us to obtain the original tree $T$ and a strongly graceful labeling $\theta$ of $T$. Thereby, the inverse of the proof of the operation $(ii)$ is just the proof of the operation $(iii)$.
\end{proof}

\begin{defn} \label{defn:mismatched-operation}
\cite{Sun-Wang-Yao-2020} Let a $(p,q)$-graph $G$ have a perfect matching $M$ and a strongly labeling $\theta$. If there are two edges $uv\in E(G)$ and $xy\notin E(G)$ holding $\theta(xy)=\theta(uv)$, then the operation of removing the edge $uv$ from $G$ and adding the edge $xy$ to the remainder $G-uv$ is called an \emph{edge-mismatched transfer-operation}, the resultant graph is denoted as $G-uv+xy$, called \emph{$\pm$-edge graph}, and it admits a strongly labeling $\pi$ induced by the strongly labeling $\theta$.
\end{defn}

By the operation defined in Definition \ref{defn:mismatched-operation} and the operations stated in Lemma \ref{thm:exchange-labels-in-perfect-matching}, Sun \emph{et al.}, in the article \cite{Sun-Wang-Yao-2020}, presented a conjecture as follows:

\begin{conj}\label{conj:strongly-graceful-matching-path-tree}
\cite{Sun-Wang-Yao-2020} Any tree with a perfect matching can be transformed into some path with a perfect matching by the edge-mismatched transfer-operation defined in Definition \ref{defn:mismatched-operation} such that they both admit strongly graceful labelings.
\end{conj}

\begin{figure}[h]
\centering
\includegraphics[width=16.4cm]{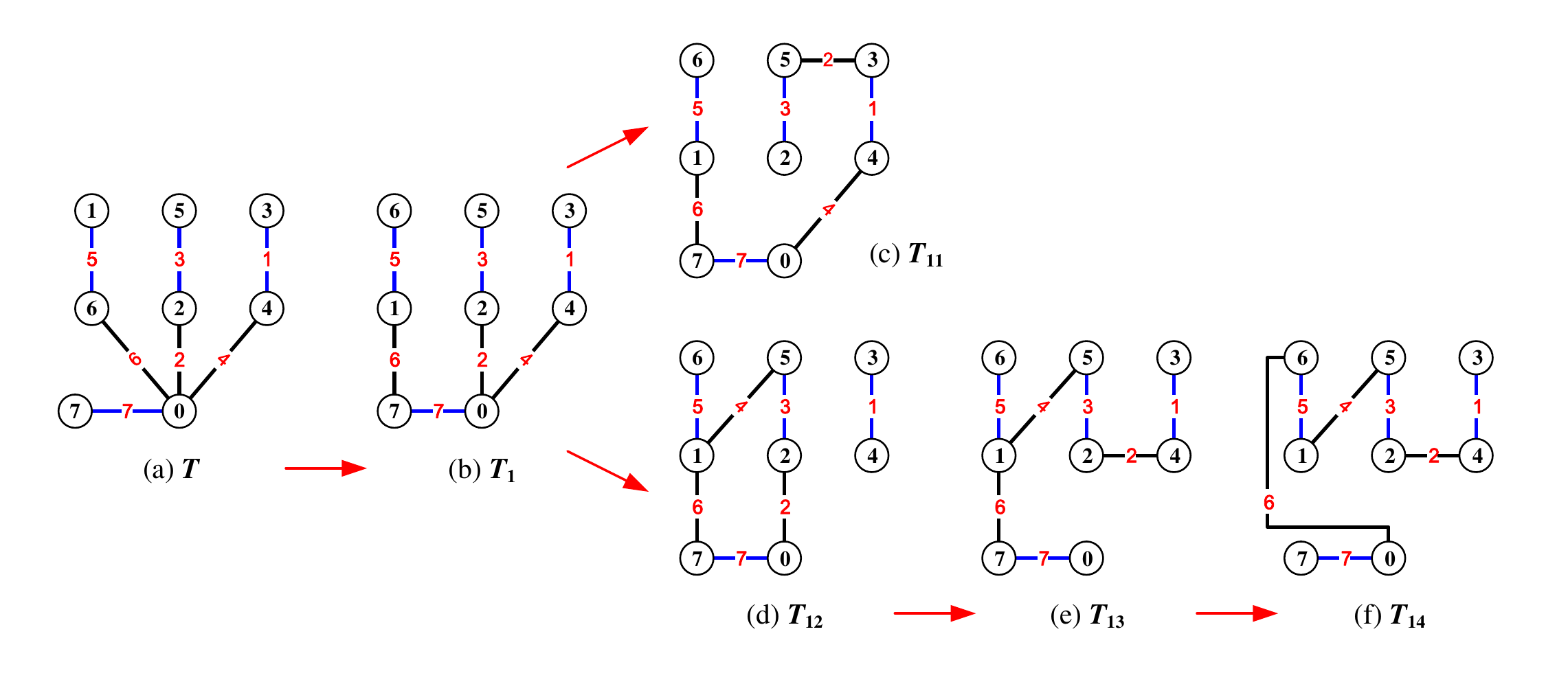}
\caption{\label{fig:matching-path-tree}{\small A scheme for illustrating Definition \ref{defn:mismatched-operation} and Conjecture \ref{conj:strongly-graceful-matching-path-tree}.}}
\end{figure}

\subsection{The $\pm e$-operations and the $\pm e$-graphs}

For the needs of theoretical research and practical application, we propose the following problems:

\begin{defn} \label{defn:add-remove-edge-tree-set}
$^*$ Let $T$ be a uncolored tree of $q$ edges. If each uncolored $\pm$-edge graph $T_{i+1}=T_i-u_iv_i+x_iy_i$ is just a tree of $q$ edges holding $u_iv_i\in E(T_i)$ and $x_iy_i\notin E(T_i)$, where $T_1=T$, then we put all uncolored $\pm$-edge trees $T_{i+1}=T_i-u_iv_i+x_iy_i$ into a set $S_{\pm e}(T)$, and write
\begin{equation}\label{eqa:e-operation-graph-homomorphism}
T_i\rightarrow ^1_{\pm e}T_{i+1}
\end{equation} called \emph{$\pm e$-operation graph homomorphism}. The set $S_{\pm e}(T)$ is called \emph{$\pm e$-tree set}, and $T$ is called \emph{source} of the $\pm e$-tree set $S_{\pm e}(T)$.\qqed
\end{defn}

\begin{problem}\label{qeu:444444} About Definition \ref{defn:add-remove-edge-tree-set}, we have the following facts:

(1) The $\pm e$-tree set $S_{\pm e}(T)$ contains a unique path of $q$ edges;

(2) There is a tree $H\in S_{\pm e}(T)$ holding diameters $D(H)\leq D(G)$ for each tree $G\in S_{\pm e}(T)$, then $H$ is a star $K_{1,q}$.\\
\textbf{Consider} the following questions:

\begin{asparaenum}[(i) ]
\item If $T$ admits a $W$-constraint coloring/labeling, \textbf{does} each tree $G\in S_{\pm e}(T)$ admit a $W$-constraint coloring/labeling too? If it is not so, \textbf{find} the subset $S_{W\textrm{-c-c}}(T)\subset S_{\pm e}(T)$, such that each tree in the subset $S_{W\textrm{-c-c}}(T)$ admits a $W$-constraint coloring/labeling.
\item For two different trees $T, T\,'$ of $q$ edges, \textbf{are} there $\pm e$-tree sets $S_{\pm e}(T)=S_{\pm e}(T\,')$?
\item If $T$ is a tree with a perfect matching as $q+1=2m$, \textbf{determine} the subset $S_{match}(T)\subseteq S_{\pm e}(T)$ such that each tree of $S_{match}(T)$ has a perfect matching. Moreover, if the perfect matching tree $T$ admits a strongly $W$-constraint coloring/labeling, \textbf{does} each tree in $S_{match}(T)$ admit a strongly $W$-constraint coloring/labeling too?

\item \textbf{Determine} all trees with diameter $D$ in $S_{\pm e}(T)$ for each diameter $D\in [3,q-1]$, and the subset $S_{dia}(T;D)\subset S_{\pm e}(T)$ collects these trees with diameter $D$.

\item \textbf{Determine} all trees with maximal degree $\Delta$ in $S_{\pm e}(T)$ for each maximal degree $\Delta\in [3,q-1]$, and we use a subset $S_{max}(T;\Delta)\subset S_{\pm e}(T)$ to collect these trees with maximal degree $\Delta$.

\item \textbf{Determine} all trees with $m$ leaves in $S_{\pm e}(T)$ for each integer $m\in [3,q-1]$, and we put them into the subset $S_{leaf}(T;m)\subset S_{\pm e}(T)$.
\item \textbf{Find} three integers $D_0,\Delta_0$ and $m_0$ falling into $3\leq D_0,\Delta_0, m_0\leq q-1$, such that
\begin{equation}\label{eqa:three-integer-D-0-Delta-0-m-0}
S_{\pm e}(T)=S_{dia}(T;D_0)\cup S_{max}(T;\Delta_0)\cup S_{leaf}(T;m_0)
\end{equation} holds true. Moreover, \textbf{consider} Eq.(\ref{eqa:three-integer-D-0-Delta-0-m-0}) for one of the following situations:

\qquad (1) $D_0,\Delta_0$ and $m_0$ are prime numbers.

\qquad (2) $D_0,\Delta_0$ and $m_0$ are odd integers.

\qquad (3) $D_0,\Delta_0$ and $m_0$ are even integers.

\qquad (4) Each of odd integers $D_0,\Delta_0$ and $m_0$ is just a product of two prime numbers.

\qquad (5) Each of even integers $D_0,\Delta_0$ and $m_0$ is just a sum of two prime numbers.
\item \textbf{Find} subsets $S_{para}(T;D,\Delta,m)\subset S_{\pm e}(T)$ for each group of integers $D,\Delta,m$ holding $3\leq D,\Delta, m\leq q-1$, such that each tree $G\in S_{para}(T;D,\Delta,m)$ just has $m$ leaves, maximal degree $\Delta(G)=\Delta$ and diameter $D(G)=D$.
\item For two trees $G\,'$ and $G\,''$ of $S_{\pm e}(T)$, if there are $T_{i+1}=T_i-u_iv_i+x_iy_i\in S_{\pm e}(T)$ for $i\in [1,r]$, such that $G\,'=T_1$ and $G\,''=T_{r+1}$, then by Eq.(\ref{eqa:e-operation-graph-homomorphism}), we get a series of \emph{$\pm e$-operation graph homomorphisms}
\begin{equation}\label{eqa:555555}
G\,'=T_1 \rightarrow ^1_{\pm e}T_2 \rightarrow ^1_{\pm e}T_3\rightarrow ^1_{\pm e}~\cdots ~\rightarrow ^1_{\pm e}T_r \rightarrow ^1_{\pm e}T_{r+1}=G\,''
\end{equation} and we write it in a short term $G\,'\rightarrow ^{[1,r]}_{\pm e}G\,''$. If $r$ is the smallest in all $\pm e$-operation graph homomorphisms $G\,'\rightarrow ^{[1,t]}_{\pm e}G\,''$, we write $r=D^{\pm e}_{homo}\langle G\,',G\,''\rangle $, called \emph{$\pm e$-operation graph homomorphism distance}. \textbf{Determine} the $\pm e$-tree set $S_{\pm e}(T)$'s diameter
\begin{equation}\label{eqa:555555}
D(S_{\pm e}(T))=\max\{D^{\pm e}_{homo}\langle G\,',G\,''\rangle :~G\,',G\,''\in S_{\pm e}(T)\}
\end{equation}
\item Suppose that a connected graph $J$ admits a vertex graph-coloring $\theta:V(J)\rightarrow U^*$, where $U^*$ is one of the $\pm e$-tree sets $S_{match}(T)$, $S_{max}(T;\Delta)$, $S_{dia}(T;D)$, $S_{leaf}(T;m)$, and $S_{\pm e}(T)$, such that $\theta (u)=\theta (v)-x_vy_v+w_uz_u$ with $x_vy_v\in E(\theta (v))$ and $w_uz_u\notin E(\theta (v))$ for each edge $uv\in E(H)$, and moreover there is an edge $u\,'v\,'\in E(H)$ holding $\theta (u\,')=T_i$ and $\theta (v\,')=T_{i+1}$ if there is $T_{i+1}=T_i-u_iv_i+x_iy_i$ for two trees $T_i,T_{i+1}\in U^*$. We call $J$ the \emph{$\pm e$-tree-intersected graph} of the $\pm e$-tree set $U^*$, like that in \emph{hypergraphs} (Ref. \cite{Yao-Ma-arXiv-2201-13354v1}). Immediately, we get two diameters $D(J)$ and $D(S_{\pm e}(T))$ holding $D(J)=D(S_{\pm e}(T))$ true.

\qquad Since the $\pm e$-tree-intersected graph $J$ describes indirectly the topological structure of each $\pm e$-tree set $U^*\in \{S_{match}(T)$, $S_{max}(T;\Delta)$, $S_{dia}(T;D)$, $S_{leaf}(T;m)$, $S_{\pm e}(T)\}$, so it is important to \textbf{characterize} the topological properties of the $\pm e$-tree-intersected graph $J$, such as, Hamilton property, Euler's property, diameter $D(J)$, maximal degree $\Delta(J)$, clique, girth, cycles, various connectivity, colorings, and so on.
\end{asparaenum}
\end{problem}

\begin{defn} \label{defn:colored-add-remove-edge-tree-set}
Suppose that a tree $T^c$ of $q$ edges admits a vertex labeling $g:V(T)\rightarrow [0,q]$ holding the vertex color set $f(V(T))=[0,q]$, then we get a tree set $S_{\pm e}(T^c)$, called \emph{colored $\pm e$-tree set}, by the \emph{$\pm e$-operation} $T^c_{i+1}=T^c_i-u_iv_i+x_iy_i$ holding $u_iv_i\in E(T^c_i)$ and $x_iy_i\notin E(T^c_i)$, where $T^c_1=T^c$, and two colored trees $T^c_{i+1}$ and $T^c_i$ form a \emph{$\pm e$-operation colored-graph homomorphism}
\begin{equation}\label{eqa:e-operation-colored-graph-homomorphism}
T^c_i\rightarrow ^{1}_{\pm e}T^c_{i+1}
\end{equation} such that any pair of trees $H\,'$ and $H\,''$ of $S_{\pm e}(T^c)$ corresponds to $T^c_{i+1}=T^c_i-u_iv_i+x_iy_i\in S_{\pm e}(T^c)$ for $i\in [1,A]$ with $H\,'=T^c_1$ and $H\,''=T^c_{A+1}$, so we have a series of \emph{$\pm e$-operation graph homomorphisms}
\begin{equation}\label{eqa:colored-graph-homomorphisms-series}
H\,'=T^c_1 \rightarrow ^1_{\pm e}T^c_2 \rightarrow ^1_{\pm e}T^c_3\rightarrow ^1_{\pm e}~\cdots ~\rightarrow ^1_{\pm e}T^c_{A} \rightarrow ^1_{\pm e}T^c_{A+1}=H\,''
\end{equation} also, $H\,'\rightarrow ^{[1,A]}_{\pm e}H\,''$, and the $\pm e$-operation graph homomorphism distance between two trees $H\,'$ and $H\,''$ is denoted as
$$D^{\pm e}_{homo}\langle H\,',H\,''\rangle =\min \left \{A:~H\,'\rightarrow ^{[1,A]}_{\pm e}H\,''\right \}
$$ over all series of $\pm e$-operation graph homomorphisms shown in Eq.(\ref{eqa:colored-graph-homomorphisms-series}). The parameter
\begin{equation}\label{eqa:555555}
D(S_{\pm e}(T^c))=\max\{D^{\pm e}_{homo}\langle H\,',H\,''\rangle :~H\,',H\,''\in S_{\pm e}(T^c)\}
\end{equation} is called the \emph{diameter} of the $\pm e$-tree set $S_{\pm e}(T^c)$.\qqed
\end{defn}

\begin{problem}\label{qeu:444444}
Recall Cayley's formula $\tau(K_n)=n^{n-2}$ (Ref. \cite{Bondy-2008}), which tells us that there are $n^{n-2}$ different colored spanning trees in a complete graph $K_n$ admitting a vertex labeling $f:V(K_n)\rightarrow [1,n]$ such that the vertex color set $f(V(K_n))=[1,n]$, so the colored $\pm e$-tree set $S_{\pm e}(T^c)$ defined in Definition \ref{defn:colored-add-remove-edge-tree-set} holds $|S_{\pm e}(T^c)|=q^{q-2}$ true.

About Definition \ref{defn:colored-add-remove-edge-tree-set}, we \textbf{propose} the following problems:
\begin{asparaenum}[\textbf{\textrm{C-}}1. ]
\item \textbf{Determine} the subset $S_{match}(T^c)$ of the colored $\pm e$-tree set $S_{\pm e}(T^c)$ if the colored tree $T^c$ is a tree with a perfect matching.
\item If the vertex labeling $g$ is just a graceful labeling of $T^c$, \textbf{find} out the subset $S_{grace}(T^c)$ of the colored $\pm e$-tree set $S_{\pm e}(T^c)$, such that each tree of $S_{grace}(T^c)$ admits a graceful labeling too.
\item If the vertex labeling $g$ is just a set-ordered graceful labeling of $T^c$, \textbf{find} out the subset $S^{order}_{grace}(T^c)$ of the colored $\pm e$-tree set $S_{\pm e}(T^c)$, such that each tree of $S^{order}_{grace}(T^c)$ admits a set-ordered graceful labeling too.
\item \textbf{Determine} the subset $S^{match}_{grace}(T^c)$ of the colored $\pm e$-tree set $S_{\pm e}(T^c)$ if the colored tree $T^c$ has a perfect matching and the vertex labeling $g$ is just a strongly graceful labeling of $T^c$, such that each tree of $S^{match}_{grace}(T^c)$ has a perfect matching and admits a strongly graceful labeling as $q+1=2m$. Clearly, for any perfect matching tree $G\in S^{match}_{grace}(T^c)$, we have a sequence $\{T^c_i-u_iv_i+x_iy_i\}^r_{i=1}$, such that $G=T^c_r-u_rv_r+x_ry_r$, also, we have a \emph{$\pm e$-operation colored-graph homomorphism}
\begin{equation}\label{eqa:555555}
T^c\rightarrow ^{[1,r]}_{\pm e}G
\end{equation} through the \emph{$\pm e$-operation colored-graph homomorphisms} $T^c_i\rightarrow^1 _{\pm e} T^c_{i+1}$ (see Eq.(\ref{eqa:e-operation-colored-graph-homomorphism})) with each $T^c_i$ is a perfect matching tree for $i\in [1,r]$. Here, the integer $r$ is the smallest $\pm e$-operation colored-graph homomorphism distance $r=D^{\pm e}_{homo}\langle T^c,G\rangle $ over all $\pm e$-operation colored-graph homomorphisms $T^c\rightarrow ^{[1,t]}_{\pm e}G$. \textbf{Find} a perfect matching tree $G^*\in S^{match}_{grace}(T^c)$ such that the $\pm e$-operation colored-graph homomorphism distance $D^{\pm e}_{homo}\langle T^c,G^*\rangle $ is largest in the subset $S^{match}_{grace}(T^c)$.
\item A connected graph $H$ admits a vertex graph-coloring $F:V(H)\rightarrow S^*$, where $S^*$ is one of subsets $S^{match}_{grace}(T^c)$, $S_{match}(T^c)$, $S^{order}_{grace}(T^c)$, and $S_{\pm e}(T^c)$ defined above, such that $F(u)=F(v)-x_vy_v+w_uz_u$ with $x_vy_v\in E(F(v))$ and $w_uz_u\notin E(F(v))$ for each edge $uv\in E(H)$, and moreover there is an edge $u\,'v\,'\in E(H)$ holding $F(u\,')=T^c_i$ and $F(v\,')=T^c_{i+1}$ if there is $T^c_{i+1}=T^c_i-u_iv_i+x_iy_i$ for two trees $T^c_i,T^c_{i+1}\in S^*$. We call $H$ the \emph{$\pm e$-tree-intersected graph} of the colored $\pm e$-tree set $S^*$.

\qquad Because of the $\pm e$-tree-intersected graph $H$ describes indirectly the topological structure of each $\pm e$-tree set $S^*\in \{S^{match}_{grace}(T^c)$, $S_{match}(T^c)$, $S^{order}_{grace}(T^c)$, $S_{\pm e}(T^c)\}$, then the topological structures of each $\pm e$-tree set $S^*$ can be \textbf{described} by the topological properties of each $\pm e$-tree-intersected graph $H$, such as, Hamilton property, Euler's property, diameter $D(H)$, maximal degree $\Delta(H)$, clique, colorings, girth, cycles, various connectivity, \emph{etc} (Ref. \cite{Yao-Ma-arXiv-2201-13354v1}). Notice that each element $F(w)$ of the Topcode-matrix $T_{code}(H, F)$ is a colored $\pm e$-tree corresponding to a Topcode-matrix, so the Topcode-matrix $T_{code}(H, F)$ is like that appeared in \emph{Tensor}.
\end{asparaenum}
\end{problem}

\begin{defn} \label{defn:pan-add-remove-edges-tree-sets}
$^*$ For a uncolored $(q+1,q)$-tree $G$, the $\pm e$-graph set $P_{\pm e}(G)$ contains all $\pm$-edge graphs $G_{i+1}=G_i-u_iv_i+x_iy_i$ with $u_iv_i\in E(G_i)$ and $x_iy_i\notin E(G_i)$, where $G_1=G$. Notice that some $\pm$-edge graphs of the $\pm e$-graph set $P_{\pm e}(G)$ are not connected, or have cycles. Clearly, $S_{\pm e}(G)\subset P_{\pm e}(G)$, and the subset $P_{\pm e}(G)\setminus S_{\pm e}(G)\neq \emptyset$, where the $\pm e$-tree set $S_{\pm e}(G)$ is defined in Definition \ref{defn:add-remove-edge-tree-set}.

As $G^c$ is a colored $(q+1,q)$-tree admitting a vertex labeling $\varphi:V(G^c)\rightarrow [0,q]$ holding the vertex color set $\varphi(V(G^c))=[0,q]$, we define the colored $\pm e$-graph set $P_{\pm e}(G^c)$ containing all colored $\pm$-edge graphs $G^c_{i+1}=G^c_i-u_iv_i+x_iy_i$ with $u_iv_i\in E(G^c_i)$ and $x_iy_i\notin E(G^c_i)$, where $G^c_1=G^c$. Obviously, some $\pm$-edge graphs of the $\pm e$-graph set $P_{\pm e}(G^c)$ are not trees, so we have $S_{\pm e}(G^c)\subset P_{\pm e}(G^c)$, and the subset $P_{\pm e}(G^c)\setminus S_{\pm e}(G^c)\neq \emptyset$, where the colored $\pm e$-tree set $S_{\pm e}(G^c)$ is defined in Definition \ref{defn:colored-add-remove-edge-tree-set}.\qqed
\end{defn}

\begin{problem}\label{qeu:pan-add-remove-edges-treesets-problem}
Let $N(t)$ be a connected network with $t\in [a,b]$, and let $P_{\pm e}(G)$ be the $\pm e$-graph set defined in Definition \ref{defn:pan-add-remove-edges-tree-sets}. We define an $e$-total tree-coloring $F$ for $N(t)$ as
\begin{equation}\label{eqa:555555}
F:V(N(t))\rightarrow P_{\pm e}(G)\setminus S_{\pm e}(G),\quad F:E(G)\rightarrow S_{\pm e}(G)
\end{equation} with each edge $uv\in E(G)$ holding the $\pm e$-constraint
\begin{equation}\label{eqa:555555}
\left\{
\begin{array}{cc}
F(u)=F(uv)-a_ub_u+x_uy_u,& a_ub_u\in E(G),~ x_uy_u\notin E(G)\\
F(v)=F(uv)-a_vb_v+x_vy_v,& a_vb_v\in E(G),~x_vy_v\notin E(G)
\end{array}
\right.
\end{equation}

And we define a $v$-total tree-coloring $\theta$ for $N(t)$ as
\begin{equation}\label{eqa:555555}
\theta:V(N(t))\rightarrow S_{\pm e}(G),\quad \theta:E(G)\rightarrow P_{\pm e}(G)\setminus S_{\pm e}(G)
\end{equation} with each edge $uv\in E(G)$ holding the $\pm e$-constraint
\begin{equation}\label{eqa:555555}
\left\{
\begin{array}{cc}
\theta (u)=\theta (uv)-s_ut_u+w_uz_u,& s_ut_u\in E(G),~ w_uz_u\notin E(G)\\
\theta (v)=\theta (uv)-s_vt_v+w_vz_v,& s_vt_v\in E(G),~w_vz_v\notin E(G)
\end{array}
\right.
\end{equation}

\textbf{Describe} the topological properties of the $\pm e$-graph set $P_{\pm e}(G)$ by the connected network $N(t)$ admitting an $e$-total tree-coloring, or a $v$-total tree-coloring, refer to \cite{Yao-Ma-arXiv-2201-13354v1}.
\end{problem}

\begin{problem}\label{qeu:colored-pan-add-remove-edges-treesets-problem}
Let $N(t)$ be a connected network with $t\in [a,b]$, and let $P_{\pm e}(G^c)$ be the colored $\pm e$-graph set defined in Definition \ref{defn:pan-add-remove-edges-tree-sets}. We define an $e$-total colored-tree-coloring $f^c$ for $N(t)$ as
\begin{equation}\label{eqa:555555}
f^c:V(N(t))\rightarrow P_{\pm e}(G^c)\setminus S_{\pm e}(G^c),\quad f^c:E(G^c)\rightarrow S_{\pm e}(G^c)
\end{equation} with each edge $uv\in E(G^c)$ holding the $\pm e$-constraint
\begin{equation}\label{eqa:555555}
\left\{
\begin{array}{cc}
f^c(u)=f^c(uv)-a_ub_u+x_uy_u,& a_ub_u\in E(G^c),~ x_uy_u\notin E(G^c)\\
f^c(v)=f^c(uv)-a_vb_v+x_vy_v,& a_vb_v\in E(G^c),~x_vy_v\notin E(G^c)
\end{array}
\right.
\end{equation}

And we define a $v$-total colored-tree-coloring $\varphi^c$ for $N(t)$ as
\begin{equation}\label{eqa:555555}
\varphi^c:V(N(t))\rightarrow S_{\pm e}(G^c),\quad \varphi^c: E(G^c)\rightarrow P_{\pm e}(G^c)\setminus S_{\pm e}(G^c)
\end{equation} with each edge $uv\in E(G^c)$ holding the $\pm e$-constraint
\begin{equation}\label{eqa:555555}
\left\{
\begin{array}{cc}
\varphi^c (u)=\varphi^c (uv)-s_ut_u+w_uz_u,& s_ut_u\in E(G^c),~ w_uz_u\notin E(G^c)\\
\varphi^c (v)=\varphi^c (uv)-s_vt_v+w_vz_v,& s_vt_v\in E(G^c),~w_vz_v\notin E(G^c)
\end{array}
\right.
\end{equation}

\textbf{Characterize} the topological structures of the $\pm e$-graph set $P_{\pm e}(G^c)$ by the connected network $N(t)$ admitting an $e$-total colored-tree-coloring, or a $v$-total colored-tree-coloring, refer to \cite{Yao-Ma-arXiv-2201-13354v1}. It is noticeable, each element $\varphi^c (w)=T^c_{i+1}$ of the Topcode-matrix $T_{code}(N(t), \varphi^c)$ corresponds to a Topcode-matrix $T_{code}(T^c_{i+1}, g_{T^c_{i+1}})$, like some thing in \emph{Tensor}.
\end{problem}

\begin{problem}\label{qeu:444444}
We consider two $\pm e$-graph sets $P_{\pm e}(G)$ and $P_{\pm e}(G^c)$ for a uncolored connected $(p,q)$-graph $G$ and another colored connected $(p,q)$-graph $G^c$ both having cycles, such that each graph of $P_{\pm e}(G)$ is $G_{i+1}=G_i-u_iv_i+x_iy_i$ holding $u_iv_i\in E(G_i)$ and $x_iy_i\notin E(G_i)$, where $G_1=G$, and each graph of $P_{\pm e}(G^c)$ is $G^c_{i+1}=G^c_i-u_iv_i+x_iy_i$ holding $u_iv_i\in E(G^c_i)$ and $x_iy_i\notin E(G^c_i)$ with $G^c_1=G^c$.

There is a $\pm e$-spanning-tree set $S_{\pm e}(T_{span})\subset P_{\pm e}(G)$ for each spanning tree $T_{span}$ of the connected $(p,q)$-graph $G$. Suppose that the connected $(p,q)$-graph $G$ has $s(G)$ different spanning trees $T^1_{span}$, $T^2_{span}$, $\dots $, $T^{s(G)}_{span}$, so we have the $\pm e$-spanning-tree sets $S_{\pm e}(T^i_{span})\subset P_{\pm e}(G)$ with $i\in [1,s(G)]$. \textbf{Do} we have $S_{\pm e}(T^i_{span})=S_{\pm e}(T^j_{span})$ for $1\leq i,j\leq s(G)$?

There is a colored $\pm e$-spanning-tree set $S_{\pm e}(T_{cspan})\subset P_{\pm e}(G^c)$ for each colored spanning tree $T_{cspan}$ of the colored connected $(p,q)$-graph $G^c$. Suppose that the colored connected $(p,q)$-graph $G^c$ has $s(G^c)$ different colored spanning trees $T^1_{cspan}$, $T^2_{cspan}$, $\dots $, $T^{s(G^c)}_{cspan}$, so we have the colored $\pm e$-spanning-tree sets $S_{\pm e}(T^i_{cspan})\subset P_{\pm e}(G^c)$ with $i\in [1,s(G^c)]$. \textbf{Is} there $S_{\pm e}(T^i_{cspan})=S_{\pm e}(T^j_{cspan})$ for $1\leq i,j\leq s(G^c)$?
\end{problem}

\section{Topics related with parameterized colorings/labelings}

\subsection{Topcode-matrices having parameterized elements}

Since a Topcode-matrix $T_{code}$ corresponds to a set $S_{gra}(T_{code})$ of graphs, such that each graph $H\in S_{gra}(T_{code})$ has its own Topcode-matrix $T_{code}(H)=T_{code}$, thereby, we will define $W$-constraint $(k,d)$-total colorings by means of parameterized Topcode-matrices in the following subsections. Number-based strings that accompanies parameterized Topcode-matrices are useful in real application. We will introduce some graphic lattices with the closure of some $W$-constraint $(k,d)$-total colorings, and other coloring problems with parameters.

Part of operations and properties on Topcode-matrices, parameterized Topcode-matrices and pan-Topcode-matrices (including set-type Topcode-matrices, string-type Topcode-matrices, graph-type Topcode-matrices, real-valued Topcode-matrices) have been introduced and investigated in \cite{Yao-Zhao-Mu-Sun-Zhang-Zhang-Yang-IAEAC-2019}, \cite{Yao-Su-Ma-Wang-Yang-arXiv-2202-03993v1}, \cite{Bing-Yao-2020arXiv}, and \cite{Yao-Zhao-Zhang-Mu-Sun-Zhang-Yang-Ma-Su-Wang-Wang-Sun-arXiv2019}.

\subsubsection{Algebraic operations on Topcode-matrices}

In \cite{Bing-Yao-2020arXiv}, the authors introduced the following operations on Topcode-matrices:

\vskip 0.2cm

\textbf{1. Union-sum operation ``$\uplus$'', union-sum Topcode-matrix and sub-Topcode-matrix.} Let $T^i_{code}=(X_i,~E_i$, $Y_i)^{T}_{3\times q_i}$ for $i=1,2$ be Topcode-matrices, where
$$
X_i=(x^i_1, x^i_2, \dots, x^i_{q_i}),~E_i=(e^i_1, e^i_2, \dots ,e^i_{q_i}),~Y_i=(y^i_1, y^i_2,\dots,y^i_{q_i}),~i=1,2
$$ The union operation ``$\uplus$'' of two Topcode-matrices $T^1_{code}$ and $T^2_{code}$ is defined by
\begin{equation}\label{eqa:union-operation-2-matrices}
T^1_{code}\uplus T^2_{code}=(X_1\uplus X_2,~E_1\uplus E_2,~Y_1\uplus Y_2)^{T}_{3\times (q_1+q_2)}
\end{equation} with

$X_1\uplus X_2=(x^1_1, x^1_2, \dots, x^1_{q_1},x^2_1, x^2_2, \dots, x^2_{q_2})$, $E_1\uplus E_2=(e^1_1, e^1_2, \dots, e^1_{q_1},e^2_1, e^2_2, \dots, e^2_{q_2})$ and

$Y_1\uplus Y_2=(y^1_1, y^1_2, \dots, y^1_{q_1},y^2_1, y^2_2, \dots, y^2_{q_2})$\\
and the process of obtaining the \emph{union-sum Topcode-matrix} $T^1_{code}\uplus T^2_{code}$ is called a \emph{union operation}, we call $T^i_{code}$ to be a \emph{sub-Topcode-matrix} of the union-sum Topcode-matrix $T^1_{code}\uplus T^2_{code}$, denoted as $T^i_{code}\subset T^1_{code}\uplus T^2_{code}$ $i=1,2$.

Moreover, we can generalize Eq.(\ref{eqa:union-operation-2-matrices}) to the following form
\begin{equation}\label{eqa:union-operation-more-matrices}
[\uplus]^m_{i=1}T^i_{code}=T^1_{code}\uplus T^2_{code}\uplus \cdots \uplus T^m_{code}=(X^*,~E^*,~Y^*)^{T}_{3\times q^*}
\end{equation} where

$X^*=X_1\uplus X_2\uplus \cdots \uplus X_m$, $E^*=E_1\uplus E_2\uplus \cdots \uplus E_m$ and

$Y^*=Y_1\uplus Y_2\uplus \cdots \uplus Y_m$, as well as $q^*=q_1+q_2+\cdots +q_m$ with $m\geq 2$.\\
And each $T^i_{code}$ is a \emph{sub-Topcode-matrix} of the union-sum Topcode-matrix $[\uplus]^m_{i=1}T^i_{code}$.

\vskip 0.2cm

\textbf{2. Subtraction operation ``$\setminus$''.} Let $T_{code}=T^1_{code}\uplus T^2_{code}$ be the union-sum Topcode-matrix defined in Eq.(\ref{eqa:union-operation-2-matrices}). We remove $T^1_{code}$ from $T_{code}$, and obtain
\begin{equation}\label{eqa:subtraction-operation-2-matrices}
{
\begin{split}
T_{code}\setminus T^1_{code}&=(X_1\uplus X_2\setminus X_1,~E_1\uplus E_2\setminus E_1,~Y_1\uplus Y_2\setminus Y_1)^{T}_{3\times (q_1+q_2-q_1)}\\
&=(X_2,~E_2,~Y_2)^{T}_{3\times q_2}=T^2_{code}
\end{split}}
\end{equation}
the process of obtaining $T^2_{code}$ from $T_{code}$ is called a \emph{subtraction operation}.

\vskip 0.2cm

\textbf{3. Intersection operation ``$\cap$''.} Let $T\,'_{code}=T^1_{code}\uplus H_{code}$ and $T\,''_{code}=T^2_{code}\uplus H_{code}$ be two Topcode-matrices obtained by doing the union operation to Topcode-matrices $T^1_{code}, H_{code}$ and $T^2_{code}$. Since the Topcode-matrix $H_{code}$ is a \emph{common sub-Topcode-matrix} of two Topcode-matrices $T\,'_{code}$ and $T\,''_{code}$, we take the largest common sub-Topcode-matrix of $T\,'_{code}$ and $T\,''_{code}$ and denote it as $T\,'_{code}\cap T\,''_{code}$ called \emph{intersected Topcode-matrix}, this is so-called the \emph{intersection operation}, so we have
\begin{equation}\label{eqa:555555}
H_{code}\subseteq T\,'_{code}\cap T\,''_{code},~T\,'_{code}\cap T\,''_{code}\subseteq T\,'_{code},~T\,'_{code}\cap T\,''_{code}\subseteq T\,''_{code}
\end{equation}

\vskip 0.2cm

\textbf{4. Union operation ``$\cup$''.} For the intersected Topcode-matrix $T\,'_{code}\cap T\,''_{code}$ based on two Topcode-matrices $T\,'_{code}$ and $T\,''_{code}$, we define the union operation ``$\cup$'' on $T\,'_{code}$ and $T\,''_{code}$ as
\begin{equation}\label{eqa:555555}
T\,'_{code}\cup T\,''_{code}=T\,'_{code}\uplus [T\,''_{code}\setminus (T\,'_{code}\cap T\,''_{code})]=[T\,'_{code}\setminus (T\,'_{code}\cap T\,''_{code})]\uplus T\,''_{code}
\end{equation}

\vskip 0.2cm

\textbf{5. Coinciding operation ``$\odot$''.} Let $T\,'_{code}=T^1_{code}\uplus H_{code}$ and $T\,''_{code}=T^2_{code}\uplus H_{code}$ be two Topcode-matrices obtained by doing the union operation to Topcode-matrices $T^1_{code}, H_{code}$ and $T^2_{code}$. So $H_{code}\subset T\,'_{code}$ and $H_{code}\subset T\,''_{code}$. We define the \emph{coinciding operation} ``$\odot$'' on two Topcode-matrices $T\,'_{code}$ and $T\,''_{code}$ as follows
\begin{equation}\label{eqa:intersection-operation-2-matrices}
[\odot H_{code}]\langle T\,'_{code},T\,''_{code}\rangle =T^1_{code}\uplus H_{code}\uplus T^2_{code}=[T\,'_{code}\setminus H_{code}]\uplus T\,''_{code}=[T\,''_{code}\setminus H_{code}]\uplus T\,'_{code}
\end{equation}

\vskip 0.2cm

\textbf{6. Splitting operation ``$\wedge $''.} We do the splitting operation to the sub-Topcode-matrix $H_{code}$ of a coincided Topcode-matrix $T^1_{code}\uplus H_{code}\uplus T^2_{code}$, such that the resultant Topcode-matrix consisted of two disjoint Topcode-matrices $T^1_{code}\uplus H_{code}$ and $T^2_{code}\uplus H_{code}$, denoted as
\begin{equation}\label{eqa:splitting-operation-2-matrices}
\{[\odot H_{code}]\langle T\,'_{code},T\,''_{code}\rangle \}\wedge H_{code}=[T^1_{code}\uplus H_{code}\uplus T^2_{code}]\wedge H_{code}
\end{equation}

\begin{example}\label{exa:8888888888}
For the following three Topcode-matrices
\begin{equation}\label{eqa:three-topcode-matrces}
\centering
{
\begin{split}
T_{code}= \left(
\begin{array}{ccccc}
x_{1} & x_{2} & \cdots & x_{n}\\
e_{1} & e_{2} & \cdots & e_{n}\\
y_{1} & y_{2} & \cdots & y_{n}
\end{array}
\right),~T^i_{code}= \left(
\begin{array}{ccccc}
x^i_{1} & x^i_{2} & \cdots & x^i_{m}\\
e^i_{1} & e^i_{2} & \cdots & e^i_{m}\\
y^i_{1} & y^i_{2} & \cdots & y^i_{m}
\end{array}
\right),~i=1,2
\end{split}}
\end{equation}
\noindent we have two union-sum Topcode-matrices

\begin{equation}\label{eqa:two-sum-topcode-matrces}
\centering
{
\begin{split}
T_{code}\uplus T^i_{code}= \left(
\begin{array}{cccccccc}
x_{1} & x_{2} & \cdots & x_{n} & x^i_{1} & x^i_{2} & \cdots & x^i_{m}\\
e_{1} & e_{2} & \cdots & e_{n} & e^i_{1} & e^i_{2} & \cdots & e^i_{m}\\
y_{1} & y_{2} & \cdots & y_{n} & y^i_{1} & y^i_{2} & \cdots & y^i_{m}
\end{array}
\right),~i=1,2
\end{split}}
\end{equation}
 A coincided Topcode-matrix is
\begin{equation}\label{eqa:matrces-coinciding-operation}
J_{code}=[\odot T^2_{code}]\langle T_{code}\uplus T^2_{code},T^1_{code}\uplus T^2_{code}\rangle
\end{equation}
and
\begin{equation}\label{eqa:topcode-matrces-coinciding-operation}
\centering
{
\begin{split}
J_{code}= \left(
\begin{array}{cccccccccccc}
x_{1} & x_{2} & \cdots & x_{n} & x^2_{1} & x^2_{2} & \cdots & x^2_{s} & x^1_{1} & x^1_{2} & \cdots & x^1_{m}\\
e_{1} & e_{2} & \cdots & e_{n} & e^2_{1} & e^2_{2} & \cdots & e^2_{s}& e^1_{1} & e^1_{2} & \cdots & e^1_{m}\\
y_{1} & y_{2} & \cdots & y_{n} & y^2_{1} & y^2_{2} & \cdots & y^2_{s}& y^1_{1} & y^1_{2} & \cdots & y^1_{m}
\end{array}
\right)=T_{code}\uplus T^2_{code}\uplus T^1_{code}
\end{split}}
\end{equation}
\end{example}

\begin{example}\label{exa:8888888888}
Given two Topcode-matrices
\begin{equation}\label{eqa:subtractive-union-intersection-1}
\centering
A= \left(
\begin{array}{ccccccccc}
\textbf{\textcolor[rgb]{0.00,0.00,1.00}{7}} & 5 & 7 &\textbf{\textcolor[rgb]{0.00,0.00,1.00}{1}}\\
\textbf{\textcolor[rgb]{0.00,0.00,1.00}{1}} & 3 & 5 &\textbf{\textcolor[rgb]{0.00,0.00,1.00}{7}} \\
\textbf{\textcolor[rgb]{0.00,0.00,1.00}{18}}& 18 & 14 &\textbf{\textcolor[rgb]{0.00,0.00,1.00}{18}}
\end{array}
\right),\quad
B= \left(
\begin{array}{ccccccccc}
\textbf{\textcolor[rgb]{0.00,0.00,1.00}{7}} &\textbf{\textcolor[rgb]{0.00,0.00,1.00}{1}} & 5 \\
\textbf{\textcolor[rgb]{0.00,0.00,1.00}{1}} &\textbf{\textcolor[rgb]{0.00,0.00,1.00}{7}}& 9\\
\textbf{\textcolor[rgb]{0.00,0.00,1.00}{18}} &\textbf{\textcolor[rgb]{0.00,0.00,1.00}{18}}& 12
\end{array}
\right)
\end{equation}
\noindent we have
\begin{equation}\label{eqa:subtractive-union-intersection-2}
\centering
A\uplus B= \left(
\begin{array}{ccccccccc}
\textbf{\textcolor[rgb]{0.00,0.00,1.00}{7}} & 5 & 7 &\textbf{\textcolor[rgb]{0.00,0.00,1.00}{1}}& \textbf{\textcolor[rgb]{0.00,0.00,1.00}{7}} &\textbf{\textcolor[rgb]{0.00,0.00,1.00}{1}} & 5 \\
\textbf{\textcolor[rgb]{0.00,0.00,1.00}{1}} & 3 & 5 &\textbf{\textcolor[rgb]{0.00,0.00,1.00}{7}}& \textbf{\textcolor[rgb]{0.00,0.00,1.00}{1}} &\textbf{\textcolor[rgb]{0.00,0.00,1.00}{7}}& 9\\
\textbf{\textcolor[rgb]{0.00,0.00,1.00}{18}}& 18 & 14 &\textbf{\textcolor[rgb]{0.00,0.00,1.00}{18}}& \textbf{\textcolor[rgb]{0.00,0.00,1.00}{18}} &\textbf{\textcolor[rgb]{0.00,0.00,1.00}{18}}& 12
\end{array}
\right),~A\cup B= \left(
\begin{array}{ccccccccc}
\textbf{\textcolor[rgb]{0.00,0.00,1.00}{7}} & 5 & 7 &\textbf{\textcolor[rgb]{0.00,0.00,1.00}{1}}& 5 \\
\textbf{\textcolor[rgb]{0.00,0.00,1.00}{1}} & 3 & 5 &\textbf{\textcolor[rgb]{0.00,0.00,1.00}{7}}& 9 \\
\textbf{\textcolor[rgb]{0.00,0.00,1.00}{18}}& 18 & 14 &\textbf{\textcolor[rgb]{0.00,0.00,1.00}{18}}& 12
\end{array}
\right)
\end{equation}

\begin{equation}\label{eqa:subtractive-union-intersection-3}
\centering
~B\setminus A= \left(
\begin{array}{ccccccccc}
 5 \\
9 \\
12
\end{array}
\right),~A\cap B= \left(
\begin{array}{ccccccccc}
\textbf{\textcolor[rgb]{0.00,0.00,1.00}{7}} &\textbf{\textcolor[rgb]{0.00,0.00,1.00}{1}} \\
\textbf{\textcolor[rgb]{0.00,0.00,1.00}{1}} &\textbf{\textcolor[rgb]{0.00,0.00,1.00}{7}}\\
\textbf{\textcolor[rgb]{0.00,0.00,1.00}{18}} &\textbf{\textcolor[rgb]{0.00,0.00,1.00}{18}}
\end{array}
\right),~A\setminus B= \left(
\begin{array}{ccccccccc}
 5 & 7 \\
 3 & 5 \\
18 & 14
\end{array}
\right)
\end{equation}

The split Topcode-matrix $(A\cup B)\wedge (A\cap B)$ consists of $A$ and $B$, and there are the following Topcode-matrix algebraic relationships:
\begin{equation}\label{eqa:555555}
{
\begin{split}
&A=(A\uplus B)\setminus B,~B=(A\uplus B)\setminus A,~(A\setminus B)\uplus (B\setminus A)=(A\cup B)\setminus (A\cap B)\\
&A\cup B=(A\setminus B)\uplus (B\setminus A)\uplus (A\cap B), ~A\cup B=[\odot (A\cap B)]\langle A, B\rangle\\
&[\odot (A\cap B)]\langle A, B\rangle=(A\setminus B)\uplus B=A\uplus (B\setminus A)
\end{split}}
\end{equation}
\end{example}

About the colored Topcode-matrix defined in Definition \ref{defn:topcode-matrix-definition} and the colored Topcode-matrix defined in Definition \ref{defn:colored-topcode-matrix}, we have some operations of Topcode-matrix algebra as follows:
\begin{asparaenum}[\textrm{\textbf{ATcode}}-1. ]
\item The values of $x_i$, $e_i$ and $y_i$ with $i\in [1,q]$ can be any things in the world. We say $e_i$ to be \emph{valued} by $x_i$ and $y_i$ if there exists a \emph{valued function} $f$ defined as $e_i=f(x_i,y_i)$ for $i\in [1,q]$, and furthermore $E=f(X,Y)$.
\item The graph $G$ has another \emph{Topcode-matrix} $T\,^{XY}(G)=(X, Y)^{T}_{2\times q}$.
\item In the view of graph colorings, we regard that $T\,^{XY}(G)$ corresponds to a \emph{vertex coloring} of $G$, $E=(e_1, e_2, \dots ,e_q)$ corresponds to an \emph{edge coloring} of $G$, and $T_{code}(G)$ corresponds to a \emph{total coloring} of $G$.
\item A Topcode-matrix $T_{code}(G)$ may be the common Topcode-matrix of many graphs $G_1,G_2,\dots,G_n$, that is $T_{code}(G)=T_{code}(G_i)$ with $i\in [1,n]$.

\item A \emph{column exchange} $C_{(i,j)}$ of vertex-vectors and edge-vector of a Topcode-matrix $T_{code}$ is defined as:
\begin{equation}\label{eqa:c3xxxxx}
{
\begin{split}
C_{(i,j)}(X)&=C_{(i,j)}(x_1, x_2, \dots ,x_{i-1},\textcolor[rgb]{0.98,0.00,0.00}{\mathbf{x_i}},x_{i+1}, \dots ,x_{j-1},\textcolor[rgb]{0.98,0.00,0.00}{\mathbf{x_j}},x_{j+1}, \dots ,x_q)\\
&=(x_1, x_2, \dots ,x_{i-1},\textcolor[rgb]{0.98,0.00,0.00}{\mathbf{x_j}},x_{i+1}, \dots ,x_{j-1},\textcolor[rgb]{0.98,0.00,0.00}{\mathbf{x_i}},x_{j+1}, \dots ,x_q)\\
C_{(i,j)}(E)&=C_{(i,j)}(e_1, e_2, \dots ,e_{i-1},\textcolor[rgb]{0.98,0.00,0.00}{\mathbf{e_i}},e_{i+1}, \dots ,e_{j-1},\textcolor[rgb]{0.98,0.00,0.00}{\mathbf{e_j}},e_{j+1}, \dots ,e_q)\\
&=(e_1, e_2, \dots ,e_{i-1},\textcolor[rgb]{0.98,0.00,0.00}{\mathbf{e_j}},e_{i+1}, \dots ,e_{j-1},\textcolor[rgb]{0.98,0.00,0.00}{\mathbf{e_i}},e_{j+1}, \dots ,e_q)\\
C_{(i,j)}(X)&=C_{(i,j)}(y_1, y_2, \dots ,y_{i-1},\textcolor[rgb]{0.98,0.00,0.00}{\mathbf{y_i}},y_{i+1}, \dots ,y_{j-1},\textcolor[rgb]{0.98,0.00,0.00}{\mathbf{y_j}},y_{j+1}, \dots ,y_q)\\
&=(y_1, y_2, \dots ,y_{i-1},\textcolor[rgb]{0.98,0.00,0.00}{\mathbf{y_j}},y_{i+1}, \dots ,y_{j-1},\textcolor[rgb]{0.98,0.00,0.00}{\mathbf{y_i}},y_{j+1}, \dots ,y_q),
\end{split}}
\end{equation} so we get a column exchange of a Topcode-matrix $T_{code}$ as:
\begin{equation}\label{eqa:555555}
C_{(i,j)}(T_{code})=C_{(i,j)}(X,E,Y)^T=(C_{(i,j)}(X),C_{(i,j)}(E),C_{(i,j)}(Y))^T
\end{equation}

For two Topcode-matrices $T\,^r_{code}$ and $T\,^s_{code}$, if there are column exchanges $C_{(i_k,j_k)}$ with $k\in [1,m]$, such that
\begin{equation}\label{eqa:555555}
C_{(i_m,j_m)}C_{(i_{m-1},j_{m-1})}\cdots C_{(i_1,j_1)}(T\,^r_{code})=T\,^s_{code}
\end{equation} we say that both Topcode-matrices $T\,^r_{code}$ and $T\,^s_{code}$ are similar from each other, write this fact as
\begin{equation}\label{eqa:555555}
T\,^r_{code}(\sim )T\,^s_{code}
\end{equation}

\item A \emph{line exchange} $L_{(i)}(T_{code})$ on the vertex-vectors of a Topcode-matrix $T_{code}$ is defined as:
\begin{equation}\label{eqa:c3xxxxx}
\left(
\begin{array}{ccccc}
x_{1}\\
e_{1}\\
y_{1}
\end{array}
\right)\bigcup \cdots \bigcup \left(
\begin{array}{ccccc}
\textcolor[rgb]{0.98,0.00,0.00}{\mathbf{x_{i}}}\\
e_{i}\\
\textcolor[rgb]{0.98,0.00,0.00}{\mathbf{y_{i}}}
\end{array}
\right)\bigcup \cdots \bigcup \left(
\begin{array}{ccccc}
x_{q}\\
e_{q}\\
y_{q}
\end{array}
\right)\rightarrow
\left(
\begin{array}{ccccc}
x_{1}\\
e_{1}\\
y_{1}
\end{array}
\right)\bigcup \cdots \bigcup \left(
\begin{array}{ccccc}
\textcolor[rgb]{0.98,0.00,0.00}{\mathbf{y_{i}}}\\
e_{i}\\
\textcolor[rgb]{0.98,0.00,0.00}{\mathbf{x_{i}}}
\end{array}
\right)\bigcup \cdots \bigcup \left(
\begin{array}{ccccc}
x_{q}\\
e_{q}\\
y_{q}
\end{array}
\right)
\end{equation} denoted as $L_{(i)}(T_{code})=(L_{(i)}(X),E,L_{(i)}(Y))^T$. For two Topcode-matrices $T\,^r_{code}$ and $T\,^s_{code}$, if there are line exchanges $L_{(i_k)}$ with $k\in [1,m]$, such that
\begin{equation}\label{eqa:555555}
L_{(i_m)}L_{(i_{m-1})}\cdots L_{(i_1)}(T\,^r_{code})=T\,^s_{code}
\end{equation} then both Topcode-matrices $T\,^r_{code}$ and $T\,^s_{code}$ are similar from each other, that is $T\,^r_{code}(\sim ) T\,^s_{code}$.
\begin{example}\label{exa:8888888888}
There two Topcode-matrices as follows:
\begin{equation}\label{eqa:column-line-exchange-top-matss}
\centering
{
\begin{split}
T\,^a_{code}= \left(
\begin{array}{ccccccc}
10 & 7 & 0 & 0 & 2 & 2 &0\\
1 & 3 & 5 & 7 & 9 & 11 &13\\
11 & 10 & 5 & 7 & 11 & 13 &13
\end{array}
\right),~T\,^b_{code}= \left(
\begin{array}{ccccccc}
0 & 10 & 10 & 0 & 7 & 2 & 2 \\
13 & 1 & 3 & 5 & 7 & 9 & 11 \\
13 & 11 & 7 & 5 & 0 & 11 & 13
\end{array}
\right)
\end{split}}
\end{equation}
We do two line exchanges $L_{(4)}$ and $L_{(2)}$ to $T\,^a_{code}$, and do a column exchange $C_{(i,j)}$ to $T\,^a_{code}$, so we get a \emph{mixed operation}
$C_{(i,j)}L_{(4)}L_{(2)}(T\,^a_{code})=T\,^b_{code}$, that is $T\,^a_{code}(\sim ) T\,^b_{code}$.\qqed
\end{example}

\item \textbf{The coinciding-operation and the joining-operation on Topcode-matrices.} Let $T_{code}(G)=(X,E,Y)^T$ be the Topcode-matrix of a $(p,q)$-graph $G$, and let $T_{code}(H)=(X\,'$, $E\,'$, $Y\,')^T$ be the Topcode-matrix of a $(p,q)$-graph $H$, and let $G$ and $H$ be disjoint from each other, that is $X\cap X\,'=\emptyset$, $E\cap E\,'=\emptyset$ and $Y\cap Y\,'=\emptyset$. If we vertex-coincide the vertex $x_{i}$ of $G$ with the vertex $x\,'_{j}$ of $H$ into one vertex $w=x_{i}\odot x\,'_{j}$, then the vertex-coincided graph $G\odot H$ has its own Topcode-matrix as follows
\begin{equation}\label{eqa:555555}
T_{code}(G\odot_w H)=T^*_{code}(G)\bigcup T^*_{code}(H)
\end{equation} called \emph{Topcode-matrix coinciding-operation}, where the Topcode-matrix $T^*_{code}(G)$ is the result of replacing each $x_{i}$ by $w$ in $T_{code}(G)$, and the Topcode-matrix $T^*_{code}(H)$ is obtain by replacing each $x\,'_{j}$ by $w$ in $T_{code}(H)$. Next, we add a new edge $e=x_{i}x\,'_{j}$ to join the vertex $x_{i}$ of $G$ with the vertex $x\,'_{j}$ of $H$ together, such that the edge-joined graph $G\ominus H$ has the Topcode-matrix defined as
\begin{equation}\label{eqa:555555}
T_{code}(G\ominus H)=T_{code}(G)\bigcup (x_{i},~e,~x\,'_{j})^T \bigcup T_{code}(H)
\end{equation} called \emph{Topcode-matrix joining-operation}.

\item \textbf{Operations on meta-matrices.}

\qquad (A) Let $T\,^i_{3\times 1}=(x_i,e_i,y_i)^T$ be a \emph{meta-matrix} of a Topcode-matrix $T_{code}$ defined in Definition \ref{defn:topcode-matrix-definition}. We define the \emph{union operation} of meta-matrices as
\begin{equation}\label{eqa:c3xxxxx}
T_{code}=\bigcup^q_{i=1} T\,^i_{3\times 1}=\left(
\begin{array}{ccccc}
x_{1}\\
e_{1}\\
y_{1}
\end{array}
\right)\bigcup \left(
\begin{array}{ccccc}
x_{2}\\
e_{2}\\
y_{2}
\end{array}
\right)\bigcup \cdots \bigcup \left(
\begin{array}{ccccc}
x_{q}\\
e_{q}\\
y_{q}
\end{array}
\right)=\bigcup^q_{i=1} (x_i,e_i,y_i)^T
\end{equation}

\qquad (B) Suppose that $x_i,e_i$ and $y_i$ of each meta-matrix $T\,^i_{3\times 1}$ with $i\in [1,q]$ are integers. The \emph{\textbf{meta-matrix number multiplication}} is defined as
\begin{equation}\label{eqa:555555}
\beta \cdot T\,^i_{3\times 1}=(\beta\cdot x_i,\beta\cdot e_i,\beta \cdot y_i)^T
\end{equation} with an integer $\beta\geq 1$, and the \emph{\textbf{meta-matrix add-operation}} is defined as
\begin{equation}\label{eqa:555555}
\alpha \cdot T\,^i_{3\times 1}+\beta \cdot T\,^j_{3\times 1}=(\alpha \cdot x_i+\beta\cdot x_j,\alpha \cdot e_i+\beta\cdot e_j,\alpha \cdot y_i+\beta\cdot y_j)^T
\end{equation} with integer $\alpha, \beta\geq 1$.

\qquad (C) \textbf{Meta-matrix lattices.} Let $\textbf{\textrm{T}}=(T\,^1_{3\times 1},T\,^2_{3\times 1},\dots ,T\,^q_{3\times 1})$ be a \emph{meta-matrix base} consisted of meta-matrices $T\,^i_{3\times 1}$ for $i\in [1,q]$. Notice that $T\,^i_{3\times 1}\cap T\,^j_{3\times 1}=\emptyset$ for $i\neq j$, so $T\,^i_{3\times 1}\cup T\,^j_{3\times 1}=T\,^i_{3\times 1}\uplus T\,^j_{3\times 1}$. We call the following set
\begin{equation}\label{eqa:c3xxxxx}
\textrm{\textbf{L}}(\cup \textbf{\textrm{T}})=\left \{\bigcup^q_{i=1} \beta_k\cdot T\,^k_{3\times 1}:\beta_k\in Z^0,T\,^k_{3\times 1}\in \textbf{\textrm{T}}\right \}=\left \{[\uplus]^q_{i=1} \beta_k\cdot T\,^k_{3\times 1}:\beta_k\in Z^0,T\,^k_{3\times 1}\in \textbf{\textrm{T}}\right \}
\end{equation} a \emph{meta-matrix lattice} based on the meta-matrix base $\textbf{\textrm{T}}=(T\,^1_{3\times 1},T\,^2_{3\times 1},\dots ,T\,^q_{3\times 1})$. In general, some matrix $\bigcup^q_{i=1} \beta_kT\,^k_{3\times 1}\in \textrm{\textbf{L}}(\cup \textbf{\textrm{T}})$ is not a Topcode-matrix of any graph.
\end{asparaenum}

Since there are $q!$ permutations for each one of $x_1x_2\cdots x_q$, $e_1e_2\cdots e_q$ and $y_1y_2\cdots y_q$, by the column exchange operation on Topcode-matrices, we have
\begin{thm} \label{thm:similar-topcode-matrices}
Any Topcode-matrix $T_{code}$ of order $3\times q$ is similar with other $(q!-1)$ Topcode-matrices of order $3\times q$.
\end{thm}

\begin{thm} \label{thm:graph-homomorphism-topcode-matrices}
If there are different graphs $G,G_1,G_2,\dots,G_n$ holding $T_{code}(G)=T_{code}(G_i)$ and $|V(G)|\leq |V(G_i)|$ for $i\in [1,n]$, then each graph $G_i$ admits a graph homomorphism to $G$, i.e. $G_i\rightarrow G$ for $i\in [1,n]$.
\end{thm}

\subsubsection{Parameterized Topcode-matrices}

Matrices having parameterized elements are called $(k,d)$-type matrices, for example, $(k,d)$-graceful matrices, $(k,d)$-edge-magic matrices, $(k,d)$-harmonious matrices, in the following discussion.

\begin{defn}\label{defn:parameterized-topcode-matrix}
$^*$ A \emph{parameterized Topcode-matrix} $P_{ara}=(X_{k,d},E_{k,d},Y_{k,d})^T$ shown in Eq.(\ref{eqa:basic-formula-1}) is defined as

$X_{k,d}=(\alpha_1k+a_1d,\alpha_2k+a_2d,\dots ,\alpha_qk+a_qd)$,

$Y_{k,d}=(\beta_1k+b_1d,\beta_2k+b_2d,\dots ,\beta_qk+b_qd)$, and

$E_{k,d}=(\gamma_1k+c_1d,\gamma_2k+c_2d,\dots ,\gamma_qk+c_qd)$\\
for two integers $k,d\geq 1$, \qqed
\end{defn}

Moreover we can partition a parameterized Topcode-matrix into the linear combination of Topcode-matrices as follows:
\begin{equation}\label{eqa:basic-formula-1}
{
\begin{split}
P_{ara}&=\left(
\begin{array}{ccccc}
\alpha_1k+a_1d & \alpha_2k+a_2d & \cdots & \alpha_qk+a_qd\\
\gamma_1k+c_1d & \gamma_2k+c_2d & \cdots & \gamma_qk+c_qd\\
\beta_1k+b_1d & \beta_2k+b_2d & \cdots & \beta_qk+b_qd
\end{array}
\right)_{3\times q}\\
&=\left(
\begin{array}{ccccc}
\alpha_1k & \alpha_2k & \cdots & \alpha_qk\\
\gamma_1k & \gamma_2k & \cdots & \gamma_qk\\
\beta_1k & \beta_2k & \cdots & \beta_qk
\end{array}
\right)_{3\times q}+\left(
\begin{array}{ccccc}
a_1d & a_2d & \cdots & a_qd\\
c_1d & c_2d & \cdots & c_qd\\
b_1d & b_2d & \cdots & b_qd
\end{array}
\right)_{3\times q}\\
&=k\cdot I_{3\times q}+d\cdot C_{3\times q}
\end{split}}
\end{equation}
where $I_{3\times q}$ is called \emph{constant Topcode-matrix} defined as
\begin{equation}\label{eqa:basic-formula-2}
{
\begin{split}
I_{3\times q}=\left(
\begin{array}{ccccc}
\alpha_1 & \alpha_2 & \cdots & \alpha_q\\
\gamma_1 & \gamma_2 & \cdots & \gamma_q\\
\beta_1 & \beta_2 & \cdots & \beta_q
\end{array}
\right)_{3\times q}=(I_X,I_E,I_Y)^T
\end{split}}
\end{equation} with three vectors $I_X=(\alpha_1,\alpha_2,\dots ,\alpha_q)$, $I_E=(\gamma_1,\gamma_2,\dots ,\gamma_q)$ and $I_Y=(\beta_1,\beta_2,\dots ,\beta_q)$; and $C_{3\times q}$ is, called \emph{main Topcode-matrix} defined as
\begin{equation}\label{eqa:basic-formula-3}
{
\begin{split}
C_{3\times q}=\left(
\begin{array}{ccccc}
a_1 & a_2 & \cdots & a_q\\
c_1 & c_2 & \cdots & c_q\\
b_1 & b_2 & \cdots & b_q
\end{array}
\right)_{3\times q}=(C_X,C_E,C_Y)^T
\end{split}}
\end{equation} for three vectors $C_X=(a_1, a_2, \dots ,a_q)$, $C_E=(c_1, c_2, \dots ,c_q)$, and $C_Y=(b_1, b_2, \dots ,b_q)$. Also, the parameterized Topcode-matrix can be written as
\begin{equation}\label{eqa:basic-formula-4}
P_{ara}=k\cdot (I_X,I_E,I_Y)^T+d\cdot (C_X,C_E,C_Y)^T=k\cdot I_{3\times q}+d\cdot C_{3\times q}.
\end{equation}

In the valued case, there is a \emph{valued function} $\theta$ to combine $C_E,C_X$ and $C_Y$ together by $C_E=\theta(C_X,C_Y)$ with $c_i=W\langle a_i,b_i\rangle $ for $i\in [1,q]$ under a $W$-constraint.

For \emph{bipartite graphs}, especially, we define the \emph{unite Topcode-matrix} as follows
\begin{equation}\label{eqa:unit-Topcode-matrix}
{
\begin{split}
I\,^0=\left(
\begin{array}{ccccc}
0 & 0 & \cdots & 0\\
1 & 1 & \cdots & 1\\
1 & 1 & \cdots & 1
\end{array}
\right)_{3\times q}=(X\,^0,~E\,^0,~Y\,^0)^T
\end{split}}
\end{equation} with two vertex-vectors $X\,^0=(0, 0, \dots ,0)_{1\times q}$ and $Y\,^0=(1, 1, \dots ,1)_{1\times q}$, and the edge-vector $E\,^0=(1, 1, \dots ,1)_{1\times q}$.

\begin{example}\label{exa:8888888888}
In Eq.(\ref{eqa:basic-formula-1}), if we have an edge-difference constraint
$$(\gamma_ik+c_id)+|(\beta_ik+b_id)-(\alpha_ik+a_id)|=(\gamma_i+\beta_i-\alpha_i)k+(c_i+b_i-a_i)d=k\cdot A_{ed}+d\cdot B_{ed}
$$ for $i\in [1,q]$, we write the above fact as $I_E+|I_Y-I_X|:=A_{ed}$ and $C_E+|C_Y-C_X|:=B_{ed}$, where $A_{ed}$ and $B_{ed}$ are constants.

If there is an edge-magic constraint
$$
(\gamma_ik+c_id)+(\beta_ik+b_id)+(\alpha_ik+a_id=(\gamma_i+\beta_i+\alpha_i)k+(c_i+b_i+a_i)d)=k\cdot A_{em}+d\cdot B_{em}
$$ for $i\in [1,q]$, we write the above fact as $I_X+I_E+I_Y:=A_{em}$ and $C_X+C_E+C_Y:=B_{em}$ for constants $A_{em}$ and $B_{em}$.

If we have a graceful-difference constraint
$$
\big ||(\beta_ik+b_id)-(\alpha_ik+a_id)|-(\gamma_ik+c_id)\big |=(\beta_i-\alpha_i-\gamma_i)k+(b_i-a_i-c_i)d=k\cdot A_{gd}+d\cdot B_{gd}
$$ for $i\in [1,q]$, we write the above fact as $\big ||I_Y-I_X|-I_E\big |:=A_{gd}$ and $\big ||C_Y-C_X|-C_E\big |:=B_{gd}$ for constants $A_{gd}$ and $B_{gd}$.

If we have felicitous-difference constraint
$$
|(\beta_ik+b_id)+(\alpha_ik+a_id)-(\gamma_ik+c_id)|=(\beta_i+\alpha_i-\gamma_i)k+(b_i+a_i-c_i)d=k\cdot A_{fd}+d\cdot B_{fd}
$$ for $i\in [1,q]$, we write the above fact as $|I_Y+I_X-I_E|:=A_{fd}$ and $|C_Y+C_X-C_E|:=B_{fd}$ for constants $A_{fd}$ and $B_{fd}$.\qqed
\end{example}

If there is no confusion, we omit ``'order $3\times q$' in the following discussion, or add a sentence `` the Topcode-matrices $I$, $T_{code}(G)$ and $P_{ara}(G)$ have the same order''.

\begin{defn} \label{defn:definition-parameterized-topcode-matrix}
$^*$ Since a Topcode-matrix $T_{code}(G)$ defined in Eq.(\ref{eqa:basic-formula-Topcode-matrix}) for a $(p,q)$-graph $G$ can be as a main Topcode-matrix in Eq.(\ref{eqa:basic-formula-3}), so we get two parameterized Topcode-matrices
\begin{equation}\label{eqa:definition-parameterized-topcode-matrix}
P_{ara}(G)=k\cdot I+d\cdot T_{code}(G), \quad P_{ara}(G,F)=k\cdot I\,^0+d\cdot T_{code}(G,f)
\end{equation}
where the Topcode-matrices $I$, $T_{code}(G)$ and $P_{ara}(G)$, $I\,^0$, $T_{code}(G,f)$ and $P_{ara}(G,F)$ have the same order, and $f$ is a $W$-constraint coloring admitted by $G$, as well as $F$ is a $W$-constraint parameterized coloring of $G$.\qqed
\end{defn}

\begin{example}\label{exa:8888888888}
In Fig.\ref{fig:k-d-matrices-graphs}, there are four colored bipartite graphs $G_1,G_2,G_3,G_4$, their $(k,d)$-type Topcode-matrices $P_{ara}(G_i,F_i)$ with $i\in [1,4]$ are shown in Eq.(\ref{eqa:44-k-d-matrices}) and Eq.(\ref{eqa:55-k-d-matrices}) respectively. By the unite Topcode-matrix $I\,^0$ defined in Eq.(\ref{eqa:unit-Topcode-matrix}) and Definition \ref{defn:definition-parameterized-topcode-matrix}, we rewrite these $(k,d)$-type Topcode-matrices by
\begin{equation}\label{eqa:333333}
P_{ara}(G_j,F_j)=k\cdot I\,^0+d\cdot T_{code}(G\,'_j,g_j),~j\in [1,4]
\end{equation} for integers $k\geq 0$ and $d\geq 1$, where the Topcode-matrices $I\,^0$, $T_{code}(G\,'_j,g_j)$ and $P_{ara}(G_j,F_j)$ have the same order $3\times 9$, four Topcode-matrices $T_{code}(G\,'_j,g_j)$ for $j\in [1,4]$ are shown in Eq.(\ref{eqa:Topcode-matrices-express11}) and Eq.(\ref{eqa:Topcode-matrices-express22}), respectively, and moreover

(1) $P_{ara}(G_1,F_1)$ is a \emph{$(k,d)$-graceful Topcode-matrix}, where $G_1$ admits a $(k,d)$-graceful labeling $F_1$ and $G\,'_1$ admits a graceful labeling $g_1$ corresponding to a Topcode-matrix $T_{code}(G\,'_1,g_1)$;

(2) $P_{ara}(G_2,F_2)$ is a \emph{$(k,d)$-felicitous Topcode-matrix}, where $G_2$ admits a $(k,d)$-felicitous labeling $F_2$ and $G\,'_2$ admits a felicitous labeling $g_2$ corresponding to a Topcode-matrix $T_{code}(G\,'_2,g_2)$;

(3) $P_{ara}(G_3,F_3)$ is a \emph{$(k,d)$-edge-magic Topcode-matrix}, where $G_3$ admits a $(k,d)$-edge-magic labeling $F_3$ and $G\,'_3$ admits an edge-magic labeling $g_3$ corresponding to a Topcode-matrix $T_{code}(G\,'_3,g_3)$;

(4) $P_{ara}(G_4,F_4)$ is a \emph{$(k,d)$-edge-antimagic Topcode-matrix}, where $G_4$ admits a $(k,d)$-edge-antimagic labeling $F_4$ and $G\,'_4$ admits an edge antimagic labeling $g_4$ corresponding to a Topcode-matrix $T_{code}(G\,'_4$, $g_4)$.\qqed
\end{example}

{\small
\begin{equation}\label{eqa:44-k-d-matrices}
\centering
{
\begin{split}
P_{ara}(G_1,F_1)&= \left(
\begin{array}{ccccccccccccccc}
2d & 2d & d & 2d & d & 0 & d & 0 & 0 \\
k & k+d & k+2d & k+3d & k+4d & k+5d & k+6d & k+7d & k+8d\\
k+2d & k+3d & k+3d & k+5d & k+5d & k+5d & k+7d & k+7d & k+8d
\end{array}
\right)\\
P_{ara}(G_2,F_2)&= \left(
\begin{array}{ccccccccccccccc}
d & 2d & 2d & 0 & 0 & d & 0 & d & 2d \\
k & k+d & k+2d & k+3d & k+4d & k+5d & k+6d & k+7d & k+8d\\
k+8d & k+8d & k+9d & k+3d & k+4d & k+4d & k+6d & k+6d & k+6d
\end{array}
\right)
\end{split}}
\end{equation}
}

{\small
\begin{equation}\label{eqa:55-k-d-matrices}
\centering
{
\begin{split}
P_{ara}(G_3,F_3)&= \left(
\begin{array}{ccccccccccccccc}
3d & 3d & 2d & 3d & 2d & 2d & 2d & d & 0 \\
k & k+d & k+2d & k+3d & k+4d & k+5d & k+6d & k+7d & k+8d\\
k+8d & k+7d & k+7d & k+5d & k+5d & k+4d & k+3d & k+3d & k+3d
\end{array}
\right)\\
P_{ara}(G_4,F_4)&= \left(
\begin{array}{ccccccccccccccc}
0 & 0 & d & 0 & d & 2d & d & 2d & 2d \\
k & k+d & k+2d & k+3d & k+4d & k+5d & k+6d & k+7d & k+8d\\
k+2d & k+3d & k+3d & k+5d & k+5d & k+5d & k+7d & k+7d & k+8d
\end{array}
\right)
\end{split}}
\end{equation}
}

\begin{figure}[h]
\centering
\includegraphics[width=15cm]{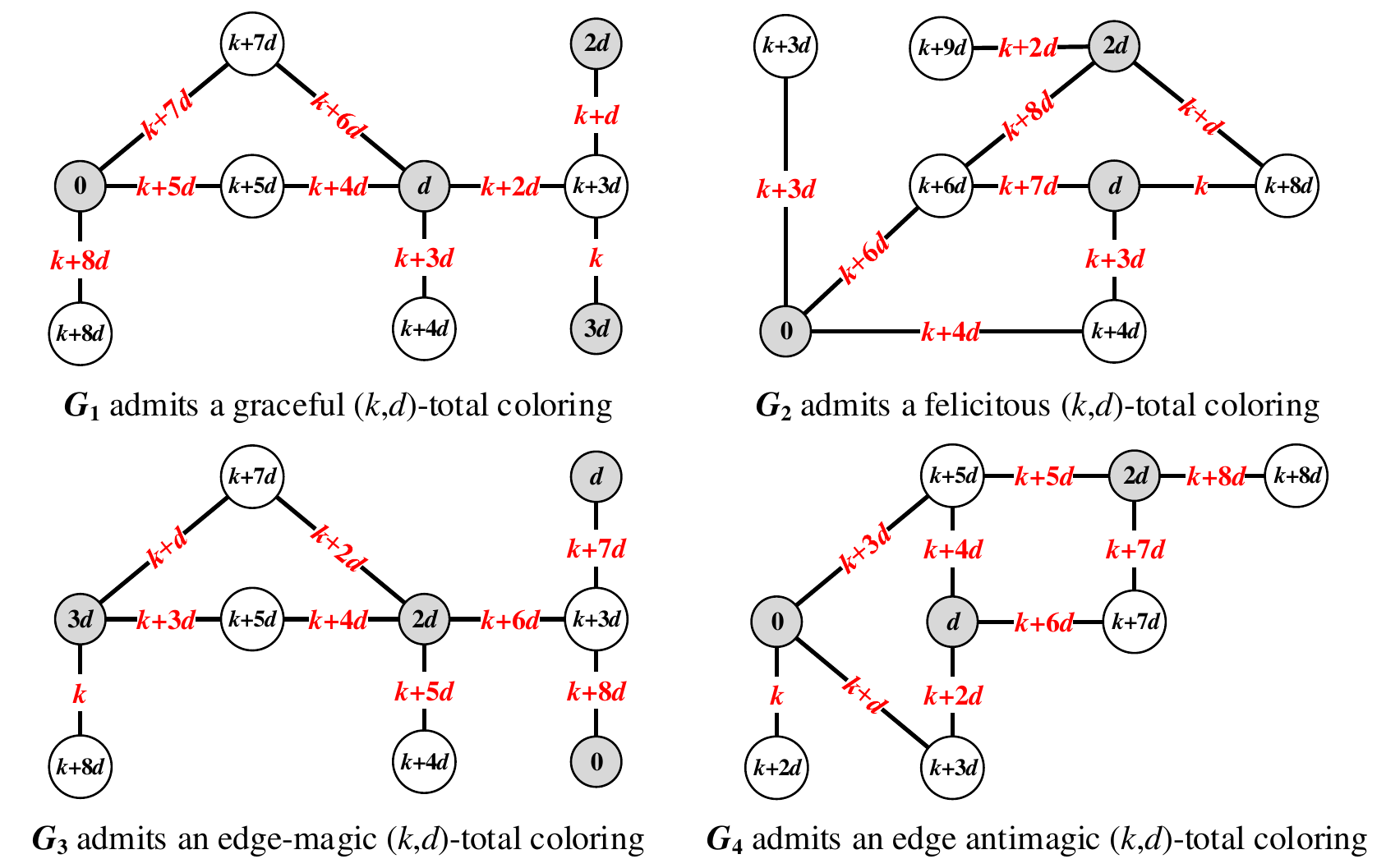}
\caption{\label{fig:k-d-matrices-graphs}{\small Four Topsnut-gpws corresponding to four $(k,d)$-type Topcode-matrices shown in Eq.(\ref{eqa:44-k-d-matrices}).}}
\end{figure}

{\footnotesize
\begin{equation}\label{eqa:Topcode-matrices-express11}
T_{code}(G\,'_1,g_1)= \left(
\begin{array}{ccccccccccccccc}
2 & 2 & 1 & 2 & 1 & 0 & 1 & 0 & 0 \\
0 & 1 & 2 & 3 & 4 & 5 & 6 & 7 & 8\\
2 & 3 & 3 & 5 & 5 & 5 & 7 & 7 & 8
\end{array}
\right)~T_{code}(G\,'_2,g_2)= \left(
\begin{array}{ccccccccccccccc}
1 & 2 & 2 & 0 & 0 & 1 & 0 & 1 & 2 \\
0 & 1 & 2 & 3 & 4 & 5 & 6 & 7 & 8\\
8 & 8 & 9 & 3 & 4 & 4 & 6 & 6 & 6
\end{array}
\right)
\end{equation}
}
{\footnotesize
\begin{equation}\label{eqa:Topcode-matrices-express22}
T_{code}(G\,'_3,g_3)= \left(
\begin{array}{ccccccccccccccc}
3 & 3 & 2 & 3 & 2 & 2 & 2 & 1 & 0 \\
0 & 1 & 2 & 3 & 4 & 5 & 6 & 7 & 8\\
8 & 7 & 7 & 5 & 5 & 4 & 3 & 3 & 3
\end{array}
\right)~T_{code}(G\,'_4,g_4)= \left(
\begin{array}{ccccccccccccccc}
0 & 0 & 1 & 0 & 1 & 2 & 1 & 2 & 2 \\
0 & 1 & 2 & 3 & 4 & 5 & 6 & 7 & 8\\
2 & 3 & 3 & 5 & 5 & 5 & 7 & 7 & 8
\end{array}
\right)
\end{equation}
}

\begin{defn}\label{defn:6C-labeling}
\cite{Yao-Sun-Zhang-Mu-Sun-Wang-Su-Zhang-Yang-Yang-2018arXiv} A total labeling $f:V(G)\cup E(G)\rightarrow [1,p+q]$ for a bipartite $(p,q)$-graph $G$ is a bijection and holds:

(i) (e-magic) $f(uv)+|f(u)-f(v)|=k$;

(ii) (ee-difference) each edge $uv$ matches with another edge $xy$ holding one of $f(uv)=|f(x)-f(y)|$ and $f(uv)=2(p+q)-|f(x)-f(y)|$ true;

(iii) (ee-balanced) let $s(uv)=|f(u)-f(v)|-f(uv)$ for each edge $uv\in E(G)$, then there exists a constant $k\,'$ such that each edge $uv$ matches with another edge $u\,'v\,'$ holding one of $s(uv)+s(u\,'v\,')=k\,'$ and $2(p+q)+s(uv)+s(u\,'v\,')=k\,'$ true;

(iv) (EV-ordered) $\min f(V(G))>\max f(E(G))$, or $\max f(V(G))<\min f(E(G))$, or $f(V(G))\subseteq f(E(G))$, or $f(E(G))$ $\subseteq f(V(G))$, or the vertex color set $f(V(G))$ is an odd-set and the edge color set $f(E(G))$ is an even-set;

(v) (ve-matching) there exists a constant $k\,''$ such that each edge $uv$ matches with one vertex $w$ such that $f(uv)+f(w)=k\,''$, and each vertex $z$ matches with one edge $xy$ such that $f(z)+f(xy)=k\,''$, except the \emph{singularity} $f(x_0)=\lfloor \frac{p+q+1}{2}\rfloor $;

(vi) (set-ordered) $\max f(X)<\min f(Y)$ (resp. $\min f(X)>\max f(Y)$) for the bipartition $(X,Y)$ of $V(G)$.

We refer to $f$ as an \emph{edge-difference 6C-labeling} of $G$.\qqed
\end{defn}

Since ``set-ordered'' is in an edge-difference 6C-labeling, so we have a new $(k,d)$-total coloring as follows:

\begin{defn} \label{defn:k-d-6C-labeling}
$^*$ If a bipartite $(p,q)$-graph $G$ admits an edge-difference 6C-labeling $f:V(G)\cup E(G)\rightarrow [1,p+q]$, which induces a Topcode-matrix $T_{code}(G,f)$, then the graph $G$ admits an \emph{edge-difference $(k,d)$-6C-labeling}
\begin{equation}\label{eqa:333333}
{
\begin{split}
&F:X\rightarrow S_{m,0,0,d}=\{0,d,\dots ,md\}\\
&F:Y\cup E(G)\rightarrow S_{p+q,k,0,d}=\{k,k+d,\dots ,k+(p+q)d\}
\end{split}}
\end{equation} where $V(G)=X\cup Y$ with $X\cap Y=\emptyset$, and the $(k,d)$-6C-labeling $F$ induces an \emph{edge-difference $(k,d)$-6C Topcode-matrix} of $G$ as follows
\begin{equation}\label{eqa:555555}
P_{ara}(G,F)=k\cdot I\,^0+d\cdot T_{code}(G,f)
\end{equation} where $I\,^0$ is the unite Topcode-matrix defined in Eq.(\ref{eqa:unit-Topcode-matrix}), and the Topcode-matrices $I\,^0$, $T_{code}(G,f)$ and $P_{ara}(G,F)$ have the same order $3\times q$ (see Fig.\ref{fig:6C-k-d}).\qqed
\end{defn}

\begin{figure}[h]
\centering
\includegraphics[width=14.4cm]{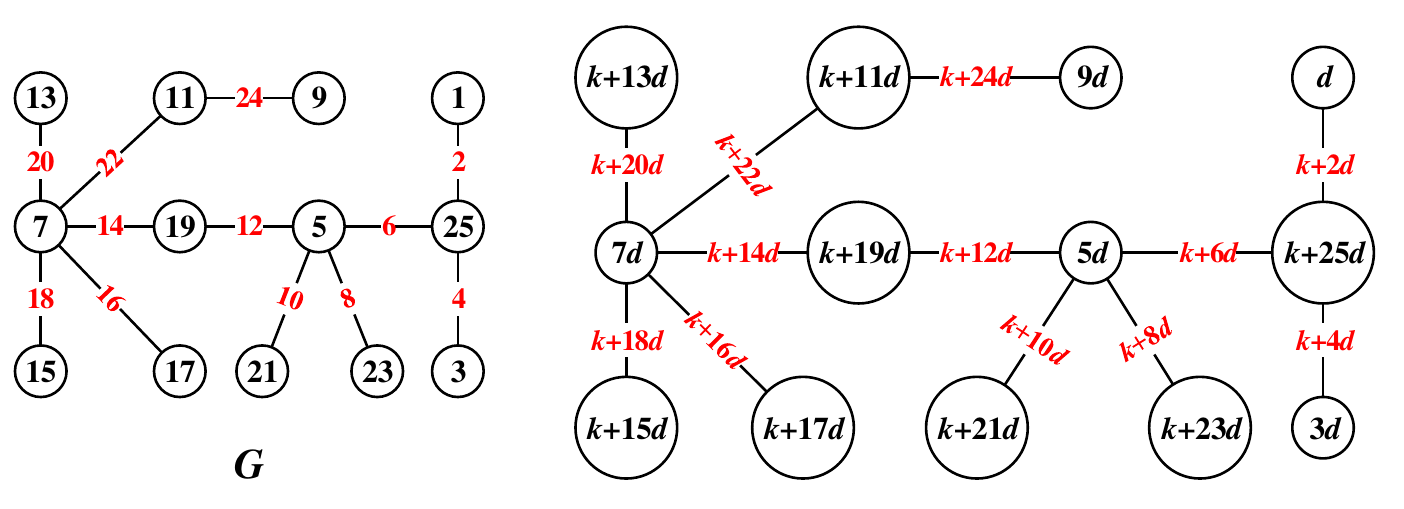}\\
\caption{\label{fig:6C-k-d} {\small A bipartite $(13,12)$-graph $G$ admits an edge-difference 6C-labeling and an edge-difference $(k,d)$-6C-labeling.}}
\end{figure}

\begin{rem}\label{rem:333333}
Since there are the (edge-difference) odd-6C-labeling and the odd-even separable (edge-difference) 6C-labeling introduced in \cite{Yao-Wang-2106-15254v1}, so we have the $(k,d)$-odd-6C Topcode-matrix and the odd-even separable $(k,d)$-6C Topcode-matrix similarly with that defined in Definition \ref{defn:k-d-6C-labeling}. \paralled
\end{rem}

Motivated from Definition \ref{defn:k-d-6C-labeling}, we present the following general definition about the $W$-constraint $(k,d)$-colorings/labelings:

\begin{defn} \label{defn:w-type-coloring-labeling-kd}
$^*$ If a bipartite $(p,q)$-graph $G$ admits a $W$-constraint coloring/labeling $f$, which induces a Topcode-matrix $T_{code}(G,f)$, then the graph $G$ admits a \emph{$W$-constraint $(k,d)$-coloring/labeling}
\begin{equation}\label{eqa:w-type-coloring-k-d}
{
\begin{split}
&F:X\rightarrow S_{m,0,0,d}=\{0,d,\dots ,md\}\\
&F:Y\cup E(G)\rightarrow S_{n,k,0,d}=\{k,k+d,\dots ,k+nd\}
\end{split}}
\end{equation} where $V(G)=X\cup Y$ with $X\cap Y=\emptyset$, and this coloring/labeling $F$ induces a \emph{$(k,d)$-$W$-constraint Topcode-matrix} of $G$ as follows
\begin{equation}\label{eqa:w-type-coloring-Topcode-ma}
P_{ara}(G,F)=k\cdot I\,^0+d\cdot T_{code}(G,f)
\end{equation} where $I\,^0$ is the unite Topcode-matrix defined in Eq.(\ref{eqa:unit-Topcode-matrix}), and $I\,^0$, $T_{code}(G,f)$ and $P_{ara}(G,F)$ are three Topcode-matrices of $3\times q$.\qqed
\end{defn}
\begin{thm}\label{thm:666666}
$^*$ Suppose that a bipartite $(p,q)$-graph $G$ admits different set-ordered total colorings $f_1,f_2,\dots, f_m$, and these total colorings are equivalent from each other. Then the graph $G$ admits $W$-constraint $(k,d)$-colorings $F_1,F_2,\dots, F_m$ defined in Definition \ref{defn:w-type-coloring-labeling-kd}, such that these $(k,d)$-$W$-constraint Topcode-matrices
\begin{equation}\label{eqa:555555}
P_{ara}(G,F_i)=k\cdot I\,^0+d\cdot T_{code}(G,f_i),~i\in [1,m]
\end{equation} are equivalent from each other, where three Topcode-matrices $I\,^0$, $T_{code}(G,f_i)$ and $P_{ara}(G,F_i)$ have the same order $3\times q$.
\end{thm}

\subsection{Number-based strings containing parameters}

Number-based strings containing parameterized numbers are called \emph{$(k,d)$-type number-based strings}.

\begin{example}\label{exa:8888888888}
There are nine Hanzi-graphs $H_1,H_2,\dots ,H_9$ shown in Fig.(\ref{fig:9-nv-zi})with their parameterized Topcode-matrices
\begin{equation}\label{eqa:333333}
P_{ara}(H_i,F_i)=k_i\cdot I\,^0+d_i\cdot T_{code}(H_i,g_i),~i\in [1,9]
\end{equation} where the Topcode-matrices $I\,^0$, $T_{code}(H_i,h_i)$ and $P_{ara}(H_i,F_i)$ for $i\in [1,9]$ have the same order.

We let the Hanzi-graph $H_1$ to be a \emph{public-key} and the Hanzi-graph $H_2$ to be a \emph{private-key}, and then the Hanzi-graph $H_3$ is just the \emph{topological authentication} of two Hanzi-graphs $H_1$ and $H_2$. On the other hands, the Hanzi-graph $H_1$ has its own private-keys $H_7$ and $H_8$, and the Hanzi-graph $H_2$ has its own private-keys $H_4,H_5,H_6$ and $H_9$.

In Eq.(\ref{eqa:strings-parameterized-22}), we can get $(k,d)$-type number-based strings with arbitrary length of bytes from two Topcode-matrices $P_{ara}(H_1,F_1)$ and $P_{ara}(H_2,F_2)$ for the values of $k_1,d_1, k_2$ and $d_2$ with random bytes. As all of $k_1$, $k_2$, $d_1$ and $d_2$ are prime numbers, moreover, we can ask for

(1) larger odd integer $m=k_1\cdot k_2$ and larger odd integer $n=d_1\cdot d_2$;

(2) larger odd integer $m=k_1\cdot k_2$ and larger even integer $n=d_1+d_2$;

(3) larger even integer $r=k_1+k_2$ and larger odd integer $s=d_1\cdot d_2$;

(4) larger even integer $r=k_1+k_2$ and larger even integer $s=d_1+d_2$.\\
This is mainly to increase the cost of attackers as much as possible.\qqed
\end{example}

{\small
\begin{equation}\label{eqa:strings-parameterized-111}
T_{code}(H_1,g_1)= \left(
\begin{array}{ccccccccccccccc}
3 & 3 & 3 & 3 & 2 & 0 & 0 & 0 & 0 \\
1 & 2 & 3 & 4 & 5 & 6 & 7 & 8 & 9\\
4 & 5 & 6 & 7 & 7 & 6 & 7 & 8 & 9
\end{array}
\right),T_{code}(H_2,g_2)= \left(
\begin{array}{ccccccccccccccc}
2 & 1 & 1 & 1 & 1 & 0 & 0\\
1 & 2 & 3 & 4 & 5 & 6 & 7\\
3 & 3 & 4 & 5 & 6 & 6 & 7
\end{array}
\right)
\end{equation}
}

{\footnotesize
\begin{equation}\label{eqa:strings-parameterized-22}
\centering
{
\begin{split}
&\quad P_{ara}(H_1,F_1)\\
&= \left(
\begin{array}{ccccccccccccccc}
3d_1 & 3d_1 & 3d_1 & 3d_1 & 2d_1 & 0 & 0 & 0 & 0\\
k_1+d_1 & k_1+2d_1 & k_1+3d_1 & k_1+4d_1 & k_1+5d_1 & k_1+6d_1 & k_1+7d_1 & k_1+8d_1 & k_1+9d_1\\
k_1+4d_1 & k_1+5d_1 & k_1+6d_1 & k_1+7d_1 & k_1+7d_1 & k_1+6d_1 & k_1+7d_1 & k_1+8d_1 & k_1+9d_1
\end{array}
\right)\\
&\quad P_{ara}(H_2,F_2)= \left(
\begin{array}{ccccccccccccccc}
2d_2 & d_2 & d_2 & d_2 & d_2 & 0 & 0 \\
k_2+d_2 & k_2+2d_2 & k_2+3d_2 & k_2+4d_2 & k_2+5d_2 & k_2+6d_2 & k_2+7d_2\\
k_2+3d_2 & k_2+3d_2 & k_2+4d_2 & k_2+5d_2 & k_2+6d_2 & k_2+6d_2 & k_2+7d_2
\end{array}
\right)
\end{split}}
\end{equation}
}

\begin{figure}[h]
\centering
\includegraphics[width=14cm]{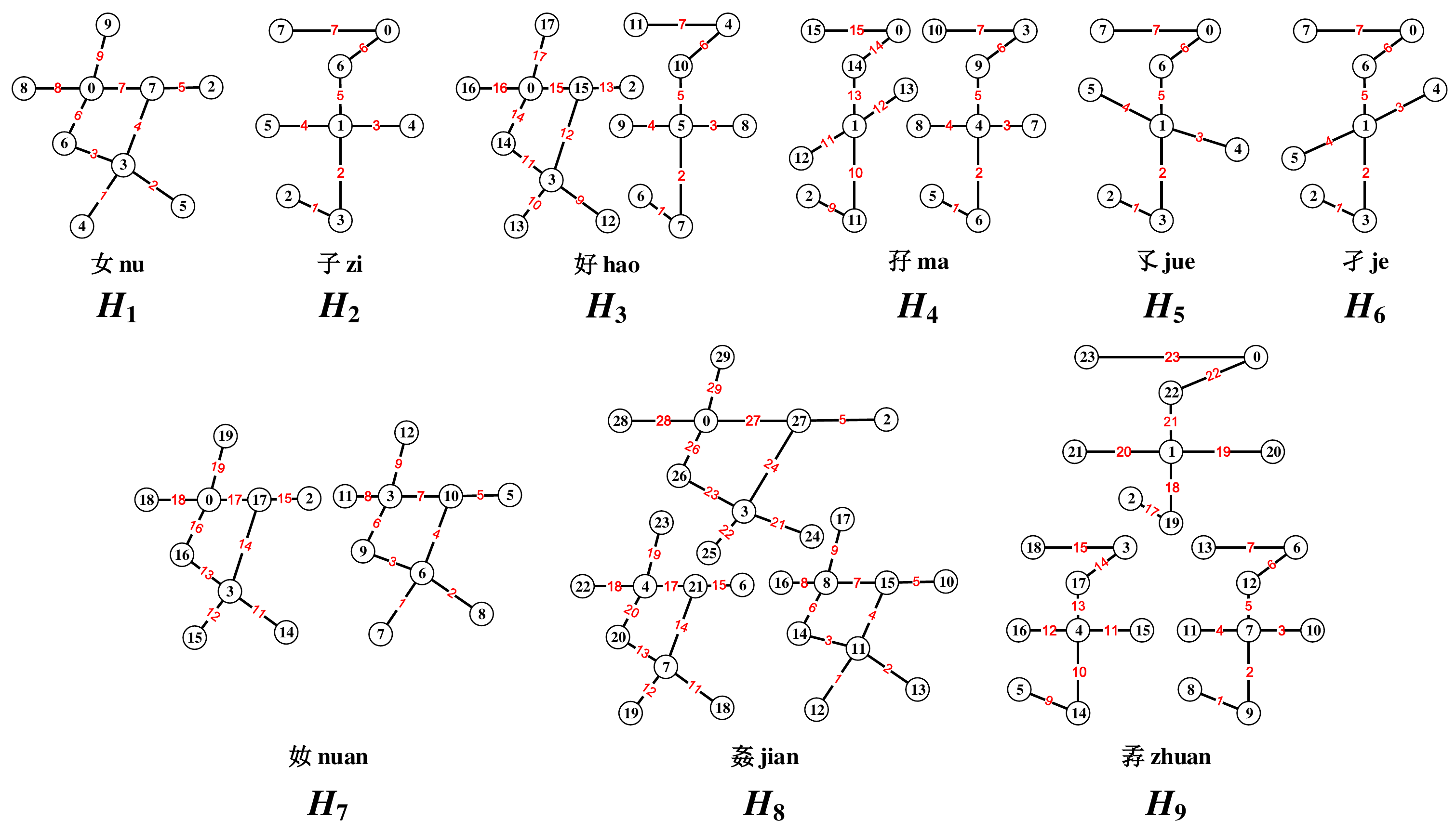}
\caption{\label{fig:9-nv-zi}{\small The Hanzi-graphs $H_1,H_2,H_5$ and $H_6$ admit set-ordered graceful labelings, and other Hanzi-graphs admit flawed set-ordered graceful labelings.}}
\end{figure}

\begin{problem}\label{qeu:PRONBS-problems}
We present a problem for the parametric reconstitution of number-based strings, called PRONBS-problem as follows:
\begin{quote}
\textbf{PRONBS-problem.} For a given number-based string $s=c_1c_2\cdots c_n$ with $c_i\in [0,9]$, rewrite this number-based string $s$ into $3q$ segments $a_1,a_2,\dots ,a_{3q}$ with $a_j=b_{j,1}b_{j,2}\cdots b_{j,m_j}$ for $j\in [1,3q]$, such that each number $c_i$ of the number-based string $s$ is in a segment $a_j$ but $c_i$ is not in other segment $a_s$ if $s\neq j$. Find two integers $k_0,d_0\geq 0$, such that each segment $a_j$ with $j\in [1,3q]$ can be expressed as $a_j=\beta_jk_0+\gamma_jd_0$ for integers $\beta_j,\gamma_j\geq 0$ and $j\in [1,3q]$; and moreover find a colored $(p,q)$-graph $G$ admitting a parameterized coloring $F$ holding its own parameterized Topcode-matrix
\begin{equation}\label{eqa:555555}
P_{ara}(G,F)_{3\times q}=k\cdot I\,^0_{3\times q}+d\cdot T_{code}(G,f)_{3\times q}
\end{equation} where $f$ is a $W$-constraint coloring of $G$, as well as the parameterized number-based string $a_1a_2\dots a_{3q}=c\,'_1c\,'_2\cdots c\,'_n$ is just generated from the parameterized Topcode-matrix $P_{ara}(G,F)_{3\times q}$ when $(k,d)=(k_0,d_0)$.
\end{quote}
\end{problem}

\begin{example}\label{exa:8888888888}
When $(k_2,d_2)=(7,5)$ in $P_{ara}(H_2,F_2)$ of Eq.(\ref{eqa:strings-parameterized-22}), we get a Topcode-matrix
{\small
\begin{equation}\label{eqa:valued-parameterized-matrix}
\centering
{
\begin{split}
P_{ara}(H_2,F_2)&= \left(
\begin{array}{ccccccccccccccc}
2\cdot 5 & 5 & 5 & 5 & 5 & 0 & 0 \\
7+5 & 7+2\cdot 5 & 7+3\cdot 5 & 7+4\cdot 5 & 7+5\cdot 5 & 7+6\cdot 5 & 7+7\cdot 5\\
7+3\cdot 5 & 7+3\cdot 5 & 7+4\cdot 5 & 7+5\cdot 5 & 7+6\cdot 5 & 7+6\cdot 5 & 7+7\cdot 5
\end{array}
\right)\\
&= \left(
\begin{array}{ccccccccccccccc}
10 & 5 & 5 & 5 & 5 & 0 & 0 \\
12 & 17 & 22 & 27 & 32 & 37 & 42\\
22 & 22 & 27 & 32 & 37 & 37 & 42
\end{array}
\right)_{3\times 7}
\end{split}}
\end{equation}
}which induces the following two number-based strings
$${
\begin{split}
D_1&=105555001217222732374222222732373742 \\
D_2&=121722273237420022222732373742105555
\end{split}}
$$ And we can rewrite the $(7,5)$-type number-based string $D_1$ into 21 segments:

$0\cdot 7+2\cdot 5$, $0\cdot 7+1\cdot 5$, $0\cdot 7+1\cdot 5$, $0\cdot 7+1\cdot 5$, $0\cdot 7+1\cdot 5$, $0\cdot 7+0\cdot 5$, $0\cdot 7+0\cdot 5$,

$1\cdot 7+1\cdot 5$, $1\cdot 7+2\cdot 5$, $1\cdot 7+3\cdot 5$, $1\cdot 7+4\cdot 5$, $1\cdot 7+5\cdot 5$, $1\cdot 7+6\cdot 5$, $1\cdot 7+7\cdot 5$,

$1\cdot 7+3\cdot 5$, $1\cdot 7+3\cdot 5$, $1\cdot 7+4\cdot 5$, $1\cdot 7+5\cdot 5$, $1\cdot 7+6\cdot 5$, $1\cdot 7+6\cdot 5$, $1\cdot 7+7\cdot 5$.\\
The above 21 segments form just the parameterized Topcode-matrix $P_{ara}(H_2,F_2)$ shown in Eq.(\ref{eqa:valued-parameterized-matrix}).

By rearranging the orders of numbers in the above $(7,5)$-type number-based strings $D_1$ and $D_2$, we rewrite these two $(7,5)$-type number-based strings as follows
$${
\begin{split}
D\,'_1&=000111222232342222223233425555777777\\
D\,'_2&=111737373737374400055552222222222222
\end{split}}
$$ Clearly, it is quite difficult to recover the $(7,5)$-type number-based strings $D_1$ and $D_2$ from these two number-based strings $D\,'_1$ and $D\,'_2$, and it is impossible to do such works for number-based strings with thousands of bytes.\qqed
\end{example}

\begin{defn} \label{defn:more-string-total-coloring}
$^*$ \textbf{Parameterized string-colorings and set-colorings}. By Definition \ref{defn:kd-w-type-colorings} and Definition \ref{defn:odd-edge-W-type-total-labelings-definition}, let $G$ be a bipartite $(p,q)$-graph, and its vertex set $V(G)=X\cup Y$ with $X\cap Y=\emptyset$ such that each edge $uv\in E(G)$ holds $u\in X$ and $v\in Y$. There are a group of colorings
$${
\begin{split}
&f_s:X\rightarrow S_{m,0,0,d}=\{0,d,\dots ,md\}\\
&f_s:Y\cup E(G)\rightarrow S_{n,k,0,d}=\{k,k+d,\dots ,k+nd\}
\end{split}}
$$ here it is allowed $f_s(u)=f_s(w)$ for some distinct vertices $u,w\in V(G)$ with $s\in [1,B]$ with integer $B\geq 2$, such that $f_s\in \{$gracefully $(k_s,d_s)$-total coloring, odd-gracefully $(k_s,d_s)$-total coloring, edge anti-magic $(k_s,d_s)$-total coloring, harmonious $(k_s,d_s)$-total coloring, odd-elegant $(k_s,d_s)$-total coloring, edge-magic $(k_s,d_s)$-total coloring, edge-difference $(k_s,d_s)$-total coloring, felicitous-difference $(k_s,d_s)$-total coloring, graceful-difference $(k_s,d_s)$-total coloring, odd-edge edge-magic $(k_s,d_s)$-total coloring, odd-edge edge-difference $(k_s,d_s)$-total coloring, odd-edge felicitous-difference $(k_s,d_s)$-total coloring, odd-edge graceful-difference $(k_s,d_s)$-total coloring$\}$ with $s\in [1,B]$. Then we have:

(i) The bipartite $(p,q)$-graph $G$ admits a \emph{parameterized total string-coloring} $F$ defined by
$$
F(u)=f_{i_1}(u)f_{i_2}(u)\cdots f_{i_B}(u),F(uv)=f_{j_1}(uv)f_{j_2}(uv)\cdots f_{j_B}(uv),F(v)=f_{s_1}(v)f_{s_2}(v)\cdots f_{s_B}(v)
$$ true for each edge $uv\in E(G)$, where $f_{i_1}(u)f_{i_2}(u)\cdots f_{i_B}(u)$ is a permutation of $f_1(u),f_2(u),\cdots $, and $f_B(u)$, $f_{j_1}(uv)f_{j_2}(uv)\cdots f_{j_B}(uv)$ is a permutation of $f_1(uv),f_2(uv),\cdots ,f_B(uv)$, as well as $f_{s_1}(v)f_{s_2}(v)$ $\cdots f_{s_B}(v)$ is a permutation of $f_1(v),f_2(v),\cdots ,f_B(v)$.

(ii) The bipartite $(p,q)$-graph $G$ admits a \emph{parameterized total set-coloring} $\theta$ defined by
$${
\begin{split}
\theta(u)&=\{f_1(u),f_2(u),\cdots ,f_B(u)\},\quad \theta(uv)=\{f_1(uv),f_2(uv),\cdots ,f_B(uv)\},\\
\theta(v)&=\{f_1(v),f_2(v),\cdots ,f_B(v)\}
\end{split}}
$$ true for each edge $uv\in E(G)$.\qqed
\end{defn}

\begin{problem}\label{qeu:444444}
Let $s_\infty=c_1c_2\dots $ be an infinite number-based string with $c_i\in [0,9]$. Does this infinite string $s_\infty$ contain any group of finite number-based strings $S_1$, $S_2$, $\dots$, $S_n$, where $S_j=a_{j,1}a_{j,2}\cdots a_{j,r_j}$ with $a_{j,t}\in [0,9]$ for $t\in [1,r_j]$, such that these finite number-based strings form an algebraic group, or a Topcode-matrix?
\end{problem}

\subsection{Graphic lattices related with parameterized total colorings}

\subsubsection{Operation-type graphic lattices}

\begin{defn} \label{defn:111111}
$^*$ Suppose that a \emph{graph base} $\textbf{\textrm{H}}=(H_1,H_2,\dots ,H_m)=(H_k)^m_{k=1}$ is made by disjoint bipartite and connected graphs $H_1$, $H_2$, $\dots $, $H_m$, and we define flawed $W$-constraint $(k,d)$-total colorings for $\textbf{\textrm{H}}$ as follows: Add the edges of an edge set $E^*$ to the base $\textbf{\textrm{H}}$ for joining disjoint bipartite and connected graphs $H_1,H_2,\dots ,H_m$ to be a connected graph, denoted as $$G=\textbf{\textrm{H}}+E^*=E^*[\ominus]^m_{k=1}H_k
$$ If $G=\textbf{\textrm{H}}+E^*$ admits a $W$-constraint $(k,d)$-total coloring $f$, then this coloring $f$ induces a $(k,d)$-total coloring $f\,'$ of the base $\textbf{\textrm{H}}$, we call $f\,'$ a \emph{flawed $W$-constraint $(k,d)$-total coloring}, and we call ``$[\ominus ]$'' to be \emph{edge-join operation}.\qqed
\end{defn}

For a graph base $\textbf{\textrm{H}}=(H_1,H_2,\dots ,H_m)=(H_k)^m_{k=1}$ and $a_k\in Z^0$, the notation $a_kH_k$ stands for $a_k$ copies of a graph $H_k$, and we use the edges of some edge set $E^*$ to join graphs $a_1H_1$, $a_2H_2$, $\dots $, $a_mH_m$ together, such that the resultant graph is a connected graph and denoted as $E^*[\ominus ]^m_{k=1}a_kH_k$, notice that the edge set $E^*$ is various, and there are more methods to get $E^*[\ominus ]^m_{k=1}a_kH_k$, then we get an \emph{edge-join graphic lattice}
\begin{equation}\label{eqa:edge-join-graphic-lattice}
\textbf{\textrm{L}}(\textbf{\textrm{E}}[\ominus ]Z^0\textbf{\textrm{H}})=\big \{E^*[\ominus ]^m_{k=1}a_kH_k:~a_k\in Z^0,H_k\in \textbf{\textrm{H}},E^*\in \textbf{\textrm{E}}\big \}
\end{equation} with $\sum ^m_{k=1}a_k\geq 1$, and call $\textbf{\textrm{H}}$ \emph{lattice base}, where $\textbf{\textrm{E}}$ is the set of edge sets.

Let $G_{i_1},G_{i_2},\dots ,G_{i_M}$ be a permutation of graphs $a_1H_1$, $a_2H_2$, $\dots $, $a_mH_m$ for $M=\sum ^m_{k=1}a_k\geq 1$. By the vertex-coinciding operation ``$[\odot ]$'', we vertex-coincide a vertex $x$ of $G_{i_j}$ with a vertex $y$ of $G_{i_{j+1}}$ into one vertex $x\odot y$ for $j\in [1,M-1]$, the resultant graph is denoted as
$$G_{i_1}[\odot ]G_{i_2}[\odot ]\dots [\odot ]G_{i_M}=[\odot _{\textrm{linear}}]^m_{k=1}a_kH_k$$
Hence, we get at least $M!$ graphs $[\odot _{\textrm{linear}}]^m_{k=1}a_kH_k$. We call the following set
\begin{equation}\label{eqa:linear-vertex-odot-graphic-lattice}
\textbf{\textrm{L}}([\odot_{\textrm{linear}}]Z^0\textbf{\textrm{H}})=\big \{[\odot _{\textrm{linear}}]^m_{k=1}a_kH_k:~a_k\in Z^0,H_k\in \textbf{\textrm{H}}\big \}
\end{equation} a \emph{linear vertex-coinciding graphic lattice} with $\sum ^m_{k=1}b_k\geq 1$.

In general, we use the vertex-coinciding operation to graphs $a_1H_1,a_2H_2,\dots ,a_mH_m$ for producing graphs $[\odot]^m_{k=1}a_kH_k$, and get a \emph{vertex-coinciding graphic lattice}
\begin{equation}\label{eqa:vertex-odot-graphic-lattice}
\textbf{\textrm{L}}([\odot]Z^0\textbf{\textrm{H}})=\big \{[\odot ]^m_{k=1}\beta_kH_k:~\beta_k\in Z^0,H_k\in \textbf{\textrm{H}}\big \}
\end{equation} with $\sum ^m_{k=1}\beta_k\geq 1$.

Similarly, we do the edge-join operation ``$[\ominus ]$'' and the vertex-coinciding operation ``$[\odot ]$'' to graphs $a_1H_1,a_2H_2,\dots ,a_mH_m$ by the mixed way, the resultant graph is denoted as $[\ominus \odot]^m_{k=1}a_kH_k$ simply. Thereby, we get an \emph{operation-mixed graphic lattice}
\begin{equation}\label{eqa:edge-join-vertex-odot-graphic-lattice}
\textbf{\textrm{L}}([\ominus \odot]Z^0\textbf{\textrm{H}})=\big \{[\ominus \odot]^m_{k=1}c_kH_k:~a_k\in Z^0,H_k\in \textbf{\textrm{H}}\big \}
\end{equation} with $\sum ^m_{k=1}c_k\geq 1$.

It is important to investigate the coloring closure of graphic lattices $\textbf{\textrm{L}}(\textbf{\textrm{E}}[\ominus ]Z^0\textbf{\textrm{H}})$, $\textbf{\textrm{L}}([\odot]Z^0\textbf{\textrm{H}})$ and $\textbf{\textrm{L}}([\ominus \odot]Z^0\textbf{\textrm{H}})$.

\begin{thm}\label{thm:vertex-split-disjoint-graphs}
Since any graph $G$ can be vertex-split into pairwise disjoint graphs $G_1,G_2,\dots ,G_m$, which is just a permutation of graphs $a_1H_1$, $a_2H_2$, $\dots $, $a_mH_m$, then each graph is in a vertex-coinciding graphic lattice $\textbf{\textrm{L}}([\odot]Z^0\textbf{\textrm{H}})$ defined in Eq.(\ref{eqa:vertex-odot-graphic-lattice}).
\end{thm}

\begin{problem}\label{qeu:444444}
Does there exist another graph base $\textbf{\textrm{H}}\,'=(H\,'_1,H\,'_2,\dots ,H\,'_m)$ holding one of
$$\textbf{\textrm{L}}([\ominus]Z^0\textbf{\textrm{H}}\,')=\textbf{\textrm{L}}([\ominus]Z^0\textbf{\textrm{H}}),~\textbf{\textrm{L}}([\odot]Z^0\textbf{\textrm{H}}\,')=\textbf{\textrm{L}}([\odot]Z^0\textbf{\textrm{H}}),~\textbf{\textrm{L}}([\ominus \odot]Z^0\textbf{\textrm{H}}\,')=\textbf{\textrm{L}}([\ominus \odot]Z^0\textbf{\textrm{H}})$$
\end{problem}

A \emph{parameter-colored tree base} $\textbf{\textrm{H}}^c=(H^c_1,H^c_2,\dots ,H^c_m)$ consists of disjoint graphs $H^c_1,H^c_2$, $\dots $, $H^c_m$, where each tree $H^c_i$ admits a $W_i$-constraint $(k_i,d_i)$-total coloring $g_i$ ($i\in [1,m]$) defined in Definition \ref{defn:kd-w-type-colorings}, and two vertex color sets $g_i(V(H^c_i))\cap g_j(V(H^c_j))\neq \emptyset$. We take is a permutation $G_1,G_2,\dots ,G_Q$ of graphs $c_1H^c_1,c_2H^c_2,\dots,c_nH^c_m$ with $Q=\sum ^m_{k=1}c_k$, and we vertex-coincide a vertex $x$ of $H^c_i$ with a vertex $y$ of $H^c_{i+1}$ into one vertex $x\odot y$ for $i\in [1,M-1]$, where $f_i(x)=f_j(y)$ according to the assume, so the resultant graph is a tree, denoted as
$$
G_1[\odot]G_2[\odot]\cdots [\odot]G_Q=[\odot_{\textrm{linear}}]^m_{k=1}c_kH^c_k
$$ and we can get $Q!$ trees for each number $Q=\sum ^m_{k=1}c_k$. We define a \emph{linear parameter-colored graphic lattice} as follows
\begin{equation}\label{eqa:linear-colored-tree-vertex-coin-lattice}
\textbf{\textrm{L}}([\odot_{\textrm{linear}}]Z^0\textbf{\textrm{H}}^c)=\big \{[\odot_{\textrm{linear}}]^m_{k=1}c_kH^c_k:~c_k\in Z^0,H^c_k\in \textbf{\textrm{H}}^c\big \}
\end{equation} with $\sum ^m_{k=1}c_k\geq 1$.

In general, we do the vertex-coinciding operation to disjoint graphs $H^c_1,H^c_2,\dots ,H^c_m$, the resultant simple graphs are denoted as $[\odot]^m_{k=1}\gamma_kH^c_k$, and we define a \emph{parameter-colored graphic lattice} as follows
\begin{equation}\label{eqa:colored-tree-vertex-coin-lattice}
\textbf{\textrm{L}}([\odot]Z^0\textbf{\textrm{H}}^c)=\big \{[\odot]^m_{k=1}\gamma_kH^c_k:~\gamma_k\in Z^0,H^c_k\in \textbf{\textrm{H}}^c\big \}
\end{equation} with $\sum ^m_{k=1}\gamma_k\geq 1$.

\subsubsection{Uncolored tree-graphic lattices}

An \emph{uncolored tree base} $\textbf{\textrm{T}}=(T_1,T_2,\dots ,T_n)=(T_k)^n_{k=1}$ consists of disjoint trees $T_1,T_2,\dots ,T_n$. We use the edges of an edge set $E^*$ to join these disjoint trees together, the resulting graph is denoted as $\textbf{\textrm{T}}+E^*$. Notice that one tree $H_i$ may be used two or more times in $\textbf{\textrm{T}}+E^*$, so we write $E^*[\ominus]^n_{k=1} \alpha_kT_k$ with $\sum ^m_{k=1}\alpha_k\geq 1$ for replacing $\textbf{\textrm{T}}+E^*$ by the \emph{edge-join operation} ``$[\ominus ]$''. We have a \emph{tree-base graphic lattice} as follows:
\begin{equation}\label{eqa:tree-edge-join-tree-lattice}
\textbf{\textrm{L}}(\textbf{\textrm{E}}[\ominus ]Z^0\textbf{\textrm{T}})=\big \{E^*[\ominus ]^m_{k=1}\mu_kT_k:~\mu_k\in Z^0,T_k\in \textbf{\textrm{T}},E^*\in \textbf{\textrm{E}}\big \}
\end{equation} with $\sum ^m_{k=1}\mu_k\geq 1$. And moreover, we have other two tree-base graphic lattices
\begin{equation}\label{eqa:tree-vertex-odot-tree-lattice}
\textbf{\textrm{L}}([\odot ]Z^0\textbf{\textrm{T}})=\big \{[\odot ]^m_{k=1}\lambda_kT_k:~\lambda_k\in Z^0,T_k\in \textbf{\textrm{T}}\big \}
\end{equation} with $\sum ^m_{k=1}\lambda_k\geq 1$ and
\begin{equation}\label{eqa:tree-vertex-edge-mixed-tree-lattice}
\textbf{\textrm{L}}([\ominus \odot]Z^0\textbf{\textrm{T}})=\big \{[\ominus \odot]^m_{k=1}\gamma_kT_k:~\gamma_k\in Z^0,T_k\in \textbf{\textrm{T}}\big \}
\end{equation} with $\sum ^m_{k=1}\gamma_k\geq 1$.

\vskip 0.4cm

For considering the coloring closure of graphic lattices, let $\textbf{\textrm{L}}_{\textrm{tree}}(\textbf{\textrm{E}}[\ominus ]Z^0\textbf{\textrm{T}})$ be the set of all trees in a tree-base graphic lattice $\textbf{\textrm{L}}(\textbf{\textrm{E}}[\ominus ]Z^0\textbf{\textrm{T}})$ defined in Eq.(\ref{eqa:tree-edge-join-tree-lattice}), and we call it \emph{all-tree-graphic lattice}. Similarly, we have other two all-tree-graphic lattices $\textbf{\textrm{L}}_{\textrm{tree}}([\odot]Z^0\textbf{\textrm{T}})\subset \textbf{\textrm{L}}([\odot]Z^0\textbf{\textrm{T}})$ defined in Eq.(\ref{eqa:tree-vertex-odot-tree-lattice}) and $\textbf{\textrm{L}}_{\textrm{tree}}([\ominus \odot]Z^0\textbf{\textrm{T}})\subset\textbf{\textrm{L}}([\ominus \odot]Z^0\textbf{\textrm{T}})$ defined in Eq.(\ref{eqa:tree-vertex-edge-mixed-tree-lattice}). By the theorems introduced in the pre-subsections, these three all-tree-graphic lattices are coloring closure for the $(k,d)$-total colorings stated in the following theorem \ref{thm:coloring-closure-kd-total-colorings}:

\begin{thm}\label{thm:coloring-closure-kd-total-colorings}
Each tree contained in the all-tree-graphic lattices $\textbf{\textrm{L}}_{\textrm{tree}}(\textbf{\textrm{E}}[\ominus ]Z^0\textbf{\textrm{T}})$, $\textbf{\textrm{L}}_{\textrm{tree}}([\odot]Z^0\textbf{\textrm{T}})$ and $\textbf{\textrm{L}}_{\textrm{tree}}([\ominus \odot]Z^0\textbf{\textrm{T}})$ admits the following ten $(k,d)$-total colorings: graceful $(k,d)$-total coloring, harmonious $(k,d)$-total coloring, (odd-edge) edge-difference $(k,d)$-total coloring, (odd-edge) graceful-difference $(k,d)$-total coloring, (odd-edge) felicitous-difference $(k,d)$-total coloring, (odd-edge) edge-magic $(k,d)$-total coloring.
\end{thm}

\subsubsection{$(k,d)$-colored tree-graphic lattices}

A \emph{parameter-colored tree base} $\textbf{\textrm{T}}^c=(T^c_1,T^c_2,\dots ,T^c_n)$ is consisted of disjoint trees $T^c_1,T^c_2$, $\dots $, $T^c_n$, where each tree $T^c_i$ admits a $W_i$-constraint $(k_i,d_i)$-total coloring $f_i$ ($i\in [1,n]$) defined in Definition \ref{defn:kd-w-type-colorings}. If $W=W_1=W_2=\cdots =W_n$, and $(k,d)=(k_1,d_1)=(k_2,d_2)=\cdots=(k_n,d_n)$, we say that $\textbf{\textrm{T}}^c$ is a \emph{uniformly $W$-constraint $(k,d)$-colored tree base}.

Assume that two vertex color sets $f_i(V(T^c_i))\cap f_j(V(T^c_j))\neq \emptyset$ for $i\neq j$ in $\textbf{\textrm{T}}^c$. We take is a permutation $I_1,I_2,\dots ,I_M$ of graphs $b_1T^c_1,b_2T^c_2,\dots,b_nT^c_n$ with $M=\sum ^n_{k=1}b_k$, and we vertex-coincide a vertex $x$ of $T^c_i$ with a vertex $y$ of $T^c_{i+1}$ into one vertex $x\odot y$ for $i\in [1,M-1]$, where $f_i(x)=f_j(y)$ according to the assume, so the resultant graph is a tree, denoted as
$$I_1[\odot]I_2[\odot]\cdots [\odot]I_M=[\odot_{\textrm{linear}}]^n_{k=1}b_kT^c_k$$
And moreover we get a \emph{linear parameter-colored all-tree-graphic lattice}
\begin{equation}\label{eqa:linear-colored-tree-vertex-coin-lattice}
\textbf{\textrm{L}}_{\textrm{tree}}([\odot_{\textrm{linear}}]Z^0\textbf{\textrm{T}}^c)=\big \{[\odot_{\textrm{linear}}]^n_{k=1}b_kT^c_k:~b_k\in Z^0,T^c_k\in \textbf{\textrm{T}}^c\big \}
\end{equation} with $\sum ^n_{k=1}b_k\geq 1$.

In general, we call the following set
\begin{equation}\label{eqa:colored-tree-vertex-coin-lattice}
\textbf{\textrm{L}}_{\textrm{tree}}([\odot]Z^0\textbf{\textrm{T}}^c)=\big \{[\odot]^n_{k=1}\lambda_kT^c_k:~\lambda_k\in Z^0,T^c_k\in \textbf{\textrm{T}}^c\big \}
\end{equation} a \emph{parameter-colored all-tree-graphic lattice} with $\sum ^n_{k=1}\lambda_k\geq 1$.

\begin{rem}\label{rem:333333}
Sometimes, a tree $H\in \textbf{\textrm{L}}_{\textrm{tree}}([\odot]Z^0\textbf{\textrm{T}}^c)$ may be colored by two or more different $W_i$-constraint $(k_i,d_i)$-total colorings defined in Definition \ref{defn:kd-w-type-colorings}. Thereby, the Topcode-matrix $T_{code}(H)$ contains tow or more groups of parameters $(k_i,d_i),(k_j,d_j),\cdots ,(k_t,d_t)$. This enables us to create very complex number-base springs. \paralled
\end{rem}

For a permutation $I_1,I_2,\dots ,I_M$ of graphs $\beta_1T^c_1,\beta_2T^c_2,\dots,\beta_nT^c_n$ with $M=\sum ^n_{k=1}\beta_k$, we add leaves to each tree $I_j$ and get a $W_i$-constraint $(k_i,d_i)$-total coloring $f^*_j$ based on the $W_i$-constraint $(k_i,d_i)$-total coloring admitted by the parameter-colored tree base $\textbf{\textrm{T}}^c$, the resultant tree is denoted as $I^*_j$, and two vertex color sets $f^*_i(V(I^*_i))\cap f^*_j(I^*_j))\neq \emptyset$ by the assume $f_i(V(T^c_i))\cap f_j(V(T^c_j))\neq \emptyset$ for $i\neq j$ in $\textbf{\textrm{T}}^c$. And then we do the vertex-coincide operation to $I^*_1,I^*_2,\dots ,I^*_M$, that is
$$I^*_1[\odot]I^*_2[\odot]\cdots [\odot]I^*_M=[\ominus_{\textrm{leaf}}\odot_{\textrm{linear}}]^n_{k=1}\beta_kT^c_k$$
and get a \emph{parameter-colored adding-leaf tree-graphic lattice}
\begin{equation}\label{eqa:colored-tree-vertex-coin-add-leaf}
\textbf{\textrm{L}}_{\textrm{tree}}([\ominus_{\textrm{leaf}}\odot_{\textrm{linear}}]Z^0\textbf{\textrm{T}}^c)=\big \{[\ominus_{\textrm{leaf}}\odot_{\textrm{linear}}]^n_{k=1}\beta_kT^c_k:~\beta_k\in Z^0,T^c_k\in \textbf{\textrm{T}}^c\big \}
\end{equation} with $\sum ^n_{k=1}\beta_k\geq 1$.

\begin{rem}\label{rem:333333}
Sometimes, a tree $H\in \textbf{\textrm{L}}_{\textrm{tree}}([\odot]Z^0\textbf{\textrm{T}}^c)$, or $H\in \textbf{\textrm{L}}_{\textrm{tree}}([\ominus_{\textrm{leaf}}\odot_{\textrm{linear}}]Z^0\textbf{\textrm{T}}^c)$ may be colored by two or more different $W_i$-constraint $(k_i,d_i)$-total colorings defined in Definition \ref{defn:kd-w-type-colorings}. Thereby, the Topcode-matrix $T_{code}(H)$ contains tow or more groups of parameters $(k_i,d_i),(k_j,d_j),\cdots ,(k_t,d_t)$. This enables us to create very complex number-base springs. \paralled
\end{rem}

\vskip 0.4cm

We have some mathematical problems as follows:
\begin{problem} \label{problem:all-tree-graphic-lattice-11}
About an all-tree-graphic lattice $\textbf{\textrm{L}}_{\textrm{tree}}(\textbf{\textrm{E}}[\ominus ]Z^0\textbf{\textrm{T}})$ defined in Eq.(\ref{eqa:tree-edge-join-tree-lattice}), we may consider the following problems:
\begin{asparaenum}[\textrm{\textbf{Ext}}-1. ]
\item \textbf{Find} a tree $H^*\in \textbf{\textrm{L}}_{\textrm{tree}}(\textbf{\textrm{E}}[\ominus ]Z^0\textbf{\textrm{T}})$ has its own diameter holding $D(H^*)\leq D(H)$ for each tree $H\in \textbf{\textrm{L}}_{\textrm{tree}}(\textbf{\textrm{E}}[\ominus ]Z^0\textbf{\textrm{T}})$.
\item \textbf{Find} a tree $H^*\in \textbf{\textrm{L}}_{\textrm{tree}}(\textbf{\textrm{E}}[\ominus ]Z^0\textbf{\textrm{T}})$ has its own maximum degree holding $\Delta(H^*)\leq \Delta(H)$ for each tree $H\in \textbf{\textrm{L}}_{\textrm{tree}}(\textbf{\textrm{E}}[\ominus ]Z^0\textbf{\textrm{T}})$.
\item \textbf{Find} a subset $S_{\varepsilon}\subset \textbf{\textrm{L}}_{\textrm{tree}}(\textbf{\textrm{E}}[\ominus ]Z^0\textbf{\textrm{T}})$, which contains all $\varepsilon$ in $\textbf{\textrm{L}}_{\textrm{tree}}(\textbf{\textrm{E}}[\ominus ]Z^0\textbf{\textrm{T}})$ for $\varepsilon \in \{$caterpillars, lobsters, spiders$\}$.
\item \textbf{Find} a tree $H^*\in \textbf{\textrm{L}}_{\textrm{tree}}(\textbf{\textrm{E}}[\ominus ]Z^0\textbf{\textrm{T}})$ has its own leaf number $n_1(H^*)$ holding $n_1(H^*)\leq n_1(H)$ for the leaf number $n_1(H)$ of each tree $H\in \textbf{\textrm{L}}_{\textrm{tree}}(\textbf{\textrm{E}}[\ominus ]Z^0\textbf{\textrm{T}})$.
\item \textbf{Estimate} the cardinality of an all-tree-graphic lattice $\textbf{\textrm{L}}_{\textrm{tree}}(\textbf{\textrm{E}}[\ominus ]Z^0\textbf{\textrm{T}})$.
\item \textbf{Find} connection between a leaf number $l_e(G)$ and the leaf numbers $l_e(T_i)$ in the uncolored tree base $\textbf{\textrm{T}}$ for $G\in \textbf{\textrm{L}}_{\textrm{tree}}(\textbf{\textrm{E}}[\ominus ]Z^0\textbf{\textrm{T}})$.
\item For the subset $S_{\textrm{matching}}\subset \textbf{\textrm{L}}_{\textrm{tree}}(\textbf{\textrm{E}}[\ominus ]Z^0\textbf{\textrm{T}})$, in which each tree has a perfect matching. Does each tree $J\in S_{\textrm{matching}}$ admits one of strongly $W$-constraint $(k,d)$-total colorings?
\end{asparaenum}
\end{problem}

\begin{problem}\label{problem:all-tree-graphic-lattice-22-33}
Consider problems shown in Problem \ref{problem:all-tree-graphic-lattice-11} about other two all-tree-graphic lattices $\textbf{\textrm{L}}_{\textrm{tree}}([\odot]Z^0\textbf{\textrm{T}})$ and $\textbf{\textrm{L}}_{\textrm{tree}}([\ominus \odot]Z^0\textbf{\textrm{T}})$.
\end{problem}

\begin{problem} \label{problem:kd-w-tupe-colorings-add-leaves}
In general, we have a colored all-tree-graphic lattice $\textbf{\textrm{L}}_{\textrm{tree}}([\bullet]Z^0\textbf{\textrm{T}})$ based on a graph operation ``$[\bullet]$'' and each tree $T\in \textbf{\textrm{L}}_{\textrm{tree}}([\bullet]Z^0\textbf{\textrm{T}})$ admits a $W$-constraint $(k,d)$-total coloring, where ``$[\bullet]$'' is one of ``$[\ominus]$'', ``$[\odot]$'' and ``$[\ominus \odot]$''. It may be good to consider the following problems:
\begin{asparaenum}[\textrm{\textbf{Prob}}-1. ]
\item \textbf{Estimate} the cardinality of an all-tree-graphic lattice $\textbf{\textrm{L}}_{\textrm{tree}}([\bullet]Z^0\textbf{\textrm{T}})$.
\item \textbf{Find} a tree $H^*\in \textbf{\textrm{L}}_{\textrm{tree}}([\bullet]Z^0\textbf{\textrm{T}})$ admitting a $W$-constraint $(k,d)$-total coloring $f$ having its own number $|f(V(H^*))|$ of different colors colored to the vertices of the tree $H^*$ under the graceful $(k,d)$-total coloring $f$ to hold $|f(V(H^*))|\leq |h(V(H))|$ for each tree $H\in \textbf{\textrm{L}}_{\textrm{tree}}([\bullet]Z^0\textbf{\textrm{T}})$ admitting a $W$-constraint $(k,d)$-total coloring $h$.
\item \textbf{Same lattice base}. For two given connected graphs $G_i$ and $G_j$, if there exists a base $H=(H_k)^m_{k=1}$ and a graph operation ``$[\bullet]$'', such that $G_i=[\bullet]^m_{k=1}a_kH_k$ and $G_j=[\bullet] ^m_{k=1}b_kH_k$, we say that both connected graphs $G_i$ and $G_j$ obey the \emph{same lattice base}. Here, we are interesting on $|V(H_k)|\geq 2$, or some $H_i$ is not a star, that is $H_i\neq K_{1,s}$.
\item If both connected graphs $G_i$ and $G_j$ contain a proper subgraph $G^*$, we contract $G^*$ into a vertex $w^*$, so we get two connected graphs $G_i\odot G^*$ and $G_j\odot G^*$, such that $G_i=(G_i\odot G^*)\odot^{-1} w^*$ and $G_j=(G_j\odot G^*)\odot^{-1} w^*$. If both connected graphs $G_i\odot G^*$ and $G_j\odot G^*$ obey the same lattice base for any proper subgraph $G^*$ holding $G^*\subset G_i$ and $G^*\subset G_j$, \textbf{can} we claim both connected graphs $G_i$ and $G_j$ obey the same lattice base?
\end{asparaenum}
\end{problem}

\subsubsection{Tree-graphic lattice homomorphisms}

\begin{defn}\label{defn:definition-graph-homomorphism}
\cite{Bondy-2008} A \emph{graph homomorphism} $G\rightarrow H$ from a graph $G$ into another graph $H$ is a mapping $f: V(G) \rightarrow V(H)$ such that $f(u)f(v)\in E(H)$ for each edge $uv \in E(G)$. \qqed
\end{defn}

\textbf{Totally-colored graph homomorphisms.} The authors in \cite{Bing-Yao-Hongyu-Wang-arXiv-2020-homomorphisms} propose a new type of graph homomorphisms by combining graph homomorphisms and graph total colorings together as follows:

\begin{defn}\label{defn:gracefully-graph-homomorphism}
\cite{Bing-Yao-Hongyu-Wang-arXiv-2020-homomorphisms} Let $G\rightarrow H$ be a graph homomorphism from a $(p,q)$-graph $G$ to another $(p\,',q\,')$-graph $H$ based on a mapping $\alpha: V(G) \rightarrow V(H)$ such that $\alpha(u)\alpha(v)\in E(H)$ for each edge $uv \in E(G)$. The graph $G$ admits a total coloring $f$, the graph $H$ admits a total coloring $g$. Write the edge color set $f(E(G))=\{f(uv):uv \in E(G)\}$ and the edge color set $g(E(H))=\{g(\alpha(u)\alpha(v)):\alpha(u)\alpha(v)\in E(H)\}$, there are the following conditions:
\begin{asparaenum}[(\textrm{C}-1) ]
\item \label{bipartite} $V(G)=X\cup Y$, each edge $uv \in E(G)$ holds $u\in X$ and $v\in Y$ true. $V(H)=W\cup Z$, each edge $\alpha(u)\alpha(v)\in E(G)$ holds $\alpha(u)\in W$ and $\alpha(v)\in Z$ true.
\item \label{edge-difference} $f(uv)=|f(u)-f(v)|$ for each edge $uv \in E(G)$, $g(\alpha(u)\alpha(v))=|g(\alpha(u))-g(\alpha(v))|$ for each $\alpha(u)\alpha(v)\in E(H)$.
\item \label{edge-homomorphism} $f(uv)=g(\alpha(u)\alpha(v))$ for each edge $uv \in E(G)$.
\item \label{vertex-color-set} $f(x)\in [1,q+1]$ for $x\in V(G)$, and $g(y)\in [1,q\,'+1]$ for $y\in V(H)$.
\item \label{odd-vertex-color-set} $f(x)\in [1,2q+2]$ for $x\in V(G)$, and $g(y)\in [1,2q\,'+2]$ for $y\in V(H)$.
\item \label{grace-color-set} $[1,q]=f(E(G))=g(E(H))=[1,q\,']$.
\item \label{odd-grace-color-set} $[1,2q-1]=f(E(G))=g(E(H))=[1,2q\,'-1]$.
\item \label{set-ordered} Set-ordered property: $\max f(X)<\min f(Y)$ and $\max g(W)<\min g(Z)$.
\end{asparaenum}

We say $G\rightarrow H$ to be:

(i) a \emph{bipartitely graph homomorphism} if (C-\ref{bipartite}) holds true;

(ii) a \emph{gracefully graph homomorphism} if (C-\ref{edge-difference}), (C-\ref{edge-homomorphism}), (C-\ref{vertex-color-set}) and (C-\ref{grace-color-set}) hold true;

(iii) a \emph{set-ordered gracefully graph homomorphism} if (C-\ref{edge-difference}), (C-\ref{edge-homomorphism}), (C-\ref{vertex-color-set}), (C-\ref{grace-color-set}) and (C-\ref{set-ordered}) hold true;

(iv) an \emph{odd-gracefully graph homomorphism} if (C-\ref{edge-difference}), (C-\ref{edge-homomorphism}), (C-\ref{odd-vertex-color-set}) and (C-\ref{odd-grace-color-set}) hold true;

(v) a \emph{set-ordered odd-gracefully graph homomorphism} if (C-\ref{edge-difference}), (C-\ref{edge-homomorphism}), (C-\ref{odd-vertex-color-set}), (C-\ref{odd-grace-color-set}) and (C-\ref{set-ordered}) hold true.\qqed
\end{defn}

A colored tree $T$ admitting a graceful $(k,d)$-total coloring $f$ shown in Fig.\ref{fig:anti-oder} admits a graph homomorphism to a colored graph $G$ admitting a graceful $(k,d)$-total coloring $g$ shown in Fig.\ref{fig:homomorphism-1} by a mapping $\varphi:V(T)\rightarrow V(G)$ such that $u=\varphi (x)$ for $x\in V(T)$ and $u\in V(G)$ if $f(x)=g(u)$, and $\varphi (x)\varphi (y)\in E(G)$ for each edge $xy\in E(T)$.

\begin{figure}[h]
\centering
\includegraphics[width=15cm]{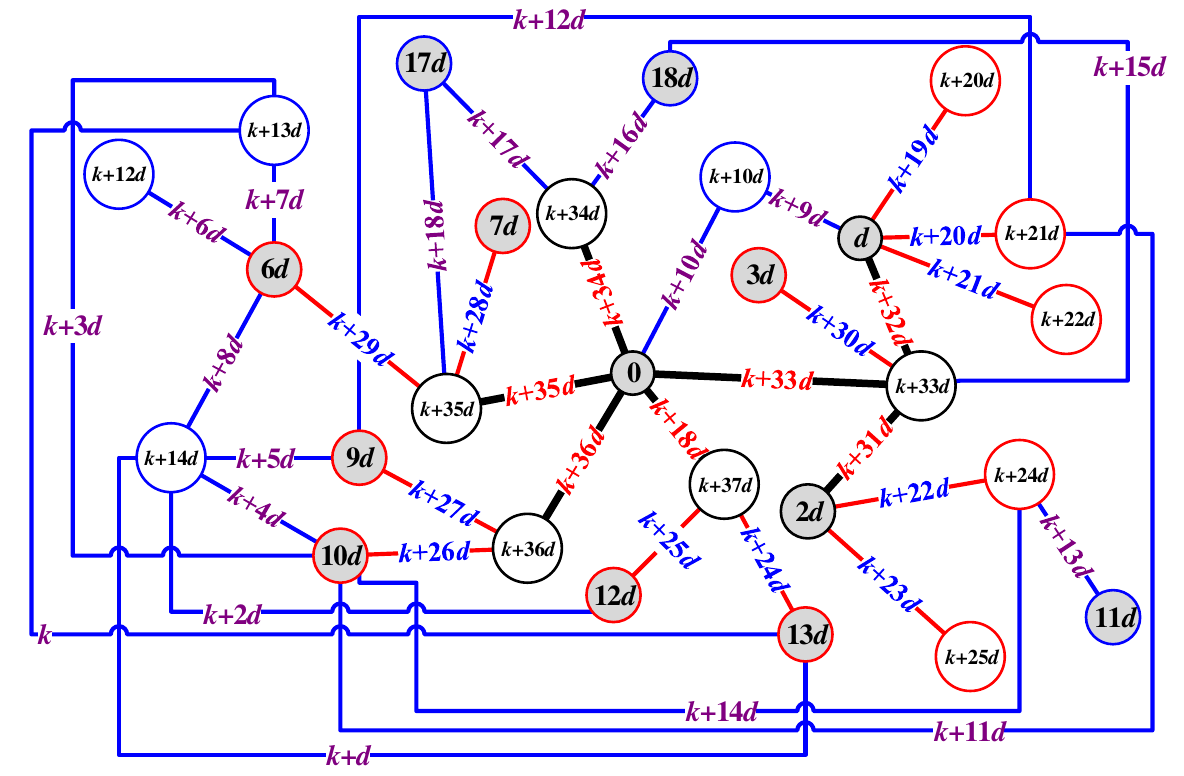}
\caption{\label{fig:homomorphism-1}{\small A graph $G$ obtained from a tree $T$ shown in Fig.\ref{fig:anti-oder} by vertex-coinciding vertices colored with the same color into one.}}
\end{figure}

Let $\textbf{\textrm{L}}^c_{\textrm{tree}}$ be the set of colored trees, in which each tree admits ten $(k,d)$-total colorings mentioned in Theorem \ref{thm:coloring-closure-kd-total-colorings}. For a colored tree $T\in \textbf{\textrm{L}}^c_{\textrm{tree}}$ we vertex-coincide some vertices colored the same color into one vertex, the resultant graph is denoted as $G$, then we get a graph homomorphism $T\rightarrow G$. All graphs $G$ holding $T\rightarrow G$ for a tree $T$ form a graph homomorphism set $S\langle T\rightarrow G\rangle$, here, $S\langle T\rightarrow G\rangle$ is related with all colorings admitted by the tree $T$.

We have a set $S_{ub}(T,\textrm{Theorem \ref{thm:coloring-closure-kd-total-colorings}})$ contains all colored trees admitting ten $(k,d)$-total colorings mentioned in Theorem \ref{thm:coloring-closure-kd-total-colorings}, such that each colored tree $H\in S_{ub}(T,\textrm{Theorem \ref{thm:coloring-closure-kd-total-colorings}})$ is isomorphic to $T$, that is $H\cong T$. Moreover, we have a graph set $S_{ub}\langle H\rightarrow G\rangle$ for $H\in S_{ub}(T,\textrm{Theorem \ref{thm:coloring-closure-kd-total-colorings}})$ and $H\rightarrow G$, so each tree $H\in S_{ub}(T,\textrm{Theorem \ref{thm:coloring-closure-kd-total-colorings}})$ is a \emph{public-key}, and each graph $G\in S_{ub}\langle H\rightarrow G\rangle$ is a \emph{private-key}.

Similarly, we have three colored all-tree-graphic lattices $\textbf{\textrm{L}}^c_{\textrm{tree}}(\textbf{\textrm{E}}[\ominus ]Z^0\textbf{\textrm{T}})$, $\textbf{\textrm{L}}^c_{\textrm{tree}}([\odot]Z^0\textbf{\textrm{T}})$ and $\textbf{\textrm{L}}^c_{\textrm{tree}}([\ominus \odot]Z^0\textbf{\textrm{T}})$, such that each tree in these three colored all-tree-graphic lattices admits every one of ten $(k,d)$-total colorings mentioned in Theorem \ref{thm:coloring-closure-kd-total-colorings} for producing topological codings in asymmetric cryptography.

Each connected graph $G$ in the edge-join graphic lattice $\textbf{\textrm{L}}(\textbf{\textrm{E}}[\ominus ]Z^0\textbf{\textrm{H}})$, the vertex-coinciding graphic lattice $\textbf{\textrm{L}}([\odot]Z^0\textbf{\textrm{H}})$ and the operation-mixed graphic lattice $\textbf{\textrm{L}}([\ominus \odot]Z^0\textbf{\textrm{H}})$ can be vertex-split into a tree $T$, so we have a graph homomorphism $T\rightarrow G$.

\subsection{Graphic groups based on parameterized total colorings}

\begin{defn} \label{defn:matching-twin-kd-labeling}
\cite{Yao-Zhang-Sun-Mu-Sun-Wang-Wang-Ma-Su-Yang-Yang-Zhang-2018arXiv} Suppose that a $(p,q)$-graph $G$ admits a $W$-constraint $(k,d)$-total labeling $f$, and another $(p\,',q\,')$-graph $H$ admits another $W$-constraint $(k,d)$-total labeling $g$. If
\begin{equation}\label{eqa:matching-twin-kd-labelingss}
(S_{q-1,0,0,d}\cup S_{q-1,k,0,d})\setminus f(V(G)\cup E(G))=g(V(H)\cup E(H))
\end{equation} then we call $g$ \emph{complementary $(k,d)$-labeling} of $f$, and both $f$ and $g$ are \emph{twin $W$-constraint $(k,d)$-total labelings}.\qqed
\end{defn}

\begin{rem}\label{rem:333333}
By Eq.(\ref{eqa:matching-twin-kd-labelingss}) defined in Definition \ref{defn:matching-twin-kd-labeling}, if two graphs $G$ and $H$ are bipartite, and \emph{twin parameterized Topcode-matrices} $P_{ara}(G,f)$ and $P_{ara}(H,g)$ have the same order, then we get a parameterized Topcode-matrix $(S_X,S_E,S_Y)^T_{(k,d)}$ defined by
\begin{equation}\label{eqa:555555}
(S_X,S_E,S_Y)^T_{(k,d)}=P_{ara}(G,F)+P_{ara}(H,F^*)
\end{equation} holding $S_X\cup S_E\cup S_Y=S_{q-1,0,0,d}\cup S_{q-1,k,0,d}$ true, where $P_{ara}(G,F)=k\cdot I\,^0+d\cdot T_{code}(G,f)$ and $P_{ara}(H,F^*)=k\cdot I\,^0+d\cdot T_{code}(H,g)$ defined in Definition \ref{defn:definition-parameterized-topcode-matrix}.

The \emph{twin $(k,d)$-labeling every-zero graphic groups} were introduced in \cite{Yao-Sun-Su-Wang-Zhao-ICIBA-2020}.\paralled
\end{rem}

\begin{defn}\label{defn:77-two-graphic-set-groups}
\cite{Yao-Wang-2106-15254v1} Let $F_m(T_{code})=\{T^1_{code},T^2_{code},\dots, T^m_{code}\}$ be a set of Topcode-matrices with each Topcode-matrix $T^i_{code}=(X_i,E_i,Y_i)^{T}$, where the \emph{v-vector} $X_i=(x_{i,1}, x_{i,2}, \dots, x_{i,q})$, the \emph{e-vector} $E_i=(e_{i,1}, e_{i,2}, \dots$, $e_{i,q})$ and the \emph{v-vector} $Y_i=(y_{i,1}, y_{i,2}, \dots ,y_{i,q})$ and there are functions $f_i$ holding $e_{i,r}=f_i(x_{i,r},y_{i,r})$ with $i\in [1,m]$ and $r\in [1,q]$.

(i) If, for a fixed positive integer $k$, there exists a constant $M$, such that
\begin{equation}\label{eqa:77-group-operation-1}
x_{i,r}+x_{j,r}-x_{k,r}~(\bmod~ M)=x_{\lambda,r}\in X_{\lambda},~y_{i,r}+y_{j,r}-y_{k,r}~(\bmod~ M)=y_{\lambda,r}\in Y_{\lambda}
\end{equation}
where $\lambda=i+j-k~(\bmod~ M)\in [1,m]$ and, $T^{\lambda}_{code}=(X_{\lambda}$, $E_{\lambda}$, $Y_{\lambda})^{T}$ is a Topcode-matrix with $e_{\lambda,r}=f(x_{\lambda,r},y_{\lambda,r})$ for $r\in [1,q]$. Then we say Eq.(\ref{eqa:77-group-operation-1}) to be an \emph{additive v-operation} on the set $F_m(T_{code})$, and write this operation by ``$\oplus$'', so $T^i_{code}\oplus _k T^j_{code}=T^{\lambda}_{code}$ under any \emph{preappointed zero} $T^k_{code}$. Then we call $F_m(T_{code})$ an \emph{every-zero Topcode$^+$-matrix group}, denoted as $\{F_m(T_{code});\oplus\}$.

(ii) If, for a fixed Topcode-matrix $T^k_{code}\in F_m$, there exists a constant $M$, such that
\begin{equation}\label{eqa:77-subtraction-group-operation-1}
x_{i,r}-x_{j,r}+x_{k,r}~(\bmod~ M)=x_{\mu,r}\in X_{\mu},~y_{i,r}-y_{j,r}+y_{k,r}~(\bmod~ M)=y_{\mu,r}\in Y_{\mu}
\end{equation}
where $\mu=i-j+k~(\bmod~ M)\in [1,m]$ and, $T^{\mu}_{code}=(X_{\mu},E_{\mu},Y_{\mu})^{T}\in F_m$, where $e_{\mu,s}=f(x_{\mu,s},y_{\mu,s})$ for $s\in [1,q]$. Then the equation (\ref{eqa:77-subtraction-group-operation-1}) defines a new operation on the set $F_m(T_{code})$, called the \emph{subtractive v-operation}, denoted as $T^i_{code}\ominus_k T^j_{code}=T^{\mu}_{code}$ under any \emph{preappointed zero} $T^k_{code}$. Thereby, we call $F_m(T_{code})$ an \emph{every-zero Topcode$^-$-matrix group}, denoted as $\{F_m(T_{code});\ominus\}$.\qqed
\end{defn}

\begin{example}\label{exa:8888888888}
A Topcode-matrix set $T_{code}(S)=\{T^1_{code},T^2_{code},T^3_{code},T^4_{code}\}$ shown in Eq.(\ref{eqa:stars-Cayley-formula-matrices})

\begin{equation}\label{eqa:stars-Cayley-formula-matrices}
\centering
T^1_{code}= \left(
\begin{array}{ccccc}
1 & 1 & 1\\
3 & 4 & 5\\
2 & 3 & 4
\end{array}
\right)~T^2_{code}= \left(
\begin{array}{ccccc}
2 & 2 & 2\\
5 & 6 & 3\\
3 & 4 & 1
\end{array}
\right)~T^3_{code}= \left(
\begin{array}{ccccc}
3 & 3 & 3\\
7 & 4 & 5\\
4 & 1 & 2
\end{array}
\right)~T^4_{code}= \left(
\begin{array}{ccccc}
4 & 4 & 4\\
5 & 6 & 7\\
1 & 2 & 3
\end{array}
\right)
\end{equation}

\noindent forms an \emph{every-zero Topcode$^+$-matrix group} under $\bmod~ 4$ and an \emph{every-zero Topcode$^-$-matrix group} under $\bmod~ 4$, respectively, according to Definition \ref{defn:77-two-graphic-set-groups}.\qqed
\end{example}

\begin{defn} \label{defn:88-two-graphic-set-groups}
\cite{Yao-Wang-2106-15254v1} \textbf{Two graphic-set groups.} Since each Topcode-matrix $T^i_{code}\in F_m(T_{code})$ defined in Definition \ref{defn:77-two-graphic-set-groups} is accompanied by a graph set $G_{rap}(T^i_{code})$ such that every graph $G\in G_{rap}(T^i_{code})$ has its own Topcode-matrix $T_{code}(G)=T^i_{code}$, so the set $F_m(G_{rap}(T_{code}))=\{G_{rap}(T^1_{code})$, $G_{rap}(T^2_{code})$, $\dots $, $G_{rap}(T^m_{code})\}$, also, forms an \emph{every-zero graphic-set group} $\{F_m(G_{rap}(T_{code}));\oplus\}$ based on the operation ``$\oplus$'' defined in Eq.(\ref{eqa:77-group-operation-1}) in Definition \ref{defn:77-two-graphic-set-groups}, that is
\begin{equation}\label{eqa:555555}
G_{rap}(T^i_{code})\oplus G_{rap}(T^j_{code})=G_{rap}(T^{\lambda}_{code})\in F_m(G_{rap}(T_{code}))
\end{equation}
with $\lambda=i+j-k~(\bmod~M)$ for any \emph{preappointed zero} $G_{rap}(T^k_{code})$. And moreover we have another every-zero graphic-set group $\{F_m(G_{rap}(T_{code}));\ominus\}$ based on the operation ``$\ominus$'' defined in Eq.(\ref{eqa:77-subtraction-group-operation-1}) in Definition \ref{defn:77-two-graphic-set-groups}, so that
\begin{equation}\label{eqa:555555}
G_{rap}(T^i_{code})\ominus G_{rap}(T^j_{code})=G_{rap}(T^{\mu}_{code})\in F_m(G_{rap}(T_{code}))
\end{equation}
with $\mu=i-j+k~(\bmod~M)$ for any \emph{preappointed zero} $G_{rap}(T^k_{code})$.\qqed
\end{defn}

\begin{defn} \label{defn:111111}
$^*$ Let a Topcode-matrix set $F_m(T_{code})=\{T^1_{code},T^2_{code},\dots, T^m_{code}\}$ be defined in Definition \ref{defn:77-two-graphic-set-groups}. There are parameterized Topcode-matrices
\begin{equation}\label{eqa:333333}
P_{ara}(i)=k\cdot I\,^0+d\cdot T^i_{code},~i\in [1,m]
\end{equation} where the Topcode-matrices $I\,^0$, $T^i_{code}$ and $P_{ara}(i)$ for $i\in [1,m]$ have the same order. The following set
\begin{equation}\label{eqa:parameterized-topcode-matrices-group}
P(F_m(T_{code})_{(k,d)})=\{P_{ara}(i):~i\in [1,m]\}
\end{equation} forms an \emph{parameterized every-zero Topcode$^+$-matrix group} $\{P(F_m(T_{code})_{(k,d)});\oplus \}$ under the additive v-operation ``$\oplus$'' defined in Definition \ref{defn:77-two-graphic-set-groups}, also, by the subtractive v-operation ``$\ominus$'' defined in Definition \ref{defn:77-two-graphic-set-groups}, the set $P(F_m(T_{code})_{(k,d)})$ forms an \emph{parameterized every-zero Topcode$^-$-matrix group} $\{P(F_m(T_{code})_{(k,d)});\ominus \}$.\qqed
\end{defn}

\begin{defn} \label{defn:111111}
$^*$ By the every-zero Topcode$^+$-matrix group $\{F_m(T_{code});\oplus\}$ and the every-zero Topcode$^-$-matrix group $\{F_m(T_{code});\ominus\}$ defined in Definition \ref{defn:77-two-graphic-set-groups}, we use each parameterized Topcode-matrix $P_{ara}(i)\in P(F_m(T_{code})_{(k,d)})$ defined in Eq.(\ref{eqa:parameterized-topcode-matrices-group}) to produce a parameterized number-based string $s_{i,k,d}=c_{i,1}c_{i,2}\cdots c_{i,N_0}$ under a fixed PNBS-algorithm for $i\in [1,m]$. Let $S_{(k,d)}=\{s_{i,k,d}:i\in [1,m]\}$. Then we get the following two parameterized number-based string groups:

(i) Any pair of parameterized number-based strings $s_{i,k,d}$ and $s_{j,k,d}$ of $S_{(k,d)}$ in the \emph{parameterized every-zero number-based string$^+$ group} $\{S_{(k,d)};\oplus \}$ hold: There exists a constant $M$, the additive PNBS-operation $s_{i,k,d}\oplus _r s_{j,k,d}=s_{\lambda,k,d}$ is defined as
\begin{equation}\label{eqa:parameterized number-based-strings-1}
c_{i,t}+c_{j,t}-c_{r,t}~(\bmod~ M)=c_{\lambda,t}\in s_{\lambda,k,d},\quad t\in [1,N_0]
\end{equation} under any \emph{preappointed zero} $s_{r,k,d}\in S_{(k,d)}$, where $\lambda=i+j-r~(\bmod~ M)\in [1,m]$;

(ii) Any pair of parameterized number-based strings $s_{i,k,d}$ and $s_{j,k,d}$ of $S_{(k,d)}$ in the \emph{parameterized every-zero number-based string$^-$ group} $\{S_{(k,d)};\ominus \}$ hold: There exists a constant $M$, the subtractive PNBS-operation $s_{i,k,d}\ominus _r s_{j,k,d}=s_{\mu,k,d}$ is defined as
\begin{equation}\label{eqa:parameterized number-based-strings-1}
c_{i,t}-c_{j,t}+c_{r,t}~(\bmod~ M)=c_{\mu,t}\in s_{\mu,k,d},\quad t\in [1,N_0]
\end{equation} under any \emph{preappointed zero} $s_{r,k,d}\in S_{(k,d)}$, where $\mu=i-j+r~(\bmod~ M)\in [1,m]$.\qqed
\end{defn}

\subsection{Various parameterized colorings}

There are several set-colorings and set-labelings introduced in \cite{Yao-Ma-arXiv-2201-13354v1}.

\begin{defn}\label{defn:graceful-intersection}
\cite{Yao-Zhang-Sun-Mu-Sun-Wang-Wang-Ma-Su-Yang-Yang-Zhang-2018arXiv} A $(p,q)$-graph $G$ admits a set-labeling $F:V(G)\rightarrow [1,q]^2$~(resp. $[1,2q-1]^2)$, and induces an edge set-color $F(uv)=F(u)\cap F(v)$ for each edge $uv \in E(G)$. If we can select a \emph{representative} $a_{uv}\in F(uv)$ for each edge color set $F(uv)$ such that $\{a_{uv}:~uv\in E(G)\}=[1,q]$ (resp. $[1,2q-1]^o$), then $F$ is called a \emph{graceful-intersection (resp. an odd-graceful-intersection) total set-labeling} of $G$.\qqed
\end{defn}

\begin{thm}\label{thm:graceful-total-set-labelings}
\cite{Yao-Zhang-Sun-Mu-Sun-Wang-Wang-Ma-Su-Yang-Yang-Zhang-2018arXiv} Each tree $T$ admits a \emph{graceful-intersection (resp. an odd-graceful-intersection) total set-labeling}.
\end{thm}

\begin{thm}\label{thm:rainbow-total-set-labelings}
\cite{Yao-Zhang-Sun-Mu-Sun-Wang-Wang-Ma-Su-Yang-Yang-Zhang-2018arXiv} Each tree $T$ of $q$ edges admits a \emph{regular rainbow intersection total set-labeling} based on a \emph{regular rainbow set-sequence} $\{[1,k]\}^{q}_{k=1}$.
\end{thm}

\begin{figure}[h]
\centering
\includegraphics[width=16cm]{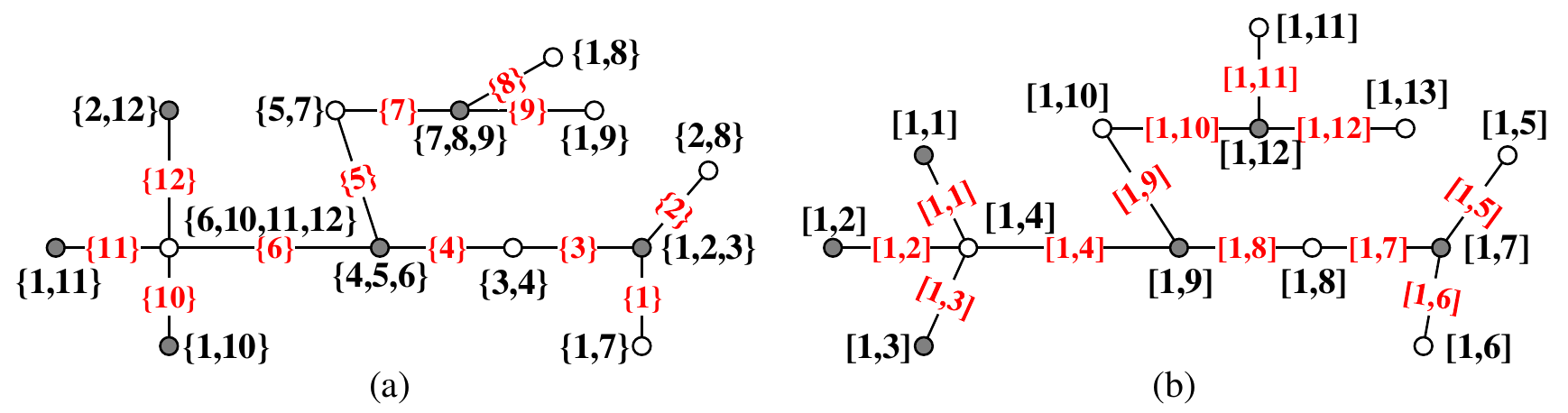}\\
\caption{\label{fig:0-graceful-intersection}{\small For illustrating Theorem \ref{thm:graceful-total-set-labelings} and Theorem \ref{thm:rainbow-total-set-labelings}: Left tree (a) admits a graceful-intersection total set-labeling, and Right tree (b) admits a regular rainbow intersection total set-labeling, cited from \cite{Yao-Zhang-Sun-Mu-Sun-Wang-Wang-Ma-Su-Yang-Yang-Zhang-2018arXiv}.}}
\end{figure}

The authors, in \cite{Yao-Ma-arXiv-2201-13354v1}, have tried to find some easy and effective techniques based on graph theory for practical application, and they have used intersected-graphs admitting set-colorings defined on hyperedge sets to observe the topological structures of hypergraphs, hypergraph-type Topcode-matrix. And moreover, the hypergraph's connectivity, colorings of hypergraphs, hypergraph homomorphism, hypernetworks, scale-free network generator, compound hypergraphs having their intersected-graphs with vertices to be hypergraphs for high-dimensional extension diagram have been investigated preliminarily.

\subsubsection{Parameterized total set-colorings}

\begin{defn} \label{defn:111111}
$^*$ For a subgraph $G\subset K_{m,n}$, since the complete bipartite graph $K_{m,n}$ admits a graceful $(k,d)$-total coloring $f$ defined in Example \ref{exa:bipartite-graph-graceful-kd-total}, also, $f$ is a $(k,d)$-total coloring of $G$, such that $f(E(G))\subset f(E(K_{m,n}))$, we call $f$ a \emph{fragmentary graceful $(k,d)$-total coloring}.\qqed
\end{defn}

\begin{thm}\label{thm:666666}
$^*$ By Theorem \ref{thm:equivalent-k-d-total-colorings}, the complete bipartite graph $K_{m,n}$ admits each one of $W$-constraint $(k,d)$-total colorings, where $W$-constraint $\in \{$harmonious, edge-difference, edge-magic, felicitous-difference, graceful-difference$\}$. Then each bipartite graph admits each one of fragmentary $W\,'$-constraint $(k,d)$-total colorings for $W\,'$-constraint $\in \{$graceful, harmonious, edge-difference, edge-magic, felicitous-difference, graceful-difference$\}$.
\end{thm}

So we ask for: \emph{Every bipartite graph admits a graceful $(k,d)$-total coloring and a harmonious $(k,d)$-total coloring}.

\begin{defn} \label{defn:44-magic-kd-set-coloring}
$^*$ Suppose that a $(p,q)$-graph $G$ admits a \emph{$W$-constraint $(k,d)$-total set-coloring} $F:V(G)\cup E(G)\rightarrow U^2$, where $U=S_{m,0,0,d}\cup S_{q-1,k,0,d}$. For each edge $uv\in E(G)$, if there is a number $c\in F(uv)$ corresponding to some $a\in F(u)$ and $b\in F(v)$ such that

(i) $c+|b-a|=k_1$, then $W$-constraint = edge-difference, and we call $F$ \emph{edge-difference $(k,d)$-total set-coloring}.

(ii) $a+c+b=k_2$, then $W$-constraint = edge-magic, and we call $F$ \emph{edge-magic $(k,d)$-total set-coloring}.

(iii) $|a+b-c|=k_3$, then $W$-constraint = felicitous-difference, and we call $F$ \emph{felicitous-difference $(k,d)$-total set-coloring}.

(iv) $\big ||a-b|-c\big |=k_4$, then $W$-constraint = graceful-difference, and we call $F$ \emph{graceful-difference $(k,d)$-total set-coloring}.\qqed
\end{defn}

\begin{defn} \label{defn:graceful-kd-set-coloring}
$^*$ Suppose that a $(p,q)$-graph $G$ admits a \emph{set-coloring} $F:V(G)\cup E(G)\rightarrow U^2$, where $U=S_{m,0,0,d}\cup S_{q-1,k,0,d}$. If there are $c_i\in F(u_iv_i)$ for each edge $u_iv_i\in E(G)$, such that each $c_i$ corresponds to some $a_i\in F(u_i)$ and $b_i\in F(v_i)$ holding $c_i=|a_i-b_i|$ with $i\in [1,q]$, and $\{c_1,c_2,\dots, c_q\}=S_{q-1,k,0,d}$, then we call $F$ \emph{graceful $(k,d)$-total set-coloring} of $G$.\qqed
\end{defn}

\begin{defn} \label{defn:harmonious-kd-set-coloring}
$^*$ Suppose that a $(p,q)$-graph $G$ admits a \emph{set-coloring} $F:V(G)\cup E(G)\rightarrow U^2$, where $U=S_{m,0,0,d}\cup S_{q-1,k,0,d}$. If there are $c_i\in F(u_iv_i)$ for each edge $u_iv_i\in E(G)$, such that each $c_i$ corresponds to some $a_i\in F(u_i)$ and $b_i\in F(v_i)$ holding $c_i=a_i+b_i~(\bmod^*qd)$ defined by $c_i-k=[a_i+b_i-k](\bmod ~qd)$ for each edge $u_iv_i\in E(G)$, and $\{c_1,c_2,\dots, c_q\}=S_{q-1,k,0,d}$, then $F$ is called a \emph{harmonious $(k,d)$-total set-coloring} of $G$.\qqed
\end{defn}

\begin{thm}\label{thm:graph-admits-6-set-colorings}
$^*$ Each connected graph $G$ admits each one of the following $W$-constraint $(k,d)$-total set-colorings for $W$-constraint $\in \{$graceful, harmonious, edge-difference, edge-magic, felicitous-difference, graceful-difference$\}$.
\end{thm}
\begin{proof} Since a connected graph $G$ can be vertex-split into a tree $T$, such that we have a graph homomorphism $T\rightarrow G$ under a mapping $\theta:V(T)\rightarrow V(G)$, and each tree admits a graceful $(k,d)$-total coloring $g_1$, a harmonious $(k,d)$-total coloring $g_2$, an edge-difference $(k,d)$-total coloring $g_3$, an edge-magic $(k,d)$-total coloring $g_4$, a felicitous-difference $(k,d)$-total coloring $g_5$ and a graceful-difference $(k,d)$-total coloring $g_6$. Then the connected graph $G$ admits a $(k,d)$-total set-coloring $F$ defined by $F(u)=\{g_i(u^*):i\in [1,6],u^*\in V(T)\}$ for each vertex $u\in V(G)$ with $u=\theta(u^*)$ for $u^*\in V(T)$, and $F(uv)=\{g_i(u^*v^*):i\in [1,6],u^*v^*\in E(T)\}$ for each edge $uv\in E(G)$ with $uv=\theta(u^*)\theta(v^*)$ for each edge $u^*v^*\in E(T)$.

The proof of the theorem is complete.
\end{proof}

\subsubsection{Set-colorings with multiple intersections}

\begin{defn} \label{defn:W-join-type-se-total-coloring}
\cite{Yao-Ma-arXiv-2201-13354v1} Let $\mathcal{E}$ be a set of subsets of a finite set $\Lambda$ such that each hyperedge $e\in \mathcal{E}$ satisfies $e\neq \emptyset$ and corresponds to another hyperedge $e\,'\in \mathcal{E}$ holding $e\cap e\,'\neq \emptyset$, as well as $\Lambda=\bigcup _{e\in \mathcal{E}}e$. Suppose that a connected graph $G$ admits a \emph{total set-labeling} $\pi: V(G)\cup E(G)\rightarrow \mathcal{E}$ with $\pi(x)\neq \pi(y)$ for distinct vertices $x,y\in V(G)$, and the edge color set $\pi (uv)$ for each edge $uv\in E(G)$ holds $\pi (uv)\neq \pi (uw)$ for $v,w\in N(u)$ and each vertex $u\in V(G)$. There are the following \textrm{intersected-type restrictive conditions}:
\begin{asparaenum}[\textbf{\textrm{Chyper}}-1.]
\item \label{join:join-edge} $\pi (u)\cap \pi (v)\subseteq \pi (uv)$ and $\pi (u)\cap \pi (v)\neq \emptyset $ for each edge $uv\in E(G)$.
\item \label{join:join-edge-r-rank} $\pi (u)\cap \pi (v)\subseteq \pi (uv)$ and $|\pi (u)\cap \pi (v)|\geq r\geq 2$ for each edge $uv\in E(G)$.
\item \label{join:join-edge-vertex} $\pi (uv)\cap \pi (u)\neq \emptyset$ and $\pi (uv)\cap \pi (v)\neq \emptyset$ for each edge $uv\in E(G)$.
\item \label{join:join-dajacent-edges} $\pi (uv)\cap \pi (uw)\neq \emptyset$ for $v,w\in N(u)$ and each vertex $u\in V(G)$.
\item \label{join:dajacent-edges-no-join} $\pi (uv)\cap \pi (uw)=\emptyset$ for $v,w\in N(u)$ and each vertex $u\in V(G)$.
\end{asparaenum}
\textbf{Then we have:}
\begin{asparaenum}[\textbf{\textrm{Cgraph}}-1.]
\item If Chyper-\ref{join:join-edge} holds true, $G$ is called a \emph{subintersected-graph} and $\pi$ is called a \emph{subintersected total set-labeling} of $G$.
\item If Chyper-\ref{join:join-edge-r-rank} holds true, $G$ is called a \emph{$r$-rank subintersected-graph} and $\pi$ is called a \emph{$r$-rank subintersected total set-labeling} of $G$.
\item If Chyper-\ref{join:join-edge} and Chyper-\ref{join:join-edge-vertex} hold true, $G$ is called an \emph{intersected-edge-intersected graph} and $\pi$ is called an \emph{intersected-edge-intersected total set-labeling} of $G$.
\item If Chyper-\ref{join:join-edge-r-rank} and Chyper-\ref{join:join-edge-vertex} hold true, $G$ is called a \emph{$r$-rank intersected-edge-intersected graph} and $\pi$ is called a \emph{$r$-rank intersected-edge-intersected total set-labeling} of $G$.
\item $G$ is called an \emph{edge-intersected graph} if Chyper-\ref{join:join-edge-vertex} holds true, and $\pi$ is called an \emph{edge-intersected total set-labeling} of $G$.
\item If Chyper-\ref{join:join-edge-vertex} and Chyper-\ref{join:join-dajacent-edges} hold true, then $G$ is called an \emph{adjacent edge-intersected graph}, and $\pi$ is called an \emph{adjacent edge-intersected total set-labeling} of $G$.
\item If Chyper-\ref{join:join-edge-vertex} and Chyper-\ref{join:dajacent-edges-no-join} hold true, then $G$ is called an \emph{individual edge-intersected graph}, and $\pi$ is called an \emph{individual edge-intersected total set-labeling} of $G$.\qqed
\end{asparaenum}
\end{defn}

\begin{problem}\label{qeu:444444}
Since $\Lambda=\bigcup _{e\in \mathcal{E}}e$ for a hyperedge set $\mathcal{E}$ based on a finite set $\Lambda$, and a connected graph $G$ admits a \emph{$W$-constraint-intersected total set-labeling} $F$ defined in Definition \ref{defn:W-join-type-se-total-coloring}, \textbf{characterize} hyperedge sets $\mathcal{E}$ and estimate extremum number $\min \big \{\max \Lambda :\Lambda=\bigcup _{e\in \mathcal{E}}e \big \}$ when each $\Lambda$ is a finite nonnegative integer set.
\end{problem}

\begin{thm}\label{thm:five-edge-join-total-set-labelings}
\cite{Yao-Ma-arXiv-2201-13354v1} Each tree admits one of subintersected total set-labeling, intersected-edge-intersected total set-labeling, edge-intersected total set-labeling, adjacent edge-intersected total set-labeling and individual edge-intersected total set-labeling defined in Definition \ref{defn:W-join-type-se-total-coloring}.
\end{thm}

\subsubsection{Parameterized hypergraphs}

We will use the following terminology and natation:

\begin{asparaenum}[$\bullet$~]
\item For a \emph{parameterized set} $\Lambda_{(m,b,n,k,a,d)}=S_{m,0,b,d}\cup S_{n,k,a,d}$ with
$$
S_{m,0,b,d}=\{bd,(b+1)d\dots, md\},~S_{n,k,a,d}=\{k+ad,k+(a+1)d,\dots ,k+(a+n)d\}
$$ the power set $\Lambda^2_{(m,b,n,k,a,d)}$ collects all subsets of the parameterized set $\Lambda_{(m,b,n,k,a,d)}$.

\item A \emph{parameterized hyperedge set} $\mathcal{E}^P\subset \Lambda^2_{(m,b,n,k,a,d)}$ holds $\bigcup _{e\in \mathcal{E}^P}e=\Lambda_{(m,b,n,k,a,d)}$.
\end{asparaenum}

Recall the concept of a hypergraph as follows:

\begin{defn}\label{defn:hypergraph-basic-definition}
\cite{Jianfang-Wang-Hypergraphs-2008} A \emph{hypergraph} $\mathcal{E}$ is a family of distinct non-empty subsets $e_1,e_2$, $\dots $, $e_n$ of a finite set $\Lambda$, and satisfies:

(i) each element $e\in \mathcal{E}$ holds $e\neq \emptyset $ true, and is called a \emph{hyperedge};

(ii) $\Lambda=\bigcup _{e\in \mathcal{E}}e$, where each element of $\Lambda$ is called a \emph{vertex}, and the cardinality $|\Lambda|$ is the \emph{order} of the hypergraph.

The notation $\mathcal{H}_{yper}=(\Lambda,\mathcal{E})$ stands for a hypergraph with its own \emph{hyperedge set} $\mathcal{E}$ defined on the \emph{vertex set} $\Lambda$, and the cardinality $|\mathcal{E}|$ is the \emph{size} of the hypergraph $\mathcal{H}_{yper}$. \qqed
\end{defn}

\begin{figure}[h]
\centering
\includegraphics[width=15cm]{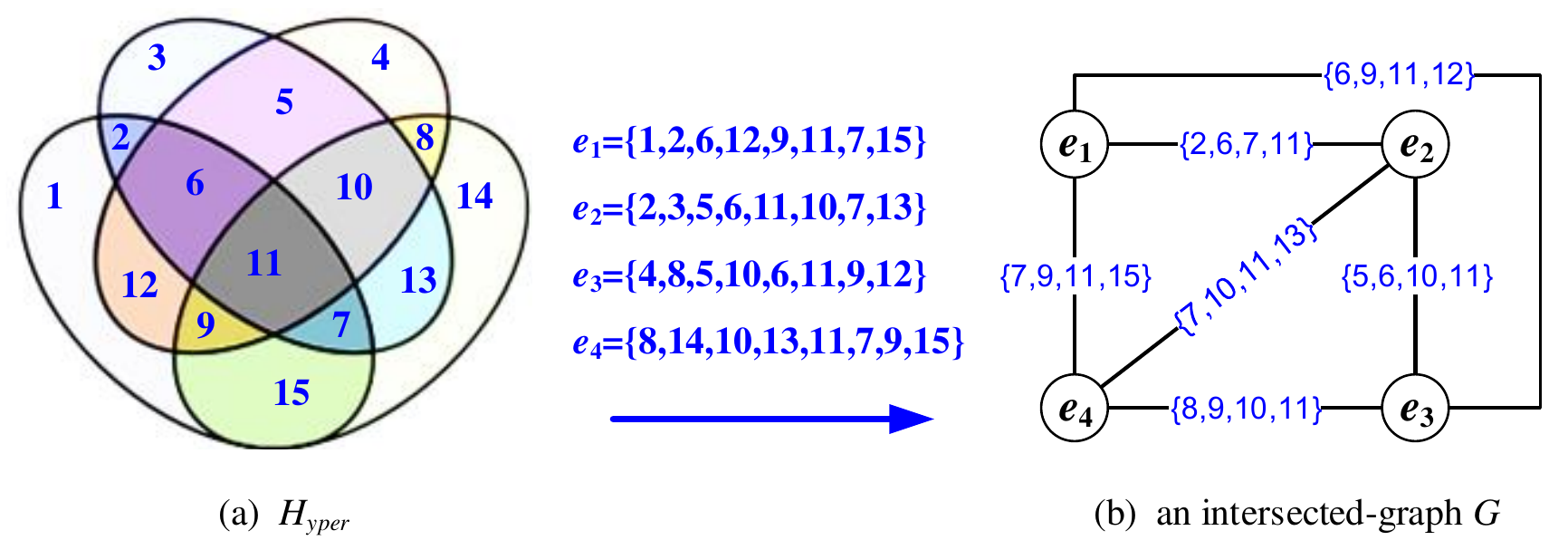}\\
\caption{\label{fig:1-example-hypergraph}{\small For understanding Definition \ref{defn:hypergraph-basic-definition}: An example from a $8$-uniform hypergraph $H_{yper}$ to an intersected-graph $G$ admitting a set-coloring, where (a) Venn's four-set diagram using four ellipses, cited from \cite{Yao-Zhang-Sun-Mu-Sun-Wang-Wang-Ma-Su-Yang-Yang-Zhang-2018arXiv}.}}
\end{figure}

Motivated from the hypergraph definition, we present the parameterized hypergraph as follows:

\begin{defn}\label{defn:parameterized-hypergraph-basic-definition}
$^*$ A \emph{parameterized hypergraph} $\mathcal{P}_{hyper}=(\Lambda_{(m,b,n,k,a,d)},\mathcal{E}^P)$ defined on a parameterized hypervertex set $\Lambda_{(m,b,n,k,a,d)}$ holds:

(i) each element of $\mathcal{E}^P$ is not empty and called a \emph{parameterized hyperedge};

(ii) $\Lambda_{(m,b,n,k,a,d)}=\bigcup _{e\in \mathcal{E}^P}e$, where each element of $\Lambda_{(m,b,n,k,a,d)}$ is called a \emph{vertex}, and the cardinality $|\Lambda_{(m,b,n,k,a,d)}|$ is the \emph{order} of $\mathcal{P}_{hyper}$;

(iii) $\mathcal{E}^P$ is called \emph{parameterized hyperedge set}, and the cardinality $|\mathcal{E}^P|$ is the \emph{size} of $\mathcal{P}_{hyper}$. \qqed
\end{defn}

\begin{problem}\label{problem:99999}
Let $S(\Lambda^P)=\big \{\mathcal{E}^P_1,\mathcal{E}^P_2,\dots ,\mathcal{E}^P_{M}\big \}$ be the set of parameterized hyperedge sets defined on a parameterized set $\Lambda_{(m,b,n,k,a,d)}$, such that each parameterized hyperedge set $\mathcal{E}^P_i$ with $i\in [1,M]$ holds $\Lambda_{(m,b,n,k,a,d)}=\bigcup _{e\in \mathcal{E}^P_i}e$ true. \textbf{Estimate} the number $M$ of parameterized hypergraphs defined on the parameterized set $\Lambda_{(m,b,n,k,a,d)}$.
\end{problem}

\begin{defn} \label{defn:operation-graphs-pa-hypergraph}
$^*$ Let $\textbf{\textrm{O}}=(O_1,O_2,\dots ,O_m)$ be an operation set with $m\geq 1$, and if an element $c$ is obtained by implementing an operation $O_i\in \textbf{\textrm{O}}$ to other two elements $a,b$, we write this fact as $c=a[O_i]b$. Suppose that a $(p,q)$-graph $G$ admits a proper $(k,d)$-total set-coloring $F: V(G)\cup E(G)\rightarrow \mathcal{E}^P$, where $\mathcal{E}^P$ is the parameterized hyperedge set of a parameterized hypergraph $\mathcal{P}_{hyper}=(\Lambda_{(m,b,n,k,a,d)},\mathcal{E}^P)$ defined in Definition \ref{defn:parameterized-hypergraph-basic-definition}.
\begin{asparaenum}[\textrm{\textbf{Pahy}}-1. ]
\item \label{pahyper:parameterized-graph-3} Only one operation $O_k\in \textbf{\textrm{O}}$ holds $c_{uv}=a_u[O_k]b_v$ for each edge $uv\in E(G)$, where $a_u\in F(u)$, $b_v\in F(v)$ and $c_{uv}\in F(uv)$.
\item \label{pahyper:parameterized-graph-1} For each operation $O_i\in \textbf{\textrm{O}}$, each edge $uv\in E(G)$ holds $c_{uv}=a_u[O_i]b_v$ for $a_u\in F(u)$, $b_v\in F(v)$ and $c_{uv}\in F(uv)$.
\item \label{pahyper:parameterized-graph-2} Each $z\in F(uv)$ for each edge $uv\in E(G)$ corresponds to an operation $O_j\in \textbf{\textrm{O}}$, such that $z=x[O_j]y$ for some $x\in F(u)$ and $y\in F(v)$.
\item \label{pahyper:parameterized-graph-v} Each $a_x\in F(x)$ for any vertex $x\in V(G)$ corresponds to an operation $O_s\in \textbf{\textrm{O}}$ and an adjacent vertex $y\in N(x)$, such that $z_{xy}=a_x[O_s]b_y$ for some $z_{xy}\in F(xy)$ and $b_y\in F(y)$.
\item \label{pahyper:parameterized-graph-4} Each operation $O_t\in \textbf{\textrm{O}}$ corresponds to some edge $xy\in E(G)$ holding $c_{xy}=a_x[O_t]b_y$ for $a_x\in F(x)$, $b_y\in F(y)$ and $c_{xy}\in F(xy)$.
\item \label{pahyper:parameterized-graph-must} If there are three different sets $e_i,e_j,e_k\in F(V(G))$ and an operation $O_r\in \textbf{\textrm{O}}$ holding $e_i[O_r]e_j\subseteq e_k$, then there exists an edge $xy\in E(G)$, such that $F(x)=e_i$, $F(y)=e_j$ and $F(xy)=e_k$.
\end{asparaenum}
\noindent \textbf{Then we call} $G$:
\begin{asparaenum}[\textrm{\textbf{Ograph}}-1. ]
\item a \emph{$(k,d)$-$\textbf{\textrm{O}}$ operation graph} of $\mathcal{P}_{hyper}$ if \textbf{Pahy}-\ref{pahyper:parameterized-graph-1} and \textbf{Pahy}-\ref{pahyper:parameterized-graph-must} hold true.
\item a \emph{$(k,d)$-$\textbf{\textrm{O}}$-edge operation graph} of $\mathcal{P}_{hyper}$ if \textbf{Pahy}-\ref{pahyper:parameterized-graph-1}, \textbf{Pahy}-\ref{pahyper:parameterized-graph-2} and \textbf{Pahy}-\ref{pahyper:parameterized-graph-must} hold true.
\item a \emph{$(k,d)$-$\textbf{\textrm{O}}$-total operation graph} of $\mathcal{P}_{hyper}$ if \textbf{Pahy}-\ref{pahyper:parameterized-graph-1}, \textbf{Pahy}-\ref{pahyper:parameterized-graph-2}, \textbf{Pahy}-\ref{pahyper:parameterized-graph-v} and \textbf{Pahy}-\ref{pahyper:parameterized-graph-must} hold true.
\item a \emph{$(k,d)$-edge operation graph} of $\mathcal{P}_{hyper}$ if \textbf{Pahy}-\ref{pahyper:parameterized-graph-2} and \textbf{Pahy}-\ref{pahyper:parameterized-graph-must} hold true.
\item a \emph{$(k,d)$-vertex operation graph} of $\mathcal{P}_{hyper}$ if \textbf{Pahy}-\ref{pahyper:parameterized-graph-v} and \textbf{Pahy}-\ref{pahyper:parameterized-graph-must} hold true.
\item a \emph{$(k,d)$-total operation graph} of $\mathcal{P}_{hyper}$ if \textbf{Pahy}-\ref{pahyper:parameterized-graph-2}, \textbf{Pahy}-\ref{pahyper:parameterized-graph-v} and \textbf{Pahy}-\ref{pahyper:parameterized-graph-must} hold true.
\item a \emph{$(k,d)$-non-uniform operation graph} of $\mathcal{P}_{hyper}$ if \textbf{Pahy}-\ref{pahyper:parameterized-graph-4} and \textbf{Pahy}-\ref{pahyper:parameterized-graph-must} hold true.
\item a \emph{$(k,d)$-one operation graph} of $\mathcal{P}_{hyper}$ if \textbf{Pahy}-\ref{pahyper:parameterized-graph-3} and \textbf{Pahy}-\ref{pahyper:parameterized-graph-must} hold true.\qqed
\end{asparaenum}
\end{defn}

\begin{example}\label{exa:8888888888}
If $\textbf{\textrm{O}}$ contains only one operation ``$\cap$'', the \emph{intersection operation} on sets, and the $(p,q)$-graph $G$ satisfies \textbf{Pahy}-\ref{pahyper:parameterized-graph-1} and \textbf{Pahy}-\ref{pahyper:parameterized-graph-must} in Definition \ref{defn:operation-graphs-pa-hypergraph}, then $G$ is a $(k,d)$-one operation graph of the parameterized hypergraph $\mathcal{P}_{hyper}$, also, $G$ is called \emph{$(k,d)$-intersected-graph} of $\mathcal{P}_{hyper}$. As $(k,d)=(1,1)$, $G$ is just the \emph{intersected-graph} of a hypergraph $H_{hyper}$ defined in \cite{Jianfang-Wang-Hypergraphs-2008}.\qqed
\end{example}

\begin{example}\label{exa:complete-graph-K-4}
Theorem \ref{thm:graph-admits-6-set-colorings} tells us: Each connected graph $G$ admits each one of the following $W$-constraint $(k,d)$-total set-colorings for $W$-constraint $\in \{$graceful, harmonious, edge-difference, edge-magic, felicitous-difference, graceful-difference$\}$. In Fig.\ref{fig:k-d-operation-graph}, doing the vertex-split to a complete graph $K_4$ produces a tree $T$, such that $E(K_4)=E(T)$, and ``$T\rightarrow K_4$'' is a graph homomorphism from the tree $T$ into the complete graph $K_4$. And moreover, there are six colored trees $T_i$ for $i\in [1,6]$, where

$T_1$ admits a graceful $(k,d)$-total labeling $f_1$;

$T_2$ admits a harmonious $(k,d)$-total labeling $f_2$;

$T_3$ admits a felicitous-difference $(k,d)$-total labeling $f_3$;

$T_4$ admits an edge-magic $(k,d)$-total labeling $f_4$;

$T_5$ admits an edge-difference $(k,d)$-total labeling $f_5$; and

$T_6$ admits a graceful-difference $(k,d)$-total labeling $f_6$.

\begin{figure}[h]
\centering
\includegraphics[width=16.4cm]{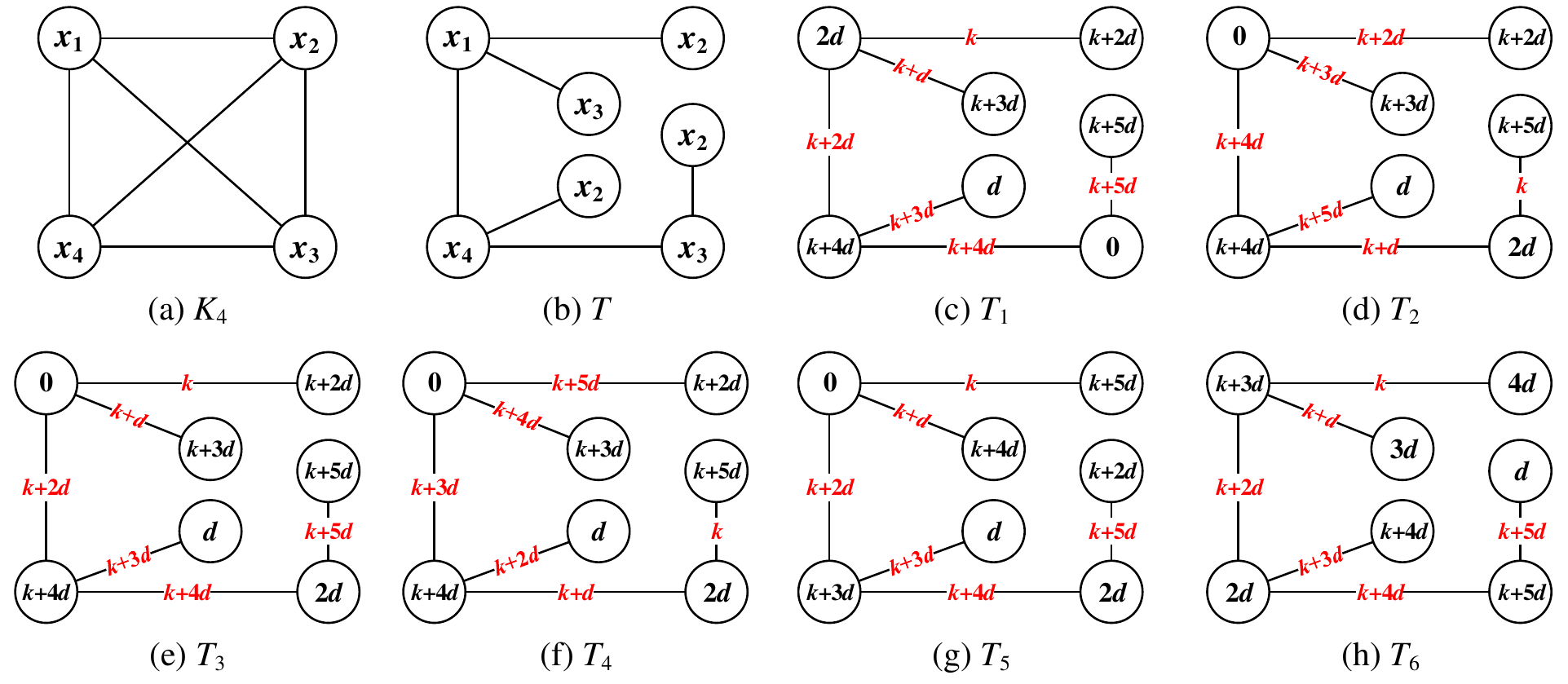}\\
\caption{\label{fig:k-d-operation-graph}{\small (a) A complete graph $K_4$; (b)-(h) six colored trees $T_i$ for $i\in [1,6]$.}}
\end{figure}

By Fig.\ref{fig:k-d-operation-graph}, we get a parameterized hypergraph $\mathcal{P}^*_{hyper}=(\Lambda_{(4,0,5,k,0,d)},\mathcal{E}^P)$ with the parameterized hypervertex set
$$\Lambda_{(4,0,5,k,0,d)}=S_{4,0,0,d}\cup S_{5,k,0,d}=\{0,d,2d,3d,4d\}\cup \{k,k+d,k+2d,k+3d,k+4d,k+5d\}$$
and the parameterized hyperedge set $\mathcal{E}^P=\{e_i:~i\in [1,10]\}$ containing $e_1=\{0,2d,k+3d\}$,

$e_2=\{0,d,3d,4d,k+2d,k+4d,k+5d\}$, $e_3=\{0,2d,3d,k+3d,k+4d,k+5d\}$,

$e_4=\{d,2d,k+3d,k+4d\}$, $e_5=\{0,k,k+2d,k+5d\}$,

$e_6=\{k+d,k+3d,k+4d\}$, $e_7=\{k+2d,k+3d,k+4d\}$,

$e_8=\{k,k+5d\}$, $e_9=\{d,k+2d,k+3d,k+5d\}$ and $e_{10}=\{k+d,k+4d\}$.

\vskip 0.2cm

Clearly, $\Lambda_{(m,b,n,k,a,d)}=\bigcup _{e_i\in \mathcal{E}^P}e_i$. The complete graph $K_4$ admits a proper $(k,d)$-total set-coloring $F: V(K_4)\cup E(K_4)\rightarrow \mathcal{E}^P$, where $F(x_1)=e_1$, $F(x_2)=e_2$, $F(x_3)=e_3$, $F(x_4)=e_4$, $F(x_1x_2)=e_5$, $F(x_1x_3)=e_6$, $F(x_1x_4)=e_7$, $F(x_2x_3)=e_8$, $F(x_2x_4)=e_9$ and $F(x_3x_4)=e_{10}$.

\vskip 0.2cm

We have an operation set $\textbf{\textrm{O}}=(O_1,O_2,\dots ,O_7)$, where
\begin{asparaenum}[(1) ]
\item The operation $O_1$ is the graceful $(k,d)$-total labeling $f_1$, such that the constraint $f_1(uv)=|f_1(u)-f_1(v)|$ for each edge $uv \in E(T_1)$ and the edge color set $f_1(E(T_1))=S_{5,k,0,d}$.
\item The operation $O_2$ is the harmonious $(k,d)$-total labeling $f_2$, such that the constraint
$$f_2(uv)=f_2(u)+f_2(v)~(\bmod~6d)
$$ for each edge $uv \in E(T_2)$ and the edge color set $f_2(E(T_2))=S_{5,k,0,d}$.
\item The operation $O_3$ is the felicitous-difference $(k,d)$-total labeling $f_3$, such that the felicitous-difference constraint
$$|f_3(u)+f_3(v)-f_3(uv)|=2d
$$ for each edge $uv \in E(T_3)$ and the edge color set $f_3(E(T_3))=S_{5,k,0,d}$.
\item The operation $O_4$ is the edge-magic $(k,d)$-total labeling $f_4$, such that the edge-magic constraint
$$f_4(u)+f_4(uv)+f_4(v)=2k+7d
$$ for each edge $uv \in E(T_4)$ and the edge color set $f_4(E(T_4))=S_{5,k,0,d}$.
\item The operation $O_5$ is the edge-difference $(k,d)$-total labeling $f_5$, such that the edge-difference constraint
$$f_5(uv)+|f_5(u)-f_5(v)|=2k+5d
$$ for each edge $uv \in E(T_5)$ and the edge color set $f_5(E(T_5))=S_{5,k,0,d}$.
\item The operation $O_6$ is the graceful-difference $(k,d)$-total labeling $f_6$, such that the graceful-difference constraint
$$\big ||f_6(u)-f_6(v)|-f_6(uv)\big |=d
$$ for each edge $uv \in E(T_6)$ and the edge color set $f_6(E(T_6))=S_{5,k,0,d}$.
\item The operation $O_7$ is the intersection operation ``$\cap $'', such that $F(x_ix_j)\cap [F(x_i)\cap F(x_j)]\neq \emptyset $ for each edge $x_ix_j \in E(K_4)$.
\end{asparaenum}

Thereby, we claim that the complete graph $K_4$ is every one of the operation graphs \textbf{Ograph}-1, \textbf{Ograph}-2, \textbf{Ograph}-3, \textbf{Ograph}-4, \textbf{Ograph}-5 and \textbf{Ograph}-6 of the parameterized hypergraph $\mathcal{P}^*_{hyper}=(\Lambda_{(4,0,5,k,0,d)},\mathcal{E}^P)$ according to Definition \ref{defn:operation-graphs-pa-hypergraph}.\qqed
\end{example}

\subsubsection{VSET-coloring algorithm}

We will use some algorithms for producing set-colorings introduced in \cite{Yao-Ma-arXiv-2201-13354v1} to make $(k,d)$-total set-colorings in this subsection.

\begin{thm}\label{thm:build-hyperedge-set}
\cite{Yao-Ma-arXiv-2201-13354v1} If a tree $T$ admits a mapping $f:V(T)\rightarrow [0,p-1]$ with $p=|V(T)|$ and $f(z)\neq f(y)$ for any pair of distinct vertices $x,y$, then $T$ admits a set-coloring defined on a hyperedge set $\mathcal{E}$ such that $T$ is a subgraph of the intersected-graph of a hypergraph $\mathcal{H}_{yper}=([0,p-1],\mathcal{E})$.
\end{thm}
\begin{proof} According the hypothesis of the theorem, we define a set-coloring $F$ for the tree $T$ by means of the following so-called \textbf{VSET-coloring algorithm}:

\textbf{Step 1.} Each leaf $w_j$ of the tree $T$ is set-colored with $F(w_j)=\{f(w_j),f(v)\}$, where the edge $w_jv\in E(T)$.

\textbf{Step 2.} Each leaf $w^1_j$ of the tree $T_1=T-L(T)$, where $L(T)$ is the set of leaves of $T$, is colored by $F(w^1_j)=\{f(w^1_j),f(z)\}$, where the edge $w^1_jz\in E(T_1)$.

\textbf{Step 3.} Each leaf $w^r_j$ of the tree $T_r=T_{r-1}-L(T_{r-1})$, where $L(T_{r-1})$ is the set of leaves of $T_{r-1}$, is colored by $F(w^r_j)=\{f(w^r_j),f(u)\}$, where the edge $w^r_ju\in E(T_r)$.

\textbf{Step 4.} Suppose $T_k$ is a star $K_{1,m}$ with vertex set $V(K_{1,m})=\big \{w^k_j,u_0:j\in [1,m]\big \}$ and edge set $E(K_{1,m})=\big \{u_0w^k_1,u_0w^k_2,\dots ,u_0w^k_m\big \}$, we color each leaf $w^k_j$ with $F(w^k_j)=\big \{f(w^k_j),f(u_0)\big \}$, and $F(u_0)=\{f(u_0)\}$.

\textbf{Step 5.} $F(uv)=F(u)\cap F(v)\neq \emptyset$ for each edge $uv\in E(T)$.

Thereby, the tree $T$ admits the set-coloring $F$ subject to $R_{est}(1)=\{c_0\}$, such that $c_0$ holds $F(uv)\supseteq F(u)\cap F(v)\neq \emptyset$ for each edge $uv\in E(T)$. So, $T$ is a subgraph of the intersected-graph of a hypergraph $\mathcal{H}_{yper}=(\Lambda,\mathcal{E})$ with its hyperedge set $\mathcal{E}=F(V(T))$, and $\Lambda=[1,p-1]=\bigcup _{e\in \mathcal{E}}e$, we are done.
\end{proof}

\begin{figure}[h]
\centering
\includegraphics[width=16.4cm]{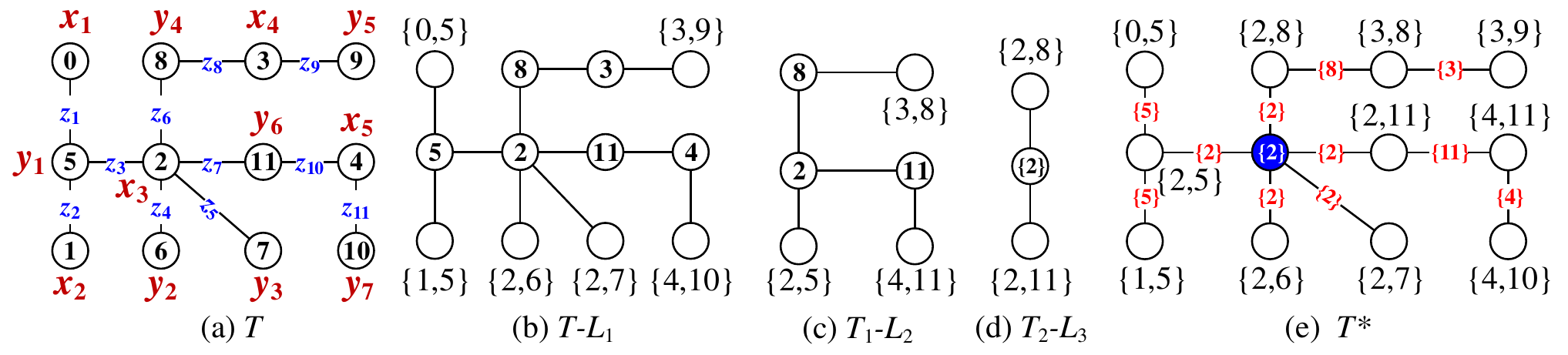}\\
\caption{\label{fig:VSET-coloring-algorithm-1}{\small An example for understanding the proof of Theorem \ref{thm:build-hyperedge-set}, cited from \cite{Yao-Ma-arXiv-2201-13354v1}.}}
\end{figure}

\begin{example}\label{exa:8888888888}
In Fig.\ref{fig:VSET-coloring-algorithm-1} (a), the tree $T$ has its own Topcode-matrix as follows
{\small
\begin{equation}\label{eqa:VSET-coloring-algorithm-matrix-1}
\centering
{
\begin{split}
T_{code}(T)&= \left(
\begin{array}{cccccccccccc}
x_1 & x_2 & x_3 & x_3 & x_3 & x_3 & x_3 & x_4 & x_4 & x_5 & x_5\\
x_1y_1 & x_2y_1 & x_3y_1 & x_3y_2 & x_3y_3 & x_3y_4 & x_3y_6 & x_4y_4 & x_4y_5 & x_5y_6 & x_5y_7\\
y_1 & y_1 & y_1 & y_2 & y_3 & y_4 & y_6 & y_4 & y_5 & y_6 &y_7
\end{array}
\right)_{3\times 11}\\
&=(X_T,E_T,Y_T)^{T}_{3\times 11}
\end{split}}
\end{equation}
}with the vertex-vector $X_T=(x_1, x_2, x_3, x_3, x_3, x_3, x_3, x_4, x_4, x_5, x_5)$, the edge-vector
$$E_T=(x_1y_1, x_2y_1, x_3y_1, x_3y_2, x_3y_3, x_3y_4, x_3y_6, x_4y_4, x_4y_5, x_5y_6, x_5y_7)=(z_1,z_2,\cdots ,z_{11})$$
and the vertex-vector $Y_T=(y_1, y_1, y_1, y_2, y_3, y_4, y_6, y_4, y_5,y_6,y_7)$, where $V(T)=X_T\cup Y_T$ and $E(T)=E_T$.

\vskip 0.4cm

Notice that the tree $T$ admits a vertex labeling $f$ holding $f(u)\neq f(v)$ for any pair of distinct vertices $u,v\in V(T)$ shown in Fig.\ref{fig:VSET-coloring-algorithm-1} (a), then we get the colored Topcode-matrix
{\footnotesize
\begin{equation}\label{eqa:VSET-coloring-algorithm-matrix-22}
\centering
{
\begin{split}
T_{code}(T,f)= \left(
\begin{array}{cccccccccccc}
0 & 1 & 2 & 2 & 2 & 2 & 2 & 3 & 3 & 4 & 4\\
f(z_1) & f(z_2) & f(z_3) & f(z_4) & f(z_5) & f(z_6) & f(z_7) & f(z_8) & f(z_9) & f(z_{10}) & f(z_{11})\\
5 & 5 & 5 & 6 & 7 & 8 & 11 & 8 & 9 & 11 & 10
\end{array}
\right)
\end{split}}
\end{equation}
}with $f(z_i)=z_i$ for $i\in [1,11]$.

\vskip 0.4cm

By Eq.(\ref{eqa:unit-Topcode-matrix}) and Eq.(\ref{eqa:VSET-coloring-algorithm-matrix-22}), the tree $T$ has its own parameterized Topcode-matrix $P_{ara}(T,F)$ defined as
{\footnotesize
\begin{equation}\label{eqa:VSET-coloring-algorithm-matrix-33}
\centering
{
\begin{split}
&\quad P_{ara}(T,F)=k\cdot I\,^0+d\cdot T_{code}(T,f)\\
&= \left(
\begin{array}{cccccccccccc}
0 & d & 2d & 2d & 2d & 2d & 2d & 3d & 3d & 4d & 4d\\
F(z_1) & F(z_2) & F(z_3) & F(z_4) & F(z_5) & F(z_6) & F(z_7) & F(z_8) & F(z_9) & F(z_{10}) & F(z_{11})\\
k+5d & k+5d & k+5d & k+6d & k+7d & k+8d & k+11d & k+8d & k+9d & k+11d & k+10d
\end{array}
\right)
\end{split}}
\end{equation}
}with $F(z_i)=k+f(z_i)\cdot d$ for $i\in [1,11]$, and we can define $f(z_i)$ to be a number obtained by some $W$-constraint coloring of graph theory.

Fig.\ref{fig:VSET-coloring-algorithm-1} (e) shows us a set-coloring $g$ of the tree $T^*$ as follows:

\textbf{(a-1)} $g(x_1)=\{0,5\}$, $g(x_2)=\{1,5\}$, $g(x_3)=\{2\}$, $g(x_4)=\{3,8\}$ and $g(x_5)=\{4,11\}$;

\textbf{(a-2)} $g(y_1)=\{2,5\}$, $g(y_2)=\{2,6\}$, $g(y_3)=\{2,7\}$, $g(y_4)=\{2,8\}$, $g(y_5)=\{3.9\}$, $g(y_6)=\{2,11\}$ and $g(y_7)=\{4,10\}$; and

\textbf{(a-3)} $g(z_1)=\{5\}$, $g(z_2)=\{5\}$, $g(z_3)=\{2\}$, $g(z_4)=\{2\}$, $g(z_5)=\{2\}$, $g(z_6)=\{2\}$, $g(z_7)=\{2\}$, $g(z_8)=\{8\}$, $g(z_9)=\{3\}$, $g(z_{10})=\{11\}$ and $g(z_{11})=\{4\}$.

\vskip 0.4cm

We have a parameterized hypervertex set
$${
\begin{split}
\Lambda_{(4,0,11,k,0,d)}&=S_{4,0,0,d}\cup S_{11,k,5,d}\\
&=\{0,d,2d,3d,4d\}\cup \{k+5d,k+6d,k+7d,k+8d,k+9d,k+10d,k+11d\}
\end{split}}
$$

Thereby, the tree $T$ admits a $(k,d)$-total set-coloring $F_{k,d}:V(T)\cup E(T)\rightarrow \Lambda^2_{(4,0,11,k,0,d)}$ defined as:

\textbf{(b-1)} The vertex $(k,d)$-colors are $F_{k,d}(x_1)=\{0,k+5d\}$, $F_{k,d}(x_2)=\{d,k+5d\}$, $F_{k,d}(x_3)=\{2d\}$, $F_{k,d}(x_4)=\{3d,k+8d\}$ and $F_{k,d}(x_5)=\{4d,k+11d\}$;

\textbf{(b-2)} The vertex $(k,d)$-colors are $F_{k,d}(y_1)=\{2d,k+5d\}$, $F_{k,d}(y_2)=\{2d,k+6d\}$, $F_{k,d}(y_3)=\{2d,k+7d\}$, $F_{k,d}(y_4)=\{2d,k+8d\}$, $F_{k,d}(y_5)=\{3d,k+9d\}$, $F_{k,d}(y_6)=\{2d,k+11d\}$ and $F_{k,d}(y_7)=\{4d,k+10d\}$; and

\textbf{(b-3)} The edge $(k,d)$-colors are $F_{k,d}(z_1)=\{k+5d\}$, $F_{k,d}(z_2)=\{k+5d\}$, $F_{k,d}(z_3)=\{2d\}$, $F_{k,d}(z_4)=\{2d\}$, $F_{k,d}(z_5)=\{2d\}$, $F_{k,d}(z_6)=\{2d\}$, $F_{k,d}(z_7)=\{2d\}$, $F_{k,d}(z_8)=\{k+8d\}$, $F_{k,d}(z_9)=\{3d\}$, $F_{k,d}(z_{10})=\{k+11d\}$ and $F_{k,d}(z_{11})=\{4d\}$.

\vskip 0.4cm

We get a \emph{parameterized hypergraph} $\mathcal{P}_{hyper}=(\Lambda_{(4,0,11,k,0,d)},\mathcal{E}^P)$, where the parameterized hyperedge set $\mathcal{E}^P=\{e_j:j\in [1,12]\}$ with elements $e_1=\{0,k+5d\}$, $e_2=\{d,k+5d\}$, $e_3=\{2d\}$, $e_4=\{3d,k+8d\}$ and $e_5=\{4d,k+11d\}$, $e_6=\{2d,k+5d\}$, $e_7=\{2d,k+6d\}$, $e_8=\{2d,k+7d\}$, $e_9=\{2d,k+8d\}$, $e_{10}=\{3d,k+9d\}$, $e_{11}=\{2d,k+11d\}$ and $e_{12}=\{4d,k+10d\}$.

Since $\Lambda_{(4,0,11,k,0,d)}=\bigcup _{e_j\in \mathcal{E}^P}e_j$, the colored tree $T$ admitting the $(k,d)$-total set-coloring $F_{k,d}$ is a subgraph of the intersected-graph of the parameterized hypergraph $\mathcal{P}_{hyper}$.\qqed
\end{example}

\subsubsection{PWCSC-algorithms on colored trees}

The sentence ``Producing $W$-constraint set-coloring algorithm'' is abbreviated as ``PWCSC-algorithm'' in the following discussion.

\vskip 0.4cm

\textbf{$^*$ PWCSC-algorithm-A based on the ordered-path.}

\textbf{Initialization-A.} Suppose that $T$ is a tree admitting a $W$-constraint labeling $f$ holding $f(uv)\neq f(xy)$ for any pair of edges $uv$ and $xy$ of $T$, and each edge $uv\in E(T)$ holds the $W$-constraint $f(uv)=W\langle f(u),f(v)\rangle $, as well as $|f(V(T))|=|V(T)|$.

\vskip 0.2cm

\textbf{Step A-1.} Do the VSET-coloring algorithm to $T$ first. We select a longest path
$$P_1=w^1_1w^1_2w^1_3\cdots w^1_{m_1-2}w^1_{m_1-1}w^1_{m_1}
$$ of $T$, then we have the neighbor set $N(w^1_2)=L(w^1_2)\cup \big \{w^1_3\big \}$, where the leaf set $L(w^1_2)=\big \{w^1_1, v^1_{2,1}, v^1_{2,2}$, $ \dots $, $v^1_{2,d_2}\big \}$ with $d_2=\textrm{deg}_T(w^1_2)-2$, and another neighbor set $N(w^1_{m_1-1})=\big \{w^1_{m_1-2}\big \}\cup L(w^1_{m_1-1})$ with the leaf set $L(w^1_{m_1-1})=\big \{w^1_{m_1}, u^1_{m_1,1}, u^1_{m_1,2},\dots , u^1_{m_1,d_{m_1}}\big \}$, where $d_{m_1}=\textrm{deg}_T(w^1_{m_1-1})-2$. We define a total set-coloring $F$ for the three $T$ as: The vertex set-colors are
\begin{equation}\label{eqa:step-a-11}
F_{path}(x)=\big \{f(x),f(w^1_2)\big \}_1,~x\in L(w^1_2);\quad F_{path}(y)=\big \{f(y),f(w^1_{m_1-1})\big \}_1,~y\in L(w^1_{m_1-1})
\end{equation}

\vskip 0.2cm

\textbf{Step A-2.} We get a tree $T_1=T-\big [L(w^1_2)\cup L(w^1_{m_1-1})\big ]$ by removing all leaves of two vertices $w^1_2$ and $w^1_{m_1-1}$ of $T$, and then do the VSET-coloring algorithm to $T_1$. Notice that the tree $T_1$ admits the set-coloring $F$, so we select a longest path
$$P_2=w^2_1w^2_2w^2_3\cdots w^2_{m_2-2}w^2_{m_2-1}w^2_{m_2}
$$ of $T_1$, then we have the neighbor set $N(w^2_2)=L(w^2_2)\cup \big \{w^2_3\big \}$, where the leaf set $L(w^2_2)=\big \{w^2_1, v^2_{2,1}$, $ v^2_{2,2}$, $\dots $, $v^2_{2,d_2}\big \}$ with $d_2=\textrm{deg}_T(w^2_2)-2$, and another neighbor set $N(w^2_{m_2-1})=\big \{w^2_{m_2-2}\big \}\cup L(w^2_{m_2-1})$ with the leaf set $L(w^2_{m_2-1})=\big \{w^2_{m_2}, u^2_{m_2,1}, u^2_{m_2,2},\dots , u^2_{m_2,d_{m_2}}\big \}$, where $d_{m_2}=\textrm{deg}_T(w^2_{m_2-1})-2$. We get the following vertex set-colors
\begin{equation}\label{eqa:step-a-22}
F_{path}(x)=\big \{f(x),f(w^2_2)\big \}_2,~x\in L(w^2_2);\quad F_{path}(y)=\big \{f(y),f(w^2_{m_2-1})\big \}_2,~y\in L(w^2_{m_2-1})
\end{equation}

\vskip 0.2cm

\textbf{Step A-3.} If the tree $T_2=T_1-\big [L(w^2_2)\cup L(w^2_{m_2-1})\big ]$ has its own diameter $D(T_2)\geq 3$, then we goto Step A-2.

\vskip 0.2cm

\textbf{Step A-4.} After $k-1$ times, we get the tree $T_k=T_{k-1}-\big [L(w^k_2)\cup L(w^k_{m_k-1})\big ]$ to be a star $K_{1,n}$ with its own diameter $D(K_{1,n})=2$, then we have the vertex set $V(K_{1,n})=\big \{x_0, y_i:i\in [1,n]\big \}$ and the edge set $E(K_{1,n})=\big \{x_0y_i:i\in [1,n]\big \}$. We have the following vertex set-colors
\begin{equation}\label{eqa:step-a-kk}
F_{path}(x_0)=\big \{f(x_0)\big \}_{k+1};\quad F_{path}(y_i)=\big \{f(y_i),f(x_0)\big \}_{k+1},~y_i\in L(x_0)
\end{equation} Notice that $|F_{path}(u)|=2$ for $u\in V(T)\setminus \big \{x_0\big \}$, and $|F_{path}(x_0)|=1$.

\vskip 0.2cm

\textbf{Step A-5.} By the $W$-constraint we recolor the edges of the tree $T$ as follows:
\begin{equation}\label{eqa:step-a-ee}
F_{path}(uv)=[F_{path}(u)\cap F_{path}(v)]\cup \big \{W\langle a,b\rangle:~a\in F_{path}(u),~b\in F_{path}(v)\big \},~uv\in E(T)
\end{equation} since $F_{path}(u)\cap F_{path}(v)\neq \emptyset$.

\vskip 0.2cm

\textbf{Step A-6.} Return the $W$-constraint proper total set-coloring $F$ of the tree $T$, since $F_{path}(s)\neq F_{path}(t)$ for any pair of adjacent, or incident elements $s,t\in V(T)\cup E(T)$.

\vskip 0.4cm

By the PWCSC-algorithm-A based on the ordered-path, we present a result as follows:

\begin{thm}\label{thm:set-ordered-graceful-PWCSC-algorithm-A-ordered-path}
$^*$ If a tree $T$ admits a set-ordered $W$-constraint labeling, then $T$ admits a $W$-constraint proper total set-coloring $F$ obtained by the PWCSC-algorithm-A based on the ordered-path, such that $|F(u)\cap F(v)|=1$ and $|F(uv)|\geq 2$ for each edge $uv\in E(T)$, and $|F(x)|=2$ for $x\in V(T)\setminus \{x_0\}$, and $|F(x_0)|=1$.
\end{thm}

\begin{example}\label{exa:PWCSC-algorithm-A-ordered-path-11}
An example for understanding the above PWCSC-algorithm-A based on the ordered-path is shown in Fig.\ref{fig:VSET-coloring-algorithm-2}. A tree $T$ shown in Fig.\ref{fig:VSET-coloring-algorithm-2} (a) admits a set-ordered graceful labeling $f$, such that $\max f(X)=f(x_6)<f(y_1)=\min f(Y)$, where $X=\{x_i:i\in [1,6]\}$ and $Y=\{y_j:j\in [1,8]\}$. The last tree $T_4$ is a star $K_{1,4}$ shown in Fig.\ref{fig:VSET-coloring-algorithm-2} (e), $T_4$ admits a $W$-constraint proper total set-coloring. The tree $T_6$ admits a graceful proper total set-coloring $F$ satisfied Theorem \ref{thm:set-ordered-graceful-PWCSC-algorithm-A-ordered-path}. We have the following facts:

\textbf{Fact-1. }The tree $T$ has its own Topcode-matrix $T_{code}(T,f)$ as
{\small
\begin{equation}\label{eqa:PWCSC-algorithm-A-ordered-path-11}
\centering
{
\begin{split}
T_{code}(T,f)=\left(
\begin{array}{cccccccccccccc}
5 & 5 & 5 & 4 & 4 & 4 & 4 & 3 & 2 & 1 & 0 & 0 & 0\\
1 & 2 & 3 & 4 & 5 & 6 & 7 & 8 & 9 & 10 & 11 & 12 & 13\\
6 & 7 & 8 & 8 & 9 & 10 & 11 & 11 & 11 & 11 & 11 & 12 & 13
\end{array}
\right)_{3\times 13}=(X_T,E_T,Y_T)^T
\end{split}}
\end{equation}
}with
$${
\begin{split}
X_T=&(5, 5, 5, 4, 4, 4, 4, 3, 2, 1, 0, 0, 0)\\
=&(f(x_6),f(x_6),f(x_6),f(x_5),f(x_5),f(x_5),f(x_5),f(x_4),f(x_3),f(x_2),f(x_1),f(x_1),f(x_1))\\
E_T=&(1, 2, 3, 4, 5, 6, 7, 8, 9, 10, 11, 12, 13)\\
=&(f(x_6y_1),f(x_6y_2),f(x_6y_3),f(x_5y_3),f(x_5y_4),f(x_5y_5),f(x_5y_6),f(x_4y_6),f(x_3y_6),f(x_2y_6),\\
&f(x_1y_6),f(x_1y_7),f(x_1y_8))\\
Y_T=&(6, 7, 8, 8, 9, 10, 11, 11, 11, 11, 11, 12, 13)\\
=&(f(y_1),f(y_2),f(y_3),f(y_3),f(y_4),f(y_5),f(y_6),f(y_6),f(y_6),f(y_6),f(y_6),f(y_7),f(y_8))
\end{split}}
$$Clearly, $\max X_T=5<6=\min Y_T$, $|E_T|=[1,13]$.

\vskip 0.2cm

\textbf{Fact-2. }The tree $T$ has its own parameterized Topcode-matrix $P_{ara}(T,\theta)$ defined as
\begin{equation}\label{eqa:VSET-coloring-algorithm-matrix-33}
\centering
{
\begin{split}
P_{ara}(T,\theta)=k\cdot I\,^0+d\cdot T_{code}(T,f)=\left(
\begin{array}{cccccccccccccc}
X_P\\
E_P\\
Y_P
\end{array}
\right)=(X_P,E_P,Y_P)^T
\end{split}}
\end{equation} with the vertex $(k,d)$-colors and the edge $(k,d)$-colors
$${
\begin{split}
X_P=&(5d,~5d,~5d,~4d,~4d,~4d,~4d,~3d,~2d,~d,~0,~0,~0)\\
E_P=&(k+d,~ k+2d,~ k+3d,~ k+4d,~ k+5d,~ k+6d,~ k+7d,~ k+8d,~ k+9d,~ k+10d,~ \\
&k+11d,~k+12d,~ k+13d)\\
Y_P=&(k+6d,~ k+7d,~ k+8d,~ k+8d,~ k+9d,~ k+10d,~ k+11d,~ k+11d,~ k+11d,~ k+11d,\\
&k+11d,~ k+12d,~ k+13d)
\end{split}}$$

\vskip 0.2cm

\textbf{Fact-3. }By the graceful proper total set-coloring $F$, the tree $T_6$ has its own \emph{set-type Topcode-matrix} $S_{et}(T_6,F)=(X_{et},E_{et},Y_{et})^T$ with
\begin{equation}\label{eqa:set-coloring}
{
\begin{split}
X_{et}=&(F(x_6),F(x_6),F(x_6),F(x_5),F(x_5),F(x_5),F(x_5),F(x_4),F(x_3),F(x_2),\\
&F(x_1),F(x_1),F(x_1))\\
E_{et}=&(F(x_6y_1),F(x_6y_2),F(x_6y_3),F(x_5y_3),F(x_5y_4),F(x_5y_5),F(x_5y_6),F(x_4y_6),\\
&F(x_3y_6),F(x_2y_6),F(x_1y_6),F(x_1y_7),F(x_1y_8))\\
Y_{et}=&(F(y_1),F(y_2),F(y_3),F(y_3),F(y_4),F(y_5),F(y_6),F(y_6),F(y_6),F(y_6),F(y_6),\\
&F(y_7),F(y_8))
\end{split}}
\end{equation} where

(i) $F(x_1)=\{11,0\}_2$, $F(x_2)=\{11,1\}_2$, $F(x_3)=\{11,2\}_2$, $F(x_4)=\{11,3\}_2$, $F(x_5)=\{4\}$, $F(x_6)=\{8,5\}_2$;

(ii) $F(y_1)=\{5,6\}_1$, $F(y_2)=\{5,7\}_1$, $F(y_3)=\{4,8\}_3$, $F(y_4)=\{4,9\}_3$, $F(y_5)=\{4,10\}_3$, $F(y_6)=\{4,11\}_3$, $F(y_7)=\{0,12\}_1$, $F(y_8)=\{0,13\}_1$;

(iii) By the $W$-constraint Eq.(\ref{eqa:step-a-ee}), the edge colors are

$F(x_6y_1)=\{5\}\cup \{0,1,2,3\}$, $F(x_6y_2)=\{5\}\cup \{0,1,3\}$, $F(x_6y_3)=\{8\}\cup \{0,1,4,3\}$,

$F(x_5y_3)=\{4\}\cup \{0\}$, $F(x_5y_4)=\{4\}\cup \{0,5\}$, $F(x_5y_5)=\{4\}\cup \{0,6\}$, $F(x_5y_6)=\{4\}\cup \{0,7\}$,

$F(x_4y_6)=\{11\}\cup \{0,1,7,8\}$, $F(x_3y_6)=\{11\}\cup \{0,2,7,9\}$, $F(x_2y_6)=\{11\}\cup \{0,3,7,10\}$,

$F(x_1y_6)=\{11\}\cup \{0,4,7\}$, $F(x_1y_7)=\{0\}\cup \{1,11,12\}$ and $F(x_1y_8)=\{0\}\cup \{2,11,13\}$.

\vskip 0.2cm

\textbf{Fact-4. }By Eq.(\ref{eqa:VSET-coloring-algorithm-matrix-33}), we get a $(k,d)$-total set-coloring $F_{k,d}$ of the tree $T_6$ and the \emph{$(k,d)$-set-type Topcode-matrix} $S_{et}(T_6,F_{k,d})=(X^{et}_{k,d},E^{et}_{k,d},Y^{et}_{k,d})^T$ with
\begin{equation}\label{eqa:k-dset-coloring}
{
\begin{split}
X^{et}_{k,d}=&(F_{k,d}(x_6),F_{k,d}(x_6),F_{k,d}(x_6),F_{k,d}(x_5),F_{k,d}(x_5),F_{k,d}(x_5),F_{k,d}(x_5),F_{k,d}(x_4),F_{k,d}(x_3),\\
&F_{k,d}(x_2),F_{k,d}(x_1),F_{k,d}(x_1),F_{k,d}(x_1))\\
E^{et}_{k,d}=&(F_{k,d}(x_6y_1),F_{k,d}(x_6y_2),F_{k,d}(x_6y_3),F_{k,d}(x_5y_3),F_{k,d}(x_5y_4),F_{k,d}(x_5y_5),F_{k,d}(x_5y_6),\\
&F_{k,d}(x_4y_6),F_{k,d}(x_3y_6),F_{k,d}(x_2y_6),F_{k,d}(x_1y_6),F_{k,d}(x_1y_7),F_{k,d}(x_1y_8))\\
Y^{et}_{k,d}=&(F_{k,d}(y_1),F_{k,d}(y_2),F_{k,d}(y_3),F_{k,d}(y_3),F_{k,d}(y_4),F_{k,d}(y_5),F_{k,d}(y_6),F_{k,d}(y_6),F_{k,d}(y_6),\\
&F_{k,d}(y_6),F_{k,d}(y_6),F_{k,d}(y_7),F_{k,d}(y_8))
\end{split}}
\end{equation} where the vertex set-$(k,d)$-colors and the edge set-$(k,d)$-colors are

(1) $F_{k,d}(x_1)=\{k+11d,0\}_2$, $F_{k,d}(x_2)=\{k+11d,d\}_2$, $F_{k,d}(x_3)=\{k+11d,2d\}_2$, $F_{k,d}(x_4)=\{k+11d,3d\}_2$, $F_{k,d}(x_5)=\{4d\}$, $F_{k,d}(x_6)=\{k+8d,5d\}_2$;

(2) $F_{k,d}(y_1)=\{5d,k+6d\}_1$, $F_{k,d}(y_2)=\{5d,k+7d\}_1$, $F_{k,d}(y_3)=\{4d,k+8d\}_3$, $F_{k,d}(y_4)=\{4d,k+9d\}_3$, $F_{k,d}(y_5)=\{4d,k+10d\}_3$, $F_{k,d}(y_6)=\{4d,k+11d\}_3$, $F_{k,d}(y_7)=\{0,k+12d\}_1$, $F_{k,d}(y_8)=\{0,k+13d\}_1$;

(3) By the $W$-constraint Eq.(\ref{eqa:step-a-ee}), the edge $(k,d)$-colors are

$F_{k,d}(x_6y_1)=\{5d\}\cup \{0,k+d,2d,k+3d\}$, $F_{k,d}(x_6y_2)=\{5d\}\cup \{0,k+2d,k+d,k+3d\}$,

$F_{k,d}(x_6y_3)=\{k+8d\}\cup \{0,d,k+3d,k+4d\}$, $F_{k,d}(x_5y_3)=\{4d\}\cup \{0\}$,

$F_{k,d}(x_5y_4)=\{4d\}\cup \{0,k+5d\}$, $F_{k,d}(x_5y_5)=\{4d\}\cup \{0,k+6d\}$,

$F_{k,d}(x_5y_6)=\{4d\}\cup \{0,k+7d\}$, $F_{k,d}(x_4y_6)=\{k+11d\}\cup \{0,d,k+7d,k+8d\}$,

$F_{k,d}(x_3y_6)=\{k+11d\}\cup \{0,2d,k+7d,k+9d\}$, $F_{k,d}(x_2y_6)=\{k+11d\}\cup \{0,3d,k+7d,k+10d\}$,

$F_{k,d}(x_1y_6)=\{k+11d\}\cup \{0,4d,k+7d\}$, $F_{k,d}(x_1y_7)=\{0\}\cup \{k+d,k+11d,k+12d\}$, and

$F_{k,d}(x_1y_8)=\{0\}\cup \{2d,k+11d,k+13d\}$.

\vskip 0.2cm

\textbf{Fact-5. }As the tree $T$ shown in Fig.\ref{fig:VSET-coloring-algorithm-2} (a) is selected as a \emph{topological public-key}, then the tree $T_6$ shown in Fig.\ref{fig:VSET-coloring-algorithm-2} (g) is just a \emph{topological private-key}.

Thereby, the bytes of a number-based string $D_T$ induced from the Topcode-matrix $T_{code}(T,f)$ is shorter than that of a number-based string $D_{T_6}$ from the set-type Topcode-matrix $S_{et}(T_6,F)$ defined in Fact-3, or the $(k,d)$-set-type Topcode-matrix $S_{et}(T_6,F_{k,d})$ defined in Fact-4, since they are related with the different ordered paths of the trees $T$ and $T_6$ according to the the PWCSC-algorithm-A based on the ordered-path. \qqed
\end{example}

\begin{figure}[h]
\centering
\includegraphics[width=16.4cm]{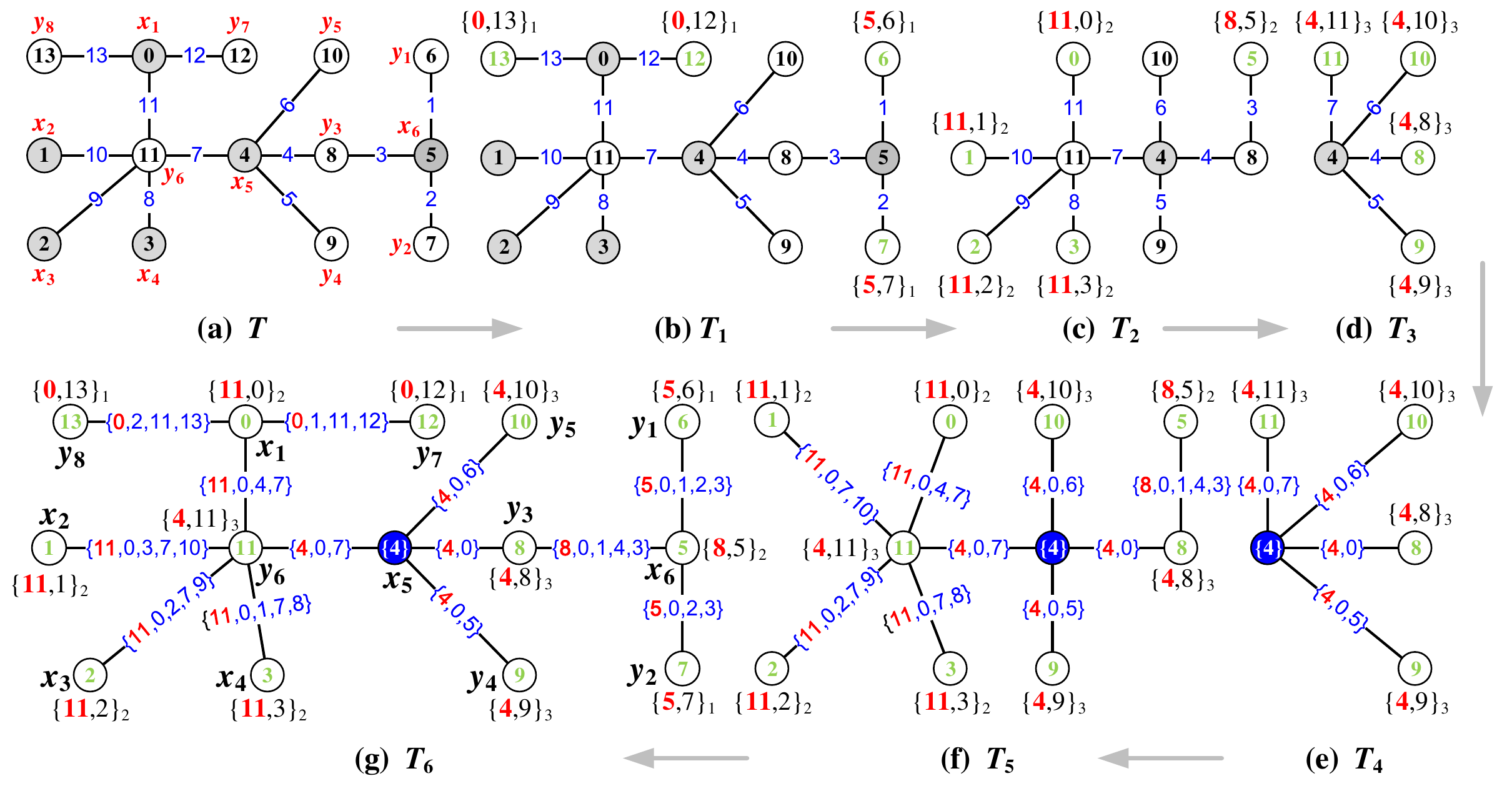}\\
\caption{\label{fig:VSET-coloring-algorithm-2}{\small An example for understanding the PWCSC-algorithm-A based on the ordered-path.}}
\end{figure}

\begin{thm}\label{thm:n-times-PWCSC-algorithm-A-ordered-path}
$^*$ After $n$ times of doing the PWCSC-algorithm-A based on the ordered-path to a tree $T$ admitting a set-ordered $W$-constraint labeling, we get a $W$-constraint set-coloring $F_n$ of the tree $T$ and $|F_n(u)\cap F_n(v)|\geq n=\lfloor \frac{D(T)}{2}\rfloor$ for each edge $uv\in E(T)$ and $F_n(x)\neq F_n(y)$ for distinct vertices $x,y\in V(T)$, where $D(T)$ is the diameter of the tree $T$.
\end{thm}

See an example shown in Fig.\ref{fig:VSET-coloring-algorithm-3} for understanding Theorem \ref{thm:n-times-PWCSC-algorithm-A-ordered-path}. By Example \ref{exa:PWCSC-algorithm-A-ordered-path-11}, we present a theorem as follows:

\begin{figure}[h]
\centering
\includegraphics[width=16.4cm]{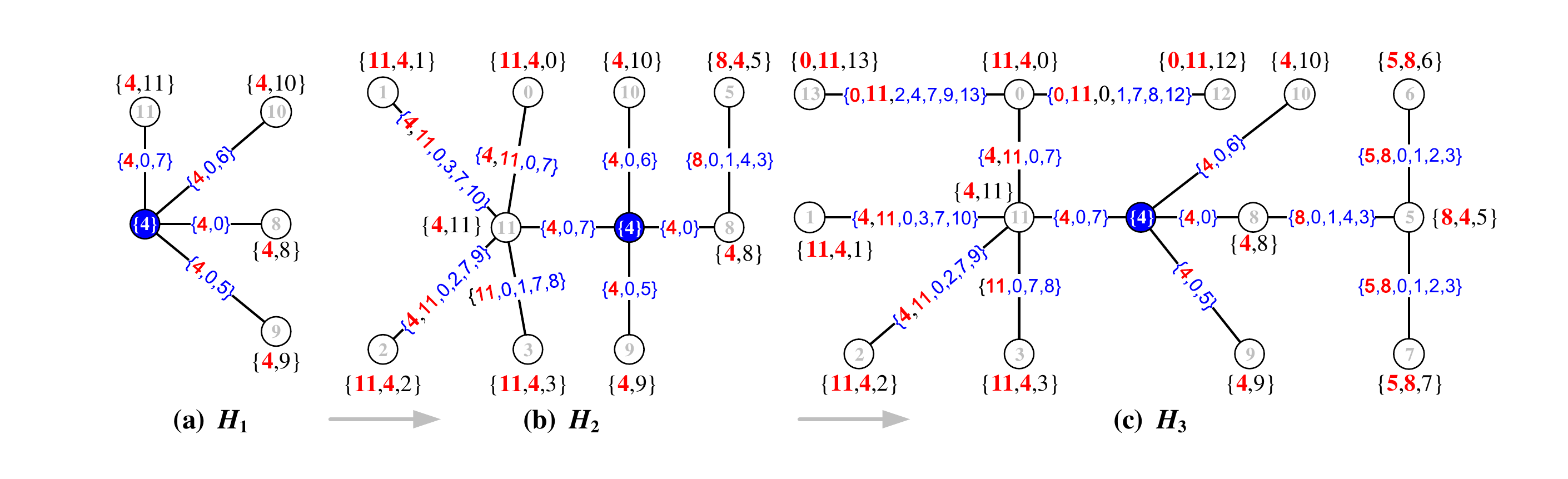}\\
\caption{\label{fig:VSET-coloring-algorithm-3}{\small The set-colored tree $H_3$ obtained by doing the PWCSC-algorithm-A based on the ordered-path to $D_6$ shown in Fig.\ref{fig:VSET-coloring-algorithm-2} (g).}}
\end{figure}

\begin{thm}\label{thm:set-ordered-w-cond-set-co}
$^*$ If a tree admits a set-ordered $W$-constraint labeling, then it admits:

(i) a $W$-constraint proper total set-coloring;

(ii) a $W$-constraint proper $(k,d)$-total coloring; and

(ii) a $W$-constraint proper $(k,d)$-total set-coloring;
\end{thm}

\vskip 0.4cm

\textbf{$^*$ PWCSC-algorithm-B based on the level-leaf.}

\textbf{Initialization-B.} Suppose that $T$ is a tree admitting a $W$-constraint labeling $f$, such that $|f(V(T))|=|V(T)|$, and the induced edge color $f(uv)$ for each edge $uv\in E(T)$ holds the $W$-constraint $f(uv)=W\langle f(u),f(v)\rangle $.

\vskip 0.2cm

\textbf{Step B-1.} Let $L(T)$ be the set of leaves of $T$. Define a set-coloring $F_{le}$ for $T$ as: $F_{le}(x)=\{f(x),f(u)\}$ for $x\in L(T)$ and $xu\in E(T)$.

\vskip 0.2cm

\textbf{Step B-2.} Let $L(T_1)$ be the set of leaves of $T_1$, where $T_1=T-L(T)$. Color each leaf $y\in L(T_1)$ with $F_{le}(y)=\{f(y),f(v)\}$ if the edge $yv\in E(T)$.

\vskip 0.2cm

\textbf{Step B-3.} After doing $k-1$ times Step B-2, if tree $T_k=T_{k-1}-L(T_{k-1})$ has its own diameter $D(T_k)\geq 3$ goto Step B-2. Otherwise, $T_k$ is a star $K_{1,n}$ with the vertex set $V(K_{1,n})=\{x_0, y_i:i\in [1,n]\}$ and the edge set $E(K_{1,n})=\{x_0y_i:i\in [1,n]\}$. Color each leaf $y_i\in L(T_k)$ with $F_{le}(y_i)=\{f(y_i),f(x_0)\}$ and $F_{le}(x_0)=\{f(x_0)\}$.

\vskip 0.2cm

\textbf{Step B-4.} By the $W$-constraint we color each edge $uv\in E(T)$ with
\begin{equation}\label{eqa:w-condition-step-b-ee}
F_{le}(uv)=[F_{le}(u)\cap F_{le}(v)]\cup \{W\langle a,b\rangle:~a\in F_{le}(u),~b\in F_{le}(v)\}
\end{equation}since $F_{le}(u)\cap F_{le}(v)\neq \emptyset$.

\vskip 0.2cm

\textbf{Step B-5.} Return the $W$-constraint proper total set-coloring $F_{le}$ of the tree $T$ by the conditions in Initialization.

\vskip 0.4cm

\textbf{$^*$ PWCSC-algorithm-C based on the neighbor-vertex.}

\textbf{Initialization-C.} Suppose that a tree $T$ admits a $W$-constraint labeling $f$ holding $|f(V(T))|=|V(T)|$, and the induced edge color $f(uv)$ for each edge $uv\in E(T)$ holds the $W$-constraint $f(uv)=W\langle f(u),f(v)\rangle $.

\vskip 0.2cm

\textbf{Step C-1.} Define a total set-coloring $F_{nv}$ as: $F_{nv}(x)=\{f(y):y\in N(x)\}$ for each vertex $x\in V(T)$. Clearly, $F_{nv}(x)\neq F_{nv}(u)$ for distinct $x,u\in V(T)$, since $|f(V(T))|=|V(T)|$.

\vskip 0.2cm

\textbf{Step C-2.} By the $W$-constraint we color the edges of the tree $T$ as
\begin{equation}\label{eqa:w-condition-step-c2-ee}
F_{nv}(uv)=[F_{nv}(u)\cap F_{nv}(v)]\cup \{W\langle a,b\rangle:~a\in F_{nv}(u),~b\in F_{nv}(v)\},~uv\in E(T)
\end{equation}since $F_{nv}(u)\cap F_{nv}(v)\neq \emptyset$.

\vskip 0.2cm

\textbf{Step C-3.} Return the $W$-constraint proper total set-coloring $F_{nv}$ of the tree $T$.

\vskip 0.4cm

\textbf{$^*$ PWCSC-algorithm-D based on the neighbor-edge.}

\textbf{Initialization-D.} Suppose that $T$ is a tree admitting a $W$-constraint total labeling $f$ holding $f(uv)\neq f(xy)$ for any pair of edges $uv$ and $xy$ of $T$, and each edge $uv\in E(T)$ holds the $W$-constraint $f(uv)=W\langle f(u),f(v)\rangle $.

\vskip 0.2cm

\textbf{Step D-1.} Define a total set-coloring $F_{ne}$ by setting $F_{ne}(x)=\{f(xz):z\in N(x)\}$ for each vertex $x\in V(T)$, clearly, $F_{ne}(x)\neq F_{ne}(u)$ for distinct vertices $x,u\in V(T)$, since $|f(E(T))|=|E(T)|$.

\vskip 0.2cm

\textbf{Step D-2.} By the $W$-constraint we color the edges of the tree $T$ as
\begin{equation}\label{eqa:w-condition-step-c2-ee}
F_{ne}(uv)=[F_{ne}(u)\cap F_{ne}(v)]\cup \{W\langle a,b\rangle:~a\in F_{ne}(u),~b\in F_{ne}(v)\},~uv\in E(T)
\end{equation}since $F_{ne}(u)\cap F_{ne}(v)\neq \emptyset$.

\vskip 0.2cm

\textbf{Step D-3.} Return the $W$-constraint proper total set-coloring $F_{ne}$ of the tree $T$.

\vskip 0.4cm

\textbf{$^*$ PWCSC-algorithm-E based on the neighbor-edge-vertex.}

\textbf{Initialization-E.} Suppose that $T$ is a tree admitting a $W$-constraint total labeling $f$ holding $f(uv)\neq f(xy)$ for any pair of distinct edges $uv,xy\in E(T)$, and each edge $uv\in E(T)$ holds the $W$-constraint $f(uv)=W\langle f(u),f(v)\rangle $.

\vskip 0.2cm

\textbf{Step E-1.} Define a total set-coloring $F_{nve}$ by setting $F_{nve}(x)=\{f(y):y\in N(x)\}\cup \{f(xz):z\in N(x)\}$ for each $x\in V(T)$, clearly, $F_{nve}(x)\neq F_{nve}(u)$ for distinct vertices $x,u\in V(T)$, since $|f(E(T))|=|E(T)|$.

\vskip 0.2cm

\textbf{Step D-2.} By the $W$-constraint we color the edges of the tree $T$ as
\begin{equation}\label{eqa:w-condition-step-c2-ee}
F_{nve}(uv)=[F_{nve}(u)\cap F_{nve}(v)]\cup \{W\langle a,b\rangle:~a\in F_{nve}(u),~b\in F_{nve}(v)\},~uv\in E(T)
\end{equation}since $F_{nve}(uv) \subseteq F_{nve}(u)\cap F_{nve}(v)$.

\vskip 0.2cm

\textbf{Step E-3.} Return the $W$-constraint proper total set-coloring $F_{nve}$ of the tree $T$.

\subsubsection{Graph homomorphisms with parameterized set-colorings}

Since a connected non-tree $(p,q)$-graph $G$ can be vertex-split into a tree $T$ of $q+1$ vertices by doing the vertex-splitting tree-operation to $G$, so we have a set $T_{ree}(G)$ of trees obtained from vertex-splitting $G$, such that each tree $T\in T_{ree}(G)$ is graph homomorphism into $G$, that is $T\rightarrow G$. So, we can use a connected non-tree $(p,q)$-graph $G$ and its tree set $T_{ree}(G)$ in asymmetric cryptosystem, and make number-based strings generated from the graph $G$ and the tree set $T_{ree}(G)$.

\vskip 0.4cm

\textbf{Situation-A.} Suppose that a connected non-tree $(p,q)$-graph $G$ is not colored by any coloring. Notice that each tree $T\in T_{ree}(G)$ admits each one of the colorings: graceful $(k,d)$-total coloring, harmonious $(k,d)$-total coloring, (odd-edge) edge-magic $(k,d)$-total coloring, (odd-edge) graceful-difference $(k,d)$-total coloring, (odd-edge) edge-difference $(k,d)$-total coloring, (odd-edge) felicitous-difference $(k,d)$-total coloring and edge-antimagic $(k,d)$-total coloring introduced in Definition \ref{defn:kd-w-type-colorings}, Definition \ref{defn:odd-edge-W-type-total-labelings-definition}, Definition \ref{defn:kd-w-type-coloring-transfoemations} and Theorem \ref{thm:equivalent-k-d-total-colorings}.

Suppose that a tree $T\in T_{ree}(G)$ admits a $W$-constraint coloring/labeling $f$. By the PWCSC-algorithms introduced above, the tree $T$ admits a $W$-constraint set-coloring $F_f$ induced from the $W$-constraint coloring/labeling $f$.

Using the graph homomorphism $T\rightarrow G$ defined on a vertex mapping $\varphi:V(T)\rightarrow V(G)$, such that $\varphi(u)\varphi(v)\in E(G)$ for each edge $uv\in E(T)$, thereby, this graph $G$ admits a $W$-constraint set-coloring $F^*_f$ defined by

(i) $F^*_f(w)=\{F_f(x): \varphi(x)=w,~x\in V(T)\}$ for each vertex $w\in V(G)$, and

(ii) $F^*_f(\varphi(u)\varphi(v))=F_f(uv)$ for each edge $uv \in E(T)$.

\vskip 0.2cm

\textbf{Analysis of complexity of Situation-A:}

\begin{asparaenum}[\textbf{\textrm{Complexity-A.}}1.]
\item \textbf{Determining} the tree set $T_{ree}(G)$ obtained by vertex-splitting the connected non-tree $(p,q)$-graph $G$ into trees will meet the Subgraph Isomorphic Problem, although each tree $T\in T_{ree}(G)$ has exactly $q$ edges.
\item There is no way to know \textbf{how many} $W$-constraint colorings/labelings admitted by each tree $T\in T_{ree}(G)$, and moreover \textbf{no algorithm} can realize all colorings/labelings holding a fixed $W$-constraint for each tree $T\in T_{ree}(G)$.
\end{asparaenum}

\vskip 0.4cm

\textbf{Situation-B.} Suppose that a connected non-tree $(p,q)$-graph $G$ admits a $W$-constraint coloring/labeling $g$, so each tree $H\in T_{ree}(G)$ admits a $W$-constraint coloring/labeling $g^*$ induced by the $W$-constraint coloring/labeling $g$.

Since there is a vertex mapping $\theta:V(H)\rightarrow V(G)$ with $\theta(u)\theta(v)\in E(G)$ for each edge $uv\in E(H)$, the \emph{colored graph homomorphism} $H\rightarrow G$ means a \emph{Topcode-matrix homomorphism}
\begin{equation}\label{eqa:topcode-matrix-homomorphism}
T_{code}(H,g^*)\rightarrow T_{code}(G,g)
\end{equation}Thereby, we have two graph sets $S_{graph}[T_{code}(H,g^*)]$ and $S_{graph}[T_{code}(G,g)]$, such that

(a) Each graph $J\in S_{graph}[T_{code}(H,g^*)]$ admits a $W$-constraint coloring/labeling $f_J$ holding $T_{code}(J,f_J)=T_{code}(H,g^*)$; and

(b) each graph $I\in S_{graph}[T_{code}(G,g)]$ holds $T_{code}(I,f_I)=T_{code}(G,g)$ for a $W$-constraint coloring/labeling $f_I$ admitted by $I$.

Hence, we get a \emph{graph-set homomorphism}
\begin{equation}\label{eqa:graph-set-homomorphism}
S_{graph}[T_{code}(H,g^*)]\rightarrow S_{graph}[T_{code}(G,g)]
\end{equation}

\vskip 0.2cm

\textbf{Analysis of complexity of Situation-B:}

\begin{asparaenum}[\textbf{\textrm{Complexity-B.}}1.]
\item \textbf{Determining} two graph sets $S_{graph}[T_{code}(H,g^*)]$ and $S_{graph}[T_{code}(G,g)]$ will meet the Subgraph Isomorphic Problem, a NP-hard problem.
\item If a graph $I\in S_{graph}[T_{code}(G,g)]$ is a \emph{public-key}, \textbf{no algorithm} is for finding a \emph{private-key} $J\in S_{graph}[T_{code}(H,g^*)]$, such that $J\rightarrow I$ is just a colored graph homomorphism.
\end{asparaenum}

\begin{figure}[h]
\centering
\includegraphics[width=16.4cm]{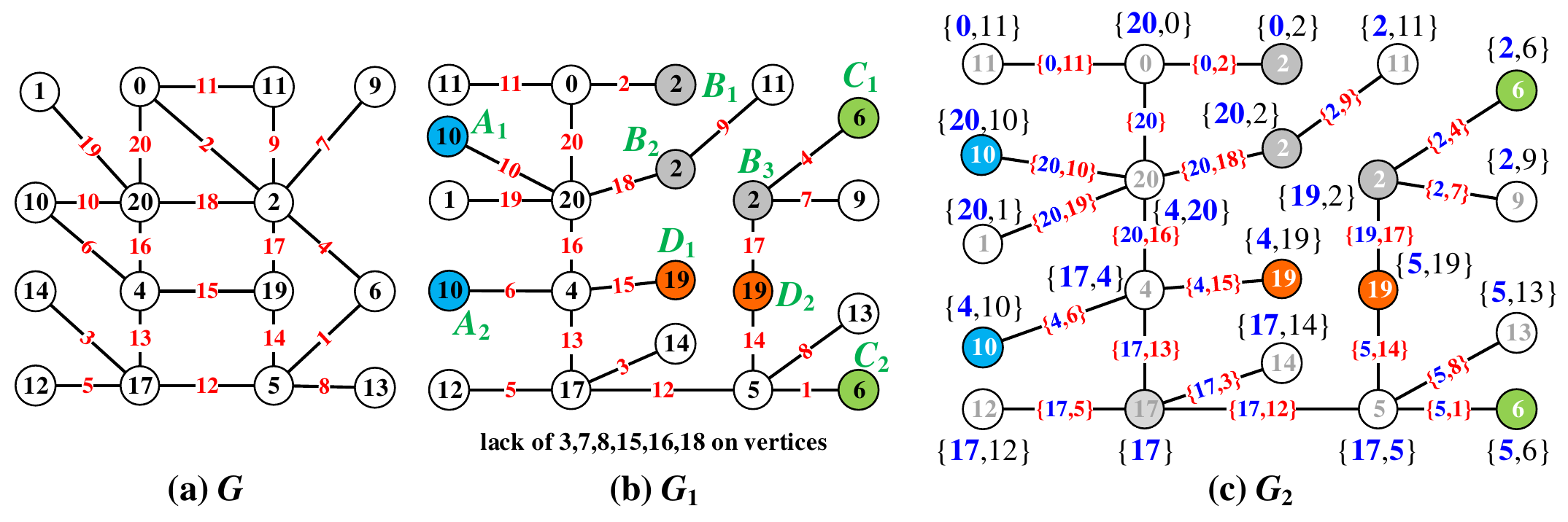}\\
\caption{\label{fig:PWCSC-algorithm-cycle-11}{\small An example for understanding the Situation-B, cited from \cite{Yao-Ma-arXiv-2201-13354v1}.}}
\end{figure}

\begin{figure}[h]
\centering
\includegraphics[width=13.4cm]{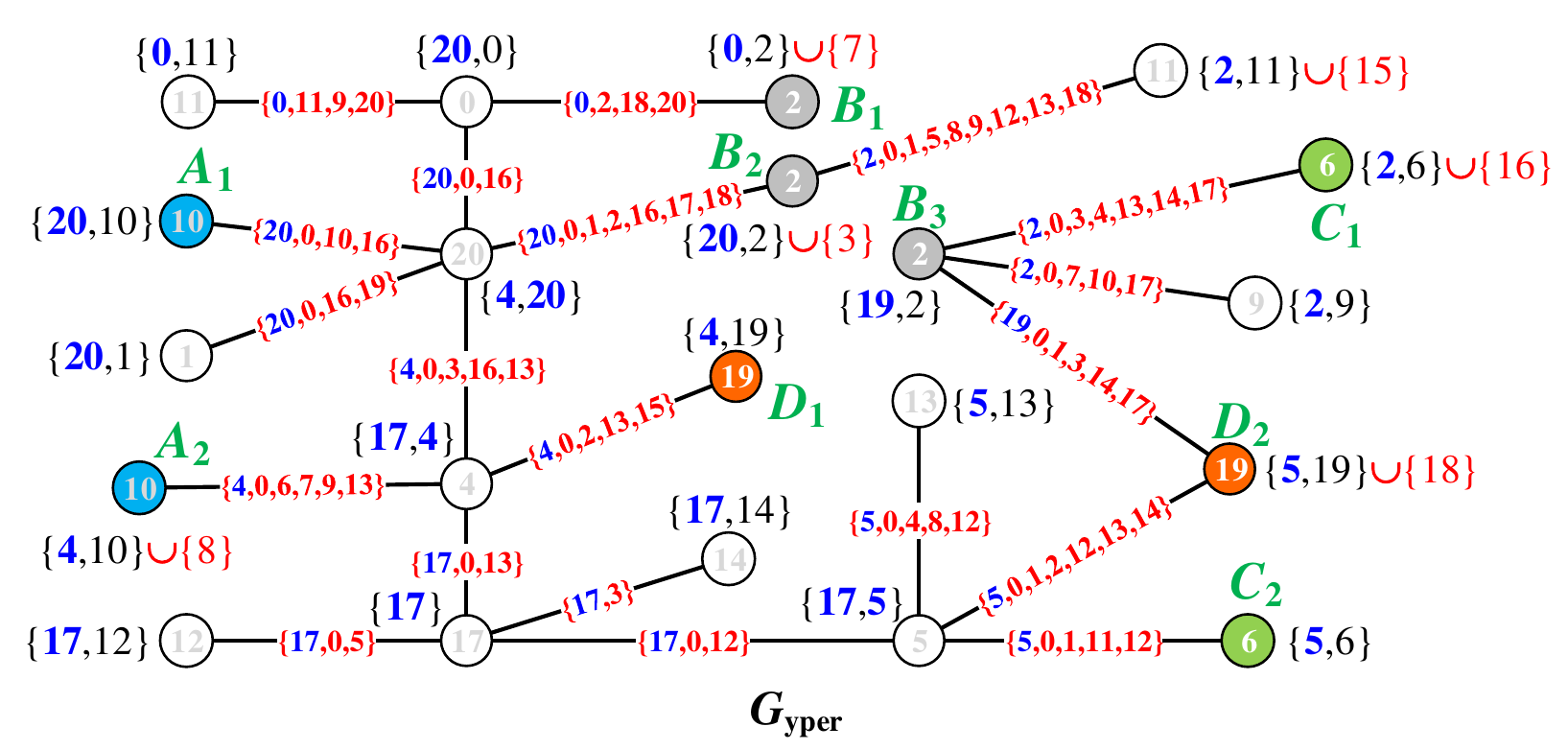}\\
\caption{\label{fig:PWCSC-algorithm-cycle-22}{\small Another example for understanding the Situation-B, cited from \cite{Yao-Ma-arXiv-2201-13354v1}.}}
\end{figure}

\subsubsection{Normal set-colorings based on hyperedge sets}

\begin{defn} \label{defn:tradition-vs-set-colorings}
$^*$ Let $G$ be a $(p,q)$-graph, and let $\Lambda$ be a finite set of numbers. There is a \emph{set-coloring} $F: S\rightarrow \mathcal{E}$ for a \emph{hyperedge set} $\mathcal{E}\subseteq \Lambda^2$, where $S\subseteq V(G)\cup E(G)$. There are the following constraints:
\begin{asparaenum}[\textbf{\textrm{Nset}}-1.]
\item \label{old:vertex} $S=V(G)$.
\item \label{old:edge} $S=E(G)$.
\item \label{old:total} $V(G)\cup E(G)$.
\item \label{old:adjacent-vertex} $F(u)\neq F(v)$ for each edge $uv\in E(G)$.
\item \label{old:adjacent-edge} $F(uv)\neq F(uw)$ for adjacent edges $uv,uw\in E(G)$, where $v,w\in N(u)$ and $u\in V(G)$.
\item \label{old:incident-edge-vertex} $F(u)\neq F(uv)$ and $F(v)\neq F(uv)$ for each edge $uv\in E(G)$.
\item \label{old:hyperedge-set} $\Lambda=\bigcup _{e\in \mathcal{E}}e$.
\item \label{old:join-oper-verticeice} $F(u)\cap F(v)\neq \emptyset$ for each edge $uv\in E(G)$.
\item \label{old:join-oper-dajacent-edges} $F(uv)\cap F(uw)\neq \emptyset$ for adjacent edges $uv,uw\in E(G)$.
\item \label{old:join-oper-vertex-edge} $F(uv)\cap F(u)\neq \emptyset$ and $F(uv)\cap F(v)\neq \emptyset$ for each edge $uv\in E(G)$.
\end{asparaenum}
\textbf{Then}

\noindent --------- \emph{traditional set-colorings}

\begin{asparaenum}[\textbf{\textrm{Scolo}}-1.]
\item $F$ is a \emph{proper set-coloring} if Nset-\ref{old:vertex} and Nset-\ref{old:adjacent-vertex} hold true.
\item $F$ is a \emph{proper edge set-coloring} if Nset-\ref{old:edge} and Nset-\ref{old:adjacent-edge} hold true.
\item $F$ is a \emph{proper total set-coloring} if Nset-\ref{old:total}, Nset-\ref{old:adjacent-vertex}, Nset-\ref{old:adjacent-edge} and Nset-\ref{old:incident-edge-vertex} hold true.

\noindent --------- \emph{hyperedge set-colorings}

\item $F$ is a \emph{proper hyperedge set-coloring} if Nset-\ref{old:vertex}, Nset-\ref{old:adjacent-vertex} and Nset-\ref{old:hyperedge-set} hold true.
\item $F$ is a \emph{proper edge hyperedge set-coloring} if Nset-\ref{old:edge}, Nset-\ref{old:adjacent-edge} and Nset-\ref{old:hyperedge-set} hold true.
\item $F$ is a \emph{proper total hyperedge set-coloring} if Nset-\ref{old:total}, Nset-\ref{old:adjacent-vertex}, Nset-\ref{old:adjacent-edge}, Nset-\ref{old:incident-edge-vertex} and Nset-\ref{old:hyperedge-set} hold true.

\noindent --------- \emph{intersected set-colorings}

\item $F$ is a \emph{proper intersected set-coloring} if Nset-\ref{old:vertex}, Nset-\ref{old:adjacent-vertex} and Nset-\ref{old:join-oper-verticeice} hold true.
\item $F$ is a \emph{proper edge intersected set-coloring} if Nset-\ref{old:edge}, Nset-\ref{old:adjacent-edge} and Nset-\ref{old:join-oper-dajacent-edges} hold true.
\item $F$ is a \emph{proper totally intersected set-coloring} if Nset-\ref{old:total}, Nset-\ref{old:adjacent-vertex}, Nset-\ref{old:adjacent-edge}, Nset-\ref{old:incident-edge-vertex}, Nset-\ref{old:join-oper-verticeice}, Nset-\ref{old:join-oper-dajacent-edges} and Nset-\ref{old:join-oper-vertex-edge} hold true.

\noindent --------- \emph{intersected hyperedge set-colorings}

\item $F$ is a \emph{proper intersected hyperedge set-coloring} if Nset-\ref{old:vertex}, Nset-\ref{old:adjacent-vertex}, Nset-\ref{old:hyperedge-set} and Nset-\ref{old:join-oper-verticeice} hold true.
\item $F$ is a \emph{proper edge intersected hyperedge set-coloring} if Nset-\ref{old:edge}, Nset-\ref{old:adjacent-edge}, Nset-\ref{old:hyperedge-set} and Nset-\ref{old:join-oper-dajacent-edges} hold true.
\item $F$ is a \emph{proper all-intersected hyperedge set-coloring} if Nset-\ref{old:total}, Nset-\ref{old:adjacent-vertex}, Nset-\ref{old:adjacent-edge}, Nset-\ref{old:incident-edge-vertex}, Nset-\ref{old:hyperedge-set}, Nset-\ref{old:join-oper-verticeice}, Nset-\ref{old:join-oper-dajacent-edges} and Nset-\ref{old:join-oper-vertex-edge} hold true.\qqed
\end{asparaenum}
\end{defn}

\begin{defn} \label{defn:normai-set-coloring-intersected-graph}
$^*$ (A-1) If a graph $G$ admits a \emph{proper intersected hyperedge set-coloring} defined in Definition \ref{defn:tradition-vs-set-colorings}, and each pair of hyperedges $e,e\,' \in \mathcal{E}$ holding $e\cap e\,'\neq \emptyset$ corresponds to an edge $xy\in E(G)$ with $F(x)=e$ and $F(y)=e\,'$ and $\Lambda=\bigcup _{e\in \mathcal{E}}e$, then we call $G$ an \emph{intersected-graph} of the hypergraph $\mathcal{H}_{yper}=(\Lambda,\mathcal{E})$ defined in \cite{Jianfang-Wang-Hypergraphs-2008}.

(A-2) Suppose that $F$ is the proper intersected hyperedge set-coloring of the intersected-graph $G$ of some hypergraph $\mathcal{H}_{yper}=(\Lambda,\mathcal{E})$, then the intersected-graph $G$ admits an \emph{edge-intersected total set-coloring} $F^*$ defined by $F^*(u)=F(u)$ for each vertex $u\in V(G)$, and $F^*(xy)=F(x)\cap F(y)$ for each edge $xy\in E(G)$.

(B-1) If a graph $G$ admits a \emph{proper edge intersected hyperedge set-coloring} defined in Definition \ref{defn:tradition-vs-set-colorings}, such that any two hyperedges $e,e\,' \in \mathcal{E}$ holding $e\cap e\,'\neq \emptyset$ correspond to two adjacent edges $xy,xw\in E(G)$ with $F(xy)=e$ and $F(xw)=e\,'$, then we call $G$ an \emph{edge-intersected graph} of the hypergraph $\mathcal{H}_{yper}=(\Lambda,\mathcal{E})$.

(B-2) Suppose that $\varphi$ is the proper edge intersected hyperedge set-coloring of the edge-intersected graph $G$ of some hypergraph $\mathcal{H}_{yper}=(\Lambda,\mathcal{E})$, then the edge-intersected graph $G$ admits a \emph{vertex-intersected total set-coloring} $\varphi^*$ defined as: $\varphi^*(u)=\{F(uv)\cap F(uw):~v,w\in N(u)\}$ for each vertex $u\in V(G)$, and each edge $xy\in E(G)$ is colored with $\varphi^*(xy)=F(xy)$.\qqed
\end{defn}

\begin{problem}\label{qeu:444444}
For a $(p,q)$-graph $G$ appeared in Definition \ref{defn:tradition-vs-set-colorings}, it may be interesting to consider the following extremum questions:
\begin{asparaenum}[\textbf{\textrm{Eque}}-1.]
\item \textbf{Find} the \emph{extremum set-number} $\Lambda_{ex}(G)=\min_F \{|\Lambda|\}$ over all proper set-colorings.
\item \textbf{Find} the \emph{extremum set-index} $\Lambda\,'_{ex}(G)=\min_F \{|\Lambda|\}$ over all proper edge set-colorings.
\item \textbf{Find} the \emph{extremum total set-number} $\Lambda\,''_{ex}(G)=\min_F \{|\Lambda|\}$ over all proper total set-colorings.
\end{asparaenum}
\end{problem}

\begin{problem}\label{qeu:444444}
Consider the set-colorings defined in Definition \ref{defn:tradition-vs-set-colorings} based on the following cases:
\begin{asparaenum}[\textbf{\textrm{Case}}-1.]
\item Each hyperedge $e\in \mathcal{E}$ has its own cardinality $|e|\geq 2$.
\item Any two hyperedges $e,e\,'\in \mathcal{E}$ hold $e\not \subset e\,'$ and $e\,'\not \subset e$.
\item Each hyperedge $e\in \mathcal{E}$ holds $|e|=k\geq 2$, so the corresponding to each set-coloring is uniformed.
\item The sequence $\{|e_i|\}$ for $\mathcal{E}=\{e_i:i\in [1,n]\}$ forms a series, such as arithmetic progression, geometric series, Fibonacci series, \emph{etc}.
\end{asparaenum}
\end{problem}

For a set $S\,^i_X=\{a\,^i_1,a\,^i_2,\dots , a\,^i_{s(i)}\}$ with integer $s(i)\geq 1$ and $i\in [1,n]$ and integers $k,d\geq 0$, we define the following two operations:
\begin{equation}\label{eqa:555555}
{
\begin{split}
d\cdot \langle S\,^i_X\rangle =&\{d\cdot a\,^i_1,d\cdot a\,^i_2,\dots , d\cdot a\,^i_{s(i)}\}\\
(k[+]d)\langle S\,^i_X\rangle =&\{k+d\cdot a\,^i_1,k+d\cdot a\,^i_2,\dots , k+d\cdot a\,^i_{s(i)}\}
\end{split}}
\end{equation}
And moreover, for a set $S^*=\{S\,^1_X,S\,^2_X,\dots ,S\,^n_X\}$, we have two operations:

(iii) $d \langle S^*\rangle =\{d\langle S\,^1_X\rangle, d\langle S\,^2_X\rangle,\dots , d\langle S\,^n_X\rangle\}$;

(iv) $(k[+]d)\langle S^*\rangle =\{(k[+]d)\langle S\,^1_X\rangle, (k[+]d)\langle S\,^2_X\rangle,\dots , (k[+]d)\langle S\,^n_X\rangle\}$.

\begin{defn} \label{defn:normai-kd-type-set-coloring}
$^*$ A connected bipartite $(p,q)$-graph $G$ has its own vertex set bipartition $V(G)=X\cup Y$ and admits a $W$-constraint set-coloring $F$ defined in Definition \ref{defn:tradition-vs-set-colorings}, then $G$ has its own set-type Topcode-matrix $T_{code}(G,F)=(X_S,E_S,Y_S)^T$ with $X_S\cap Y_S=\emptyset$, where each element in three \emph{set vectors} $X_S,E_S$ and $Y_S$ is a \emph{set}. So we have a set-type parameterized Topcode-matrix
\begin{equation}\label{eqa:66666666666}
\centering
{
\begin{split}
P^{set}_{ara}(G,\Phi)=k\cdot I\,^0[+]d\cdot T_{code}(G,F)=(d\cdot X_S,~(k[+]d) \langle E_S\rangle,~(k[+]d)\langle Y_S\rangle )^T
\end{split}}
\end{equation} which defines a $(k,d)$-type set-coloring
\begin{equation}\label{eqa:555555}
\Phi:S\rightarrow \mathcal{E}^P,~S\subseteq V(G)\cup E(G),~\mathcal{E}^P\subseteq \Lambda^2_{(m,b,n,k,a,d)}
\end{equation} for the connected bipartite $(p,q)$-graph $G$.\qqed
\end{defn}

\subsection{Miscellaneous colorings with parameters}

\subsubsection{Total colorings subject to $(abc)$-magic-type functions}

\begin{defn} \label{defn:combinatoric-definition-total-coloring-abc}
\cite{Yao-Wang-Ma-Su-Wang-Sun-2020ITNEC, Bing-Yao-2020arXiv} Suppose that a bipartite graph $G$ admits a proper total coloring $f:V(G)\cup E(G)\rightarrow [1,M]$ with $V(G)=X\cup Y$ and $X\cap Y=\emptyset $. We define an \emph{$(abc)$-magic-type function} $c_f(uv)(a,b,c)$ with three non-negative integers $a,b,c$ for each edge $uv\in E(G)$ with $u\in X$ and $v\in Y$, and define a parameter
\begin{equation}\label{eqa:edge-difference-total-coloring}
B^*_{\alpha}(G,f,M)(a,b,c)=\max_{uv \in E(G)}\{c_f(uv)(a,b,c)\}-\min_{xy \in E(G)}\{c_f(xy)(a,b,c)\}.
\end{equation}

If $B^*_{\alpha}(G,f,M)(a,b,c)=0$, we call $f$ \emph{$(abc)$-$W$-constraint proper total coloring}, the smallest number
\begin{equation}\label{eqa:minimum}
\chi\,''_{\alpha}(G)(a,b,c) =\min_f \{M:~B^*_{\alpha}(G,f,M)(a,b,c)=0\}
\end{equation}
over all $(abc)$-$W$-constraint proper total colorings of $G$ is called \emph{$(abc)$-$W$-constraint proper total chromatic number}, and $f$ is called a \emph{perfect $(abc)$-$W$-constraint proper total coloring} if $\chi\,''_{\alpha}(G)(a,b,c)=\chi\,''(G)$. If $B^*_{\alpha}(G,f,M)(a,b,c)=1$, we call $f$ \emph{equitably $(abc)$-$W$-constraint proper total coloring}.
\begin{asparaenum}[\textrm{\textbf{TCol}}-1. ]
\item We call $f$ an \emph{(resp. perfect) (equitably) $(abc)$-edge-magic proper total coloring} of $G$ if the $(abc)$-edge-magic function
\begin{equation}\label{eqa:555555}
c_f(uv)(a,b,c)=af(u)+bf(v)+cf(uv)
\end{equation} we rewrite
\begin{equation}\label{eqa:555555}
B^*_{\alpha}(G,f,M)(a,b,c)=B^*_{emt}(G,f, M)(a,b,c)
\end{equation} and call the number $\chi\,''_{\alpha}(G)(a,b,c)=\chi\,''_{emt}(G)(a,b,c)$ \emph{(equitably) $(abc)$-edge-magic total chromatic number}.
\item We call $f$ an \emph{(resp. perfect) (equitably) $(abc)$-edge-difference proper total coloring} of $G$ if the $(abc)$-edge-difference function
\begin{equation}\label{eqa:555555}
c_f(uv)(a,b,c)=cf(uv)+|af(u)-bf(v)|
\end{equation} we rewrite
\begin{equation}\label{eqa:555555}
B^*_{\alpha}(G,f,M)(a,b,c)=B^*_{edt}(G,f, M)(a,b,c)
\end{equation} and call the number $\chi\,''_{\alpha}(G)(a,b,c)=\chi\,''_{edt}(G)(a,b,c)$ \emph{(equitably) $(abc)$-edge-difference total chromatic number}.
\item We call $f$ an \emph{(resp. perfect) (equitably) $(abc)$-felicitous-difference proper total coloring} of $G$ if the $(abc)$-felicitous-difference function
\begin{equation}\label{eqa:555555}
c_f(uv)(a,b,c)=|af(u)+bf(v)-cf(uv)|
\end{equation} and we rewrite
\begin{equation}\label{eqa:555555}
B^*_{\alpha}(G,f,M)(a,b,c)=B^*_{fdt}(G,f,M)(a,b,c)
\end{equation} and call the number $\chi\,''_{\alpha}(G)(a,b,c)=\chi\,''_{fdt}(G)(a,b,c)$ \emph{(equitably) $(abc)$-felicitous-difference total chromatic number}.
\item We refer to $f$ as an \emph{(resp. perfect) (equitably) $(abc)$-graceful-difference proper total coloring} of $G$ if the $(abc)$-graceful-difference function
\begin{equation}\label{eqa:555555}
c_f(uv)(a,b,c)=\big ||af(u)-bf(v)|-cf(uv)\big |
\end{equation} and we rewrite
\begin{equation}\label{eqa:555555}
B^*_{\alpha}(G,f,M)(a,b,c)=B^*_{gdt}(G,f,M)(a,b,c)
\end{equation} and call the number $\chi\,''_{\alpha}(G)(a,b,c)=\chi\,''_{gdt}(G)(a,b,c)$ \emph{(equitably) $(abc)$-graceful-difference total chromatic number}.\qqed
\end{asparaenum}
\end{defn}

\begin{rem}\label{rem:333333}
We can put forward various requirements for $(a,b,c)$ in Definition \ref{defn:combinatoric-definition-total-coloring-abc} to increase the difficulty of attacking topological coding, since the ABC-conjecture (also, Oesterl\'{e}-Masser conjecture, 1985) involves the equation $a+b=c$ and the relationship between prime numbers. Proving or disproving the ABC-conjecture could impact: Diophantine (polynomial) math problems including Tijdeman's theorem, Vojta's conjecture, Erd\"{o}s-Woods conjecture, Fermat's last theorem, Wieferich prime and Roth's theorem, and so on \cite{Cami-Rosso2017Abc-conjecture}.\paralled
\end{rem}

\begin{defn} \label{defn:111111}
$^*$ Let $T_{code}(G,f)$ be a Topcode-matrix of a bipartite graph $G$ admitting a total coloring $f$ defined in Definition \ref{defn:combinatoric-definition-total-coloring-abc}, and let $(a,c,b)^T$ be a matrix of order $3\times 1$ for three non-negative integers $a,b,c$. By Definition \ref{defn:colored-topcode-matrix}, we define a new Topcode-matrix
\begin{equation}\label{eqa:555555}
(a,c,b)^T\cdot T_{code}(G,f)=(a,c,b)^T\cdot (X_f,E_f,Y_f)^T=(aX_f,cE_f,bY_f)^T
\end{equation} with three vectors $aX_f=(a\cdot x_1,a\cdot x_2,\dots,a\cdot x_q)$, $cE_f=(c\cdot e_1,c\cdot e_2,\dots, c\cdot e_q)$ and $bY_f=(b\cdot y_1$, $b\cdot y_2$, $\dots, b\cdot y_q)$.\qqed
\end{defn}

\begin{thm}\label{thm:666666}
$^*$ If a bipartite $(p,q)$-graph $G$ admits a set-ordered $W$-constraint total coloring/labeling $f$, for constants $a, b , c $ and $ \lambda$, there are

(i) the $(abc)$-edge-magic constraint $af(x)+bf(y)+cf(xy)=\lambda$ for each edge $xy\in E(G)$;

(ii) the $(abc)$-felicitous-difference constraint $|af(x)+bf(y)-cf(xy)|=\lambda$ for each edge $xy\in E(G)$;

(iii) the $(abc)$-graceful-difference constraint $\big ||af(x)-bf(y)|-cf(xy)\big |=\lambda$ for each edge $xy\in E(G)$; ans

(iv) the $(abc)$-edge-difference constraint $cf(xy)+|af(x)-bf(y)|=c$ for each edge $xy\in E(G)$.

Then this graph $G$ admits a $W$-constraint $(k,d)$-total coloring/labeling $F$ defined by Eq.(\ref{eqa:w-type-coloring-k-d}) and Eq.(\ref{eqa:w-type-coloring-Topcode-ma}) in Definition \ref{defn:w-type-coloring-labeling-kd} holding the above four constraints.
\end{thm}
\begin{proof} Let $(X,Y)$ be the vertex bipartition of a bipartite $(p,q)$-graph $G$, so a set-ordered $W$-constraint total coloring/labeling $f$ of $G$ holds $\max f(X)<\min f(Y)$. Since $F(x)=f(x)d$, $F(y)=k+f(y)d$ and $F(xy)=k+f(xy)d$ obtained from the Topcode-matrix $P_{ara}(G,F)$, so we have:

(i) If the $(abc)$-edge-magic constraint $af(x)+bf(y)+cf(xy)=\lambda$ holds true for each edge $xy\in E(G)$, then the $(abc)$-edge-magic constraint
\begin{equation}\label{eqa:555555}
{
\begin{split}
aF(x)+bF(y)+cF(xy)&=af(x)d+b[k+f(y)d]+c[k+f(xy)d]\\
&=[af(x)+bf(y)+cf(xy)]d+(b +c)k\\
&=d \cdot \lambda+(b +c)k
\end{split}}
\end{equation} holds true for each edge $xy\in E(G)$.

(ii) If the $(abc)$-felicitous-difference constraint $|af(x)+bf(y)-cf(xy)|=\lambda$ holds true for each edge $xy\in E(G)$, then we get the $(abc)$-felicitous-difference constraint
\begin{equation}\label{eqa:555555}
{
\begin{split}
\big |aF(x)+bF(y)-cF(xy)\big |&=\big |af(x)d+b[k+f(y)d]-c[k+f(xy)d]\big |\\
&=\big |[af(x)+bf(y)-cf(xy)]d+(b -c)k\big |\\
&\leq \big |af(x)+bf(y)-cf(xy)\big |d+\big |b -c\big |k\\
&=d \cdot \lambda+|b-c|k
\end{split}}
\end{equation} for each edge $xy\in E(G)$.

(iii) If the $(abc)$-graceful-difference constraint $\big ||af(x)-bf(y)|-cf(xy)\big |=\lambda$ holds true for each edge $xy\in E(G)$, thus, the $(abc)$-graceful-difference constraint
\begin{equation}\label{eqa:555555}
{
\begin{split}
\big ||aF(x)-bF(y)|-cF(xy)\big |&=\big ||af(x)d-b[k+f(y)d]|-c[k+f(xy)d]\big |\\
&\leq\big |bf(y)-af(x)-cf(xy)\big |d+|b-c|k\\
&=d \cdot \lambda+|b -c|k
\end{split}}
\end{equation} or
\begin{equation}\label{eqa:555555}
{
\begin{split}
\big ||aF(x)-bF(y)|-cF(xy)\big |&=\big ||af(x)d-b[k+f(y)d]|-c[k+f(xy)d]\big |\\
&\leq\big |af(x)-bf(y)-cf(xy)\big |d+|b-c|k\\
&=d \cdot \lambda+|b -c|k
\end{split}}
\end{equation} holds true for each edge $xy\in E(G)$.

(iv) If the $(abc)$-edge-difference constraint $cf(xy)+|af(x)-bf(y)|=\lambda$ holds true for each edge $xy\in E(G)$, then we have the $(abc)$-graceful-difference constraint
\begin{equation}\label{eqa:555555}
{
\begin{split}
cF(xy)+\big |aF(x)-bF(y)\big |&=c[k+f(xy)d]+|af(x)d-b[k+f(y)d] |\\
&\leq [cf(xy)+|af(x)-bf(y)|]d+(b+c)k\\
&=d \cdot \lambda+(b+c)k
\end{split}}
\end{equation} for each edge $xy\in E(G)$.

Moreover, we have other two relations as follows:
$${
\begin{split}
|aF(x)-bF(y)|&=|b[k+f(y)d]-af(x)d|\leq |bf(y)-af(x)|d+bk\\
&=c[k+f(xy)d]+(b-c)k\\
&=cF(xy)+(b-c)k
\end{split}}
$$ if $cf(xy)=|bf(y)-af(x)|$, and
$${
\begin{split}
aF(x)+bF(y)&=af(x)d+b[k+f(y)d]=[bf(y)+af(x)]d+bk\\
&=c[k+f(xy)d]+(b-c)k\\
&=cF(xy)+(b-c)k
\end{split}}
$$ if $cf(xy)=af(x)+bf(y)$.

The proof of the theorem is complete.
\end{proof}

\subsubsection{Real-valued parameterized total colorings}

As known, graph colorings/labelings were defined on integer sets, but real number sets. Notice that some real-valued labelings have been investigated in \cite{A-Vietri-Ars2011}, he allowed the vertex labels of a graph
with $q$ edges to be real numbers in a real number interval $[0,q]^r$. For a simple graph $G$, he defined an injective mapping $f$ from $V(G)$ to $[0,q]^r$ to be a real-graceful labeling of $G$, such that the edge color is induced by $f(uv)=|f(u)-f(v)|$ for each edge $uv\in E(G)$, and $f(uv)\neq f(xy)$ if $uv\neq xy$ in $E(G)$. Moreover, if $f(x)$ is an integer for each vertex $x\in V(G)$, then $f$ is just the traditional graceful labeling.

\begin{defn}\label{defn:real-valued-total-cs}
\cite{Yao-Wang-Ma-Su-Wang-Sun-2020ITNEC} Let $R_S$ be a non-negative real number set, and let $G$ be a $(p,q)$-graph. If there exists a real number $k>0$, for any positive real number $\varepsilon$, $G$ admits a real-valued total coloring $f_{\varepsilon}$ on $R_S$, which induces an \emph{edge-function} $C(f_{\varepsilon},uv)$ for each edge $uv\in E(G)$, such that
\begin{equation}\label{eqa:converges-edge-magic-00}
|C(f_{\varepsilon},uv)-k|<\varepsilon
\end{equation}
then we say that $G$ converges to a \emph{$C$-edge-magic constant} $k$, each $f_{\varepsilon}$ is a \emph{$C$-edge-magic real-valued total coloring}, and write this fact as
\begin{equation}\label{eqa:converges-edge-magic-11}
\lim_{\varepsilon\rightarrow 0}C(f_{\varepsilon},uv)_{uv\in E(G)}=k
\end{equation}
We have:
\begin{asparaenum}[(i) ]
\item as the edge-function $C(f_{\varepsilon},uv)=f_{\varepsilon}(u)+f_{\varepsilon}(uv)+f_{\varepsilon}(v)$, we call $f_{\varepsilon}$ \emph{edge-magic real-valued total coloring}, and $k$ \emph{edge-magic real constant};

\item as the edge-function $C(f_{\varepsilon},uv)=f_{\varepsilon}(uv)+|f_{\varepsilon}(u)-f_{\varepsilon}(v)|$, we call $f_{\varepsilon}$ \emph{edge-difference real-valued total coloring}, and $k$ \emph{edge-difference real constant};

\item as the edge-function $C(f_{\varepsilon},uv)=\big ||f_{\varepsilon}(u)-f_{\varepsilon}(v)|-f_{\varepsilon}(uv)\big |$, we call $f_{\varepsilon}$ \emph{graceful-difference real-valued total coloring}, and $k$ \emph{graceful-difference real constant};
\item as the edge-function $C(f_{\varepsilon},uv)=|f_{\varepsilon}(u)+f_{\varepsilon}(v)-f_{\varepsilon}(uv)|$, we call $f_{\varepsilon}$ \emph{felicitous-difference real-valued total coloring}, and $k$ \emph{felicitous-difference real constant}.\qqed
\end{asparaenum}
\end{defn}

\begin{rem}\label{rem:333333}
It is not easy to estimate the spaces of all $C$-edge-magic real-valued total colorings of a $(p,q)$-graph $G$ made by the above real-valued total colorings, since these colorings are related with real numbers. \paralled
\end{rem}

\begin{thm} \label{thm:real-valued-total-coloring-1}
\cite{Yao-Wang-Ma-Su-Wang-Sun-2020ITNEC} For each $W$-constraint $\in \{$edge-magic, graceful-difference, edge-difference, felicitous-difference$\}$, if a connected graph $G$ admits a $W$-constraint real-valued total coloring, then the graph $G$ admits a $W$-constraint real-valued total coloring.
\end{thm}
\begin{proof}Suppose that a connected graph $G$ admits a proper total coloring $g$. We select a pair of non-negative functions $\alpha(\varepsilon)$ and $\beta(\varepsilon)$ with $\alpha(\varepsilon)\rightarrow 0$ and $\beta(\varepsilon)\rightarrow 1$ as $\varepsilon\rightarrow 0$. Thereby, we define a real-valued total coloring $f_{\varepsilon}$ as
\begin{equation}\label{eqa:0-1-real-valued-total-coloring}
f_{\varepsilon}(w)=\alpha(\varepsilon)+\beta(\varepsilon)g(w)
\end{equation}
for each element $w\in V(G)\cup E(G)$, and call $f_{\varepsilon}(w)$ defined in (\ref{eqa:0-1-real-valued-total-coloring}) a \emph{real-valued $(0,1)$-total coloring} of $G$. As we take $\beta(\varepsilon)=1-\alpha(\varepsilon)$, the equation (\ref{eqa:0-1-real-valued-total-coloring}) is changed as
\begin{equation}\label{eqa:0-1-real-valued-total-coloring-dual}
f_{\varepsilon}(w)=\alpha(\varepsilon)+[1-\alpha(\varepsilon)]g(w)=g(w)+[1-g(w)]\alpha(\varepsilon).
\end{equation}

\textbf{The edge-magic constraint.} Since the edge-function
\begin{equation}\label{eqa:555555}
C_1(f_{\varepsilon},uv)=3\alpha(\varepsilon)+\beta(\varepsilon)[g(u)+g(uv)+g(v)]
\end{equation} for each edge $uv\in E(G)$ in an edge-magic real-valued total coloring $f_{\varepsilon}$, so we have
{\small
$$
{
\begin{split}
|C_1(f_{\varepsilon},uv)-k_1|&=|C(f_{\varepsilon},uv)-\beta(\varepsilon)k_1+\beta(\varepsilon)k_1-k_1|\\
&=\big |3\alpha(\varepsilon)+\beta(\varepsilon)[g(u)+g(uv)+g(v)-k_1]+k_1[\beta(\varepsilon)-1]\big |\\
&\leq 3\alpha(\varepsilon)+\beta(\varepsilon)\varepsilon+k_1\cdot |\beta(\varepsilon)-1|
\end{split}}
$$
}which means that $|C_1(f_{\varepsilon},uv)-k_1|\rightarrow 0$ as $\varepsilon\rightarrow 0$.

\textbf{The edge-difference constraint.} A difference edge-magic real-valued total coloring $f_{\varepsilon}$ induces the edge-function
\begin{equation}\label{eqa:555555}
C_2(f_{\varepsilon},uv)=\alpha(\varepsilon)+\beta(\varepsilon)[g(uv)+|g(u)-g(v)|]
\end{equation} for each edge $uv\in E(G)$, then
{\small
$$
{
\begin{split}
|C_2(f_{\varepsilon},uv)-k_2|&=|C(f_{\varepsilon},uv)-\beta(\varepsilon)k_2+\beta(\varepsilon)k_2-k_2|\\
&=\big |\alpha(\varepsilon)+\beta(\varepsilon)[g(uv)+|g(u)-g(v)|-k_2]+k_2[\beta(\varepsilon)-1]\big |\\
&\leq \alpha(\varepsilon)+\beta(\varepsilon)\varepsilon+k_2\cdot |\beta(\varepsilon)-1|
\end{split}}
$$
}immediately, $\lim \limits_{\varepsilon\rightarrow 0}C_2(f_{\varepsilon},uv)_{uv\in E(G)}=k_2$.

\textbf{The graceful-difference constraint.} By the edge-function
\begin{equation}\label{eqa:555555}
C_3(f_{\varepsilon},uv)=\big |\beta(\varepsilon)\big [|g(u)-g(v)|-g(uv)\big ]-\alpha(\varepsilon)\big |
\end{equation} for each edge $uv\in E(G)$ under a mixed-edge-magic real-valued total coloring $f_{\varepsilon}$, we can estimate
{\small
$$
{
\begin{split}
|C_3(f_{\varepsilon},uv)-k_3|&=|C(f_{\varepsilon},uv)-\beta(\varepsilon)k_3+\beta(\varepsilon)k_3-k_3|\\
&\leq \alpha(\varepsilon)+\big |\beta(\varepsilon)[|g(u)-g(v)|-g(uv)-k_3]\big |+k_3|\beta(\varepsilon)-1|\\
&\leq \alpha(\varepsilon)+\beta(\varepsilon)\varepsilon+k_3\cdot |\beta(\varepsilon)-1|
\end{split}}
$$
}we claim that $\lim \limits_{\varepsilon\rightarrow 0}C_3(f_{\varepsilon},uv)_{uv\in E(G)}=k_3$.

\textbf{The felicitous-difference constraint. }Notice that the edge-function
\begin{equation}\label{eqa:555555}
C_4(f_{\varepsilon},uv)=\big |\alpha(\varepsilon)+\beta(\varepsilon)[g(u)+g(v)-g(uv)]\big |
\end{equation} for each edge $uv\in E(G)$ under a graceful edge-magic real-valued total coloring $f_{\varepsilon}$, the following deduction
$$
{
\begin{split}
|C_4(f_{\varepsilon},uv)-k_4|&=|C(f_{\varepsilon},uv)-\beta(\varepsilon)k_4+\beta(\varepsilon)k_4-k_4|\\
&\leq \alpha(\varepsilon)+\big |\beta(\varepsilon)[g(u)+g(v)-g(uv)-k_4]\big |+k_4|\beta(\varepsilon)-1|\\
&\leq \alpha(\varepsilon)+\beta(\varepsilon)\varepsilon+k_4\cdot |\beta(\varepsilon)-1|
\end{split}}
$$ tells us $\lim \limits_{\varepsilon\rightarrow 0}C_4(f_{\varepsilon},uv)_{uv\in E(G)}=k_4$.
\end{proof}

\begin{defn}\label{defn:real-valued-vertex-coloring}
\cite{Yao-Wang-Ma-Su-Wang-Sun-2020ITNEC} Let $R_S$ be a non-negative real number set, and let $G$ be a $(p,q)$-graph with edge set $E(G)=\{u_kv_k:~k\in [1,q]\}$. If a real-valued vertex coloring $f: V(G) \rightarrow R_S$ induces the edge color $f(u_kv_k)=|f(u_k)-f(v_k)|$ for each edge $u_kv_k\in E(G)$, such that
$$|f(u_kv_k)-k|<\varepsilon<1/2 \quad (\textrm{resp.}~ |f(u_kv_k)-(2k-1)|<\varepsilon<1/2)
$$ for each $k\in [1,q]$, we call $f$ a \emph{graceful real-valued vertex coloring} (resp. an \emph{odd-graceful real-valued vertex coloring}).\qqed
\end{defn}

For a graceful labeling $h: V(G) \rightarrow [0,q]$, we define a real-valued vertex coloring as
\begin{equation}\label{eqa:real-valued-vertex-coloring}
f(x)=\lambda(\varepsilon)+\theta(\varepsilon)h(x)
\end{equation}for $x\in V(G)$, where $\lambda(\varepsilon)\rightarrow 0$ and $\theta(\varepsilon)\rightarrow 1$ as $\varepsilon\rightarrow 0$. Since
$$
f(u_kv_k)=|f(u_k)-f(v_k)|=\theta(\varepsilon)\cdot |h(u_k)-h(v_k)|=\theta(\varepsilon)h(u_kv_k)
$$ so we have
$$
{
\begin{split}
|f(u_kv_k)-k| =&|\theta(\varepsilon)h(u_kv_k)-\theta(\varepsilon)k+\theta(\varepsilon)k-k|\\
\leq &\theta(\varepsilon)|h(u_kv_k)-k|+k\cdot |\theta(\varepsilon)-1|\\
\leq & \theta(\varepsilon)\cdot \varepsilon +k\cdot |\theta(\varepsilon)-1|,
\end{split}}
$$
which claims that $f$ is a graceful real-valued vertex coloring of $G$. So, we have shown the following result:

\begin{thm} \label{thm:real-valued-vertex-coloring}
\cite{Yao-Wang-Ma-Su-Wang-Sun-2020ITNEC} If a connected graph admits a graceful labeling (resp. an odd-graceful labeling), then it admits a graceful real-valued vertex coloring (resp. an odd-graceful real-valued vertex coloring).
\end{thm}

\begin{example}\label{exa:8888888888}
A real-valued vertex coloring $f$ has its own \emph{totally-dual real-valued vertex coloring} $\overline{f}$ defined by
$$\overline{f}(x)=M_f+m_f-f(x),~x\in V(G)~(\textrm{resp.}\quad x\in V(G)\cup E(G))
$$ where
$$
M_f=\max \{f(x):x\in V(G)\},\quad m_f=\min \{f(x):x\in V(G)\}
$$ resp.
$$
M_f=\max \{f(x):x\in V(G)\cup E(G)\},\quad m_f=\min \{f(x):x\in V(G)\cup E(G)\}
$$

Thereby, the totally-dual real-valued vertex coloring $\overline{f}$ of a graceful/odd-graceful real-valued vertex coloring $f$ is graceful/odd-graceful too. Thereby, the $(p,q)$-graph $G$ admits a parameterized real-valued coloring $F$ defined as \begin{equation}\label{eqa:555555}
F(x)=kf(x)+d\overline{f}(x)=kf(x)+d[M_f+m_f-f(x)]=(k-d)f(x)+d[M_f+m_f]
\end{equation} for each vertex $x\in V(G)$ (resp. $x\in V(G)\cup E(G)$). Let $f(uv)$ with $uv\in E(G)$ be the edge $uv$'s color induced by $f$, we get:

\textbf{Case 1.} The graceful constraint (resp. odd-graceful constraint)
$$F(uv)=|F(u)-F(v)|=|k-d|\cdot |f(u)-f(v)|=|k-d|\cdot f(uv)
$$ for each edge $uv\in E(G)$ holds true.

\textbf{Case 2.} If the edge-difference constraint $f(uv)+|f(u)-f(v)|=c_1$ with $uv\in E(G)$, then the edge-difference constraint
$$F(uv)+|F(u)-F(v)|=|k-d|\cdot [f(uv)+|f(u)-f(v)|]=|k-d|\cdot c_1
$$ for each edge $uv\in E(G)$ holds true.

\textbf{Case 3.1.} If there are the edge-magic constraint $f(u)+f(uv)+f(v)=c_2$ with $uv\in E(G)$ and $k>d$, then the edge-magic constraint
$${
\begin{split}
F(u)+F(uv)+F(v)=&(k-d)f(u)+d[M_f+m_f]+|k-d|\cdot f(uv)\\
&+(k-d)f(v)+d[M_f+m_f]\\
=&(k-d)\cdot [f(u)+f(uv)+f(v)]+2d[M_f+m_f]\\
=&(k-d)\cdot c_2+2d[M_f+m_f]
\end{split}}
$$ for each edge $uv\in E(G)$ holds true.

\textbf{Case 3.2.} If there are $f(u)-f(uv)+f(v)=c_3$ with $uv\in E(G)$ and $k<d$, then the edge-magic constraint
$${
\begin{split}
F(u)+F(uv)+F(v)=&(k-d)f(u)+d[M_f+m_f]+|k-d|\cdot f(uv)\\
&+(k-d)f(v)+d[M_f+m_f]\\
=&(k-d)\cdot [f(u)-f(uv)+f(v)]+2d[M_f+m_f]\\
=&2d[M_f+m_f]-(d-k)c_3
\end{split}}
$$ for each edge $uv\in E(G)$ holds true.

\textbf{Case 4.} The graceful-difference constraint $\big ||F(u)-F(v)|-F(uv)\big |=0$ for each edge $uv\in E(G)$ holds true.

\textbf{Case 5.1.} If there are the felicitous-difference constraint $|f(u)+f(v)-f(uv)|=c_4$ with $uv\in E(G)$ and $k>d$, then the felicitous-difference constraint
$${
\begin{split}
\big |F(u)+F(v)-F(uv)\big |=&\big |(k-d)f(u)+d[M_f+m_f]+(k-d)f(v)\\
&+d[M_f+m_f]-|k-d|\cdot f(uv)\big |\\
=&(k-d)|f(u)+f(v)-f(uv)|+2d[M_f+m_f]\\
=&(k-d)\cdot c_4+2d[M_f+m_f]
\end{split}}
$$ for each edge $uv\in E(G)$ holds true.

\textbf{Case 5.2.} If there are the edge-magic constraint $f(u)+f(v)+f(uv)=c_5$ and $k<d$, then the felicitous-difference constraint
$${
\begin{split}
\big |F(u)+F(v)-F(uv)\big |=&\big |(k-d)f(u)+d[M_f+m_f]+(k-d)f(v)+\\
&d[M_f+m_f]-|k-d|\cdot f(uv)\big |\\
=&(k-d)[f(u)+f(v)+f(uv)]+2d[M_f+m_f]\\
=&2d[M_f+m_f]-(d-k)c_5
\end{split}}
$$ for each edge $uv\in E(G)$ holds true.\qqed
\end{example}

\begin{defn}\label{defn:real-valued-vertex-coloring}
\cite{Yao-Wang-Ma-Su-Wang-Sun-2020ITNEC} A connected $(p,q)$-graph $G$ admits a real-valued vertex coloring $\alpha:V(G)\rightarrow R_S$, and $|\alpha(u)-\alpha(v)|\neq |\alpha(u)-\alpha(w)|$ for any pair of adjacent edges $uv,uw$ for each vertex $u\in V(G)$, then $\alpha$ induces a real-valued total coloring $\beta$: $\beta(x)=\alpha(x)$ for $x\in V(G)$, $\beta(xy)=|\alpha(x)-\alpha(y)|$ for $xy\in E(G)$. We call $\beta$ \emph{real-valued total coloring} of $G$. If each edge color $\beta(u_kv_k)=|\beta(u_k)-\beta(v_k)|$ for each edge $u_kv_k\in E(G)$ holds
$$
|\beta(u_kv_k)-k|<\varepsilon<1/2\quad (\textrm{resp.}~|\beta(u_kv_k)-(2k-1)|<\varepsilon<1/2),~k\in [1,q]
$$ we call $\beta$ \emph{graceful real-valued total coloring} (resp. \emph{odd-graceful real-valued total coloring}).\qqed
\end{defn}

An algebraic operation on Topcode-matrices is introduced in \cite{Yao-Wang-Ma-Su-Wang-Sun-2020ITNEC}. In the Topcode-matrix $I_{code}=(X, E, Y)^{T}$ with $x_i=1$, $e_i=1$ and $y_i=1$ for $i\in [1,q]$ is called the \emph{identity Topcode-matrix}. For two Topcode-matrices $T^j_{code}=(X^j, E^j, Y^j)^{T}$ with $j=1,2$, where
$$
X^j=(x^j_1, x^j_2, \dots , x^j_q),\quad E^j=(e^j_1, e^j_2, \dots , e^j_q),\quad Y^j=(y^j_1, y^j_2, \dots , y^j_q)
$$ the \emph{coefficient multiplication} of a function $f(x)$ and a Topcode-matrix $T^j_{code}$ is defined by
$${
\begin{split}
f(x)\cdot T^j_{code}=f(x)\cdot(X^j, E^j, Y^j)^{T}=(f(x)\cdot X^j, f(x)\cdot E^j, f(x)\cdot Y^j)^{T}
\end{split}}
$$
where

$f(x)\cdot X^j=(f(x)\cdot x^j_1, f(x)\cdot x^j_2, \dots , f(x)\cdot x^j_q)$,

$f(x)\cdot E^j=(f(x)\cdot e^j_1, f(x)\cdot e^j_2, \dots , f(x)\cdot e^j_q)$ and

$f(x)\cdot Y^j=(f(x)\cdot y^j_1, f(x)\cdot y^j_2, \dots , f(x)\cdot y^j_q)$.\\
And the \emph{addition} between two Topcode-matrices $T^1_{code}$ and $T^2_{code}$ is denoted as
\begin{equation}\label{eqa:real-valued-topcode-matrix00}
T^1_{code}+T^2_{code}=(X^1+X^2, E^1+E^2, Y^1+Y^2)^{T}
\end{equation}
where
\begin{equation}\label{eqa:real-valued-topcode-matrix11}
{
\begin{split}
&X^1+X^2=(x^1_1+x^2_1, x^1_2+x^2_2, \cdots , x^1_q+x^2_q),\quad E^1+E^2=(e^1_1+e^2_1, e^1_2+e^2_2, \cdots , e^1_q+e^2_q)\\
&Y^1+Y^2=(y^1_1+y^2_1, y^1_2+y^2_2, \cdots , y^1_q+y^2_q)
\end{split}}
\end{equation} and
We have a \emph{real-valued Topcode-matrix} $R_{code}$ defined as:
\begin{equation}\label{eqa:real-valued-topcode-matrix22}
R_{code}=\alpha(\varepsilon)T^1_{code}+\beta(\varepsilon)T^2_{code}
\end{equation}

\begin{defn} \label{defn:111111}
$^*$ A $(p,q)$-graph $G$ admits a real-valued coloring $f_\varepsilon$ defined by the following real-valued Topcode-matrix
\begin{equation}\label{eqa:real-valued-topcode-matrix}
R_{code}(G,f_\varepsilon)=\alpha(\varepsilon)I_{code}+\beta(\varepsilon)T_{code}(G,f)
\end{equation} where two functions $\alpha(\varepsilon)$ and $\beta(\varepsilon)$ are real-valued, $I_{code}$ is the identity Topcode-matrix, and $T_{code}(G,f)$ is a Topcode-matrix of the $(p,q)$-graph $G$ admitting a $W$-constraint coloring $f$.\qqed
\end{defn}

\begin{example}\label{exa:8888888888}
The number-based strings induced by the real-valued Topcode-matrix $R_{code}(G,f_\varepsilon)$ are complex than that induced by the Topcode-matrix $T_{code}(G,f)$. By a Topcode-matrix $T_{code}(G,f)$ shown in Eq.(\ref{eqa:integer-topcode-matrix}), we get the following number-based strings
\begin{equation}\label{eqa:text-based-1}
{
\begin{split}
&D_1=10700220131197531111057111313,~D_2=10111103705577029111311201313\\
&D_3=11110531075711700913211131320
\end{split}}
\end{equation}

The Topcode-matrix $T_{code}(G,f)$ shown in Eq.(\ref{eqa:integer-topcode-matrix}) can provide us with a total of $(21)!$ number-based strings like $D_1,D_2$ and $D_3$ shown in Eq.(\ref{eqa:text-based-1}).

\begin{equation}\label{eqa:integer-topcode-matrix}
\centering
{
\begin{split}
T_{code}(G,f)&= \left(
\begin{array}{ccccccc}
10 & 7 & 0 & 0&2&2&0\\
1 & 3 & 5 & 7&9&11&13\\
11 &10 & 5&7&11&13&13
\end{array}
\right)
\end{split}}
\end{equation}

For example, after taking $\alpha(\varepsilon)=0.32$ and $\beta(\varepsilon)=1-0.32=0.68$, we get a real-valued Topcode-matrix
$$
R_{code}(G,f_\varepsilon)=0.32\cdot I_{code}+0.68\cdot T_{code}(G,f)
$$ where $R_{code}(G,f_\varepsilon)$ is
\begin{equation}\label{eqa:real-valued-Topcode-matrix}
\centering
{
\begin{split}
R_{code}(G,f_\varepsilon)=\left(
\begin{array}{ccccccc}
7.12 & 5.08 & 0.32 & 0.32&1.68&1.68&0.32\\
1 & 2.36 & 3.72 & 5.08&6.44&7.8&9.16\\
7.8 &7.12 &3.72&5.08&7.8&9.16&9.16
\end{array}
\right)
\end{split}}
\end{equation}

We apply $R_{code}(G,f_\varepsilon)$ to produce the following number-based strings:
$${
\begin{split}
D^r_1=7125080320321681680329167864450837223617871237250878916916\\
D^r_2=7817123722367125083725087850803203264491691678168168916032
\end{split}}$$ with $58$ bytes.

Obviously, both number-based strings $D^r_1$ and $D^r_1$ are complex than the number-based strings $D_1,D_2$ and $D_3$ shown in Eq.(\ref{eqa:text-based-1}). We have the following relationships between a Topcode-matrix $T_{code}(G,f)$ and a real-valued Topcode-matrix $R_{code}(G,f_\varepsilon)$:

(1) The graceful constraint $e_i=|x_i-y_i|$ in a Topcode-matrix $T_{code}(G,f)$ of a $(p,q)$-graph $G$ corresponds to $\alpha(\varepsilon)+\beta(\varepsilon)|x_i-y_i|$ of the real-valued Topcode-matrix $R_{code}(G,f_\varepsilon)$;

(2) the edge-magic constraint $x_i+e_i+y_i=k$ in $T_{code}(G,f)$ corresponds to $3\alpha(\varepsilon)+\beta(\varepsilon)\cdot k$ of $R_{code}(G,f_\varepsilon)$;

(3) the edge-difference constraint $e_i+|x_i-y_i|=k$ in $T_{code}(G,f)$ corresponds to $\alpha(\varepsilon)+\beta(\varepsilon)\cdot k$ of $R_{code}(G,f_\varepsilon)$;

(4) the felicitous-difference constraint $|x_i+y_i-e_i|=k$ in $T_{code}(G,f)$ corresponds to $\alpha(\varepsilon)+\beta(\varepsilon)\cdot k$ of $R_{code}(G,f_\varepsilon)$ if $e_i-(x_i+y_i)\geq 0$, otherwise $|x_i+y_i-e_i|=k$ corresponds to $|\beta(\varepsilon)\cdot k-\alpha(\varepsilon)|$;

(5) the graceful-difference constraint $\big ||x_i-y_i|-e_i\big |=k$ in $T_{code}(G,f)$ corresponds to $\alpha(\varepsilon)+\beta(\varepsilon)\cdot k$ of $R_{code}(G,f_\varepsilon)$ if $|x_i+y_i|-e_i< 0$, otherwise $\big ||x_i-y_i|-e_i\big |=k$ corresponds to $|\beta(\varepsilon)\cdot k-\alpha(\varepsilon)|$.\qqed
\end{example}

\begin{defn} \label{defn:111111}
$^*$ \textbf{Real-valued parameterized total coloring.} A \emph{$W$-constraint real-valued $(k,d)$-total coloring} is obtained by taking real values of $k$ and $d$ in each $W$-constraint real-valued $(k,d)$-total coloring defined in Definition \ref{defn:kd-w-type-colorings}.\qqed
\end{defn}

\begin{defn} \label{defn:probabilistic-real-valued-total-coloring}
$^*$ Suppose that a $(p,q)$-graph $G$ admits two total colorings $f$ and $g$. For each element $w\in V(G)\cup E(G)$, we define a \emph{probabilistic total coloring} $F$ for $G$ as
\begin{equation}\label{eqa:probabilistic-total-coloring}
F(w)=p_{\varepsilon}f(w)+(1-p_{\varepsilon})g(w)
\end{equation} where $p_{\varepsilon}$ is a \emph{probabilistic function} with $\varepsilon\in (0,1)$, such that $\displaystyle \lim _{\varepsilon\rightarrow 0}p_{\varepsilon}=0$ and $\displaystyle \lim _{\varepsilon\rightarrow 1}p_{\varepsilon}=1$.\qqed
\end{defn}

\begin{thm}\label{thm:666666}
By Eq.(\ref{eqa:real-valued-topcode-matrix22}) and Eq.(\ref{eqa:probabilistic-total-coloring}), the probabilistic total coloring $F$ corresponds to the following probabilistic Topcode-matrix equation
\begin{equation}\label{eqa:probabilistic-topcode-matrix}
T_{code}(G,F)=p_{\varepsilon}T_{code}(G,f)+(1-p_{\varepsilon})T_{code}(G,g)
\end{equation}
Since a Topcode-matrix $T_{code}$ corresponds to two or more colored graphs having more edges, as two colored graph homomorphisms $H\rightarrow G$ (or $G\rightarrow H$) and $T\rightarrow G$ (or $G\rightarrow T$) hold true, we get a probabilistic Topcode-matrix equation
\begin{equation}\label{eqa:555555}
T_{code}(G,F)=p_{\varepsilon}T_{code}(H,f)+(1-p_{\varepsilon})T_{code}(T,g)
\end{equation} for three different graphs $G,H$ and $T$, according to Definition \ref{defn:probabilistic-real-valued-total-coloring}.
\end{thm}

\begin{rem}\label{rem:333333}
Since $\displaystyle \lim _{\varepsilon\rightarrow 0}F(w)=g(w)$, $\displaystyle \lim _{\varepsilon\rightarrow 1}F(w)=f(w)$ according to Eq.(\ref{eqa:probabilistic-total-coloring}), then Eq.(\ref{eqa:probabilistic-topcode-matrix}) enables us to get
$$\lim _{\varepsilon\rightarrow 0}T_{code}(G,F)=T_{code}(G,g),~\lim _{\varepsilon\rightarrow 1}T_{code}(G,F)=T_{code}(G,f)$$

By Definition \ref{defn:probabilistic-real-valued-total-coloring}, we have the edge-magic constraint

$${
\begin{split}
F(u)+F(uv)+F(v)=&[p_{\varepsilon}f(u)+(1-p_{\varepsilon})g(u)]+[p_{\varepsilon}f(uv)+(1-p_{\varepsilon})g(uv)]\\
&+[p_{\varepsilon}f(v)+(1-p_{\varepsilon})g(v)]\\
=&p_{\varepsilon}[f(u)+f(uv)+f(v)]+(1-p_{\varepsilon})[g(u)+g(uv)+g(v)]\\
=&p_{\varepsilon}c_1+(1-p_{\varepsilon})c_2
\end{split}}
$$ if two edge-magic constraints
$$
f(u)+f(uv)+f(v)=c_1,\quad g(u)+g(uv)+g(v)=c_2
$$ hold true for edges $uv\in E(G)$. Thereby, the total coloring $F$ obeys the edge-magic constraint and is called \emph{edge-magic probabilistic total coloring}, where $p_{\varepsilon}c_1+(1-p_{\varepsilon})c_2$ is a \emph{probabilistic constant}.

If $G$ is a bipartite graph with $V(G)=X\cup Y$ and $X\cap Y=\emptyset$, and $f$ and $g$ are \emph{set-ordered}, i.e.
$$\max f(X)<\min f(Y),\quad \max g(X)<\min g(Y)
$$ then, for each edge $uv\in E(G)$ with $v\in X$ and $u\in Y$, we have
$${
\begin{split}
F(uv)=&p_{\varepsilon}f(uv)+(1-p_{\varepsilon})g(uv)=p_{\varepsilon}[f(u)-f(v)]+(1-p_{\varepsilon})[g(u)-g(v)]\\
=&[p_{\varepsilon}f(u)+(1-p_{\varepsilon})g(u)]-[p_{\varepsilon}f(v)+(1-p_{\varepsilon})g(v)]\\
=&F(u)-F(v)>0
\end{split}}
$$ also, the total coloring $F$ is \emph{set-ordered} with $\max F(X)<\min F(Y)$.

If two felicitous-difference constraints
$$
\big |f(u)+f(v)-f(uv)\big |=c_3,\quad \big |g(u)+g(v)-g(uv)\big |=c_4
$$ hold true for edges $uv\in E(G)$, then we have the felicitous-difference constraint
$${
\begin{split}
\big |F(u)+F(v)-F(uv)|=&\big |[p_{\varepsilon}f(u)+(1-p_{\varepsilon})g(u)]+[p_{\varepsilon}f(v)+(1-p_{\varepsilon})g(v)]-\\
&-[p_{\varepsilon}f(uv)+(1-p_{\varepsilon})g(uv)]\big |\\
=&\big |p_{\varepsilon}[f(u)+f(v)-f(uv)]+(1-p_{\varepsilon})[g(u)+g(v)-g(uv)]\big |\\
\leq &p_{\varepsilon}c_3+(1-p_{\varepsilon})c_4
\end{split}}
$$ for each edge $uv\in E(G)$, also, the total coloring $F$ obeys the felicitous-difference constraint and is called \emph{felicitous-difference probabilistic total coloring}, where $p_{\varepsilon}c_3+(1-p_{\varepsilon})c_4$ is a \emph{probabilistic constant}.\paralled
\end{rem}

\subsubsection{Parameterized Topcode-matrices to non-isomorphic graphs}

There are three different disjoint colored graphs $G$, $H$ and $T$ of $q$ edges, such that $G$, $H$ and $T$ admit three total colorings $F$, $f$ and $g$, respectively. So, we have:

1. The edge color set $E(G)=\{e_i:~i\in [1,q]\}$ with $e_i=x_iy_i$, and $x_i,y_i\in V(G)$, as well as the Topcode-matrix $P_{ara}(G,F)=(X_G,E_G,Y_G)^T$ with
$$X_G=(F(x_1),F(x_2),\dots ,F(x_q)),~E_G=(F(e_1),F(e_2),\dots ,F(e_q)),~Y_G=(F(y_1),F(y_2),\dots ,F(y_q))$$

2. The edge color set $E(H)=\{a_i:~i\in [1,q]\}$ with $a_i=u_iv_i$, and $u_i, v_i\in V(H)$, as well as the Topcode-matrix $T_{code}(H,f)=(X_H,E_H,Y_H)^T$ with
$$X_H=(f(u_1),f(u_2),\dots ,f(u_q)),~E_H=(f(a_1),f(a_2),\dots ,f(a_q)),~Y_H=(f(v_1),f(v_2),\dots ,f(v_q))$$

3. The edge color set $E(T)=\{b_i:~i\in [1,q]\}$ with $b_i=s_it_i$, and $s_i,t_i\in V(T)$, as well as the Topcode-matrix $T_{code}(T,g)=(X_T,E_T,Y_T)^T$ with
$$X_T=(g(s_1),g(s_2),\dots ,g(s_q)),~E_T=(g(b_1),g(b_2),\dots ,g(b_q)),~Y_T=(g(t_1),g(t_2),\dots ,g(t_q))$$

Here, the total coloring $F$ of $G$ is defined as follows
\begin{equation}\label{eqa:componud-topcode-matrices-333}
{
\begin{split}
F(x_i)&= k\cdot f(u_i)+d\cdot g(s_i),~ x_i\in V(G), u_i\in V(H), ~ s_i\in V(T);\\
F(e_i)&= k\cdot f(a_i)+ d\cdot g(b_i),~ e_i\in E(G),~a_i\in E(H),~ b_i\in E(T);\\
F(y_i)&= k\cdot f(v_i)+d\cdot g(t_i),~ y_i\in V(G),~ v_i\in V(H), ~ t_i\in V(T).
\end{split}}
\end{equation} Thereby, we get
\begin{equation}\label{eqa:coordinates-topcode-matrices-4444}
{
\begin{split}
X_G=& k\cdot X_H+ d\cdot X_T=(k\cdot f(u_1)+d\cdot g(s_1),k\cdot f(u_2)+d\cdot g(s_2),\dots ,k\cdot f(u_q)+d\cdot g(s_q))\\
E_G=& k\cdot E_H+ d\cdot E_T=(k\cdot f(a_1)+d\cdot g(b_1),k\cdot f(a_2)+d\cdot g(b_2),\dots ,k\cdot f(a_q)+d\cdot g(b_q))\\
Y_G=& k\cdot Y_H+ d\cdot Y_T=(k\cdot f(t_1)+d\cdot g(b_1),k\cdot f(v_2)+d\cdot g(t_2),\dots ,k\cdot f(v_q)+d\cdot g(t_q))
\end{split}}
\end{equation}

According to the above algebraic preparation, we present the following definition

\begin{defn} \label{defn:more-parameterized-matrices}
$^*$ If there are two colored graphs $H$ and $T$ having $q$ edges and their Topcode-matrices to be $T_{code}(H,f)$ and $T_{code}(T,g)$ defined in Definition \ref{defn:colored-topcode-matrix}, then we have another colored $(p,q)$-graph $G$ holding $q=|E(G)|=|E(H)|=|E(T)|$ and admitting a total set-coloring $F$ defined in Eq.(\ref{eqa:componud-topcode-matrices-333}), such that $G$ corresponds to a parameterized Topcode-matrix
\begin{equation}\label{eqa:two-parameterized-matrices}
P_{ara}(G,F)=k\cdot T_{code}(H,f)+ d\cdot T_{code}(T,g)
\end{equation} defined in Eq.(\ref{eqa:coordinates-topcode-matrices-4444}), where integers $k\geq 0$ and $d\geq 1$, and the Topcode-matrices $T_{code}(H)$, $T_{code}(T)$ and $P_{ara}(G,F)$ have the same order. \qqed
\end{defn}

\begin{rem}\label{rem:333333}
In general, a colored graph $G$ admitting a total set-coloring has one its own $\{k_i\}^n_{i=1}$-parameterized Topcode-matrix defined as
\begin{equation}\label{eqa:more-parameterized-matrices}
{
\begin{split}
T_{code}(G,\{f_i\}^n_{i=1})&=k_1\cdot T_{code}(H_1,f_1)+ k_2\cdot T_{code}(H_2,f_2)+ \cdots + k_n\cdot T_{code}(H_n,f_n)\\
&=\sum ^n_{i=1} k_i\cdot T_{code}(H_i,f_i)
\end{split}}
\end{equation} with integers $k_i\geq 0$ and $n\geq 2$, where $|E(G)|=|E(H_i)|$ with $i\in [1,n]$.

Here, for the needs of making \emph{public-keys} and \emph{private-keys} in practical application, we can ask the integer factorization $m=k_1k_2\cdots k_n~(n\geq 2)$ with each $k_i$ is a prime number, and the integer partition $m=k_1+k_2+\cdots +k_n~(n\geq 2)$ with each $k_i$ is a prime number.\paralled
\end{rem}

\begin{example}\label{exa:8888888888}
Suppose that a graph $G$ admits different total colorings $f_1,f_2,\dots, f_n$, then each total coloring $f_i$ induces a Topcode-matrix $T_{code}(G,f_i)$, so, by Definition \ref{defn:more-parameterized-matrices}, we have a $\{k_i\}^n_{i=1}$-parameterized Topcode-matrix
\begin{equation}\label{eqa:more-parameterized-matrices-example}
T_{code}(G,\{f_i\}^n_{i=1})=k_1\cdot T_{code}(G,f_1)+ k_2\cdot T_{code}(G,f_2)+ \cdots + k_n\cdot T_{code}(G,f_n)
\end{equation} with integers $k_i\geq 0$ and $n\geq 2$, which is just a \emph{linear-combinatoric total set-coloring} of $G$. Moreover, we get a \emph{graph-coloring Topcode-matrix lattice}
\begin{equation}\label{eqa:555555}
\textbf{\textrm{L}}(Z^0[+]\{f_i\}^n_{i=1})=\left \{T_{code}(G,\{f_i\}^n_{i=1})=\sum^n_{i=1} k_i\cdot T_{code}(G,f_i): k_i\in Z^0\right \}
\end{equation} with $\sum^n_{i=1} k_i\geq 1$, and $\{f_i\}^n_{i=1}$ is the \emph{graph-coloring base}.\qqed
\end{example}

\begin{example}\label{exa:8888888888}
It is easy to see that five graphs shown in Fig.\ref{fig:4-sum-topcode-matrices} are not isomorphic from each other. By the following Topcode-matrices
\begin{equation}\label{eqa:4-sum-Topcode-matrices-express11}
T_{code}(T_1,f_1)= \left(
\begin{array}{ccccccccccccccc}
2 & 1 & 0 & 0 & 0 & 0 & 0 \\
1 & 2 & 3 & 4 & 5 & 6 & 7 \\
3 & 3 & 3 & 4 & 5 & 6 & 7
\end{array}
\right),~T_{code}(T_2,f_2)= \left(
\begin{array}{ccccccccccccccc}
3 & 2 & 1 & 0 & 0 & 0 & 0 \\
1 & 2 & 3 & 4 & 5 & 6 & 7 \\
4 & 4 & 4 & 4 & 5 & 6 & 7
\end{array}
\right)
\end{equation}
and
\begin{equation}\label{eqa:4-sum-Topcode-matrices-express22}
T_{code}(T_3,f_3)= \left(
\begin{array}{ccccccccccccccc}
2 & 2 & 2 & 2 & 1 & 0 & 0 \\
1 & 2 & 3 & 4 & 5 & 6 & 7 \\
3 & 4 & 5 & 6 & 6 & 6 & 7
\end{array}
\right),~T_{code}(T_4,f_4)= \left(
\begin{array}{ccccccccccccccc}
3 & 2 & 2 & 2 & 1 & 0 & 0 \\
1 & 2 & 3 & 4 & 5 & 6 & 7 \\
4 & 4 & 5 & 6 & 6 & 6 & 7
\end{array}
\right)
\end{equation}we get the following parameterized Topcode-matrices

{\footnotesize
$$
{
\begin{split}
&\quad \sum^2_{i=1}k_iT_{code}(T_i,f_i)= \left(
\begin{array}{ccccccccccccccc}
2k_1+3k_2 & k_1+2k_2 & k_2 & 0 & 0 & 0 & 0 \\
k_1+k_2 & 2(k_1+k_2) & 3(k_1+k_2) & 4(k_1+k_2) & 5(k_1+k_2) & 6(k_1+k_2) & 7(k_1+k_2) \\
3k_1+4k_2 & 3k_1+4k_2 & 3k_1+4k_2 & 4(k_1+k_2) & 5(k_1+k_2) & 6(k_1+k_2) & 7(k_1+k_2)
\end{array}
\right)
\end{split}}
$$
}
and
{\footnotesize
$$
{
\begin{split}
&\quad \sum^4_{i=3}k_iT_{code}(T_i,f_i)= \left(
\begin{array}{ccccccccccccccc}
2k_3+3k_4 & 2(k_3+k_4) &2(k_3+k_4) & 2(k_3+k_4) & k_3+k_4 & 0 & 0 \\
k_3+k_4 & 2(k_3+k_4) & 3(k_3+k_4) & 4(k_3+k_4) & 5(k_3+k_4) & 6(k_3+k_4) & 7(k_3+k_4) \\
3k_3+4k_4 & 4(k_3+k_4) & 5(k_3+k_4) & 6(k_3+k_4) & 6(k_3+k_4) & 6(k_3+k_4) & 7(k_3+k_4)
\end{array}
\right)
\end{split}}
$$
}

\begin{figure}[h]
\centering
\includegraphics[width=16.4cm]{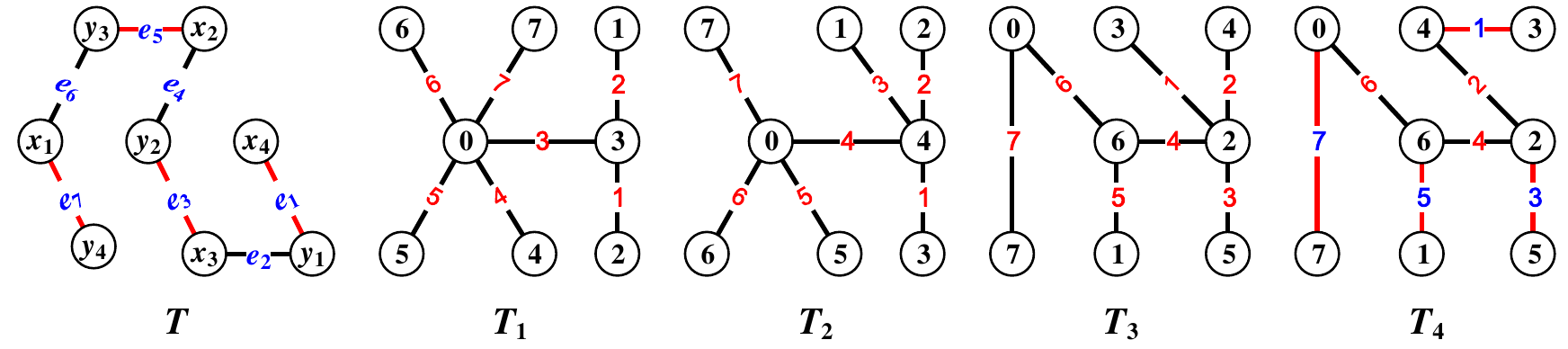}
\caption{\label{fig:4-sum-topcode-matrices}{\small Five graphs are not isomorphic from each other.}}
\end{figure}

Thereby, we get a $\{k_i\}^4_{i=1}$-parameterized Topcode-matrix of the graph $T$ as follows
\begin{equation}\label{eqa:44-Topcode-matrices-express11}
T_{code}(T,\{k_j\}^4_{j=1})=\sum^4_{j=1}k_jT_{code}(T_j,f_j)= \left(
\begin{array}{ccccccccccccccc}
x_4 & x_3 & x_3 & x_2 & x_2 & x_1 & x_1 \\
e_1 & e_2 & e_3 & e_4 & e_5 & e_6 & e_7 \\
y_1 & y_1 & y_2 & y_2 & y_3 & y_3 & y_4
\end{array}
\right)_{3\times 7}
\end{equation} then $T$ admits a linear-combinatoric total set-coloring defined by

$x_1=\{0\}$, $x_2=\{k_3+k_4,~2(k_3+k_4)\}$, $x_3=\{k_1+2k_2+2(k_3+k_4),~k_2+2(k_3+k_4)\}$, $x_4=\{2(k_1+k_3)+3(k_2+k_4)\}$;

$e_i=\{i(k_1+k_2+k_3+k_4)\}$ for $i\in [1,7]$;

$y_1=\{3(k_1+k_3)+4(k_2+k_4),~ 3k_1+4(k_2+k_3+k_4)\}$,

$y_2=\{3k_1+4k_2+5(k_3+k_4),~ 4(k_1+k_2)+6(k_3+k_4)\}$,

$y_3=\{5(k_1+k_2)+6(k_3+k_4),~ 6(k_1+k_2+k_3+k_4)\}$ and $y_4=\{7(k_1+k_2+k_3+k_4)\}$. \qqed
\end{example}

\begin{rem}\label{rem:333333}
Suppose that there are two colored graphs $H$ and $T$ having $q$ edges and their Topcode-matrices to be $T_{code}(H,f)$ and $T_{code}(T,g)$ defined in Definition \ref{defn:colored-topcode-matrix}. Let $G_j$ be the $j$th copy of a $(p,q)$-graph $G$ with the edge color set $E(G_j)=\{e_{j,1},e_{j,2}, \dots ,e_{j,q}\}$ with $j\in [1,q!]$, where $e_{j,1},e_{j,2}, \dots ,e_{j,q}$ is the $j$th permutation of $q$ edges of $G$. We define a limit-form total coloring $f^j_{\varepsilon}$ for $G_j$, called \emph{limit-form real-valued total set-coloring}, by
\begin{equation}\label{eqa:limited-adjustment-parameterized-matrices}
P_{ara}(G_j,f^j_{\varepsilon})=p_{\varepsilon}\cdot T_{code}(H,f)+ (1-p_{\varepsilon})\cdot T_{code}(T,g)
\end{equation} with a \emph{probabilistic function} $p_{\varepsilon}$ for $\varepsilon\in (0,1)$ and $0< p_{\varepsilon} <1$, such that
$$P_{ara}(G_j,f^j_{0})=\lim_{\varepsilon\rightarrow 0}P_{ara}(G_j,f^j_{\varepsilon})=T_{code}(T,g)
$$ as $\displaystyle \lim_{\varepsilon\rightarrow 0}p_{\varepsilon}=0$ and moreover
$$P_{ara}(G_j,f^j_{1})= \lim_{\varepsilon\rightarrow 1}P_{ara}(G_j,f^j_{\varepsilon})=T_{code}(H,f)
$$ as $\displaystyle \lim_{\varepsilon\rightarrow 1}p_{\varepsilon}=1$.

Notice that there are $q!$ limit-form real-valued total set-colorings of the $(p,q)$-graph $G$ produced by $P_{ara}(G_j,f^j_{\varepsilon})$ defined in Eq.(\ref{eqa:limited-adjustment-parameterized-matrices}). \paralled
\end{rem}

\begin{defn} \label{defn:4magic-parameterized-matrices}
$^*$ Suppose that three Topcode-matrices $P_{ara}(G,F)$, $T_{code}(H,f)$ and $T_{code}(T,g)$ hold
\begin{equation}\label{eqa:two-parameterized-matrices}
P_{ara}(G,F)=k\cdot T_{code}(H,f)+ d\cdot T_{code}(T,h)
\end{equation} according to Definition \ref{defn:more-parameterized-matrices}, where
\begin{equation}\label{eqa:555555}
T_{code}(H,f)=(X_H,~E_H,~Y_H)^T,~T_{code}(T,g)=(X_T,~E_T,~Y_T)^T
\end{equation} with $X_H=(f(u_1),f(u_2),\dots ,f(u_q))$, $E_H=(f(u_1v_1), f(u_2v_2), \dots ,f(u_qv_q))$, $Y_H=(f(v_1),f(v_2)$, $\dots $, $f(v_q))$, $X_T=(g(s_1),g(s_2),\dots ,g(s_q))$, $E_T=(g(s_1t_1), g(s_2t_2), \dots ,g(s_qt_q))$ and $Y_T=(g(t_1),g(t_2)$, $\dots $, $g(t_q))$. For constants $c_{\varepsilon}$ and $c\,'_{\varepsilon}$ with $\varepsilon\in \{ed, em$, $ gd$, $fd\}$, we have:
\begin{asparaenum}[\textbf{\textrm{Pama}}-1.]
\item If there are the edge-difference constraints
$$f(u_iv_i)+|f(v_i)-f(u_i)|=c_{ed},~g(s_it_i)+|g(t_i)-g(s_i)|=c\,'_{ed},~i\in [1,q]$$ we write these two facts by $E_H+|Y_H-X_H|:=c_{ed}$ and $E_T+|Y_T-X_T|:=c\,'_{ed}$ respectively, and say $P_{ara}(G,F)$ to be \emph{$(k,d)$-edge-difference Topcode-matrix}; and moreover $P_{ara}(G,F)$ is called \emph{uniformly $(k,d)$-edge-difference Topcode-matrix} if $c_{ed}=c\,'_{ed}$.
\item If there are the edge-magic constraints
$$f(u_i)+f(u_iv_i)+f(v_i)=c_{em},~g(s_i)+g(s_it_i)+g(t_i)=c\,'_{em},~i\in [1,q]$$ we write these two facts by $X_H+E_H+Y_H:=c_{ed}$ and $X_T+E_T+Y_T:=c\,'_{ed}$ respectively, and say $P_{ara}(G,F)$ to be \emph{$(k,d)$-edge-magic Topcode-matrix}; and moreover $P_{ara}(G,F)$ is called \emph{uniformly $(k,d)$-edge-magic Topcode-matrix} if $c_{em}=c\,'_{em}$.
\item If there are the graceful-difference constraints $$\big ||f(v_i)-f(u_i)|-f(u_iv_i)\big |=c_{gd},~\big ||g(t_i)-g(s_i)|-g(s_it_i)\big |=c\,'_{gd},~i\in [1,q]$$ we write these two facts by $\big ||Y_H-X_H|-E_H\big |:=c_{gd}$ and $\big ||Y_T-X_T|-E_T\big |:=c\,'_{gd}$ respectively, and say $P_{ara}(G,F)$ to be \emph{$(k,d)$-graceful-difference Topcode-matrix}; and moreover $P_{ara}(G,F)$ is called \emph{uniformly $(k,d)$-graceful-difference Topcode-matrix} if $c_{gd}=c\,'_{gd}$.
\item If there are the felicitous-difference constraints $$\big |f(v_i)+f(u_i)-f(u_iv_i)\big |=c_{fd},~\big |g(t_i)+g(s_i)-g(s_it_i)\big |=c\,'_{fd},~i\in [1,q]$$ we write these two facts by $\big |Y_H+X_H-E_H\big |:=c_{gd}$ and $\big |Y_T+X_T-E_T\big |:=c\,'_{gd}$ respectively, and say $P_{ara}(G,F)$ to be \emph{$(k,d)$-felicitous-difference Topcode-matrix}; and moreover $P_{ara}(G,F)$ is called \emph{uniformly $(k,d)$-felicitous-difference Topcode-matrix} if $c_{fd}=c\,'_{fd}$.\qqed
\end{asparaenum}
\end{defn}

\subsubsection{Colorings/labelings with fractional numbers}

\begin{example}\label{exa:8888888888}
Recall, a general graph-book $B\langle S_X,W,\{H_j\}^M_{j=1}\rangle$ defined in Definition \ref{defn:general-graph-book-defn} is a \emph{$W$-graph-book} if $\{q_j\}^M_{j=1}$ is a \emph{$W$-sequence}. If $\displaystyle \lim_{j\rightarrow \infty} q_j=q$, then the general graph-book $B\langle S_X,W,\{H_j\}^M_{j=1}\rangle$ converges to a regular graph-book $R_{book}=[\odot_X]^M_{j=1}G_j$ defined in Example \ref{exa:regular-graph-book-expression}.\qqed
\end{example}

\begin{defn} \label{defn:limit-k-d-colorings}
$^*$ For a $W$-constraint $(k,d)$-total coloring (or $W$-constraint $(k,d)$-total set-coloring) $f$ of a bipartite $(p,q)$-graph $G$, if two sequences $\{k_n\}^n_1$ and $\{d_n\}^n_1$ converge to two integers $k_0$ and $d_0$ respectively, that is
$$
\lim _{n\rightarrow \infty}k_n=k_0,~\lim _{n\rightarrow \infty}d_n=d_0
$$
then a $W$-constraint $(k_0,d_0)$-total coloring (or $W$-constraint $(k_0,d_0)$-total set-coloring) forms a \emph{$\{k_n\}^n_1$-$\{d_n\}^n_1$-sequence coloring} $F\,^n$ as $|k_n-k_0|<\varepsilon$ and $|d_n-d_0|<\varepsilon$, such that this graph $G$ has its own $(k_n,d_n)$-type Topcode-matrix
\begin{equation}\label{eqa:555555}
P^n_{ara}(G,F\,^n)=k_n\cdot I\,^0+d_n\cdot T_{code}(G,f)
\end{equation} where three Topcode-matrices $I\,^0$, $T_{code}(G,f)$ and $P^n_{ara}(G,F\,^n)$ have the same order $3\times q$. Moreover, we have a \emph{limit-form Topcode-matrix}
\begin{equation}\label{eqa:555555}
{
\begin{split}
P_{ara}(G,F_0)&=\lim _{n\rightarrow \infty}P^n_{ara}(G,F\,^n)=\lim _{n\rightarrow \infty}[k_n I\,^0+d_n T_{code}(G,f)]\\
&=\lim _{n\rightarrow \infty}k_n I\,^0+\lim _{n\rightarrow \infty}d_n T_{code}(G,f)\\
&=k_0I\,^0+d_0T_{code}(G,f)
\end{split}}
\end{equation} where the \emph{limit-form coloring} $\displaystyle F_0=\lim _{n\rightarrow \infty}F\,^n$.\qqed
\end{defn}

\begin{rem}\label{rem:333333}
By Definition \ref{defn:limit-k-d-colorings}, we have a \emph{limit-form coloring} $\beta$ of a graph $G$ admitting a $W$-constraint total coloring $f$ defined by the following parameterized Topcode-matrix
\begin{equation}\label{eqa:555555}
P(G,\beta)=\frac{s}{r}\cdot I\,^0+\frac{a}{b}\cdot T_{code}(G,f)
\end{equation}
or
\begin{equation}\label{eqa:fractional-topcode-matrices}
(rb)\cdot P(G,\beta)=(sb)\cdot I\,^0+(ar)\cdot T_{code}(G,f)
\end{equation} where $\frac{s}{r}$ and $\frac{a}{b}$ are fractional numbers. Notice that $rb$, $sb$ and $ar$ are integers in Eq.(\ref{eqa:fractional-topcode-matrices}).\paralled
\end{rem}

\section{Conclusion and further researching topics}

For designing more powerful and effective techniques of topological coding in information security and resisting the intelligent attack equipped with quantum computers, we have defined new parameterized colorings and set-colorings, and shown:

\begin{asparaenum}[(i) ]
\item Any tree $T$ with diameter $D(T)\geq 3$ admits at least $2^m$ different $W$-constraint $(k,d)$-total colorings for $m+1=\left \lceil \frac{D(T)}{2}\right \rceil $, where $W$-constraint is one of graceful, harmonious, (odd-edge) edge-magic, (odd-edge) graceful-difference, (odd-edge) edge-difference, (odd-edge) felicitous-difference and edge-antimagic in the previous sections.

\item By the technique of adding leaves to graphs, we have shown the following polynomial algorithms: RLA-algorithm-A for the graceful-difference $(k,d)$-total coloring, RLA-algorithm-B for the edge-difference $(k,d)$-total coloring, RLA-algorithm-C for the felicitous-difference $(k,d)$-total coloring and RLA-algorithm-D for the edge-magic $(k,d)$-total coloring.

\item If a connected graph $G$ admits a set-ordered graceful labeling, then $G$ admits a graceful-difference $(k,d)$-total coloring and an odd-edge graceful-difference $(k,d)$-total coloring.

\item A connected graph $G$ admits a $W$-constraint $(k,d)$-total coloring if and only if there exists a tree $T$ admitting a $W$-constraint $(k,d)$-total coloring such that the result of vertex-coinciding each group of vertices colored the same colors into one vertex is just $G$, where $W$-constraint is one of graceful, harmonious, (odd-edge) edge-magic, (odd-edge) graceful-difference, (odd-edge) edge-difference and (odd-edge) felicitous-difference.
\item In the Problem of Integer Partition And Trees proposed in Section 1, we were asked for ``How many trees with this degree sequence $m_1,m_2,\dots ,m_n$ are there?'' Our answer is: There are infinite tress with this degree sequence $m_1,m_2,\dots ,m_n$, since
\begin{equation}\label{eqa:555555}
n_1=2+\sum_{3\leq d\leq \Delta} (d-2)n_d=2+\sum_{2\leq d\leq \Delta} (d-2)n_d
\end{equation}
with arbitrary positive integer $n_2$, where $n_d$ is the number of vertices of degree $d$ in a tree (Ref. \cite{Yao-Zhang-Yao-2007} and \cite{Yao-Zhang-Wang-Sinica-2010}).
\item For the purpose of application, we have discussed number-based strings containing parameters, graphic lattices related with parameterized total colorings, $(k,d)$-colored tree-graphic lattices, tree-graphic lattice homomorphisms, and graphic groups based on parameterized total colorings.
\item We have shown more complex homomorphism phenomenon, such as, from a \emph{colored graph homomorphism} $H\rightarrow G$ to a \emph{Topcode-matrix homomorphism} $T_{code}(H,g^*)\rightarrow T_{code}(G,g)$ in Eq.(\ref{eqa:topcode-matrix-homomorphism}), and then to a \emph{graph-set homomorphism} $S_{graph}[T_{code}(H,g^*)]\rightarrow S_{graph}[T_{code}(G,g)]$ in Eq.(\ref{eqa:graph-set-homomorphism}).
\item As the generalization of parameterized colorings, we have recalled total colorings subject to parameterized magic-functions, real-valued parameterized total colorings, and introduced fractional colorings/labelings by limitation method.
\end{asparaenum}

\vskip 0.2cm

For the in-depth study and topological authentication application of parameterized colorings and parameterized labelings introduced here, we propose the following researching topics:

\begin{asparaenum}[\textbf{\textrm{Topic}}-1. ]
\item As a special branch of topological coding, Topcode-matrix algebra was first proposed and investigated in \cite{Bing-Yao-2020arXiv} and \cite{Yao-Zhao-Zhang-Mu-Sun-Zhang-Yang-Ma-Su-Wang-Wang-Sun-arXiv2019}. We have used Topcode-matrices to define various parameterized colorings, since they have two basic advantages: topological structures and mathematical constraints. However, we point that Topcode-matrices mentioned here differ from adjacent matrices of graphs of graph theory. An adjacent matrix corresponds to a unique graph, however a Topcode-matrix may correspond to two or more graphs (colored graphs, Topsnut-gpws) with the same number of edges. Despite the results of studying Topcode-matrices, but we still need to discover more algebraic theories and topological properties of Topcode-matrices.

\item Cayley's formula $\tau(K_n)=n^{n-2}$ (Ref. \cite{Bondy-2008}) tells us that there are $n^{n-2}$ different colored spanning trees in a complete graph $K_n$ admitting a vertex labeling $f:V(K_n)\rightarrow [1,n]$ such that the vertex color set $f(V(K_n))=[1,n]$. Each spanning tree $T$ of these $n^{n-2}$ spanning trees admits each one of the following $W$-constraint $(k,d)$-total set-colorings for $W$-constraint $\in \{$graceful, harmonious, edge-difference, edge-magic, felicitous-difference, graceful-difference$\}$ according to Theorem \ref{thm:graph-admits-6-set-colorings}, and each spanning tree $T$ admits a set-coloring defined on a hyperedge set $\mathcal{E}$ such that $T$ is a subgraph of the intersected-graph of a hypergraph $\mathcal{H}_{yper}=([1,n],\mathcal{E})$ by Theorem \ref{thm:build-hyperedge-set}.

\qquad Clearly, finding a colored spanning tree $T$ from these $n^{n-2}$ spanning trees is a terrible job for supercomputers and quantum computers, since it will meet the Subgraph Isomorphic Problem.
\item Since any graph $G$ can be vertex-split into pairwise disjoint graphs $G_1,G_2$, $\dots $, $G_m$, which is just a permutation of graphs $a_1H_1$, $a_2H_2$, $\dots $, $a_mH_m$, then each graph is in some vertex-coinciding graphic lattice $\textbf{\textrm{L}}([\odot]Z^0\textbf{\textrm{H}})$ defined in Eq.(\ref{eqa:vertex-odot-graphic-lattice}) and Theorem \ref{thm:vertex-split-disjoint-graphs}. If this vertex-coinciding graphic lattice is closed to a $W$-constraint (set-)coloring, then $G$ admits this $W$-constraint (set-)coloring.
\item Although there are no reports on the study of parameterized hypergraphs, we have made a preliminary attempt to study parameterized hypergraphs. The authors in \cite{Yao-Ma-arXiv-2201-13354v1} have applied set-colorings and intersected-graphs to observe indirectly and characterize properties of hypergraphs. They have pointed that hypergraph theory will be an important application in the future resisting AI attack equipped quantum computer, although the research achievements and application reports of hypergraphs are far less than that of popular graphs. Studying what properties parameterized hypergraphs have may refer to \cite{Yao-Ma-arXiv-2201-13354v1}.
\item Since Hanzi-graphs admit many colorings/labelings defined in this article, applying them to the field of information using Chinese characters will greatly help people using Chinese characters to be unimpeded in digital finance and personal privacy protection.
\item We are not only to design various practical technologies for asymmetric cryptography, but also we define some new coloring objects with new problems for graph theory, a branch of mathematics. For example, Topcode-matrix algebra, real-valued parameterized (set-)colorings, colored graph homomorphism, Topcode-matrix homomorphism, graph-set homomorphism, graph-coloring Topcode-matrix lattice, and so on.

\qquad As known, each connected $(p,q)$-graph can be vertex-split into a tree of $q$ edges. We have proven: \emph{Each connected graph $G$ admits each one of $W$-constraint $(k,d)$-total \textbf{set-colorings}}, since every tree admits each one of $W$-constraint $(k,d)$-total \textbf{colorings} in Section 3 and Section 4, where $W$-constraint $\in \{$graceful, harmonious, edge-difference, edge-magic, felicitous-difference, graceful-difference$\}$ according to Theorem \ref{thm:graph-admits-6-set-colorings}.

\qquad Let $F$ be a $W$-constraint $(k,d)$-total set-coloring of the connected graph $G$. If each set cardinality $|F(w)|\leq r$ for $w\in V(G)\cup E(G)$, we say that this set-coloring $F$ is $r$-bounded. We propose: \textbf{Find} a $W$-constraint $(k,d)$-total set-coloring $F$ with $r$-boundedness, such that $r$ is the minimal number over all $W$-constraint $(k,d)$-total set-colorings of $G$.

\qquad It may be interesting and meaningful to build up parameterized colorings/labelings for part of colorings/labelings defined in \cite{Gallian2021}.
\item We know very little about various graphic lattices, Topcode-matrix lattices, although they have many of the advantages of lattice ciphers.
\end{asparaenum}

\vskip 0.4cm

Part of anti-quantum computing foundations of our techniques mentioned here are based on the following works:

1. \textbf{Recall} the PRONBS-problem stated in Problem \ref{qeu:PRONBS-problems}. For a given number-based string $s=c_1c_2\cdots c_n$ with $c_i\in [0,9]$, \textbf{do} the following works:
\begin{asparaenum}[\textrm{\textbf{Pronbs}}-1. ]
\item \textbf{Rewrite} the string $s$ into $3q$ segments $a_1,a_2,\dots ,a_{3q}$ with $a_j=b_{j,1}b_{j,2}\cdots b_{j,m_j}$ for $j\in [1,3q]$, such that each number $c_i$ of the number-based string $s$ is in a segment $a_j$ but $c_i$ is not in other segment $a_s$ if $s\neq j$. In this work, we will meet the \textbf{String Partition Problem}, since there are different colored graphs, which induce the same number-based strings by Theorem \ref{thm:different-graphs-same-number-based-strings}.

\item \textbf{Find} two integers $k_0,d_0\geq 0$, each segment $a_j$ with $j\in [1,3q]$ can be expressed as $a_j=\beta_jk_0+\gamma_jd_0$ for integers $\beta_j,\gamma_j\geq 0$ and $j\in [1,3q]$. In this work, we will meet the \textbf{Indefinite Equation Problem}.

\item \textbf{Find} a colored $(p,q)$-graph $G$ admitting a $W$-constraint (set-)coloring $f$, such that $G$ admits a parameterized coloring $F$ defined by
\begin{equation}\label{eqa:555555}
P_{ara}(G,F)_{3\times q}=k\cdot I\,^0_{3\times q}+d\cdot T_{code}(G,f)_{3\times q}
\end{equation} which is just a parameterized Topcode-matrix. The work of \textbf{Determining} $G$ from a huge amount of graphs will meet the \textbf{Subgraph Isomorphic Problem} (refer to two numbers $G_{23}$ and $G_{24}$ shown in Eq.(\ref{eqa:number-graphs-23-24-vertices})), and the desired $W$-constraint (set-)coloring $f$ is related with the \textbf{Indefinite Equation Problem}, more important thing is that there are different colored graphs, which induce the same number-based strings by Theorem \ref{thm:different-graphs-same-number-based-strings}.

\item \textbf{Use} this Topcode-matrix $P_{ara}(G,F)$ to produce just the desired number-based string $a_1a_2\dots a_{3q}=c\,'_1c\,'_2\cdots c\,'_n$ when $(k,d)=(k_0,d_0)$ determined by the solution $k_0,d_0$ of the above indefinite equations $a_j=\beta_jk_0+\gamma_jd_0$ for $j\in [1,3q]$.
\end{asparaenum}

\vskip 0.4cm

2. \textbf{Constructing} topological structures related with integer factorization and integer partition. For a positive integer $m$, we have:

(i) \textbf{Integer factorization}:
\begin{equation}\label{eqa:problem-integer-factorization}
m=p_1p_2\cdots p_n~(n\geq 2)
\end{equation} with each $p_i$ is a prime number.

(ii) \textbf{Integer partition}:
\begin{equation}\label{eqa:problem-integer-partition}
m=m_1+m_2+\cdots +m_n~(n\geq 2)
\end{equation} with each $m_i$ is a prime number.

\begin{asparaenum}[\textrm{\textbf{Reotost}}-1. ]
\item \textbf{Find} a graph $G$ having degrees $\textrm{deg}_G(x_i)=p_i$ for $i\in [1,n]$, where these numbers $p_1,p_2,\dots ,p_n$ holds Eq.(\ref{eqa:problem-integer-factorization}). Moreover, we can add more constraints to $G$, for example,

\qquad (i) this graph $G$ holds some particular topological structures, such as Hamilton cycle, Euler's graph, bipartite graph, \emph{etc};

\qquad (ii) this graph $G$ admits some $W$-constraint colorings/labelings.
\item \textbf{Find} a graph $H$ having degrees $\textrm{deg}_H(u_i)=m_i$ for $i\in [1,n]$, where these numbers $m_1,m_2,\dots ,m_n$ holds Eq.(\ref{eqa:problem-integer-partition}). Moreover, this graph $H$ (i) \textbf{has} some particular topological structures, such as Hamilton cycle, Euler's graph, bipartite graph, \emph{etc}; and (ii) \textbf{admits} some $W$-constraint colorings/labelings.
\item \textbf{Adding} $m$ leaves to a graph $G$ for realizing some graph properties, or some $W$-constraint colorings/labelings, such that $m$ holds:

\qquad (i) one of Eq.(\ref{eqa:problem-integer-factorization}) and Eq.(\ref{eqa:problem-integer-partition}); or

\qquad (ii) both Eq.(\ref{eqa:problem-integer-factorization}) and Eq.(\ref{eqa:problem-integer-partition}), simultaneously.

\item \textbf{Find} a maximal planar graph $G$, such that $G=[\ominus]^4_{k=1}a_kT_{S,k}$ with $m=a_1+a_2+a_3+a_4$, see Eq.(\ref{eqa:4-color-triangle-base-mpgs}) and Eq.(\ref{eqa:4-color-planar-lattice}), where $m$ holds one of Eq.(\ref{eqa:problem-integer-factorization}) and Eq.(\ref{eqa:problem-integer-partition}), or $m$ holds two Eq.(\ref{eqa:problem-integer-factorization}) and Eq.(\ref{eqa:problem-integer-partition}) simultaneously.
\item \textbf{Find} a connected graph $G$, such that $G$ can be vertex-split into edge-disjoint subgraphs $G_1,G_2,\dots ,G_n$, such that

\qquad (i) Each vertex number $|V(G_i)|=p_i$ for $i\in [1,n]$, where each prime number $p_i$ is in Eq.(\ref{eqa:problem-integer-factorization}). Or each edge number $|E(G_i)|=m_i$ for $i\in [1,n]$, where each prime number $m_i$ is in Eq.(\ref{eqa:problem-integer-partition}).

\qquad (ii) Each edge number $|E(G_i)|=p_i$ for $i\in [1,n]$, where each prime number $p_i$ is in Eq.(\ref{eqa:problem-integer-factorization}). Or each vertex number $|V(G_i)|=m_i$ for $i\in [1,n]$, where each prime number $m_i$ is in Eq.(\ref{eqa:problem-integer-partition}).

\qquad (iii) All of vertex numbers $|V(G_i)|$ and all of edge numbers $|E(G_i)|$ for $i\in [1,n]$ hold two Eq.(\ref{eqa:problem-integer-factorization}) and Eq.(\ref{eqa:problem-integer-partition}), simultaneously.

\item Let $n_d(T)$ be the number of vertices of degree $d$ of a tree $T$. \textbf{Find} a tree $T$ holds $m=\sum _{d\geq 3}n_d(T)$, where $m$ holds:

\qquad (i) one of Eq.(\ref{eqa:problem-integer-factorization}) and Eq.(\ref{eqa:problem-integer-partition}); or

\qquad (ii) both Eq.(\ref{eqa:problem-integer-factorization}) and Eq.(\ref{eqa:problem-integer-partition}), simultaneously.

Moreover, \textbf{find} some $W$-constraint colorings/labelings admitted by this tree $T$.
\end{asparaenum}

\section{Acknowledgment}

The author, \emph{Bing Yao}, was supported by the National Natural Science Foundation of China under grants No. 61163054, No. 61363060 and No. 61662066. The author, \emph{Hongyu Wang}, thanks gratefully the National Natural Science Foundation of China under grants No. 61902005, and China Postdoctoral Science Foundation Grants No. 2019T120020 and No. 2018M641087.

\end{document}